\documentclass[11pt]{article} 
 
\usepackage[utf8]{inputenc} 

\usepackage{xr} 

\usepackage{comment}
\usepackage{lipsum}
\usepackage[dvipsnames]{xcolor}
\usepackage{bm} 
\usepackage{eufrak,booktabs} 
\usepackage{algorithm}
\usepackage[noend]{algpseudocode}

\algnewcommand\algorithmicoutput{\textbf{Output:}}
\algnewcommand\Output{\item[\algorithmicoutput]}
\usepackage{upgreek,mathrsfs}
\usepackage{enumitem}
\newcommand{\intmeasure}{em}
\newcommand{\interval}[2][\relax]{%
  \underbrace{\ifx#1\relax
      \rule{2\intmeasure}{0pt}
    \else
      \rule{#1\intmeasure}{0pt}
    \fi}_{\text{#2}}
}
\usepackage[letterpaper,
			left = 0.75truein,  
			right = 0.75truein, 
			top = 1.1truein, 
			bottom = 1.1truein]{geometry} 
\usepackage{amsthm, amsfonts, amssymb, amsmath,enumitem, bbm, mathabx}
\usepackage{mathtools}
\mathtoolsset{showonlyrefs,showmanualtags}
\numberwithin{equation}{section}
\allowdisplaybreaks
\usepackage[
            CJKbookmarks=true,
            bookmarksnumbered=true,
            bookmarksopen=true, 
            colorlinks=true,
            citecolor=red,
            linkcolor=blue,
            anchorcolor=red,
            urlcolor=blue
            ]{hyperref}
\usepackage[numbers,sort]{natbib} 
\usepackage{empheq}

\newcommand{\wt}{\widetilde}
\newcommand{\wh}{\widehat}
\newcommand{\wb}{\widebar}

\newcommand{\hdm}{\bm X}
\newcommand{\ldm}{\wt{\bm X}}
\newcommand{\rshdm}{\wb{\bm X}}

\newcommand{\hatl}{\wh{\bm L}}

\newcommand{\hd}{\wh{\bm D}}
\newcommand{\hslm}{\wh{\bm S}^L}
\newcommand{\hsrm}{\wh{\bm S}^R}
\newcommand{\hslv}{\wh{\bm \Sigma}^L}
\newcommand{\hsrv}{\wh{\bm \Sigma}^R}
\newcommand{\bu}{\wb{\bm U}}
\newcommand{\bv}{\wb{\bm V}}
\newcommand{\hslmp}{\wh{\bm S}^{L,\mathrm{pc}}}
\newcommand{\hsrmp}{\wh{\bm S}^{R,\mathrm{pc}}}
\newcommand{\hslvp}{\wh{\bm \Sigma}^{L,\mathrm{pc}}}
\newcommand{\hsrvp}{\wh{\bm \Sigma}^{R,\mathrm{pc}}}

\newcommand{\bq}{Q}
\newcommand{\tq}{\wt{Q}}
\newcommand{\calf}{\mathcal{F}}
\newcommand{\caltf}{\wt{\mathcal{F}}}
\newcommand{\rcheck}{\widehat{\bm R}}
\newcommand{\hatuq}{\wh{U}^{\mathrm Q, 0}}

\newcommand{\calc}{\mathcal{C}}

\newcommand{\usig}{\wh{U}^{\mathrm{Q},\mathsf{signal}}}
\newcommand{\unoise}{\wh{U}^{\mathrm{Q},\mathsf{noise}}}
\newcommand{\varnoise}{\bm \Sigma^{\mathrm Q,H}}
\newcommand{\rmq}{\mathrm Q}

\newcommand{\radobs}{\wh b}

\newcommand{\mobs}{\wh U}
\newcommand{\hatum}{\wh U^H}
\newcommand{\hatul}{\wh U^L}

\newcommand{\hslmor}{\bm S^{L,\orc}}
\newcommand{\hsrmor}{\bm S^{R,\orc}}
\newcommand{\hslvor}{\bm \Sigma^{L,\orc}}
\newcommand{\hsrvor}{\bm \Sigma^{R,\orc}}
\newcommand{\hslmorp}{\bm S^{L,\mathrm{pc}}}
\newcommand{\hsrmorp}{\bm S^{R,\mathrm{pc}}}
\newcommand{\hslvorp}{\bm \Sigma^{L,\mathrm{pc}}}
\newcommand{\hsrvorp}{\bm \Sigma^{R,\mathrm{pc}}}
\newcommand{\orc}{\circ}
\newcommand{\uor}{{\bm U}^{\orc}}
\newcommand{\vor}{{\bm V}^{\orc}}
\newcommand{\upca}{{\bm U}^{\mathrm{pc}}}
\newcommand{\vpca}{{\bm V}^{\mathrm{pc}}}
\newcommand{\tupca}{\wt{\bm U}^{\mathrm{pc}}}
\newcommand{\buor}{\wb{\bm U}^{\orc}}
\newcommand{\bvor}{\wb{\bm V}^{\orc}}

\newcommand{\rmxl}{\mathrm{X}^{L}}
\newcommand{\rmxr}{\mathrm{X}^{R}}
\newcommand{\rmtxl}{\wt{\mathrm{X}}^{L}}

\newcommand{\rmufunc}{\mathsf{u}}
\newcommand{\rmvfunc}{\mathsf{v}}
\newcommand{\gammab}{\bm \Gamma^L}
\newcommand{\gamman}{\bm \Gamma^R}

\newcommand{\gsu}{\bm U^{\mathsf{gs}}}
\newcommand{\projpar}{\bm P^{\parallel}}
\newcommand{\projperp}{\bm P^{\perp}}
\newcommand{\gsuperp}{\bm U^{\mathsf{gs},\perp}}
\newcommand{\adX}{\bm X^{\perp}}
\newcommand{\adv}{\bm V^\perp}
\newcommand{\adu}{\bm U^\perp}
\newcommand{\advs}{\wb{\bm V}^\perp}
\newcommand{\adus}{\wb{\bm U}^\perp}
\newcommand{\adxl}{Y^{L,\perp}}
\newcommand{\adxr}{Y^{R,\perp}}
\newcommand{\alphalt}{\bm{A}^{L,\perp}}
\newcommand{\betalt}{\bm{B}^{L,\perp}}
\newcommand{\tault}{\bm{T}^{L,\perp}}
\newcommand{\alphart}{\bm{A}^{R,\perp}}
\newcommand{\betart}{\bm{B}^{R,\perp}}
\newcommand{\taurt}{\bm{T}^{R,\perp}}
\newcommand{\funcm}{\mathtt{S}}
\newcommand{\funcv}{\mathtt{\Sigma}}
\newcommand{\tildSl}{\bm{S}^{L,\perp}}
\newcommand{\tildSr}{\bm{S}^{R,\perp}}
\newcommand{\tildSgl}{\bm{\Sigma}^{L,\perp}}
\newcommand{\tildSgr}{\bm{\Sigma}^{R,\perp}}
\newcommand{\tildfuncu}{u^{\perp,\orc}}
\newcommand{\tildfuncv}{v^{\perp,\orc}}
\newcommand{\nperpu}{{\bm N}^{\perp,L}}
\newcommand{\nbaru}{\wb{\bm N}^{\perp,L}}
\newcommand{\nperpv}{{\bm N}^{\perp,R}}
\newcommand{\nbarv}{\wb{\bm N}^{\perp,R}}
\newcommand{\tildgenu}{u^{\mathtt{g}}}
\newcommand{\tildgenv}{v^{\mathtt{g}}}
\newcommand{\genericm}{\bm G}
\newcommand{\genu}{{\bm N}^{G,L}}
\newcommand{\genub}{\wb{\bm N}^{G,L}}
\newcommand{\genv}{{\bm N}^{G,R}}
\newcommand{\genvb}{\wb{\bm N}^{G,R}}
\newcommand{\taul}{\bm T^{G,L}}
\newcommand{\taur}{\bm T^{G,R}}
\newcommand{\jacu}{\bm J^{G,L}}
\newcommand{\jacv}{\bm J^{G,R}}
\newcommand{\sfnr}{\boldsymbol{\mathsf{N}}^{G,R}}
\newcommand{\sfnbr}{\boldsymbol{\wb{\mathsf{N}}}^{G,R}}
\newcommand{\sfnl}{\boldsymbol{\mathsf{N}}^{G,L}}
\newcommand{\sfnbl}{\boldsymbol{\wb{\mathsf{N}}}^{G,L}}
\newcommand{\sfE}{\boldsymbol{\mathsf E}}
\newcommand{\sfP}{\boldsymbol{\mathsf P}}
\newcommand{\rfF}{\boldsymbol{\mathrm N}}

\newcommand{\jac}{\bm J}
\newcommand{\jacl}{\bm J^L}
\newcommand{\jacr}{\bm J^R}
\newcommand{\jacor}{\bm J^\orc}
\newcommand{\jaclor}{\bm J^{L,\orc}}
\newcommand{\jacror}{\bm J^{R,\orc}}
\newcommand{\jaclperp}{\bm J^{L,\perp}}
\newcommand{\jacrperp}{\bm J^{R,\perp}}

\newcommand{\orchAMP}{{OrchAMP}}

\newtheorem{thm}{Theorem}[section]
\newtheorem{lem}{Lemma}[section]
\newtheorem{asm}{Assumption}[section]
\newtheorem{cor}{Corollary}[section]
\newtheorem{prop}{Proposition}[section]
\theoremstyle{remark}
\newtheorem{rem}{Remark}[section] 

\newcommand{\R}{\mathbb{R}}

\usepackage{subcaption}

\newcommand{\revzm}[1]{\textcolor{black}{#1}}
\newcommand{\revzms}[1]{\textcolor{black}{#1}}
\newcommand{\revsn}[1]{\textcolor{black}{#1}}
\newcommand{\revsns}[1]{\textcolor{black}{#1}}
\newcommand{\revsag}[1]{\textcolor{black}{#1}}

\renewcommand{\arraystretch}{2}

\DeclareMathOperator*{\argmax}{arg\,max}

\usepackage{graphicx} 
\usepackage{color, float}
\usepackage{booktabs,rotating,multirow,pdflscape,adjustbox} 
\usepackage{array} 
\usepackage{verbatim} 

\title{Multimodal data integration and cross-modal querying via orchestrated approximate message passing}
\author{
Sagnik Nandy$^{1}$\thanks{Email: nandy.15@osu.edu}
\and
Zongming Ma$^{2}$\thanks{Email: zongming.ma@yale.edu}
}

\date{
$^{1}$Department of Statistics, The Ohio State University\\
$^{2}$Department of Statistics and Data Science, Yale University
}

\begin{document}

	\maketitle
	
\begin{abstract}
The need for multimodal data integration arises naturally when multiple complementary sets of features are measured on the same sample, such as in single-cell multi-omics studies.
Under a dependent multifactor model, we develop a fully data-driven orchestrated approximate message passing algorithm for integrating information across these feature sets to achieve statistically optimal signal recovery. In practice, these reference datasets are often queried later by new subjects that are only partially observed. 
Leveraging on asymptotic normality of estimates generated by our data integration method, we further develop an asymptotically valid prediction set for the latent factor representation of any such query subject.
We demonstrate the prowess of
both the data integration and the prediction set construction algorithms synthetic examples and real world single-cell datasets.
\end{abstract}
	






\section{Introduction}

\subsection{Motivation}


Identifying and documenting different cell populations and cell states from single-cell RNA sequencing (scRNAseq) data is an important problem in cell biology \cite{Papalexi2017-dn, Chen2019a}. 
Funding agencies have organized consortium level efforts to build large scale reference datasets (a.k.a.~atlases) for different tissues and organs in different biological conditions. 
Such efforts include but are not limited to the Human Biomolecular Atlas Project (HuBMAP) \cite{hubmap2019human}, the Human Tumor Atlas Network (HTAN) \cite{rozenblatt2020human}, the Molecular Atlas of Lung Development Program (LungMAP) \cite{gaddis2024lungmap}, etc.
They have resulted in high-quality expert-annotated reference datasets or atlases that can be later used for classifying and annotating new cells of similar types. 
The practice is known as ``label transfer''.
Though scRNAseq can help separate heterogeneous groups of cells, transcriptomics alone often is unable to separate molecularly similar but functionally different types of cells. 
To resolve this issue, biologists often rely on other molecular or structural features observed along with scRNAseq data.

Multimodal technologies that measure different single-cell level characteristics have gained tremendous popularity in recent years \cite{teichmann2020method}. 
These technologies allow for the availability of transcriptome data along with other single-cell level biochemical properties, such as surface protein abundance \cite{Stoeckius2017} and chromatin accessibility \cite{doi:10.1126/science.aau0730, Chen2019}, among others. With the availability of multimodal data, it becomes important to design efficient data integration techniques to combine information from these heterogenous modalities to create a unified low-dimensional embedding to segregate cells having molecular and functional dissimilarities. 
\revsn{In practical settings, multimodal reference atlases are frequently queried by new cells measured in a single modality, or even a subset of features within one modality. 
This task, often referred to as \emph{label transfer}, involves inferring cell types of new observations by
\revzm{mapping them to the embedding space learned from reference data.}
While conceptually related to supervised learning, label transfer differs in that it often requires cross-modality mapping and domain adaptation, as the new data may not share the same feature space or modality as the reference atlas. These challenges make it essential not only to produce accurate predictions of cell identity, but also to quantify the uncertainty associated with each prediction, particularly when query data are noisy, incomplete, or potentially unmatched to known cell types in the reference.
}

In \cite{HAO20213573}, the authors developed a weighted-nearest-neighbor (WNN) approach for integrating multimodal data. 
Their method can handle multiple datasets with different dimensions and is state-of-the-art in the literature to efficiently segregate different cells into interpretable categories based on the structural and molecular similarities inferred from 
the multimodal single-cell observations. 
Other popular methods include CiteFuse \cite{kim2020citefuse}, MOFA+ \cite{argelaguet2020mofa+}, and totalVI \cite{gayoso2021joint}.
However, these methods suffer from a common shortcoming from 
a
statistical perspective: 
They are built on heuristics and lack solid mathematical foundations that enable theoretically justifiable generalization beyond the data examples shown in the original papers in which they were proposed.
Moreover, to the best of our knowledge, none of these methods generates an atlas that allows uncertainty quantification of the predicted cell state when queried.

To address these issues, we propose in this paper a new multimodal data integration method that generates atlases with provable performance guarantees under a stylized model and provides an asymptotically valid prediction set for any desired coverage probability when the constructed atlas is queried by a new subject with partial observation of the modalities and the features used in atlas construction.

Although we have motivated the problem studied by this manuscript mainly with single-cell multi-omics data analysis, the model we shall propose is general and so are the developed methods and the related theory.


\subsection{Problem formulation}

In general, one can view multimodal data as multiple datasets with correspondence of measured subjects. 
If each dataset is stored as a subject-by-feature matrix, then these data matrices have corresponding rows.
\revzms{We model}
the latent state encodings of the same subject across different datasets as 
dependent
\revzms{signal}
vectors and design a multi-orbit approximate message passing algorithm in each iteration of which we synchronously integrate information across different datasets to denoise each estimate of the latent space representation vector specific to an individual dataset.
The resulting \revzms{algorithm}
has a number of desirable characteristics: fast computation, provable statistical optimality under reasonable modeling assumptions, and the availability of asymptotically valid confidence statements. 


We develop our method based on the following dependent multimodal factor model.
For any positive integer $\ell$, let $[\ell] = \{1,\dots,\ell\}$ be the count set.
Suppose all available datasets come in the form of $m$ high-dimensional and $\wt m$ low-dimensional datasets,
where $m \geq 1$, $\wt m \geq 0$, and $m + \wt{m}\geq 2$.
Suppose \revsn{for each $h \in [m]$, the high-dimensional data matrix} $\hdm_h \in \mathbb{R}^{N \times p_h}$
satisfies
\begin{equation}
    \label{eq:def_matrix_prop_asymp}
    \hdm_h = \frac{1}{\sqrt{N}}\bm U_h \bm D_h \bm V_h^\top + \bm W_h,
\end{equation}
where $\bm U_h \in \R^{N \times r_h}$, $\bm V_h \in \R^{p_h \times r_h}$, $\bm D_h \in \R^{r_h \times r_h}$ and $\bm W_h \in \R^{N \times p_h}$. 
Hence, each $\hdm_h$ can be viewed as a noisy observation of the low-rank signal $\frac{1}{\sqrt{N}}\bm U_h \bm D_h \bm V_h^\top$ representing effective information present in the matrix.
\revzms{We assume that each $\bm D_h$ is diagonal with positive diagonal entries, ${\bm U}_h^\top {\bm U}_h\approx N {\bm I}_{r_h}$, and ${\bm V}_h^\top {\bm V}_h \approx p_h {\bm I}_h$. 
In other words, up to scaling, each $\bm U_h \bm D_h \bm V_h^\top$ is essentially the singular value decomposition of the signal component.
We do not enforce strict orthonormalization of ${\bm U}_h$'s and ${\bm V}_h$'s as we shall assume that their rows are sampled from certain priors and we shall make this assumption precise below.}
In addition, 
assume each low-dimensional matrix $\ldm_\ell \in \mathbb{R}^{N \times {\wt r}_\ell}$
satisfies
\begin{equation}
    \label{eq:def_matrix_low_dim_asymp}
    \ldm_\ell = \wt{\bm U}_\ell\bm L^\top_\ell+\wt{\bm W}_\ell,
\end{equation}
where $\wt{\bm U}_\ell \in \R^{N \times {\wt r}_\ell}$ and $\bm L_\ell \in \R^{{\wt r}_\ell \times {\wt r}_\ell}$. We assume that for all $\ell \in [\wt m]$, the matrices $\bm L_{\ell}$ are symmetric and full-rank. 
The low dimensional matrix $\ldm_\ell$ can be viewed as a noisy observation of the low-dimensional signal $\wt{\bm U}_\ell\bm L^\top_\ell$ for $\ell \in [\wt m]$.
\revzms{Similarly, we assume $\wt{\bm U}_\ell^\top \wt{\bm U}_\ell \approx N \bm I_{\wt{r}_\ell}$ as the rows of $\wt{\bm U}_\ell$'s are also sampled from prior.}
Throughout the paper, we differentiate high- and low-dimensional datasets in asymptotic sense: 
for high-dimensional datasets, we assume proportional growth asymptotics where $p_h/N \rightarrow \gamma_h \in (0,\infty)$ \revzm{and $r_h = O(1)$ is fixed} as both $N$ and $p_h \rightarrow \infty$ for $h\in [m]$, and
for low-dimensional datasets, we assume that ${\wt r}_\ell = O(1)$ is fixed for all $\ell \in [\wt m]$ as $N \rightarrow \infty$.

\revsn{The assumption of a low-rank latent structure for the signal component of the high-dimensional modalities is widely used in both empirical and theoretical works across several disciplines. In single-cell biology, it is widely accepted that a relatively small number of biological programs, such as differentiation trajectories and cell-cycle phases,
drive the observed high-dimensional variation in gene
expressions. This motivates the modeling of each modality as arising from a low-dimensional latent representation, a framework that has been adopted in numerous recent works including \cite{HAO20213573, gayoso2021joint, argelaguet2020mofa+, JMLR:v22:20-589, denault2025covariate}. 
Beyond biology, similar assumptions have been adopted in machine learning and artificial intelligence, where it is often assumed that high-dimensional observations lie near a low-dimensional manifold. 
In the multimodal setting, the different data modalities are typically assumed to be governed by a set of dependent low-dimensional latent factors \cite{lock2013jive,radford2021learning,gui2025multi} 
which capture coordinated structure across modalities, making low-rank modeling a natural and interpretable choice.}
\revsn{The additional low-dimensional modalities in equation \eqref{eq:def_matrix_low_dim_asymp} are motivated by an experimental design in single-cell genomics where a \emph{targeted panel} \revzm{of biomarkers} is used. 
Such a panel involves a curated set of pre-selected genes or proteins (typically tens 
of features) that are known to be biologically informative, and 
are often collected alongside more comprehensive but sparse or noisy measurements of the full transcriptome or \revzm{the full} epigenome. 
}

In data integration and atlas construction,
our goal is to \revzms{jointly} recover the low-rank signals in \eqref{eq:def_matrix_prop_asymp} and \eqref{eq:def_matrix_low_dim_asymp} from the collection of matrices $\{\hdm_h: h \in [m] \}$ and $\{\ldm_\ell: \ell \in [\wt m] \}$.
For the $i$th subject, we regard the $i$th rows in the subject effect matrices $\bm U_h$ for $h \in [m] $ and $\wt{\bm U}_\ell$ for $\ell \in [\wt m] $ as latent state encodings of this subject provided by different data modalities. 
Since the subject is the same and the modalities are different, these encoding vectors are expected to be distinct while dependent. 
In particular, 
we assume that $\{((\bm U_1)_{i*},\dots, (\bm U_m)_{i*}, ({\wt{\bm U}}_1)_{i*},\dots, ({\wt{\bm U}}_{\wt m})_{i*} ):i\in [N] \}$ are i.i.d.~samples from a joint prior distribution $\mu$ (\revsn{supported on a subset of $\mathbb R^{\sum_{i=1}^{m}r_i+\sum_{j=1}^{\wt m}\wt r_j}$}),
through which we are allowed to incorporate any dependence structure among different modalities. 
The 
dependence of the modality-specific latent state encoding vectors, and hence the signals, necessitates information pooling for their reconstruction. 
For other components in \eqref{eq:def_matrix_prop_asymp}-\eqref{eq:def_matrix_low_dim_asymp}, we assume that the rows 
of each modality effect matrix $\bm V_h$ are sampled independently from some prior $\nu_h$ and they are mutually independent across $h \in [m]$.
The matrices $\{\bm D_h: h \in [m] \}$ are assumed to be diagonal matrices with positive diagonal entries representing signal strengths characterizing the association between subject effects and modality effects. Furthermore, to ensure identifiability of the columns in $\bm U_h$ and $\bm V_h$ for each $h \in [m]$, we assume that the entries $\{(\bm D_h)_{ii}: i \in [r_h]\}$ are distinct for all $h \in [m]$. 
\revsn{This assumption is not necessary for the validity of the OrchAMP algorithm or its associated state evolution equations, but is imposed to guarantee identifiability of the singular vectors \revzm{and to facilitate prediction set construction for query subjects}.}
The entries of the error matrices $\bm W_h$ for $h \in [m]$ and $\wt{\bm W}_\ell$ for $\ell \in [\wt m]$ are i.i.d.~$N(0,1)$ random variables.


In atlas querying and prediction set construction, we assume that we observe, on a query subject, subsets of features on a subset of modalities seen in the reference datasets used for atlas construction. 
To be precise, we observe $\{\bq_{h_1},\ldots,\bq_{h_d},\tq_{\ell_1},\ldots,\tq_{\ell_{\wt d}}\}$ on a query subject, where the vectors $\bq_{h_1},\ldots,\bq_{h_d}$ are from high-dimensional modalities $\{h_1,\ldots,h_d\} \subseteq [m]$ and $\tq_{\ell_1},\ldots,\tq_{\ell_{\wt d}}$ low-dimensional modalities $\{\ell_1,\ldots,\ell_{\wt d}\} \subseteq [\wt m]$.
For the $k$-th observed high-dimensional modality,
$\bq_{h_k} \in \R^{|\calf_{h_k}|}$ 
where $\calf_{h_k} \subseteq [p_{h_k}]$ collects indices of the observed features in this modality.
For the $k$-th observed low-dimensional modality,
$\tq_{\ell_k} \in \R^{|\wt\calf_{\ell_k}|}$ 
where $\wt\calf_{\ell_k} \subseteq [\wt r_{\ell_k}]$ denotes indices of the observed features.
The observed information on query is assumed to be partial observation of an independent realization of the rows in model \eqref{eq:def_matrix_prop_asymp}--\eqref{eq:def_matrix_low_dim_asymp} with the same $\{\bm V_h, \bm D_h: h\in [m]\}$, $\{\bm L_\ell: \ell\in [\wt m]\}$, and the same latent factor prior distribution $\mu$ as in the reference data. 
Here, our goal is to leverage both the atlas and the partial observation on query to construct a prediction set for the query's full latent representation $({U}_1^{\mathrm{Q}},\dots,{U}_m^{\mathrm{Q}}, \wt{U}_1^{\mathrm{Q}},\dots, \wt{U}_{\wt{m}}^{\mathrm{Q}})$.


\subsubsection{Main contributions}

The main contributions of this work include the following:
\begin{enumerate}

\item \textbf{Proposal of the OrchAMP algorithm for data integration and its theoretical justification}.
We propose and justify a fully data-driven algorithm for the integrative estimation of signal components in model \eqref{eq:def_matrix_prop_asymp}--\eqref{eq:def_matrix_low_dim_asymp}. 
We establish the ``large $N$'' asymptotic normality of all iterates of the algorithm under general conditions for a wide range of priors on the signal.
Under an additional unique-fixed-point condition on the approximate message passing state evolution, we further establish the ``large $N$, large $t$'' asymptotic Bayes optimality of the estimators produced by the algorithm.

\item \textbf{Design of an asymptotically valid prediction set for query.} Leveraging on the asymptotic normality of the iterates in the OrchAMP algorithm, we construct asymptotically valid prediction sets for query data in which only a subset of modalities and/or subsets of features within modalities are observed.
The prediction set is valid even when the unique-fixed-point condition on state evolution is not satisfied. Hence, it is widely applicable.
To the best of our knowledge, the present work is the first to construct asymptotically valid prediction sets for queries.


\item \textbf{Validation on a single-cell data example.}
We validate the proposed atlas building and prediction set construction algorithms on a tri-modal TEA-seq dataset \cite{10.7554/eLife.63632}.
While both algorithms have been motivated by and theoretically validated under model \eqref{eq:def_matrix_prop_asymp}--\eqref{eq:def_matrix_low_dim_asymp} in the first place, we demonstrate on this real dataset that the OrchAMP algorithm achieves comparable state-of-the-art cell state differentiation capacity to its competitors such as WNN, and that the prediction sets constructed by our algorithm are highly informative when we query the constructed atlas with a new cell measured in a single modality.

\end{enumerate}

A demonstration of the proposed methods on the single-cell data example is available at \url{https://github.com/Sagnik-Nandy/OrchAMP}.

\subsection{Related works}

Approximate message passing (AMP) was developed by \citet{doi:10.1073/pnas.0909892106} in the context of compressed sensing. 
Related ideas were also explored in the context of belief propagation in graphical models \cite{PEARL198829} and statistical physics \cite{10.1093/acprof:oso/9780198570837.001.0001,doi:10.1142/0271}. 
Furthermore, the idea of state evolution constituting a major component of the message-passing algorithms leading to the development of AMP was numerically studied in \cite{6409458,5464920,doi:10.1073/pnas.0909892106}, among others.
\revsn{The foundational work of \citet{BM11journal} developed a rigorous framework for analyzing the approximate message passing (AMP) algorithm in the context of high-dimensional linear regression, including the LASSO, 
\revzm{and} \citet{JM12} adapted this framework to study AMP for low-rank matrix recovery problems. 
\revzm{They showed}
that AMP can achieve asymptotically Bayes-optimal recovery of low-rank signals when an appropriate class of denoising functions is employed within the iterative updates. 
\revzm{See also \cite{Berthier, AbbeMonYash}.}
In parallel, \citet{Lesieur2017ConstrainedLM} studied constrained low-rank matrix estimation using AMP, precisely characterizing phase transitions in recovery performance under structural priors and non-Gaussian noise, thereby broadening AMP’s applicability beyond classical spiked models. 
\revzm{See also \cite{FletcherRangan2018, yuting, montanari2021,eb_pca}.}}
\revzm{In particular, \citet{eb_pca} proposed an empirical Bayes technique to estimate the underlying prior of the signal components and hence established a completely data-driven method to compute the denoising functions in AMP for low-rank matrix recovery in the single-view setting. The present manuscript generalizes this empirical Bayes idea to estimation in multi-view settings and to prediction set construction. The latter is beyond the scope of \cite{eb_pca} even in single-view settings.}


The use of AMP in information retrieval with multi-view data, where more than one data sources (views) contain independent pieces of information about the same underlying signal of interest, has been considered in the Statistics/Machine Learning literature. 
When one data source is high-dimensional and the other is low-dimensional, researchers have designed specific AMP algorithms around the high-dimensional data source by treating low-dimensional data as side information \revzm{\cite{Rush,doi:10.1073/pnas.1802705116}}. 
\revsn{When there are multiple high-dimensional sources, \citet{ma_nandy} studied a special case where the signal of interest is a vector \revzm{(hence of rank one) that is identical across sources and} whose entries are i.i.d.~random numbers from a (known) Rademacher
prior in a stylized contextual multi-layer two-block network model. 
They proposed \revzm{an elementary form of {\orchAMP} with multi-orbit message passing to ensure}
asymptotically Bayes optimal reconstruction of the signal.
\revzm{The present manuscript develops this idea from identical rank-one signal to dependent signals of unequal ranks, from an oracle initialization to spectral, from known Rademacher prior to generic unknown priors, and from only low-rank signal recovery to both signal recovery and cross-modal prediction set construction.}}
\revzm{See \cite{nandy2023bayes} 
for the generalization to high-dimensional regression with network side information.}
\citet{gerbelot2022graphbased} proposed a unified \revzms{perspective}
of these AMP algorithms as message passing on edges of an oriented graph whose vertices contain different related modalities.
When preparing this manuscript, we became aware of the recent work by \citet{keup2024optimal}, which considered a related multimodal set-up consisting of two high-dimensional rank-one spiked random matrices with correlated factors. 
They linearized the belief propagation updates to obtain a version of AMP and used it to reconstruct the latent factors from the data matrices. 
Furthermore, they quantified the gain in the minimum signal-to-noise ratio required for weak recovery of the correlated factors by integrating the two data matrices compared to that from using a single matrix.
In comparison,
we tackle here the more challenging and realistic scenario where the underlying signals to be recovered from multiple data sources are correlated/dependent rather than identical and have potentially different ranks, the signal priors are unknown, and there exist both high-dimensional and low-dimensional information sources.

\subsection{Notations}
The notation $\R^k$ denotes the set of $k$ dimensional vectors with real co-ordinates. 
For a vector $a \in \R^k$, the Euclidean norm is denoted by $\|a\|$ and the $p$-th norm for $p\in \mathbb N\setminus\{2\}$ is denoted by $\|a\|_p$. 
The notation $\R^{k \times \ell}$ denotes the set of $k\times \ell$ matrices with real entries. 
For any set $\mathcal F$, $|\mathcal F|$ denotes its cardinality. 
All the matrices will be denoted by upper case bold letters throughout the text. 
For any matrix $\bm A$, $\bm A_{i*}$ denotes its $i$-th row and $\bm A_{*i}$ its $i$-th column. 
Furthermore, $\bm A_{ij}$ denotes its $(i,j)$-th entry. 
For any subset $\mathcal F$ of $\{1,\ldots,n\}$, the submatrix formed by the rows with indices in $\mathcal F$ is denoted by $\bm A_{\mathcal F,*}$ and the submatrix formed by the columns with indices in $\mathcal F$ is denoted by $\bm A_{*,\mathcal F}$. 
Next, for $\bm A \in \R^{k \times k}$, $\bm A^\top$ denotes its transpose and $\mathrm{Tr}(\bm A)= \sum_{i=1}^{k}\bm A_{ii}$ denotes its trace. 
The Frobenius norm of a matrix $\bm A$ is denoted by $\|\bm A\|_F$ 
and the spectral norm of $\bm A$ is denoted by $\|\bm A\|$. 
The notation $\mbox{diag}(a_1,\ldots,a_k)$ defines a $k \times k$ diagonal matrix with the $j$-th diagonal element $a_j$ for $j \in [k]$. 
For a vector $v \in \R^n$, $v^{\otimes 2}$ denotes the $n\times n$ matrix $vv^\top$. 
For two matrices $\bm A$ and $\bm B$, $\bm A \preceq \bm B$ means $\bm B-\bm A$ is positive semi-definite. The class of all $m \times m$ dimensional positive definite matrices is denoted by $\mathbb S^{m}_+$. The notation $\mu_n \xrightarrow{w} \mu$ means that the sequence of measures $\mu_n$ converges weakly to $\mu$. Finally, let $B_p(r)$ denote the set $\{x \in \R^p:\|x\| \le r\}$.

\section{Methods}
\label{sec:methods}
 
Suppose on $N$ subjects, we observe multimodal data $\{\hdm_h: h\in[m]\}$ and $\{\ldm_\ell:\ell\in [\wt m] \}$ as generated in \eqref{eq:def_matrix_prop_asymp}--\eqref{eq:def_matrix_low_dim_asymp}.
In this section, we first introduce an algorithm for the integrative estimation of the signal components in these data matrices. 
In addition, we propose a method for 
\revzms{constructing prediction set of latent state encodings}
across all modalities for a new subject (i.e., a query) on which only a subset of modalities/features are observed.

\subsection{Integrative signal recovery with {\orchAMP}}
\label{sec:estimation}
\subsubsection{An initial singular value decomposition}
To facilitate the description of our method, we rescale the high-dimensional matrices by $\sqrt{N}$ 
to define
\begin{align}
\label{eq:rescaled_matrices}
  \rshdm_h= \frac{1}{\sqrt{N}}\hdm_h,\quad h\in [m].
\end{align}
Recall that the
signal 
in 
$\rshdm_h$ is of rank $r_h$.
We write the best rank $r_h$ approximation to $\rshdm_h$ as
\begin{equation}
    \label{eq:best_rank_k}
    \frac{1}{N}\upca_{0,h}\bm D_{0,h} (\vpca_{0,h})^\top = 
    \frac{1}{N}\sum_{k=1}^{r_h}(\bm D_{0,h})_{kk}\;(\upca_{0,h})_{*k}\;(\vpca_{0,h})^\top_{*k},
\end{equation}
where $\bm D_{0,h}$ is the diagonal matrix consisting of the top $r_h$ 
singular values of 
$\rshdm_h$ and $\upca_{0,h}$ and $\vpca_{0,h}$ are the matrices of the corresponding sample left and right singular vectors rescaled such that 
\begin{equation}
    \label{eq:sample-evec-scale}
    (\upca_{0,h})^\top\upca_{0,h}=N\,\bm I_{r_h} 
    \quad \mbox{and}\quad 
    (\vpca_{0,h})^\top\vpca_{0,h}=p_h\,\bm I_{r_h}. 
\end{equation}
Further, we  assume that for all $h \in [m]$ and $j \in [r_h]$
\begin{align}
\label{eq:identifiability_sign_eqn}
(\upca_{0,h})^\top_{*j}(\bm U_h)_{*j} \ge 0, \quad \mbox{and} \quad (\vpca_{0,h})^\top_{*j}(\bm V_h)_{*j} \ge 0.
\end{align}
\revsn{This assumption is taken to enforce identifiability of the latent factors since 
\revzm{otherwise they}
are \revzm{only} identifiable up to \revzm{columnwise} sign flips.}

\revsn{When the population singular values $\{(\bm D_h)_{kk} : k \in [r_h]\}$ are sufficiently large\footnote{\revzm{To be exact, when each $(\bm D_h)_{kk} > \gamma_h^{-1/4}$ 
as 
$N, p_h \to \infty$ with $p_h/N \to \gamma_h \in (0, \infty)$.
This condition 
corresponds to a Baik–Ben Arous–Péché (BBP) type threshold \cite{baik2005} that ensures separation of the signal components from the bulk of the spectrum.}}, one can show that\footnote{For the precise formulation of the approximation in \eqref{eq:emp_sin_vec_asymp}, see Proposition~\ref{prop:singular_values}.},}
\begin{align}
\label{eq:emp_sin_vec_asymp}
    \upca_{0,h} \approx \bm U_h(\hslmorp_{0,h})^\top+\bm Z^L_h(\hslvorp_{0,h})^{1/2}, \hskip 1.5em \vpca_{0,h} \approx \bm V_h(\hsrmorp_{0,h})^\top+\bm Z^R_h(\hsrvorp_{0,h})^{1/2}.
\end{align}
Here $\bm Z^L_h$ and $\bm Z^R_h$ are independent white noise matrices with $(\bm Z^L_h)_{ik}, (\bm Z^R_h)_{jk} \stackrel{iid}{\sim} N(0,1)$ for 
$i\in [N], j\in [p_h]$, and $k\in [r_h]$.
In addition, for $\star\in \{L,R\}$,
\begin{equation}
\label{eq:initializer_1}
  {\bm S}^{\star,\mathrm{pc}}_{0,h} = \mathrm{diag}((s^\star_{0,h})_{1},\dots, (s^\star_{0,h})_{r_h}),
  \quad
  {\bm \Sigma}^{\star,\mathrm{pc}}_{0,h} = \mathrm{diag}((\sigma^\star_{0,h})^2_{1},\dots, (\sigma^\star_{0,h})_{r_h}^2),
\end{equation}
where
for all $h \in [m]$ and $k \in [r_h]$,
\begin{equation}
\label{eq:initializers}
\begin{aligned}
    (\sigma^L_{0,h})^2_{k} = \frac{1+(\bm D_h)^2_{kk}}{(\bm D_h)^2_{kk}\{\gamma_h\,(\bm D_h)^2_{kk}+1\}}, &\hskip 1.5em  (\sigma^R_{0,h})^2_{k} = \frac{1+\gamma_h\,(\bm D_h)^2_{kk}}{\gamma_h\,(\bm D_h)^2_{kk}\{(\bm D_h)^2_{kk}+1\}},\\
    (s^L_{0,h})_{k} = \sqrt{1-(\sigma^L_{0,h})^2_{k}}, &\hskip 1.5em  (s^R_{0,h})_{k} = \sqrt{1-(\sigma^R_{0,h})^2_{k}}.
\end{aligned}
\end{equation}

The representations \eqref{eq:emp_sin_vec_asymp}--\eqref{eq:initializers} and the model \eqref{eq:def_matrix_prop_asymp}--\eqref{eq:def_matrix_low_dim_asymp}
enable one to infer the population singular vectors in $\bm U_h$ and $\bm V_h$ based on their sample counterparts and $\wt{\bm U}_\ell$ based on the leading left singular vectors $\tupca_{0,\ell}$ of $\wt{\bm X}_\ell$ within each dataset. 
However,  
when 
there is dependence among these signals, it is necessary to integrate information across all datasets for statistical efficiency even if one is only interested in the signal within a single \revzms{modality}.

\subsubsection{\revzms{Estimation of nuisance parameters}}
Under appropriate assumptions,
we can consistently estimate the signal strengths $(\bm D_h)_{kk}$ for all $h \in [m]$ and $k \in [r_h]$ by the square roots of 
\begin{align}
\label{eq:est_sing_val}
(\hd_{h})^2_{kk}=\frac{\gamma_h(\bm D_{0,h})^2_{kk}-(1+\gamma_h)+\sqrt{[\gamma_h(\bm D_{0,h})^2_{kk}-(1+\gamma_h)]^2-4\gamma_h}}{2\gamma_h}.
\end{align}
Consequently, we can construct consistent estimators of ${\bm S}^{\star,\mathrm{pc}}_{0,h},{\bm \Sigma}^{\star,\mathrm{pc}}_{0,h}$ for $\star\in \{L,R\}$ and $h \in [m]$ as
\begin{equation}
    \label{eq:scale-est}
  \wh{\bm S}^{\star,\mathrm{pc}}_{0,h} = \mathrm{diag}((\wh{s}^\star_{0,h})_{1},\dots, (\wh{s}^\star_{0,h})_{r_h}),
  \quad
  \wh{\bm \Sigma}^{\star,\mathrm{pc}}_{0,h} = \mathrm{diag}((\wh{\sigma}^\star_{0,h})^2_{1},\dots, (\wh{\sigma}^\star_{0,h})_{r_h}^2),
  \end{equation}
where for $k\in [r_h]$, $(\wh{s}^\star_{0,h})_{k}$ and $(\wh{\sigma}^\star_{0,h})_{k}$ are obtained from replacing $({\bm D}_h)_{kk}$ with $(\wh{\bm D}_h)_{kk}$ in \eqref{eq:initializers}.
For low-dimensional modalities, we can consistently estimate the matrices $\bm L_\ell$ for $\ell \in [\wt m]$ by 
\begin{equation}
    \label{eq:sqrt_L}
    \wh{\bm L}_\ell=\Big(\frac{1}{N}\wt{\bm X}^\top_\ell\wt{\bm X}_\ell-\bm I_{{\wt r}_\ell}\Big)^{1/2}.
\end{equation}

\subsubsection{Iterative refinement via {\orchAMP}}
To incorporate information across multiple modalities,
we design an iterative refinement algorithm for $\{\upca_{0,h},\vpca_{0,h}: h\in [m]\}$, and $\{\tupca_{0,\ell}:\ell\in [\wt{m}]\}$, 
which generates a sequence of estimates 
$\{\bm U_{t,h},\bm V_{t,h}: h\in [m]\}$ and $\{\wt{\bm U}_{t,\ell}:\ell\in [\wt{m}]\}$ for $t\geq 0$, and information integration across all datasets occurs at each iteration.

In the ideal situation, the optimal updates in each iteration require knowledge of the priors $\mu$ and $\{\nu_h: h\in [m]\}$. 
In practice, this knowledge is unavailable and we adopt the empirical Bayes approach in \cite{eb_pca} to obtain data-based estimators of these priors.

\paragraph{Estimating priors}
Let $r = \sum_{h=1}^m r_h$ and $\wt{r} = \sum_{\ell=1}^{\wt{m}}\wt{r}_\ell$.
The approximation in \eqref{eq:emp_sin_vec_asymp} suggests that
the corresponding rows in the matrices $\{\bm U_{0,h}:h\in [m]\}$ and
$\{\ldm_{\ell}:\ell\in [\wt{m}]\}$
can be viewed as i.i.d.~samples from the joint distribution of 
\[
(Y_{0,1},\ldots,Y_{0,m},\wt Y_{0,1},\ldots,\wt Y_{0,\wt m}) \in \R^{r + \wt{r}}
\]
where conditional on $(U_1,\ldots,U_m,\tilde{U}_1,\ldots,\tilde{U}_{\wt m})$,
\begin{equation}
\label{eq:compund_decision_singular_vector}    
\begin{aligned}
    Y_{0,h} \sim N_{r_h}(\hslmorp_{0,h}U_h,\hslvorp_{0,h}),\, h\in [m], \quad 
    \wt{Y}_{0,\ell} \sim N_{\wt{r}_\ell}(\bm L_{\ell}\wt{U}_\ell,\bm I_{{\wt r}_\ell}),\, \ell\in [\wt{m}].
\end{aligned}
\end{equation}
Hence, the (approximate) likelihood of $\mu$ given $\{\upca_{0,h}:h\in [m]\}$ and
$\{\ldm_{\ell}:\ell\in [\wt{m}]\}$
is 
\begin{equation}
    \label{eq:lik-U}
\begin{aligned}
& L(\mu; \{ \hslmorp_{0,h}, \hslvorp_{0,h}:h\in [m] \}, \{{\bm L}_{\ell}:\ell\in [\wt{m}]\} ) \\
& ~~~\propto~~ \prod_{i=1}^{N}\Bigg[ \int \prod_{h=1}^{m}\exp\left(-\frac{1}{2}[(\upca_{0,h})_{i*}-\hslmorp_{0,h}(\bm U_h)_{i*}]^\top(\hslvorp_{0,h})^{-1}[(\upca_{0,h})_{i*}-\hslmorp_{0,h}(\bm U_h)_{i*}]\right)\\
& \hskip 6.4em \times\prod_{\ell=1}^{\wt m}\exp\left(-\frac{1}{2}[(\ldm_\ell)_{i*}-\bm L_{\ell}(\wt{\bm U}_\ell)_{i*}]^\top[(\ldm_\ell)_{i*}-\bm L_{\ell}(\wt{\bm U}_\ell)_{i*}]\right)\\
& \hskip 6.4em ~~~~~\times d\mu((\bm U_1)_{i*},\ldots,(\bm U_m)_{i*},(\wt{\bm U}_1)_{i*},\ldots,(\wt{\bm U}_{\wt m})_{i*})\Bigg].
\end{aligned}
\end{equation}
When the prior $\mu$ belongs to \revsn{some class of prior distributions $\mathcal P$ on $\mathbb{R}^{r+\widetilde{r}}$}, we define its approximate MLE as 
\begin{equation}
\label{eq:emp_bayes_1}
    \wh{\mu} = \argmax_{\mu \in \mathcal P}
    L(\mu; \{ \hslmp_{0,h}, \hslvp_{0,h}:h\in [m] \}, \{\wh{\bm L}_{\ell}:\ell\in [\wt{m}]\} ).
    \end{equation}
Note that in \eqref{eq:emp_bayes_1}, we have replaced $\{ \hslmorp_{0,h}, \hslvorp_{0,h}:h\in [m] \}$ and $\{{\bm L}_{\ell}:\ell\in [\wt{m}]\}$ in the definition of $L$ with their sample counterparts in \eqref{eq:scale-est} and \eqref{eq:sqrt_L}, and hence $\wh{\mu}$ is fully data-driven.

Similarly, the rows of $\bm V_{0,h}$ can be regarded as iid~samples from the model
\begin{equation}
\label{eq:compound_decision_model_2}
    Y^R_{0,h}\,\big|\,V_h \sim N_{r_h}(\hsrmorp_{0,h}V_h,\hsrvorp_{0,h}) \quad \mbox{with} \quad V_h \sim \nu_h \quad \mbox{for $h\in[m]$.} 
\end{equation}
As the $V_h$'s are mutually independent,
for each $h\in [m]$,
the approximate MLE of $\nu_h$ is given by 
\begin{align}
\label{eq:emp_bayes_2}
\widehat{\nu}_h &= \argmax_{\nu_h \in \mathcal P_{\nu_h}}\prod_{i=1}^{p_h}\Bigg[\int \exp\left(-\frac{1}{2}[(\vpca_{0,h})_{i*}-\hsrmp_{0,h}(\bm V_{h})_{i*}]^\top (\hsrvp_{0,h})^{-1}[(\vpca_{0,h})_{i*}-\right.\nonumber\\
& \hskip 20em\left. \hsrmp_{0,h}(\bm V_{h})_{i*}]\right)
d\nu_h((\bm V_{h})_{i*})\Bigg],
\end{align}
\revsn{where $\mathcal P_{\nu_h}$ is a class of priors on $\mathbb{R}^{r_h}$. }
We defer
examples
of possible classes of priors $\mathcal P$ and $\mathcal P_{\nu_h}$ to Appendix J.1-J.2.

\paragraph{Updating estimators for the latent factors}
To initialize the iterative updates, we define $\bm V_{0,h}=\vpca_{0,h}$ and $\bu_{-1,h}=\upca_{0,h}(\hsrv_{0,h})^{1/2}$ for $h \in [m]$ where $\hsrv_{0,h}=\hsrvp_{0,h}$ for $\hsrvp_{0,h}$ defined in \eqref{eq:scale-est}. Additionally, we set $\hsrm_{0,h}=\hsrmp_{0,h}$ for all $h \in [m]$.
The proposed {\orchAMP} updates $\{\bu_{t,h},\bv_{t,h}:h\in [m]\}$ for estimating 
$\{{\bm U}_h,{\bm V}_h:h\in [m]\}$
are defined recursively for all integer $t\geq 0$ as:
\begin{equation}
\label{eq:orc_amp_emp_bayes}    
\begin{aligned}
    \bv_{t,h} &= v_{t,h}(\bm V_{t,h};\wh\nu_h), \\
    {\bm U}_{t,h} &= \wb{\hdm}_h\bv_{t,h}-\gamma_h\bu_{t-1,h}(\bm J^R_{t,h}\big(\bm V_{t,h};\wh\nu_h)\big)^\top,\\
    \bu_{t,h} &= u_{t,h}(\bm U_{t,1},\ldots,\bm U_{t,m},\ldm_1,\ldots,\ldm_{\wt m};\wh\mu),\\
    \hskip 0.8em \bm V_{t+1,h} &= \wb{\hdm}^\top_h\bu_{t,h}-\bv_{t,h}(\bm J^L_{t,h}
    \big(\bm U_{t,1},\ldots,\bm U_{t,m},\ldm_1,\ldots,\ldm_{\wt m};\wh\mu)\big)^\top.
\end{aligned}
\end{equation}
\revzms{where the}
\emph{denoisers} 
$v_{t,h}:\mathbb{R}^{r_h} \rightarrow \mathbb{R}^{r_h}$ and $u_{t,h}:\mathbb R^{r+{\wt r}} \rightarrow \mathbb R^{r_h}$ 
are 
\begin{equation}
\label{eq:opt_denoisers}
\begin{aligned}
   & v_{t,h}(x;\wh\nu_h) =\mathbb{E}_{\wh\nu_h}\left[V_h \,\Big|\,\hsrm_{t,h}V_h+(\hsrv_{t,h})^{1/2}Z_h=x \right],\\
   & u_{t,h}(x_1,\ldots,x_m,\wt x_1,\ldots,\wt x_{\wt m};\wh\mu)\\
   & ~~~~~~~~~~~~~=\mathbb{E}_{\wh\mu}\left[U_h \,\Big|\,\wh Y_{t,1}=x_1,\ldots,\wh Y_{t,m}=x_m,\wt Y_{1}=\wt x_1,\ldots,\wt Y_{\wt m}=\wt x_{\wt m} \right]. 
\end{aligned}
\end{equation}
\revzms{In \eqref{eq:opt_denoisers}, }
$Z_h\sim N_{r_h}(0, {\bm I}_{r_h})$, $V_h\sim \wh{\nu}_h$ and $V_h$, $Z_h$, and $(\hsrm_{t,h},\hsrv_{t,h})$ are mutually independent.
In addition, 
\[
\wh Y_{t,h}\,\big|\, U_h \sim N_{r_h}(\hslm_{t,h}U_h,\hslv_{t,h}), 
\;\wt Y_{\ell}\,\big|\, \wt{U}_\ell \sim N_{\wt r_\ell}(\hatl_{\ell}\wt U_\ell,\bm I_{\wt r_\ell}),\; (U_1,\ldots,U_m,\tilde{U}_1,\ldots,\tilde{U}_{\wt m}) \sim \wh\mu\]
 and are independent of $\{\hslm_{t,h}, \hslv_{t,h}:h\in [m]\}$ and $\{\hatl_{\ell}:\ell\in [\wt{m}]\}$. Furthermore, conditional on $\{\hslm_{t,h}, \hslv_{t,h}:h\in [m]\}$, $\{\hatl_{\ell}:\ell\in [\wt{m}]\}$, $(U_1,\ldots,U_m,\tilde{U}_1,\ldots,\tilde{U}_{\wt m})$, elements in the collection $\{\wh{Y}_{t,1},\dots, \wh{Y}_{t,m}, \wt{Y}_1,\dots, \wt{Y}_{\wt{m}}\}$ are mutually independent.
The matrices $\{\hslm_{t,h}, \hslv_{t,h}, \hsrm_{t,h}, \hsrv_{t,h}:t\ge 0\}$ are \emph{state evolution parameters} and are defined by the following state evolution recursions:
\begin{equation}
\label{eq:state_evol_emp_bayes}    
\begin{aligned}
    \hslv_{t,h} = \frac{1}{N}(\bv_{t,h})^\top\bv_{t,h},\quad & \hsrv_{t+1,h} =\frac{1}{N}(\bu_{t,h})^\top\bu_{t,h}, \\
    \hslm_{t,h} = \hslv_{t,h}\wh{\bm D}_h,\quad &
    \hsrm_{t+1,h} = \hsrv_{t+1,h}\wh{\bm D}_h.
\end{aligned}
\end{equation}
With \eqref{eq:opt_denoisers}, the matrices $v_{t,h}(\bm V_{t,h}; \wh \nu_h) \in \mathbb R^{p_h \times r_h}$ and $u_{t,h}(\bm U_{t,1},\ldots,\bm U_{t,m},\ldm_1,\ldots,\ldm_{\wt m}; \wh \mu) \in \mathbb R^{N \times r_h}$ are defined by applying the denoisers in a rowwise fashion to the matrix arguments:

\begin{equation}
    \label{eq:rowwise}
\begin{aligned}
  &v_{t,h}(\bm V_{t,h}; \wh \nu_h)_{k*}=v_{t,h}((\bm V_{t,h})_{k*}; \wh \nu_h),\\  
    &u_{t,h}(\bm U_{t,1},\ldots,\bm U_{t,m},\ldm_1,\ldots,\ldm_{\wt m}; \wh \mu)_{k*}\\
    &~~~~~~~~~~~~~~~~~=u_{t,h}((\bm U_{t,1})_{k*},\ldots,(\bm U_{t,m})_{k*},(\ldm_1)_{k*},\ldots,(\ldm_{\wt m})_{k*}; \wh \mu).
\end{aligned}
\end{equation}
Furthermore, the matrix $\jacr_{t,h}(\bm V_{t,h};\wh\nu_h) \in \R^{r_h \times r_h}$ is defined entrywise as follows:  
\[
\left[\jacr_{t,h}(\bm V_{t,h};\wh\nu_h)\right]_{ij}= \frac{1}{p_h}\sum_{k=1}^{p_h}\frac{\partial v_{t,h,i}}{\partial x_j}\left((\bm V_{t,h})_{k*};\wh\nu_h\right), 
\quad (i,j) \in [r_h] \times [r_h],
\]
where $v_{t,h}(\cdot)=(v_{t,h,1}(\cdot),\ldots,v_{t,h,r_h}(\cdot))^{\top}$ and $\partial v_{t,h,i}/\partial x_j$ is the partial derivative of $v_{t,h,i}$ with respect to its $j$-th argument. 
Next, for $h \in [m]$, let 
\begin{equation}
\label{eq:mathcali}
\mathcal I_h= \left\{\sum_{k=1}^{h-1}r_{k}+1,\ldots,\sum_{k=1}^{h}r_{k} \right\}
\end{equation}
and for $(x_1,\ldots,x_m,\wt x_1,\ldots,\wt x_{\wt m}) \in \R^{r+\wt r}$, 
let $\bm J_{t,h}(x_1,\ldots,x_m,\wt x_1,\ldots,\wt x_{\wt m};\wh \mu)$
be the Jacobian matrix of the function 
$u_{t,h}(x_1,\ldots,x_m,\wt x_1,\ldots,\wt x_{\wt m};\wh \mu)$ and 
\[\left[\jac_{t,h}(x_1,\ldots,x_m,\wt x_1,\ldots,\wt x_{\wt m};\wh \mu)\right]_{*\mathcal I_h} \in \mathbb{R}^{(r+\wt{r})\times r_h}\] be the submatrix of $\bm J_{t,h}$ that contains only the columns in $\mathcal I_h$. 
Then, in the last line of \eqref{eq:orc_amp_emp_bayes}, the quantity
\begin{equation*}
\begin{aligned}
&\jacl_{t,h}(\bm U_{t,1},\ldots,\bm U_{t,m},\ldm_1,\ldots,\ldm_{\wt m};\wh\mu)\\
&~~~~~~~~~= \frac{1}{N}\sum_{k=1}^{N}\left[\jac_{t,h}((\bm U_{t,1})_{k*},\ldots,(\bm U_{t,m})_{k*},(\ldm_1)_{k*},\ldots,(\ldm_{\wt m})_{k*};\wh\mu)\right]_{*\mathcal I_h}.
\end{aligned}    
\end{equation*}

In summary, the {\orchAMP} updates  in \eqref{eq:orc_amp_emp_bayes} can be roughly decomposed into three parts:
\begin{enumerate}
\item Computing $\wb{\hdm}_h\bv_{t,h}$ and $\wb{\hdm}_h^\top\bu_{t,h}$ is \revsn{a} power iteration for improving estimation accuracy of left and right singular vectors, respectively.
    
\item The subtraction of $\bv_{t,h}(\bm J^L_{t,h})^\top$ from $\hdm^\top_h\bu_{t,h}$ and of $\gamma_h\bu_{t-1,h}(\bm J^R_{t,h}(\bm V_{t,h};\wh\nu_h))^\top$ from $\hdm_h\bv_{t,h}$ removes the dependence on past iterates from the current, which enables tracking of the asymptotic distributions of $\{\bu_{t,h},\bv_{t,h}:h\in [m] \}$.
    
\item The matrices $\bm U_{t,h}$ and $\bm V_{t,h}$ are intermediate estimates of $\bm U_h$ and $\bm V_h$ in the $t$-th iteration. 
They are further improved by the denoisers $u_{t,h}$ and $v_{t,h}$, respectively, to construct the final estimates $\bu_{t,h}$ and $\bv_{t,h}$ in the $t$-th iteration.
\end{enumerate}

\paragraph{Updating estimators for $\{\wt{\bm U}_\ell: \ell\in [\wt{m}]\}$}
For each fixed $\ell\in [\wt{m}]$, define
the denoiser at the $t$-th iteration
$\wt u_{t,\ell}(x_1,\ldots,x_m,\wt x_1,\ldots,\wt x_{\wt m};\wh \mu):\mathbb R^{r+\wt r} \rightarrow \R^{{\wt r}_{\ell}}$ 
as
\begin{equation}
\label{eq:theta_t_k}    
\begin{aligned}
   &\wt u_{t,\ell}(x_1,\ldots,x_m,\wt x_1,\ldots,\wt x_{\wt m};\wh\mu)\\
   &=\mathbb E_{\wh\mu}\left[\wt{U}_\ell\,\Big|\,\wh Y_{t,1}=x_1,\ldots,\wh Y_{t,m}=x_m,\wt Y_{1}=\wt x_1,\ldots,\wt Y_{\wt m}=\wt x_{\wt m} \right],
\end{aligned}
\end{equation}
where the joint distribution of $\{\wh{Y}_{t,1},\dots, \wh{Y}_{t,m}, \wt{Y}_1,\dots, \wt{Y}_{\wt{m}}\}$ is the same as that in \eqref{eq:opt_denoisers}.

For $\ell \in [\wt m]$ and $t \ge 0$, we propose the following denoised estimates for the subject effects of the low-dimensional modalities:
\begin{equation}
	\label{eq:low-dim-iter}
\wt{\bm U}_{t,\ell}=\wt u_{t,\ell}(\bm U_{t,1},\ldots,\bm U_{t,m},\wt{\bm X}_1,\ldots,\wt {\bm X}_{\wt m};\wh\mu).	
\end{equation}
The matrices, $\wt u_{t,\ell}(\bm U_{t,1},\ldots,\bm U_{t,m},\wt{\bm X}_1,\ldots,\wt {\bm X}_{\wt m};\wh\mu) \in \R^{N \times \wt r_{\ell}}$ satisfies 
\begin{equation}
\begin{aligned}
 \label{eq:low_dim_rowwise}
     &\wt u_{t,\ell}(\bm U_{t,1},\ldots,\bm U_{t,m},\wt{\bm X}_1,\ldots,\wt {\bm X}_{\wt m};\wh\mu)_{k*}\\
     &~~~~~~~~~~~~~~=\wt u_{t,\ell}((\bm U_{t,1})_{k*},\ldots,(\bm U_{t,m})_{k*},(\ldm_1)_{k*},\ldots,(\ldm_{\wt m})_{k*};\wh\mu), \quad k \in [N].
 \end{aligned}    
\end{equation}

\subsubsection{\revzms{Summary}}
\revzms{All details of the integrative signal recovery procedure are summarized in}
Algorithm \ref{alg:orchamp}.

\begin{algorithm}[tb]
\caption{Orchestrated Approximate Message Passing}\label{alg:orchamp}
\begin{algorithmic}[1]
\Require
\begin{enumerate}
    \item Data matrices $\wb{\bm X}_h \in \R^{N \times p_h}$ for $h \in [m]$ and $\wt{\bm X}_\ell \in \R^{N \times {\wt r}_\ell}$ for $\ell \in [\wt m]$;
    \item Signal ranks of high-dimensional modalities $\{r_1,\ldots,r_m\}$;
    \item Maximum number of iterates $T$.
\end{enumerate}
\State {\bf Initialization:}
\begin{enumerate} 
\item Perform SVD on each $\wb{\bm X}_h$ to get its best rank $r_h$ approximation $\frac{1}{N}\upca_{0,h}\bm D_{0,h} (\vpca_{0,h})^\top$ for $h\in [m]$.
\item Compute $\wh{\bm D}_h$,  $\hslmp_{0,h},\hslvp_{0,h},\hsrmp_{0,h}, \hsrvp_{0,h}$ for $h\in [m]$ as in \eqref{eq:est_sing_val}-\eqref{eq:scale-est} and $\hatl_\ell$ for $\ell \in [\wt m]$ as in \eqref{eq:sqrt_L}.
\item Compute empirical Bayes priors $\wh{\mu}$ and $\wh{\nu}_h$ for $h \in [m]$ as in \eqref{eq:emp_bayes_1} and \eqref{eq:emp_bayes_2}.
\item Set $\hsrm_{0,h} \leftarrow \hsrmp_{0,h}$, $\hsrv_{0,h} \leftarrow \hsrvp_{0,h}$, $\bm V_{0,h}=\vpca_{0,h}$, and $\bu_{-1,h}=\upca_{0,h}(\hsrv_{0,h})^{1/2}$ for $h \in [m]$.
\end{enumerate}
\For{$t=0,\ldots,T$} 
\Comment{Iterative refinement}
\For{$h=1,\ldots,m$}
\State $\bv_{t,h} \leftarrow v_{t,h}(\bm V_{t,h};\wh{\nu}_h)$;
\Comment{Update the estimate of $\bm{V}_h$}
\State $\bm U_{t,h} \leftarrow \wb{\bm X}_h\bv_{t,h}-\gamma_h\bu_{t-1,h}(\bm J^R_{t,h}(\bm V_{t,h};\wh\nu_h))^\top$; 
\State $\hslv_{t,h} \leftarrow \frac{1}{N}(\bv_{t,h})^\top\bv_{t,h}$ \quad \mbox{and} \quad $\hslm_{t,h} \leftarrow \hslv_{t,h} \wh{\bm D}_h$;
\State $\bu_{t,h} \leftarrow u_{t,h}(\bm U_{t,1},\ldots,\bm U_{t,m},\wt{\bm X}_1,\ldots,\wt{\bm X}_{\wt m};\wh{\mu})$;
\Comment{Update the estimate of $\bm{U}_h$ via integration}
\State $\bm V_{t+1,h} \leftarrow \wb{\bm X}^\top_h\bu_{t,h}-\bv_{t,h}(\bm J^L_{t,h}(\bm U_{t,1},\ldots,\bm U_{t,m},\ldm_1,\ldots,\ldm_{\wt m};\wh\mu))^\top$; 
\State $\hsrv_{t+1,h} \leftarrow \frac{1}{N}(\bu_{t,h})^\top\bu_{t,h}$ \quad \mbox{and} \quad $\hsrm_{t+1,h} \leftarrow \hsrv_{t+1,h}\wh{\bm D}_h$.
\EndFor
\For{$\ell=1,\ldots,\wt m$}
\State $\wt{\bm U}_{t,\ell} \leftarrow \wt{u}_{t,\ell}(\bm U_{t,1},\ldots,\bm U_{t,m},\wt{\bm X}_1,\ldots,\wt{\bm X}_{\wt m};\wh{\mu})$.
\Comment{Update the estimate of $\wt{\bm U}_\ell$ via integration}
\EndFor
\EndFor
\Output $\{\bu_{T,1},\ldots,\bu_{T,m},
\wt{\bm U}_{T,1},\ldots,\wt{\bm U}_{T,\wt m}\}$ and 
$\{\bv_{T,1},\ldots,\bv_{T,m}\}$.
\end{algorithmic}
\end{algorithm}

\begin{rem}
When different modalities have different signal ranks, 
there exists modality-specific information and such information must not be destroyed in the process of integration. 
Since our algorithm can handle different signal ranks for different modalities, it helps incorporating modality-specific information in the construction of the refined estimators $\bu_{t,h}$, $\bv_{t,h}$, and $\wt{\bm U}_{t,\ell}$.
\end{rem}

\subsection{Latent representation prediction set of query data}
\label{pred_0}

\subsubsection{Problem formulation}
We now consider predicting the latent representation of a query subject with partially observed data about it:
for a query subject, we observe $d \in [m]$ high-dimensional modalities and $\wt{d} \in [\wt m]$ low-dimensional modalities, and for each observed modality, potentially only a subset of modality-specific features are observed. 
Our goal is to construct a prediction set for the realized value of the query's latent factor representation $({U}_1^{\mathrm{Q}},\dots,{U}_m^{\mathrm{Q}}, \wt{U}_1^{\mathrm{Q}},\dots, \wt{U}_{\wt{m}}^{\mathrm{Q}})$ by leveraging the availability of $\{\hdm_1,\dots,\hdm_m, \ldm_1,\dots,\ldm_{\wt m}\}$ that will be referred to as the reference data in the following discussion.

To be precise, we observe $\{\bq_{h_1},\ldots,\bq_{h_d},\tq_{\ell_1},\ldots,\tq_{\ell_{\wt d}}\}$ on a query subject, where $\bq_{h_1},\ldots,\bq_{h_d}$ are from high-dimensional modalities $\{h_1,\ldots,h_d\} \subseteq [m]$ and $\tq_{\ell_1},\ldots,\tq_{\ell_{\wt d}}$ low-dimensional modalities $\{\ell_1,\ldots,\ell_{\wt d}\} \subseteq [\wt m]$.
For the $k$-th observed high-dimensional modality,
$\bq_{h_k} \in \R^{|\calf_{h_k}|}$ can be written as
\begin{equation}
	\label{eq:query-high}
  \bq_{h_k} =  \frac{1}{\sqrt N}(\bm V_{h_k})_{\calf_{h_k} *}\bm D_{h_k}U^{\mathrm Q}_{h_k}  + W^{\mathrm Q}_{h_k}, \quad \mbox{for $W^{\mathrm Q}_{h_k} \sim N_{|\calf_{h_k}|}(0,\bm I_{|\calf_{h_k}|})$.}
\end{equation}
In \eqref{eq:query-high}, $\calf_{h_k} \subseteq [p_{h_k}]$ collects indices of the observed features and $U^{\mathrm Q}_{h_k} \in \R^{r_{h_k}}$ denotes the query's realized subject effect in the modality. 
We assume that
the matrices $\bm D_{h_k}$ and $\bm V_{h_k}$ are shared by the query data $\bq_{h_k}$ and the reference data $\bm X_{h_k}$, for all $k \in [d]$, 
and that 
the noise vectors $\{W^{\mathrm Q}_{h_k}:k \in [d]\}$ are mutually independent and are independent of the reference data. 
Furthermore, we assume that the number of observed features is comparable to the total number of features in reference data for any observed high-dimensional modality in the query: 
for all $k \in [d]$, ${|\calf_{h_k}|}/{p_{h_k}} \rightarrow \lambda_{h_k} \in (0,1]$, as $p_{h_k} \rightarrow \infty$.
For the $k$-th observed low-dimensional modality,
$\tq_{\ell_k} \in \R^{|\wt\calf_{\ell_k}|}$ can be written as
\begin{equation} 
	\label{eq:query-low}
\tq_{\ell_k}=(\bm L_{\ell_k})_{\caltf_{\ell_k}*}\wt{U}^{\mathrm Q}_{\ell_k}+\wt{W}^{\mathrm Q}_{\ell_k}, \quad \mbox{for $\wt{W}^{\mathrm Q}_{\ell_k} \sim N_{|\wt\calf_{\ell_k}|}(0,\bm I_{|\wt\calf_{\ell_k}|})$.}
\end{equation}
In \eqref{eq:query-low}, $\wt\calf_{\ell_k} \subseteq [\wt r_{\ell_k}]$ denotes indices of the observed features in this modality, 
${\wt U}^{\rmq}_{\ell_k} \in \R^{\wt r_{\ell_k}}$ represents the query's realized subject effect in the modality, and the noise vector $\wt{W}^{\mathrm Q}_{\ell_k}$ is independent of that in other modalities and of those in reference data.
Finally, we assume that 
$$
(U^{\mathrm Q}_{h_1},\dots,U^{\mathrm Q}_{h_d},\wt{U}^{\mathrm Q}_{\ell_1},\dots,\wt{U}^{\mathrm Q}_{\ell_{\wt d}}) \sim \mu_{h_1,\ldots,h_d;\ell_1,\ldots,\ell_{\wt d}},
$$ 
where $\mu_{h_1,\ldots,h_d;\ell_1,\ldots,\ell_{\wt d}}$ is the marginal distribution of $(U_{h_1},\ldots,U_{h_d},\wt U_{\ell_1},\ldots,\wt U_{\ell_{\wt d}})$ when $(U_1,\ldots,U_m,\wt U_1,\ldots,\wt U_{\wt m}) \sim \mu$ that is the data generating prior for the reference datasets. 
 
\subsubsection{Prediction set construction}
We are to construct a prediction set in the form of an $\ell_2$ ball.
We start with point prediction of the latent subject effect vector of the query that serves as the center of the prediction set.
For the $k$-th observed high-dimensional modality in the query, we first estimate $(\bm V_{h_k})_{\calf_{h_k}*}$ by $(\bv_{T,h_k})_{\calf_{h_k}*}$ where $\bv_{T,h_k}$ is the estimator of ${\bm V}_{h_k}$ generated by Algorithm \ref{alg:orchamp}.
The point prediction of the subject effect of the query in the $k$-th observed high-dimensional modality is then obtained by least squares as
\begin{align}
\label{eq:point_est_cell_eff} \hatuq_{h_k}:=(\rcheck^\top_{h_k}\rcheck_{h_k})^{-1}\rcheck^\top_{h_k}\wb\bq_{h_k},
\end{align}
where $\wb\bq_{h_k}={1 \over \sqrt{N}}Q_{h_k}$ and for $\hd_{h_k}$ defined in \eqref{eq:est_sing_val},
$\rcheck_{h_k}:=\frac{1}{N}(\bv_{T,h_k})_{\calf_{h_k}*}\hd_{h_k} \in \R^{|\calf_{h_k}| \times r_{h_k}}$.

We shall later show that if $N$, the number of cells in the reference dataset, tends to infinity, then
\begin{equation}
\label{eq:hat_uq_def}
    \hatuq_{h_k} = U^{\mathrm Q}_{h_k} +\lambda^{-1/2}_{h_k}\revsns{\left(\hd^{-1}_{h_k}(\hslv_{T,h_k})^{-1}\hd^{-1}_{h_k}\right)^{1/2}}Z_{h_k}+o_p(1),
\end{equation}
where $Z_{h_k} \sim N_{r_{h_k}}(0,\bm I_{r_{h_k}})$ and 
$\hslv_{T,h_k}=\frac{1}{N}(\bv_{T,h_k})^\top\bv_{T,h_k}$
.
To further integrate with other observed modalities, 
let 
$r^{\rmq} = \sum_{k=1}^{d}r_{h_k}$ and $\wt{r}^{\rmq} = \sum_{k=1}^{\wt d}|\wt {\calf}_{\ell_k}|$.
In the same spirit as \eqref{eq:opt_denoisers} and \eqref{eq:theta_t_k}, we define denoisers $u^{\rmq}_{h}:\mathbb R^{r^\rmq+\wt r^\rmq}\rightarrow \mathbb R^{r_{h}}$ and $\wt u^{\rmq}_{\ell}:\mathbb R^{r^\rmq+\wt r^\rmq}\rightarrow \mathbb R^{\wt r_{\ell}}$ for $h \in [m]$, $\ell \in [\wt m]$, and
any 
$x=(x_{h_1},\ldots,x_{h_d},\wt x_{\ell_1},\ldots,\wt x_{\ell_{\wt d}}) \in \R^{r^\rmq+\wt r^\rmq}$ as
\begin{align}
\label{eq:denoisers_query_hd}
    u^{\rmq}_{h}(x;\wh \mu) &= \mathbb{E}_{\wh \mu}\left[U_{h} \,\Big|\,Y^{\rmq}_{h_1}=x_{h_1},\ldots,Y^{\rmq}_{h_d}=x_{h_d},\wt Y^{\rmq}_{\ell_1}=\wt x_{\ell_1},\ldots,\wt Y^{\rmq}_{\ell_{\wt d}}=\wt x_{\ell_{\wt d}} \right],\\
\label{eq:denoisers_query_ld}
    \wt u^{\rmq}_{\ell}(x;\wh \mu) &= \mathbb{E}_{\wh \mu}\left[\wt U_{\ell} \,\Big|\,Y^{\rmq}_{h_1}=x_{h_1},\ldots,Y^{\rmq}_{h_d}=x_{h_d},\wt Y^{\rmq}_{\ell_1}=\wt x_{\ell_1},\ldots,\wt Y^{\rmq}_{\ell_{\wt d}}=\wt x_{\ell_{\wt d}} \right].
\end{align}
In \eqref{eq:denoisers_query_hd}--\eqref{eq:denoisers_query_ld}, the sequence of priors $\wh{\mu}$ weakly converges to $\mu$ as $N \rightarrow \infty$.
For $k \in [d]$ and $\wt k \in [\wt d]$, 
\begin{equation}
\label{eq:pred_y_q}
Y^{\rmq}_{h_k} \sim N_{r_{h_k}}(U^\rmq_{h_k},\,\lambda^{-1}_{h_k}\hd^{-1}_{h_k}(\hslv_{T,h_k})^{-1}\hd^{-1}_{h_k}), \quad  
\wt Y^{\rmq}_{\ell_{\wt k}} \sim N_{|\wt{\calf}_{\ell_{\wt k}}|}((\hatl_{\ell_{\wt k}})_{\wt{\calf}_{\ell_{\wt k}}*}\wt U^\rmq_{\ell_{\wt k}},\,\bm I_{|\wt{\calf}_{\ell_{\wt k}}|}),
\end{equation}
where $(U^\rmq_{h_1},\ldots,U^\rmq_{h_d},\wt U^\rmq_{\ell_1},\ldots,\wt U^\rmq_{\ell_{\wt d}}) \sim \wh \mu_{ h_1,\ldots,h_d;\ell_1,\ldots,\ell_{\wt d}}$
with $\wh \mu_{h_1,\ldots,h_d;\ell_1,\ldots,\ell_{\wt d}}$ being the marginal distribution of 
$(U^\rmq_{1},\ldots,U^\rmq_{m},\wt U^\rmq_{1},\ldots,\wt U^\rmq_{\wt m}) \sim \wh \mu$.
In addition, we assume that the random vectors $(U^\rmq_{h_1},\ldots,U^\rmq_{h_d},\wt U^\rmq_{\ell_1},\ldots,\wt U^\rmq_{\ell_{\wt d}})$ are independent of $\{\wh{\bm D}_{h_k},\wh{\bm \Sigma}^L_{T,h_k},\wh{\bm L}_{\ell_{\wt k}}:k \in [d],\; \mbox{and}\; \wt k \in [\wt d]\}$ and that the collection of variables 
\[
(Y^{\rmq}_{h_1},\ldots,Y^{\rmq}_{h_d},\wt Y^{\rmq}_{\ell_1},\ldots,\wt Y^{\rmq}_{\ell_{\wt d}})
\]
 are mutually independent given 
 \[
 (U^\rmq_{h_1},\ldots,U^\rmq_{h_d},\wt U^\rmq_{\ell_1},\ldots,\wt U^\rmq_{\ell_{\wt d}}),
 \]
  $\{\wh{\bm D}_{h_k},\wh{\bm \Sigma}^L_{T,h_k}:k \in [d]\}$, and $\{\wh{\bm L}_{\ell_{\wt k}}: \wt k \in [\wt d]\}$.
Finally, the point predictor for the subject effects of the query subject is given by:
\begin{equation}
	\label{eq:point-pred}
\mobs:=
   \big[(\hatum_{1})^\top ~ \ldots ~ (\hatum_{m})^\top ~ (\hatul_{1})^\top ~ \ldots ~ (\hatul_{\wt m})^\top
   \big]^\top
\end{equation}
where, for all $h \in [m]$ and $\ell \in [\wt m]$,
\[
\hatum_{h} = u^{\rmq}_{h}((\hatuq_{h_1},\ldots,\hatuq_{h_d},\tq_{\ell_1},\ldots,\tq_{\ell_{\wt d}});\wh \mu)\]
and
\[\hatul_{\ell} = \wt u^{\rmq}_{\ell}((\hatuq_{h_1},\ldots,\hatuq_{h_d},\tq_{\ell_1},\ldots,\tq_{\ell_{\wt d}});\wh \mu).\]

Our $100\times(1-\alpha)\%$ prediction set for the query's latent representation vector is a ball centered at $\mobs$. 
To determine its radius, let $Y^{\rmq}=(Y^{\rmq}_{h_1},\ldots,Y^{\rmq}_{h_d}) \in \R^{r^\rmq}$ and $\wt Y^{\rmq}=(Y^{\rmq}_{\ell_1},\ldots,\wt Y^{\rmq}_{\ell_{\wt d}}) \in \R^{\wt{r}^\rmq}$ with components defined as in \eqref{eq:pred_y_q}.
Now for
\begin{align}
\label{eq:observed_predictors}
y^{\rmq} = (\hatuq_{h_1},\ldots,\hatuq_{h_d} )\in \R^{r^\rmq} \quad \mbox{and} \quad \wt y^{\rmq} = (\tq_{\ell_1},\ldots,\tq_{\ell_{\wt d}})\in \R^{\wt{r}^\rmq},
\end{align}
define the prediction ball radius $\wh b_{\alpha}>0$ for $\alpha \in (0,1)$ as 
\begin{align}
\label{eq:rad_pred_set}
    \wh b_{\alpha}&=\inf\bigg\{b>0: \mathbb P_{\wh \mu}\left[(U^\rmq_1,\ldots,U^\rmq_m,\wt{U}^\rmq_1,\ldots,\wt{U}^\rmq_{\wt m}) \in B_{r+\wt r}(\mobs,b)\,\Big|\,Y^{\rmq}=y^{\rmq},\wt Y^{\rmq}=\wt y^{\rmq}\right]\nonumber\\
    &\hskip 30em \ge 1-\alpha\bigg\},
\end{align}
where $(U^\rmq_{1},\ldots,U^\rmq_{m},\wt U^\rmq_{1},\ldots,\wt U^\rmq_{\wt m}) \sim \wh \mu$. 
Thus, the $100\times(1-\alpha)\%$ prediction set is given by
\begin{align}
\label{eq:full_pred_set}
  \calc_{\alpha} & =B_{r+\wt r}(\mobs,\radobs_{\alpha}).
\end{align}
The construction of the preceding prediction set is summarized in Algorithm \ref{alg:pred}.

\begin{rem}
    The radius $\wh{b}_\alpha$ can approximated by simulating a large number of samples from the conditional distribution described in \eqref{eq:rad_pred_set} and using the $(1-\alpha)$-th sample quantile of the Euclidean distance between the sampled points and $\wh U$.
\end{rem}

\begin{algorithm}[tb]
\caption{Prediction Set of the subject effects for the Query Subjects}\label{alg:pred}
\begin{algorithmic}[1]
\Require 
\begin{enumerate}
	\item A query data entry with observed high-dimensional modalities $Q_{h_k} \in \R^{|\calf_{h_k}|}$ with observed features indexed by $\calf_{h_k}$  for $k \in [d]$ and observed low-dimensional modalities $\wt{Q}_{\ell_l} \in \R^{|\wt{\calf}_{\ell_l}|}$ with observed features indexed by $\wt{\calf}_{\ell_l}$ for $l\in [\wt d]$; 
	\item Matrices $\{\bv_{T,h_1},\ldots,\bv_{T,h_d}\}$, $\{\wh{\bm D}_{h}:h\in [m]\}$,  and $\{\wh{\bm L}_{\ell}:\ell \in [\wt m]\}$ estimated from a reference dataset with $N$ subjects according to Algorithm \ref{alg:orchamp}, \eqref{eq:est_sing_val}, and \eqref{eq:sqrt_L}, respectively;
	\item A sequence of priors $\wh \mu$ converging to $\mu$ weakly;
	\item Confidence level $100\times (1-\alpha)\%$.
\end{enumerate}
\State 
Compute $\rcheck_{h_k}:=\frac{1}{N}(\bv_{T,h_k})_{\calf_{h_k}*}\hd_{h_k}$ for $k \in [d]$.
\State 
Compute $\hatuq_{h_k}:=(\rcheck^\top_{h_k}\rcheck_{h_k})^{-1}\rcheck^\top_{h_k}\wb\bq_{h_k}$ for $k \in [d]$, where $\wb\bq_{h_k}=\bq_{h_k}/\sqrt{N}$.
\For{$h=1,\ldots,m$}
\State {$\hatum_{h} \leftarrow u^{\rmq}_{h}((\hatuq_{h_1},\ldots,\hatuq_{h_d},\tq_{\ell_1},\ldots,\tq_{\ell_{\wt d}});\wh \mu)$.}
\EndFor
\For{$\ell=1,\ldots,\wt m$}
\State {$\hatul_{\ell} \leftarrow \wt u^{\rmq}_{\ell}((\hatuq_{h_1},\ldots,\hatuq_{h_d},\tq_{\ell_1},\ldots,\tq_{\ell_{\wt d}});\wh \mu)$.}
\EndFor
\State Construct $\wh U$ as in \eqref{eq:point-pred} and $\wh b_\alpha$ as in \eqref{eq:rad_pred_set};
\Output The prediction set $\calc_{\alpha} = B_{r+\wt{r}}({\wh U}, {\wh b_\alpha})$.
\end{algorithmic}
\end{algorithm}

\section{Numerical Experiments}
\revzm{In this section, we present numerical experiments on synthetic data to benchmark the performance of Algorithm \ref{alg:orchamp} in signal recovery and the finite-sample coverage of the prediction set in Algorithm \ref{alg:pred}.}

\subsection{\revzm{Integrative signal recovery}}
\label{sec:eff_data_integration}

\revzm{We considered a simulation setting consisting of two high-dimensional modalities $\bm X_i\in \mathbb{R}^{N\times p_i}$, $i=1,2$, and one low-dimensional modality $\wt{\bm X}_1\in \R^{N\times \wt r}$, generated according to the models \eqref{eq:def_matrix_prop_asymp} and \eqref{eq:def_matrix_low_dim_asymp}. 
We had let $r_1 = 2$,  $r_2 = 3$,
$\wt r = 2$, $p_1 = \lfloor\gamma_1 N\rfloor$ and $p_2 = \lfloor\gamma_2 N\rfloor$ for $N= 10000$. 
We fixed $\gamma_1 = \gamma_2 = 0.25$ for this subsection. 
Our primary interest here was to benchmark whether Algorithm \ref{alg:orchamp} leverages the inter-modality dependence of signal components to achieve better signal recovery.} 


\revzm{To this end, we let the rows of $\bm U_1 \in \R^{N\times 2}$ be sampled as
$(\bm U_1)_{i*} \overset{\text{i.i.d.}}{\sim} \sqrt{\tfrac{2}{3}}(T_{1i} + T_{2i})$,
where $T_{1i} \overset{\text{i.i.d.}}{\sim} \mathrm{Unif}(\mathbb{S}^1)$ and $T_{2i} \overset{\text{i.i.d.}}{\sim} \mathrm{Unif}(\{-1, 1\}^2)$ were mutually independent.}
\revzm{The first column of $\bm U_2$ was generated independently by
$(\bm U_2)_{i,1} \overset{\text{i.i.d.}}{\sim} \mathrm{Laplace}\left(0, \tfrac{1}{\sqrt{2}}\right)$.}
\revsn{The \revzm{next} two columns of $\bm U_2$ were then generated via nonlinear transformations of the columns of $\bm U_1$. 
Each entry in the \revzm{second} column of ${\bm U}_2$ was generated as
\[
(\bm U_2)_{i2} = 
\rho \cdot 
\frac{\max(0, (\bm U_1)_{i2}) - \mathbb{E}[\max(0, (\bm U_1)_{i2})]}{\sqrt{\operatorname{Var}[\max(0, (\bm U_1)_{i2})]}}
+ \sqrt{1 - \rho^2} \cdot B_i,  
\]
\revzm{where $B_i  \overset{\text{i.i.d.}}{\sim}\mathrm{Rademacher}(\pm 1)$.}
Each entry in the third column of ${\bm U}_2$ was generated as
\[
\begin{aligned}
(\bm U_2)_{i3} &= 
\rho \cdot \left\{ \frac{\mathsf{h}((\bm U_1)_{i1}) - \mathbb{E}[\mathsf{h}((\bm U_1)_{i1})]}{\sqrt{\mathrm{Var}[\mathsf{h}((\bm U_1)_{i1})]}} \right\} + \sqrt{1 - \rho^2} \cdot Z_i,
\end{aligned}
\]
where $Z_i \overset{\text{i.i.d.}}{\sim} {N}(0,1)$ and $\mathsf{h}(x):= \mathrm{sgn}(x)\min\{|x|,0.25\}
$.
\revzm{Moreover, $B_i$'s, $Z_i$'s and $\bm U_1$ were mutually independent.} 
}
\revsn{The signal $\wt{\bm U}_1$ \revzm{in $\wt{\bm X}_1$} was constructed \revzm{as}
\[
(\wt{\bm U}_1)_{i*} = \rho \cdot (\bm U_1)_{i*} + \sqrt{1 - \rho^2} \cdot T_{3i}, \quad T_{3i} \overset{\text{i.i.d.}}{\sim} \mathrm{Unif}(\{-1,1\}^2),
\]
where $T_{3i}$'s were independent of $\bm U_1$, $Z_i$'s, and $B_i$'s.}
\revsn{The loading matrices $\bm V_1 \in \R^{p_1 \times r_1}$ and $\bm V_2 \in \R^{p_2 \times r_2}$ were independently generated row-wise as:
$(\bm V_h)_{i*} \overset{\text{i.i.d.}}{\sim} \sqrt{r_h} \cdot \mathrm{Unif}(\mathbb{S}^{r_h - 1})$, 
for
$h \in \{1,2\}$.}
\revsn{The \revzm{diagonal} signal strength matrices in high-dimensional modalities were 
defined \revzm{by}
\begin{align}
\label{eq:diag_entry}
(\bm D_1)_{ii} = \lambda+3-i, \quad i = 1,2; 
\qquad (\bm D_2)_{ii} = 2 \times \revzm{(4-i)},
\quad i = 1,\dots,\revzm{3}.
\end{align}
\revzm{Finally, we set $\bm L_1 = 5 \bm I_2 + 2.5 \bm 1_2 \bm 1_2^\top$ in the low-dimensional modality.
In the foregoing definitions, the hyperparameter $\rho \in [0,1]$ governs the level of inter-modality dependence of the signals, where larger $\rho$ enforces stronger dependence.
The hyperparameter $\lambda > 0$ alters signal strength in the second high-dimensional modality.}}


\revzm{We benchmarked the performance of {\orchAMP} against two others approaches:}

\begin{itemize}
    
    \item \revsn{{EB-PCA}~\cite{eb_pca}, which applies 
    denoising \revzm{based on empirical priors}, but separately 
    \revzm{within}
    each modality without information sharing.}
    
    \item \revsn{{SVD}, which applies \revzm{neither denoising nor cross-modal integration}. 
    Latent factors $\bm U_h$ are \revzm{estimated by}
    top left singular vectors $\upca_{0,h}$ for $h \in \{1,2\}$, and 
    $\wt{\bm U}_1$ is estimated via least-squares projection $\wt{\bm X}_1 \wh{\bm L}_1(\wh{\bm L}_1^\top \wh{\bm L}_1)^{-1}$, where $\wh{\bm L}_1$ is defined in \eqref{eq:sqrt_L}.}
\end{itemize}

\revsn{All methods were applied to the triplet 
$(\bm X_1, \bm X_2, \wt{\bm X}_{\revzm{1}})$ \revzm{with the knowledge of the true signal ranks}.
\revzm{For any estimator triplet $(\wb{\bm U}_1,\wb{\bm U}_2,\widecheck{\bm U}_1)$,}
\revzm{reconstruction accuracies were} 
quantified by the following normalized \revzm{squared} Frobenius loss:
\begin{align*}
    \ell(\bm{U}_h, \wb{\bm U}_h)
    &= \frac{1}{N^2 r_h} \left\| \wb{\bm U}_h \wb{\bm U}_h^\top - \bm U_h \bm U_h^\top \right\|_F^2, 
    \, h \in \{1,2\};\,
    \ell(\wt{\bm U}_1, \widecheck{\bm U}_1)
    =
    \frac{1}{N^2 \wt r} \left\| \widecheck{\bm U}_1 \widecheck{\bm U}_1^\top - \wt{\bm U}_1 \wt{\bm U}_1^\top \right\|_F^2.
\end{align*}
The relative performance of the three methods \revzm{are}
expected to depend on both the signal strength $\lambda$ and the strength of 
\revzm{inter-modality}
dependence $\rho$,
\revzm{and we varied these simulation hyperparameters}
over the following grid \revzm{to demonstrate such a dependence}
\[
\lambda \in \{2.0, 2.4, 2.8\}, \quad \rho \in \{0.8, 0.9, 1.0\}.
\]}
\revzm{In both {\orchAMP} and EB-PCA, we restricted each estimated prior to be a Gaussian mixture model (GMM) with four components and considered the results after $25$ AMP iterations.}

\revzm{Table \ref{tab:modality_lambda_by_rho} reported average losses and the associated standard errors over $50$ replicates for each $(\lambda, \rho)$ pair.
By Table \ref{tab:modality_lambda_by_rho}, both {\orchAMP} and EB-PCA outperformed SVD across all $(\lambda,\rho)$ configurations.
Moreover, the additional performance gain by {\orchAMP} was more pronounced for smaller values of $\lambda$, which is aligned with the intuition that borrowing strength from other modalities is most helpful for modalities with relatively low signal-to-noise ratios (SNRs). 
Furthermore, 
the performance of {\orchAMP} was not monotone in $\lambda$ for the reconstruction of latent factors from Modality 2, since the difficulty in prior estimation is simultaneously affected by $\lambda$.
\revsn{In summary, these experiments demonstrate that \orchAMP~is particularly effective when the latent components of interest have weak SNRs and cannot be reliably recovered from individual modalities alone. In such cases, the use of information across the modalities allows for a more accurate recovery.}
Within each modality, when the SNR increases, the performance of {\orchAMP} may be surpassed by that of EB-PCA as {\orchAMP} needs to estimate a higher-dimensional joint prior on latent factors, while EB-PCA only focuses on its marginal within each modality. This is the price that one pays for retaining the capacity of cross-modal querying, which EB-PCA does not offer.}


\begin{table}[tbh]
\centering
\scriptsize
\setlength{\tabcolsep}{3pt}
\renewcommand{\arraystretch}{0.9}
\begin{tabular}{@{}llccc@{}}
\toprule
\multicolumn{5}{l}{\textbf{Modality 1}}\\
\midrule
 & Method & $\rho=0.8$ & $\rho=0.9$ & $\rho=1$ \\
\midrule
\multirow{3}{*}{\(\lambda=2\)} & OrchAMP &   \textbf{0.442344 (0.008088)} & \textbf{0.374381 (0.000860)} & \textbf{0.425845 (0.005339)} \\
& EBPCA & 0.460190 (0.000892) & 0.460190 (0.000892) & 0.460190 (0.000892) \\
& SVD & 0.555038 (0.000874) & 0.555038 (0.000874) & 0.555038 (0.000874) \\
\addlinespace[2pt]
\multirow{3}{*}{\(\lambda=2.4\)} & OrchAMP & \textbf{0.386132 (0.006364)} & \textbf{0.335412 (0.000896)} & \textbf{0.381269 (0.005446)} \\
& EBPCA & 0.393059 (0.000747) & 0.393059 (0.000747) & 0.393059 (0.000747) \\
& SVD & 0.459870 (0.000730) & 0.459870 (0.000730) & 0.459870 (0.000730) \\
\addlinespace[2pt]
\multirow{3}{*}{\(\lambda=2.8\)} & OrchAMP & 0.350279 (0.006589) & \textbf{0.301414 (0.000758)} & 0.345909 (0.005424) \\
& EBPCA &  \textbf{0.338840 (0.000702)} & 0.338840 (0.000702) & \textbf{0.338840 (0.000702)}  \\
& SVD &  0.386666 (0.000616) & 0.386666 (0.000616) & 0.386666 (0.000616) \\
\addlinespace[2pt]

\midrule
\multicolumn{5}{l}{\textbf{Modality 2}}\\
\midrule
 & Method & $\rho=0.8$ & $\rho=0.9$ & $\rho=1$ \\
\midrule
\multirow{3}{*}{\(\lambda=2\)} & OrchAMP & 0.519971 (0.008662) & \textbf{0.482710 (0.002066)} & \textbf{0.485770 (0.005907)} \\
& EBPCA &  \textbf{0.496249 (0.001140)} & 0.505828 (0.001081) & 0.513239 (0.001150) \\
&SVD & 0.627213 (0.001231) & 0.627244 (0.001248) & 0.627341 (0.001284)  \\
\addlinespace[2pt]
\multirow{3}{*}{\(\lambda=2.4\)}& OrchAMP & 0.506314 (0.007229) & \textbf{0.483782 (0.002131)} & \textbf{0.479342 (0.005626)} \\
&EBPCA & \textbf{0.496249 (0.001140)} & 0.505828 (0.001081) & 0.513239 (0.001150) \\
&SVD & 0.627213 (0.001231) & 0.627244 (0.001248) & 0.627341 (0.001284) \\
\addlinespace[2pt]
\multirow{3}{*}{\(\lambda=2.8\)}& OrchAMP & 0.511756 (0.007916) & \textbf{0.485434 (0.002248)} & \textbf{0.475111 (0.005041)} \\
&EBPCA & \textbf{0.496249 (0.001140)} & 0.505828 (0.001081) & 0.513239 (0.001150) \\
&SVD & 0.627213 (0.001231) & 0.627244 (0.001248) & 0.627341 (0.001284) \\
\addlinespace[2pt]
\midrule
\multicolumn{5}{l}{\textbf{Modality 3}}\\
\midrule
 & Method & $\rho=0.8$ & $\rho=0.9$ & $\rho=1$ \\
\midrule
\multirow{3}{*}{\(\lambda=2\)} & OrchAMP & \textbf{0.048622 (0.000191)} & \textbf{0.048131 (0.000102)} & \textbf{0.050549 (0.000163)} \\
& EBPCA & 0.049676 (0.000144) & 0.049513 (0.000113) & 0.051961 (0.000205)  \\
& SVD & 1.854730 (0.009864) & 1.852470 (0.009474) & 1.851095 (0.008641)  \\
\addlinespace[2pt]
\multirow{3}{*}{\(\lambda=2.4\)} & OrchAMP  &  \textbf{0.048314 (0.000158)} & \textbf{0.048053 (0.000100)} & \textbf{0.050537 (0.000169)} \\
& EBPCA & 0.049676 (0.000144) & 0.049513 (0.000113) & 0.051961 (0.000205) \\
& SVD & 1.854730 (0.009864) & 1.852470 (0.009474) & 1.851095 (0.008641)  \\
\addlinespace[2pt]
\multirow{3}{*}{\(\lambda=2.8\)} & OrchAMP  & \textbf{0.048334 (0.000175)} & \textbf{0.047973 (0.000100)} & \textbf{0.050593 (0.000182)}  \\
& EBPCA & 0.049676 (0.000144) & 0.049513 (0.000113) & 0.051961 (0.000205)  \\
& SVD & 1.854730 (0.009864) & 1.852470 (0.009474) & 1.851095 (0.008641) \\
\addlinespace[2pt]
\bottomrule
\end{tabular}
\caption{\revzm{Comparison of average within modality reconstruction errors (standard errors in parentheses) 
for 
OrchAMP, EBPCA, and SVD over $50$ replications: $\gamma_1= \gamma_2 =0.25$. 
The lowest average (i.e.~best performance) within each modality is given in boldface for each $(\lambda,\rho)$ combination.} 
}
\label{tab:modality_lambda_by_rho}
\end{table}

\subsection{Finite-sample 
\revzms{coverage}
of prediction sets} 
\revsn{Next we investigated the finite-sample coverage 
of the prediction sets constructed via Algorithm~\ref{alg:pred}. 
In this experiment, we considered two high-dimensional data matrices $\bm X_1 \in \mathbb{R}^{N \times p_1}$ and $\bm X_2 \in \mathbb{R}^{N \times p_2}$ generated according to the data generation scheme described in Section~\ref{sec:eff_data_integration} (but without the low-dimensional modality), with ambient dimensions $p_1 = p_2 = \lfloor0.25N\rfloor$.
\revzm{We fixed $\rho = 0.8$, and varied the sample size $N$ and signal strength parameter $\lambda$ as
\[
N \in \{2000, 3000, 4000\}, \quad \lambda \in \{5, 7, 9\}.
\]}}
\revzm{In each experiment, after obtaining the {\orchAMP} estimators on the training data, we queried them with ${\bm X}_1^\mathrm{Q}$ (i.e., modality $1$) of an independent query sample generated from the same distribution. 
In particular, we supplied ${\bm X}_1^\mathrm{Q}$ and the {\orchAMP} estimators to Algorithm~\ref{alg:pred} to compute a $95\%$ prediction set for the concatenated latent embedding $(U^\rmq_1, U^\rmq_2) \in \mathbb{R}^5$ and tested if it contained the true latent embedding of the query sample. 
\revsn{To estimate the radius $\wh b_\alpha$ of the prediction set, we simulated $10000$ samples from the posterior distribution specified in~\eqref{eq:rad_pred_set} and took the $0.95$ quantile of the corresponding Euclidean norms.} 
All other implementation details were identical to those used in Section \ref{sec:eff_data_integration}.}

\revsn{
For each $(N, \lambda)$ pair, we repeated the procedure described above over $50$ independent replicates.
\revzm{Table~\ref{tab:pred_calibrate} reports the proportions of replicates}
in which the true latent embedding $(U^\rmq_1, U^\rmq_2)$ fell within the constructed prediction set.
\revzm{Across all configuration, the coverage was close to the nominal $95\%$ level,} demonstrating that our prediction sets were reasonably well-calibrated even in finite samples.}
\revzm{It is possible that  conformal prediction techniques could be adapted to the current context for achieving exact finite-sample calibration, and we leave this interesting problem for future research.} 

\begin{table}[tbh]
\centering
\scriptsize
\setlength{\tabcolsep}{3pt}
\renewcommand{\arraystretch}{0.95}
    \begin{tabular}{c c c c}
    \toprule
     & $\lambda = 5$ & $\lambda = 7$ &  $\lambda = 9$ \\
    \midrule
    $N = 2000$ & 0.94 & 0.92 & 0.92 \\
    $N = 3000$ & 0.94 & 0.94 & 0.92 \\
    $N = 4000$ & 0.94 & 0.96 & 0.98 \\
    \toprule
    \end{tabular}
    \caption{Empirical coverage probabilities of the constructed $95\%$ prediction sets.}
    \label{tab:pred_calibrate}
\end{table}

\section{A \revzms{single-cell multi-omics} data example}

In this section, we illustrate on a real data example the prowess of Algorithm \ref{alg:orchamp} for constructing cell atlas with single-cell multi-omics data and of Algorithm \ref{alg:pred} for probabilistic querying of the constructed atlas with uncertainty quantification where the query cell is only measured in one modality.

We focused on
a human peripheral blood mononuclear cell (PBMC) single-cell TEA-seq data 
\cite{10.7554/eLife.63632},
\revzms{which}
comprised three modalities: chromatin accessibility (\revzms{the} ATAC modality) information for the whole genome (around $66828$ read pairs), gene expression levels (\revzms{the} RNA modality) for the whole genome (around $36601$ transcripts), and abundance levels of $48$ surface proteins (\revzms{the} Protein modality), for $8213$ PBMC cells. 
After pre-processing, we retained
a subset of $6323$ cells with high data quality. 
In addition, after feature screening, we selected $5000$ highly variable read pairs for the ATAC modality, $2000$ highly variable genes for the RNA modality, \revzms{and $40$ highly variable proteins for the Protein modality}.
To conform with model \eqref{eq:def_matrix_prop_asymp}-\eqref{eq:def_matrix_low_dim_asymp}, we applied the $\log 1p$ transform $(x \mapsto \log(1+x))$ entry-wise and subsequently centered each column (feature) in each modality\footnote{\revsn{The $\log 1p$ transformation maps the (total-counts-per-cell-normalized) discrete count measurements to a continuous scale and, by stabilizing the variance, mitigates the strong mean--variance dependence and skewness. 
This step is standard in the pre-processing of single-cell RNA and ATAC data.
The \revzm{post-transform} variation around the latent signal is more homogeneous across features, making \eqref{eq:def_matrix_prop_asymp}-\eqref{eq:def_matrix_low_dim_asymp} more reasonable approximation to data.
}}.
Thus, the three modalities were represented by two high-dimensional matrices $\bm X_1 \in \R^{6323 \times 5000}$ for ATAC and $\bm X_2 \in \R^{6323 \times 2000}$ for RNA, and one low-dimensional matrix $\bm {\wt X}_1 \in \R^{6323 \times 40}$ for Protein.

For each cell, we regarded its original cell type annotation in \cite{10.7554/eLife.63632} as the ground-truth \revzms{for evaluation purpose}.
\revzms{Neither Algorithm \ref{alg:orchamp} or \ref{alg:pred} had access to these annotations.}

\subsection{Cell atlas construction}

Based on the comparison of the empirical singular value distribution with the square root of the (properly rescaled) Marchenko-Pastur distribution,
we set $r_1=15$ and $r_2=20$. 
\revzm{See Fig \ref{fig:sidebyside_scree_plot} for scree plots of these two modalities.}
We adopted the empirical Bayes technique for prior estimation with a Gaussian mixture class of priors (described in Appendix J.1) with 4 mixing components for estimating $\mu$, 11 mixing components to estimate $\nu_1$, and 6 mixing components to estimate $\nu_2$\footnote{ 
\revsn{We selected the number of mixture components by Bayesian Information Criterion (BIC) \cite{schwarz1978bic}. Specifically, for the latent‐factor prior $\mu$ we fit GMMs with $K\in\{2,\ldots,\lfloor \sqrt[3]{N}\rfloor\}$, and for the loading priors $\nu_1$ and $\nu_2$ we used $K\in\{2,\ldots,\lfloor \sqrt[3]{p_1}\rfloor\}$ and $K\in \{2,\ldots,\lfloor \sqrt[3]{p_2}\rfloor\}$, respectively, retaining the $K$ that minimized BIC for each prior.}}.
We ran Algorithm \ref{alg:orchamp} with $T = 15$ iterative updates and 
learned
embeddings 
\revzms{$\wb{\bm U}_{1} \in \R^{6323 \times 15}$, $\wb{\bm U}_{2} \in \R^{6323 \times 20}$, and $\wt{\bm U}_{1} \in \R^{6323 \times 40}$
for the ATAC, the RNA, and the Protein modalities, respectively.}
Concatenating these embeddings for each cell, we got the following matrix of integrated multimodal embeddings as the constructed cell atlas:
\[
\bm U = [\wb{\bm U}_{1} \quad \wb{\bm U}_{2} \quad \wt{\bm U}_{1} ] \in \R^{6323 \times 75}.
\]

\revsn{Figure~\ref{fig:atlas_tea_seq} compares UMAP embeddings of the atlases built by Algorithm~\ref{alg:orchamp} (middle) with those by Seurat WNN~\cite{HAO20213573} (left, \revzm{with 50 PCs for ATAC and RNA and 40 PCs for Protein for optimized performance in metrics reported in Table \ref{tab:clustering_metrics_tea_seq}},
\revzm{and those by }
MOFA+~\cite{argelaguet2020mofa+} \revzm{(right, with default tuning parameter choices)}.}
There appears to be no substantial difference \revzms{across the three atlas UMAPs in terms of separating the cell types}.


\begin{figure}[tb]
    \centering
    \includegraphics[scale = 0.28]{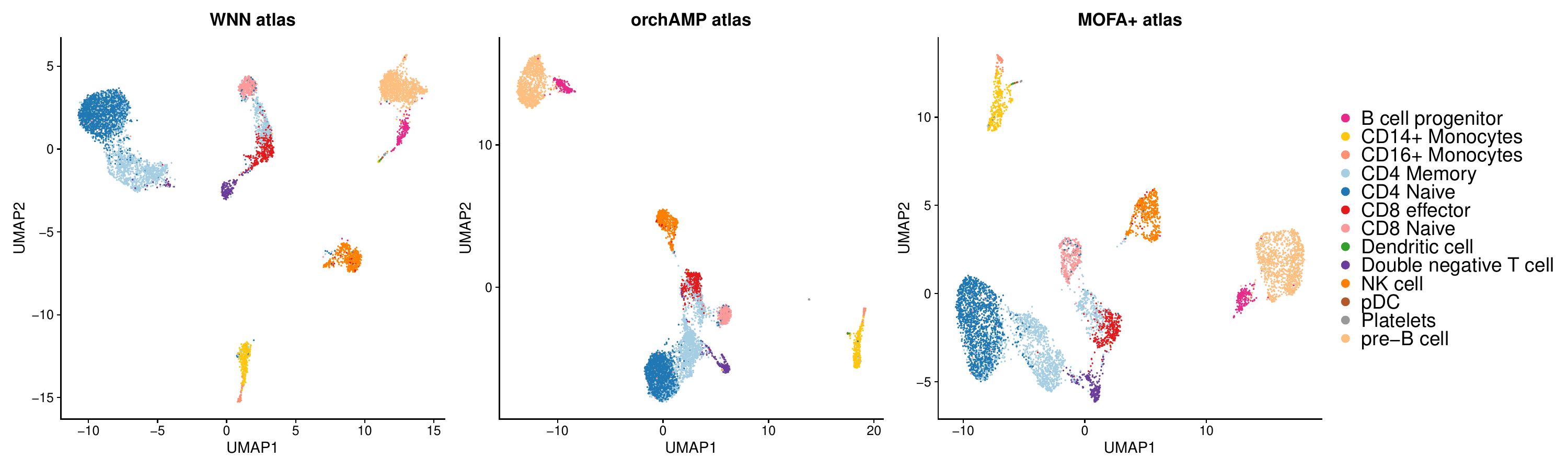}
    \caption{\revsn{UMAPs of PBMC atlases from $6323$ TEA-seq cells by Seurat WNN (left), OrchAMP (middle), and MOFA+ (right),
    colored according to original cell type annotations in \cite{10.7554/eLife.63632}.}}
    \label{fig:atlas_tea_seq}
\end{figure} 

\revsn{\revzm{In addition to the qualitative visual inspection enabled by UMAPs,}
we 
report quantitative 
metrics to \revzm{further evaluate and compare the capacity of these atlases in delineating cell types.}
Specifically, we compared \revzm{these atlases}
using four clustering metrics: the average silhouette score \cite{rousseeuw1987silhouettes}, the V-measure \cite{rosenberg2007vmeasure}, the Adjusted Rand Index \cite{HubertArabie1985}, and the celltype Local Inverse Simpson’s Index (cLISI) \cite{Korsunsky2019}. \footnote{\revsn{See Section \ref{sec:clusreing_metrics} for further description of the clustering metrics.}}
These metrics were computed \revzm{to compare the ground-truth cell type annotation with the cluster labels}
of 
the \revzm{learned representations of all cells by} 
each method \revzm{where the cluster labels were obtained from Louvain clustering \cite{blondel2008fast} of the 20-nearest-neighbor graph constructed from the respective learned representations}.
These 
metrics \revzm{collectively} offer an objective assessment of \revzm{how well clustering structure in the learned representations reflects ground-truth cell types}: higher values of the average silhouette score, Adjusted Rand Index, and V-measure indicate stronger separation of major cell types \revzm{in the resulting clusters}, whereas lower cLISI values reflect improved clustering fidelity.
The \revzm{metrics}, summarized in Table~\ref{tab:clustering_metrics_tea_seq}, \revzm{provided additional evidence showing} that OrchAMP achieves comparable performance \revzm{in cell atlas construction to state-of-the-art} methods while \revzm{being theoretically justifiable and} provides additional flexibility for downstream \revzm{querying} tasks.}
\revzm{See Section \ref{sec:bench_cite_seq} for an additional multi-omics example.}

\begin{table}[ht]
\centering
\scriptsize
\setlength{\tabcolsep}{3pt}
\renewcommand{\arraystretch}{0.95}
\begin{tabular}{cccc}
\toprule
\textbf{Clustering Metric} & \textbf{WNN} & \textbf{OrchAMP} & \textbf{MOFA+}\\
\midrule
Average Silhouette Score & 0.387  & {\bf 0.395}  & 0.394 \\
V-measure                & 0.732  & {\bf 0.762}  & 0.730  \\
Adjusted Rand Index                & 0.592  & {\bf 0.634}  & 0.533  \\
cLISI                     & 1.200  & 1.200  & {\bf 1.181} \\
\toprule
\end{tabular}
\caption{Comparison of clustering metrics among WNN, OrchAMP, and MOFA+ on the TEA-seq dataset. \revzms{The best value of each metric is given in boldface.}}
\label{tab:clustering_metrics_tea_seq}
\end{table}

\begin{figure}[t]
    \centering
    \includegraphics[scale = 0.4]{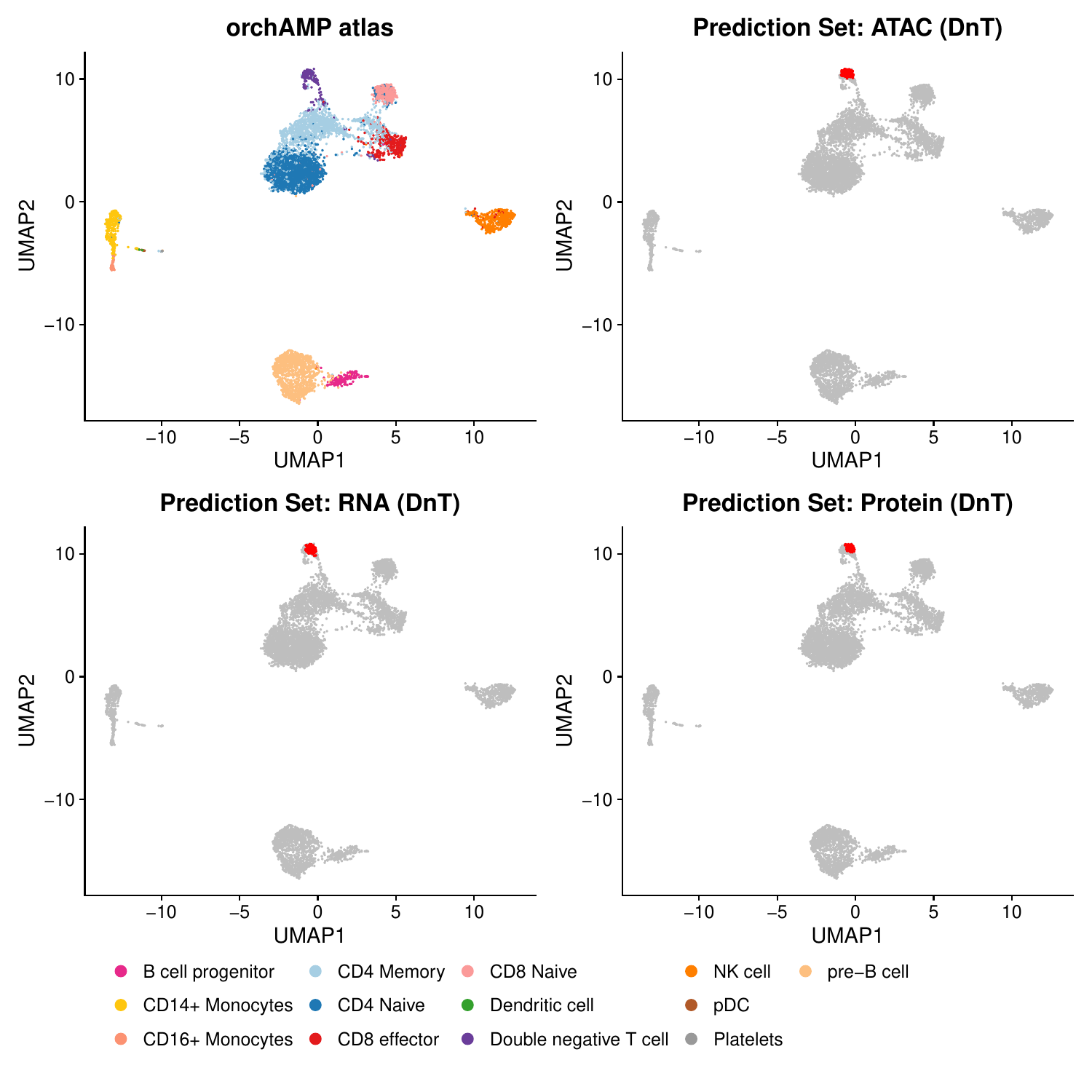}
    \caption{Visualization of prediction sets constructed by Algorithm \ref{alg:pred}. 
	Top left: OrchAMP cell atlas constructed 
    on $6320$ TEA-seq cells. Cells are colored according to their cell type annotations in \cite{10.7554/eLife.63632}.
	\revzms{Visualizations of the $95\%$ prediction set ($500$ randomly sampled points from the set in red and atlas cells in grey)} when queried by the ATAC observation of a held-out \textsf{double Negative T} cell \revzms{(top right), its RNA observation (bottom left), and its Protein observation (bottom right), respectively}.
    }
    \label{fig:atlas_tea_seq_pred}
\end{figure}

\subsection{Cell atlas querying}

We demonstrate the performance of Algorithm \ref{alg:pred} in constructing 
prediction sets for query cells' joint latent representations with a hold-out experiment.

First, we randomly held out one \textsf{CD8 effector} cell, one  \textsf{double negative T} cell, and one \textsf{pre-B} cell as test data and performed the same procedure as in the previous subsection on the remaining $6320$ cells to obtain the OrchAMP cell atlas (Figure \ref{fig:atlas_tea_seq_pred} top left).
Next,
we queried the atlas with the observation within each single modality (i.e. 
\revzms{ATAC, RNA, and Protein})
\revzms{respectively}
for each held-out cell, constructed the $95\%$ prediction set of the joint latent representation of the query cell using Algorithm \ref{alg:pred}.
\revzms{Finally, we visualized each prediction set by first sampling $500$ points in the set uniformly at random and then mapping
them} to the constructed atlas using the \texttt{MapQuery} function in the Seurat V4 \texttt{R} package. 

\revzm{For illustration, we visualized the prediction sets for the held-out \textsf{double negative T} cell when queried only with its ATAC, RNA, or Protein observations, respectively, in the top right, the bottom left, and the bottom right panels of Figure \ref{fig:atlas_tea_seq_pred}, respectively.}
\revsn{
Analogous plots for the other two held-out cells \revzm{are deferred to} the Section \ref{sec:app_tea_seq} \revzm{as the main message remains similar}. 
\revzm{With reference to cell type locations in the atlas in the top left panel, all three prediction sets were contained predominantly within the correct cell population. 
\revsn{Moreover, the prediction sets typically occupied only a subset of the full cell-type cluster, thereby conveying information beyond a coarse cell-type label.}
}} 

\section{Theoretical results}
\label{theo_res}

In this section, we study the asymptotic properties of the OrchAMP estimators constructed by Algorithm~\ref{alg:orchamp}, as well as the asymptotic coverage probabilities of the prediction sets obtained from Algorithm~\ref{alg:pred}. 

\revsn{Our theoretical analysis \revzm{of Algorithm \ref{alg:orchamp}}
builds upon \revzm{and further develops} the framework 
in \cite{eb_pca}: 
Assumption~5.1 \revzm{below}
generalizes Assumption~5.1 of \cite{eb_pca}, 
Proposition~5.1 extends Proposition~5.2 therein, 
and Theorem~5.1 \revzm{generalizes}
the conclusions of their Theorem~5.4 \revzm{to multi-modal settings}. 
While the proof strategies are \revzm{closely connected},
adapting them to \revzm{the multi-modal}
setting requires substantial 
work. 
Specifically, \citet{eb_pca} \revzm{analyzed} a single AMP orbit associated with \revzm{a single} data matrix, whereas \revzm{our} orchestrated framework involves multiple parallel AMP orbits, each corresponding to a different high-dimensional modality whose leading singular vectors are to be recovered. 
Therefore, the analysis must account for synchronous information integration \revzm{at each iterate} across these orbits via 
\revzm{respective} denoisers, a \revzm{signature in our analysis that was} absent in the single-orbit setting.} 

\revsn{Beyond \revzm{estimation}, we contribute new theoretical results on the asymptotic validity of the prediction sets constructed by Algorithm~\ref{alg:pred}. 
Establishing these \revzm{results} requires novel arguments based on empirical process techniques, which \revzm{is orthogonal to} the scope of \cite{eb_pca} and forms a distinct component of our \revzm{theory}.}

\subsection{Assumptions}
We begin with the following assumptions on the model \eqref{eq:def_matrix_prop_asymp}--\eqref{eq:def_matrix_low_dim_asymp}.

\begin{asm}
\label{asm:prior_1_mom}
The matrices $\{\bm U_h, \bm D_h, \bm V_h, \bm W_h: h \in [m]\}$, $\{\wt{\bm U}_\ell,\wt{\bm W}_\ell: \ell \in [\wt m]\}$ and the prior classes $\mathcal P$ and $\{\mathcal P_{\nu_h}: h \in [m]\}$ satisfy the following assumptions:
\begin{enumerate}
    \item All $\mu \in \mathcal P$ are absolutely continuous with respect to the Lebesgue measure and have bounded density.
    \item For all $\mu \in \mathcal P$, we have $\mathbb{E}_\mu[U_hU^\top_k]<\infty$, $\mathbb{E}_\mu[U_h\wt{U}^\top_\ell]<\infty$ and $\mathbb{E}_\mu[\wt{U}_\ell\wt{U}^\top_g]<\infty$, where $h,k \in [m]$; $\ell,g \in [\wt m]$ and $(U_1,\ldots,U_{m},\wt{U}_1,\ldots,\wt{U}_{\wt m}) \sim \mu$. Furthermore, $\mathbb E_\mu[U_hU^\top_h]=\bm I_{r_h}$ and $\mathbb E_\mu[\wt U_\ell\wt U^\top_\ell]=\bm I_{\wt r_\ell}$ for all $h \in [m]$ and $\ell \in [\wt m]$. 
    \item For all $h \in [m]$, $\nu_h \in \mathcal P_{\nu_h}$ satisfies $\mathbb{E}_{\nu_h}[V_{h}V^\top_{h}]=\bm I_{r_h}$, where $V_{h} \sim \nu_h$. 
    \item The collection of the matrices $\{\bm U_1,\ldots,\bm U_m,\wt{\bm U}_1,\ldots,\wt{\bm U}_{\wt m}\}$ is independent of the collection of the matrices $\{\bm V_1,\ldots,\bm V_m\}$.
	\item For all $h \in [m]$ and $k \in [r_h]$, the signal strengths $(\bm D_{h})_{kk}>\gamma^{-1/4}_h$. 
    \item For any collection of positive semi-definite matrices $\{\bm A_{h}, \bm B_{h} \in \mathbb{R}^{r_h \times r_h}: h \in [m]\}$, and  $\{\wt{\bm A}_\ell \in \mathbb{R}^{{\wt r}_\ell \times {\wt r}_\ell}: \ell \in [\wt m]\}$ and $\mu \in \mathcal{P}$, there exists an open neighborhood $O_\mu$ of $\mu$ in the topology of weak convergence, such that for
    \[
     Y^G_{h} \sim N_{r_h}(\bm A_{h}U_h,\bm B_{h}),\quad \wt{Y}^G_{\ell} \sim N_{{\wt r}_\ell}(\wt{\bm A}_\ell \wt U_\ell,\bm I_{{\wt r}_\ell}),\quad \mbox{and} \quad (U_1,\ldots,U_m,\wt{U}_1,\ldots,\wt{U}_{\wt m}) \sim \mu,
     \]
    the functions $u^G_{h}:\mathbb R^{r+\wt r} \rightarrow \R^{r_h}$ and $\wt{u}^G_{\ell}:\mathbb R^{r+\wt r} \rightarrow \R^{\wt r_\ell}$, given by 
 \begin{align}
 \label{eq:oracle_denoisers_g}
   &u^G_{h}(x_1,\ldots,x_m,\wt x_1,\ldots,\wt x_{\wt m};\mu)\\
   &=\mathbb E_\mu[U_h|Y^G_{1}=x_1,\ldots,Y^G_{m}=x_m,\wt{Y}^G_{1}=\wt x_1,\ldots,\wt{Y}^G_{\wt m}=\wt x_{\wt m}],  \quad \text{and}\\
 \label{eq:oracle_ld_denoiser_g}
   &\wt{u}^G_{\ell}(x_1,\ldots,x_m,\wt x_1,\ldots,\wt x_{\wt m};\mu)\\
   &=\mathbb E_\mu[\tilde{U}_\ell|Y^G_{1}=x_1,\ldots,Y^G_{m}=x_m,\wt{Y}^G_{1}=\wt x_1,\ldots,\wt{Y}^G_{\wt m}=\wt x_{\wt m}],  
 \end{align}
 respectively, are uniformly Lipschitz over the neighborhood $O_\mu$ for all $h \in [m]$ and $\ell \in [\wt m]$. 
    \item For any positive semi-definite matrices $\bm A^R_h,\bm B^R_h \in \R^{r_h \times r_h}$, and $\nu_h \in \mathcal P_{\nu_h}$, $h \in [m]$, there exists an open neighborhood $O_{\nu_h}$ of $\nu_h$ in the topology of weak convergence, such that for 
    \[
    Y^G_{h,R} \sim N_{r_h}(\bm A^R_{h}V_h,\bm B^R_{h}), \quad \mbox{where} \quad V_h \sim \nu_h,
    \]
    the function $v^G_{h}:\R^{r_h} \rightarrow \R^{r_h}$ given by
    \begin{align}
    \label{eq:oracle_denoisers_g_r}
   &v^G_{h}(x;\nu_h)=\mathbb E_{\nu_h}[V_h|Y^G_{h,R}=x],  
 \end{align}
 is uniformly Lipschitz over the neighborhood $O_{\nu_h}$, for all $h \in [m]$.
\end{enumerate}
\end{asm}

\revsn{The first assumption on $\mu$ is a regularity condition that is \revzm{only} required for 
\revzm{establishing} the asymptotic coverage of the constructed prediction sets. \revzm{It} is valid for a wide range of distribution classes $\mathcal P$ including Gaussian mixture models. 
The second and \revzm{the} third assumptions guarantee identifiability of the latent factors $\bm U_h$, $\wt{\bm U}_\ell$, and the factor loading matrices $\bm V_h$ for $h \in [m]$ and $\ell \in [\wt m]$. 
The fourth assumption \revzm{can be generalized slightly to allow dependence among the $\bm{V}_h$'s, while we state it in its current form to avoid estimating higher-dimensional priors without clear benefits.}
}
Part 5 of Assumption \ref{asm:prior_1_mom}
is an eigengap assumption which lands all high-dimensional modality in the super-critical regime \cite[Theorem 2.8]{BENAYCHGEORGES2012120} and ensures for all $h \in [m]$ that the leading $r_h$ sample singular values of $\bm X_h$ are separated from the bulk
and so
the leading $r_h$ singular vectors of $\hdm_{h}$ have non-zero correlations with the true signals. 
The latter property is critical to the successful initialization of Algorithm \ref{alg:orchamp} with sample SVD.
While seemingly strong, it
holds for most motivating data examples in single-cell biology.
In the rest of this section,
we assume oracle knowledge of the matrices $\bm D_1,\ldots,\bm D_m$ for notational convenience\footnote{ 
Under Assumption \ref{asm:prior_1_mom} part 5,
the matrices $\{\bm D_h: h\in [m]\}$ are consistently estimated by the matrices $\{\wh{\bm D}_h: h \in [m]\}$ in \eqref{eq:est_sing_val} \cite[Theorem 2.8]{BENAYCHGEORGES2012120}, and
the same asymptotic analysis goes through with minimal modification if we replace the matrices $\{\bm D_h: h \in [m]\}$ with the estimated signal strength matrices $\{\wh{\bm D}_h: h \in [m]\}$.}.
The last two parts of Assumption \ref{asm:prior_1_mom} are satisfied by many commonly considered classes of priors such as priors supported on compact subsets of $\R^{r+\wt r}$ and $\R^{r_h}$, respectively, as well as Gaussian mixture models with bounded component means and component covariances bounded away from zero.

\subsection{Asymptotic properties of OrchAMP estimators}
\subsubsection{Estimation of nuisance parameters and priors}
A function $\psi:\mathbb{R}^{p} \rightarrow \mathbb R$ is called \emph{pseudo-Lipschitz} if there exists a constant $C>0$ such that for any $\bm x, \bm y \in \mathbb{R}^{p}$ we have
\begin{align}
\label{eq:pseudo_lips}
    |\psi(\bm x)-\psi(\bm y)| \le C\left(1+\|\bm x\|+\|\bm y\|\right)\|\bm x-\bm y\|.
\end{align}
In the super-critical regime, the matrices $\{\upca_{0,h}, \vpca_{0,h}: h \in [m]\}$, satisfy the following properties.
\begin{prop}
\label{prop:singular_values}
For all $h \in [m]$, consider the matrices $\wb{\bm X}_h$ and their best rank $r_h$ approximations given by $\frac{1}{N}\upca_{0,h}\bm D_{0,h} (\vpca_{0,h})^\top$. 
Let the matrices $\{\hslmorp_{0,h},\hslvorp_{0,h},\hsrmorp_{0,h}, \revsn{\hsrvorp_{0,h}}:h\in [m]\}$ be defined by \eqref{eq:initializer_1} and \eqref{eq:initializers}.
Then, for any sequence of deterministic subsets $\mathcal F_L \subset [N]$ satisfying $\frac{|\mathcal F_L|}{N} \rightarrow \lambda_{L} \in (0,1]$, as $N \rightarrow \infty$, and any pseudo-Lipschitz function $\psi:\mathbb R^{2(r+\wt r)} \rightarrow \mathbb R$ \revzms{as defined in \eqref{eq:pseudo_lips}}, we have
\begin{equation} 
\label{eq:lim_u_pca}	
\begin{aligned}
    &\lim_{N \rightarrow \infty}\frac{1}{|\mathcal F_L|}\sum_{i\in \mathcal F_L}\psi\big((\upca_{0,1})_{i*},\dots,(\upca_{0,m})_{i*},(\wt{\bm X}_{1})_{i*},\dots,(\wt{\bm X}_{\wt m})_{i*},\\
    &\hskip 15em(\bm U_{1})_{i*},\ldots,(\bm U_{m})_{i*},(\wt{\bm U}_{1})_{i*},\ldots,(\wt {\bm U}_{\wt m})_{i*}\big)\\
    &\hskip 4em\overset{a.s.}{=} \mathbb{E}_{\mu}\left[\psi(Y_{0,1},\ldots,Y_{0,m},\wt{Y}_{0,1},\ldots,\wt{Y}_{0,\wt m},U_1,\ldots,U_m,\wt{U}_1,\ldots,\wt{U}_{\wt m})\right].
\end{aligned}
\end{equation}
Here, 
$\{Y_{0,h}:h\in [m]\}$ and $\{\wt{Y}_{0,\ell}:\ell\in [\wt m] \}$ are defined by \eqref{eq:compund_decision_singular_vector} and the collection
$(U_1,\ldots,U_m,\wt{U}_1,\ldots,\wt{U}_{\wt m}) \sim \mu$. 
Furthermore, conditional on $(U_1,\ldots,U_m,\wt U_1,\ldots,\wt U_{\wt m})$, the collection of random vectors $\{Y_{0,1},\ldots,Y_{0,m},\wt Y_{0,1},\ldots,\wt Y_{0,\wt m}\}$ are mutually independent.
Similarly, for all $h \in [m]$, if we consider any sequence of deterministic subsets $\mathcal F_{R,h} \subset [p_{h}]$ that satisfy $\frac{|\mathcal F_{R,h}|}{p_{h}} \rightarrow \lambda_{R,h} \in (0,1]$, as $p_h \rightarrow \infty$ and any pseudo-Lipschitz function $\varphi_h:\mathbb R^{2r_h} \rightarrow \mathbb R$, we have
\begin{align}
\label{eq:lim_v_pca}
    \lim_{p_h \rightarrow \infty}\frac{1}{|\mathcal F_{R,h}|}\sum_{j \in \mathcal F_{R,h}}\varphi_h((\vpca_{0,h})_{j*},(\bm V_{h})_{j*}) \overset{a.s.}{=} \mathbb{E}_{\nu_h}\left[\varphi_h(Y^R_{0,h},V_h)\right],
\end{align}
where $\{Y^R_{0,h}:h\in [m]\}$ is defined by \eqref{eq:compound_decision_model_2} and $V_h \sim \nu_h$.
\end{prop}

This proposition substantiates the approximation in \eqref{eq:emp_sin_vec_asymp} and lays the theoretical foundation for the {empirical Bayes estimation of priors 
described in \eqref{eq:emp_bayes_1} and \eqref{eq:emp_bayes_2}.

\begin{rem}
The expectations on the right hand sides of \eqref{eq:lim_u_pca} and \eqref{eq:lim_v_pca} are with respect to the priors $\mu$ and $\{\nu_h:h \in [m]\}$ and the Gaussian measures used in defining $\{Y_{0,h},Y^R_{0,h}:h \in [m]\}$ and $\{\wt Y_{0,\ell}:\ell \in [\wt m]\}$. 
Thus, the quantities on the right sides are deterministic.
We have deliberately suppressed the Gaussian measures in the subscript of the expectation for conciseness. 
A similar convention is followed throughout the manuscript.
\end{rem}
Next, let us consider the estimates of the nuisance parameters used to initialize Algorithm \ref{alg:orchamp}, defined in \eqref{eq:est_sing_val}, \eqref{eq:scale-est} and \eqref{eq:sqrt_L}. The following lemma shows that these estimates are consistent estimators of the population quantities $\{\bm D_h:h \in [m]\}$, $\{\bm L_\ell:\ell \in [\wt m]\}$ and $\{\bm S^{\star,\mathrm{pc}}_{0,h},\bm \Sigma^{\star,\mathrm{pc}}_{0,h}:\star \in \{L,R\},\, h \in [m]\}$ defined in \eqref{eq:def_matrix_prop_asymp}, \eqref{eq:def_matrix_low_dim_asymp} and \eqref{eq:initializers}.

\begin{lem}
    \label{lem:consistency_nuisance}
    Under Assumptions \ref{asm:prior_1_mom}, the following properties hold as $N \rightarrow \infty$:
    \begin{enumerate}
        \item For all $h \in [m]$, 
        $
        \wh{\bm D}_h \xrightarrow{a.s} \bm D_h,
        $
        where $\wh{\bm D}_h \in \R^{r_h \times r_h}$ is defined in \eqref{eq:est_sing_val}.
        \item For all $\ell \in [\wt m]$, the estimator of $\bm L_\ell$ defined in \eqref{eq:sqrt_L} satisfies
$
\wh{\bm L}_\ell \xrightarrow{a.s.} \bm L_\ell.
$
\item For all $h \in [m]$ and $\star \in \{L,R\}$, $\wh{\bm S}^{\star,\mathrm{pc}}_{0,h} \xrightarrow{a.s}{\bm S}^{\star,\mathrm{pc}}_{0,h}$ and $\wh{\bm \Sigma}^{\star,\mathrm{pc}}_{0,h} \xrightarrow{a.s}{\bm \Sigma}^{\star,\mathrm{pc}}_{0,h}$, where the estimators $\wh{\bm S}^{\star,\mathrm{pc}}_{0,h},\wh{\bm \Sigma}^{\star,\mathrm{pc}}_{0,h}$ are defined in \eqref{eq:scale-est}.
    \end{enumerate}
\end{lem}
Finally, consider the estimated priors $\wh \mu$ and $\{\wh \nu_{h}:h \in [m]\}$ defined in \eqref{eq:emp_bayes_1} and \eqref{eq:emp_bayes_2}, respectively. The following lemma shows that they converge to the true data generating priors $\mu$ and $\{\nu_h:h \in [m]\}$ as $N \rightarrow \infty$ if the nonparametric classes are properly chosen.
\begin{lem}
\label{lem:weak_conv_emp_bayes}
    If $\mu \in \mathcal P$ and $\nu_h \in \mathcal P_{\nu_h}$ for all $h \in [m]$; then as $N \rightarrow \infty$, $\wh{\mu} \overset{w}{\rightarrow} \mu$ and $\wh{\nu}_{h} \overset{w}{\rightarrow} \nu_h$, almost surely.
\end{lem}

\subsubsection{Asymptotics of OrchAMP iterates}

To analyze the asymptotic behavior of OrchAMP iterates \eqref{eq:orc_amp_emp_bayes} and \eqref{eq:low-dim-iter},
we begin with their oracle counterparts
constructed using the true priors $\mu$ and $\{\nu_h:h\in [m]\}$. 
Define
$\hslmor_{t,h}$, $\hslvor_{t,h}$, $\hsrmor_{t,h}$ and $\hsrvor_{t,h} \in \mathbb{R}^{r_h \times r_h}$, for all integers $t \ge 0$, recursively as follows:
\begin{equation}
\begin{aligned}
\label{eq:state_evol_gen}
    \hslvor_{t,h} &= \gamma_h\; \mathbb{E}_{\nu_h}[v^{\orc}_{t,h}(Y^{R,\orc}_{t,h};\nu_h)^{\otimes 2}],\\ 
    \hslmor_{t,h} &= \gamma_h\; \mathbb{E}_{\nu_h}[v^{\orc}_{t,h}(Y^{R,\orc}_{t,h};\nu_h)V^\top_h]\bm D_h,\\
    \hsrvor_{t+1,h} &= \mathbb{E}_{\mu}[u^{\orc}_{t,h}(Y^{\orc}_{t,1},\ldots,Y^{\orc}_{t,m},\wt{Y}_{0,1},\ldots,\wt{Y}_{0,\wt m};\mu)^{\otimes 2}], \\
    \hsrmor_{t+1,h} &= \mathbb{E}_{\mu}[u^{\orc}_{t,h}(Y^{\orc}_{t,1},\ldots,Y^{\orc}_{t,m},\wt{Y}_{0,1},\ldots,\wt{Y}_{0,\wt m};\mu)U^\top_h]\bm D_h,
\end{aligned}
\end{equation}
with $\hsrvor_{0,h}=\hsrvorp_{0,h}$ and $\hsrmor_{0,h}=\hsrmorp_{0,h}$.
In the above equations, for all $h \in [m]$ and $\ell \in [\wt m]$,
\begin{align}
\label{eq:y_orcs}
    Y^{\orc}_{t,h} \sim N_{r_h}(\hslmor_{t,h} U_h,\hslvor_{t,h})\quad \mbox{and}\quad \wt{Y}_{0,\ell} \sim N_{{\wt r}_h}(\bm L_\ell \wt{U}_\ell,\bm I_{\wt r_{\ell}}),
\end{align}
with $(U_1,\ldots,U_m,\wt U_1,\ldots,\wt U_{\wt m}) \sim \mu$, and conditional on $(U_1,\ldots,U_m,\wt U_1,\ldots,\wt U_{\wt m})$, the elements in the collection $\{Y^\orc_{t,1},\ldots,Y^\orc_{t,m},\wt Y_{0,1},\ldots,\wt Y_{0,\wt m}\}$ are mutually independent. 
Similarly, for all $h \in [m]$,
\begin{equation}
\label{eq:y_orc_2}
    Y^{R,\orc}_{t,h} \sim N_{r_h}(\hsrmor_{t,h} V_h,\hsrvor_{t,h}),\; \mbox{where} \; V_h \sim \nu_h.
\end{equation}
In addition, for $h \in [m]$ and $\ell \in [\wt m]$, the \emph{oracle denoisers}, $\{v^\orc_{t,h}\}:\mathbb{R}^{r_h} \rightarrow \mathbb{R}^{\revzms{r_h}}$, $\{u^\orc_{t,h}\}:\mathbb R^{r+{\wt r}} \rightarrow \mathbb R^{r_h}$ and $\{\wt u^\orc_{t,\ell}\}:\mathbb R^{r+{\wt r}} \rightarrow \mathbb R^{\wt r_\ell}$, are defined as follows:
\begin{equation}
\label{eq:oracle_denoisers}	
\begin{aligned}
   &v^\orc_{t,h}(x;\nu_h)=\mathbb{E}_{\nu_h}[V_h|Y^{R,\orc}_{t,h}=x]\\
   &u^\orc_{t,h}(x_1,\ldots,x_m,\wt{x}_1,\ldots,\wt{x}_{\wt m};\mu)\\
   &\hskip 2em=\mathbb{E}_{\mu}[U_h|Y^{\orc}_{t,1}=x_1,\ldots,Y^{\orc}_{t,m}=x_m,\wt{Y}_{0,1}=\wt x_1,\ldots,\wt{Y}_{0,\wt m}=\wt x_{\wt m}],\\
   & \wt{u}^\orc_{t,\ell}(x_1,\ldots,x_m,\wt{x}_1,\ldots,\wt{x}_{\wt m};\mu)\\
   &\hskip 2em=\mathbb{E}_{\mu}[\wt{U}_\ell|Y^{\orc}_{t,1}=x_1,\ldots,Y^{\orc}_{t,m}=x_m,\wt{Y}_{0,1}=\wt x_1,\ldots,\wt{Y}_{0,\wt m}=\wt x_{\wt m}].
\end{aligned} 
\end{equation}
Now, for $h \in [m]$ consider the set of iterates, $\{\buor_{t,h},\bvor_{t,h}\}_{t \ge 0}$ given by:
\begin{equation}
\label{eq:orc_amp_basic}	
\begin{aligned}
    & \bvor_{t,h} = v^\orc_{t,h}(\vor_{t,h};\nu_h),\\
    & \uor_{t,h} = \wb{\bm X}_h\bvor_{t,h}-\gamma_h\buor_{t-1,h}(\jacror_{t,h}(\bm V^\orc_{t,h};\nu_h))^\top\\
    &\buor_{t,h} = u^\orc_{t,h}(\uor_{t,1},\ldots,\uor_{t,m},\wt{\bm X}_1,\ldots,\wt{\bm X}_{\wt m};\mu),\\
    & \vor_{t+1,h} = \wb{\bm X}^\top_h\buor_{t,h}-\bvor_{t,h}(\jaclor_{t,h}(\bm U^\orc_{t,1},\ldots,\bm U^\orc_{t,m},\ldm_1,\ldots,\ldm_{\wt m};\mu))^\top,
\end{aligned}
\end{equation}
where $\buor_{-1,h}=\wb{\bm U}_{-1,h}$ and $\bm V^\orc_{0,h}=\bm V_{0,h}$.
Here, the matrices 
\[u^\orc_{t,h}(\bm U^\orc_{t,1},\ldots,\bm U^\orc_{t,m},\ldm_1,\ldots,\ldm_{\wt m}; \mu) \in \mathbb R^{N \times r_h}\quad \mbox{and}\quad v^\orc_{t,h}(\bm V^\orc_{t,h}; \nu_h) \in \mathbb R^{p_h \times r_h}\] are constructed by applying $u^\orc_{t,h}$ to the collection of rows of the embeddings $(\bm U^\orc_{t,1},\ldots,\bm U^\orc_{t,m},\newline\ldm^\orc_1,\ldots,\ldm^\orc_{\wt m})$ and $v^\orc_{t,h}$ to the rows of $\bm V^\orc_{t,h}$, respectively, as in \eqref{eq:rowwise}.
Further, the matrix $\jacror_{t,h}(\bm V^\orc_{t,h};\nu_h) \in \R^{r_h \times r_h}$ is defined entrywise as follows: For $(i,j) \in [r_h] \times [r_h]$,
\begin{align}
\label{eq:jacror}
\left[\jacror_{t,h}(\bm V^\orc_{t,h};\nu_h)\right]_{ij}= \frac{1}{p_h}\sum_{k=1}^{p_h}\frac{\partial v^\orc_{t,h,i}}{\partial x_j}\left((\bm V^\orc_{t,h})_{k*};\nu_h\right),
\end{align}
where $v^\orc_{t,h}(\cdot)=(v^\orc_{t,h,1}(\cdot),\ldots,v^\orc_{t,h,r_h}(\cdot))^\top$ and $\partial v^\orc_{t,h,i}/\partial x_j$ is the partial derivative of $v^\orc_{t,h,i}$ with respect to its $j$-th argument. Similarly, for $(x_1,\ldots,x_m,\wt x_1,\ldots,\wt x_{\wt m}) \in \R^{r+\wt r}$, let $\jacor_{t,h}(x_1,\ldots,x_m,\wt x_1,\ldots,\wt x_{\wt m};\mu)$ be the Jacobian of the function $u^\orc_{t,h}(x_1,\ldots,x_m,\wt x_1,\ldots,\wt x_{\wt m};\mu)$ and  $[\jacor_{t,h}(x_1,\ldots,x_m,\wt x_1,\ldots,\wt x_{\wt m};\mu)]_{*\mathcal I_h}$ be the submatrix of $\jacor_{t,h}$ formed by selecting the columns in $\mathcal I_h$, where $\mathcal I_h$ is defined in \eqref{eq:mathcali}. 
Then, the matrix $\jaclor_{t,h}$
in the last line of \eqref{eq:orc_amp_basic} is defined as
\begin{equation}
\label{eq:jaclor}	
\begin{aligned}
&\jaclor_{t,h}(\bm U^\orc_{t,1},\ldots,\bm U^\orc_{t,m},\ldm_1,\ldots,\ldm_{\wt m};\mu)\\
&~~~~~~~~~= \frac{1}{N}\sum_{k=1}^{N}[\jacor_{t,h}((\bm U^\orc_{t,1})_{k*},\ldots,(\bm U^\orc_{t,m})_{k*},(\ldm_1)_{k*},\ldots,(\ldm_{\wt m})_{k*};\mu)]_{*\mathcal I_h},
\end{aligned}
\end{equation}
Finally,
for $\ell \in [\wt m]$ and all $t\geq 0$, define 
\[
\wt{\bm U}^\orc_{t,\ell} = \wt{u}^\orc_{t,\ell}(\uor_{t,1},\ldots,\uor_{t,m},\wt{\bm X}_1,\ldots,\wt{\bm X}_{\wt m};\mu),
\]
where the matrices $\wt{u}^\orc_{t,\ell}(\uor_{t,1},\ldots,\uor_{t,m},\wt{\bm X}_1,\ldots,\wt{\bm X}_{\wt m};\mu) \in \R^{N \times \wt r_\ell}$ are constructed by applying the function $\wt{u}^\orc_{t,\ell}$ rowwise to the collection of matrices $\uor_{t,1},\ldots,\uor_{t,m}$ and $\wt{\bm X}_1,\ldots,\wt{\bm X}_{\wt m}$ as in \eqref{eq:low_dim_rowwise}.

With the aid of the foregoing definitions of oracle quantities, the following theorem characterizes the large sample asymptotics of the (empirical Bayes) OrchAMP iterates.
\begin{thm}
\label{thm:em_bayes_state_evol}
Under Assumption \ref{asm:prior_1_mom}, we have the following for all $t \ge 0$:
    \begin{enumerate}
        \item As $N \rightarrow \infty$, for all $h \in [m]$,
        $\frac{1}{N}\|\bm U_{t,h}-\uor_{t,h}\|^2_F \xrightarrow{a.s.} 0$,
        $\frac{1}{p_h}\|\bm V_{t,h}-\vor_{t,h}\|^2_F \xrightarrow{a.s.} 0$,
        and for all $\ell \in [\wt m]$,
        $\frac{1}{N}\|\wt{\bm U}_{t,\ell}-\wt{\bm U}^\orc_{t,\ell}\|^2_F \xrightarrow{a.s.} 0$.
		
\item As $N \rightarrow \infty$, for all $h \in [m]$,
$\wh{\bm S}^{\star}_{t,h} \xrightarrow{a.s.} {\bm S}^{\star,\orc}_{t,h}$,
  and
  $\wh{\bm \Sigma}^{\star}_{t,h} \xrightarrow{a.s.} {\bm \Sigma}^{\star,\orc}_{t,h}$,
  \mbox{for $\star\in \{L,R\}$.
  }
		
        \item For any sequence of deterministic subsets $\mathcal F_L \subset [N]$ satisfying $\frac{|\mathcal F_L|}{N} \rightarrow \lambda_{L} \in (0,1]$ as $N \rightarrow \infty$, and any  pseudo-Lipschitz function $\psi:\mathbb R^{2(r+\wt r)} \rightarrow \mathbb R$, 
\begin{equation}
\label{eq:orc_iterate_lim_u}	
\begin{aligned}
     &\lim_{N \rightarrow \infty}\frac{1}{|\mathcal F_L|}\sum_{i \in \mathcal F_L}\psi\Big((\bm U_{t,1})_{i*},\ldots,(\bm U_{t,m})_{i*},(\wt{\bm X}_{1})_{i*},\dots,(\wt{\bm X}_{\wt m})_{i*},\\
     & \hskip 15em (\bm U_{1})_{i*},\ldots,(\bm U_{m})_{i*},
	 (\wt{\bm U}_1)_{i*},\ldots,(\wt{\bm U}_{\wt m})_{i*}\Big)\\
     &\hskip 4em\overset{a.s.}{=} \mathbb{E}_\mu\left[\psi(Y^\orc_{t,1},\ldots,Y^\orc_{t,m},\wt{Y}_{0,1},\ldots,\wt{Y}_{0,\wt m},U_1,\ldots,U_m,\tilde U_1,\ldots,\wt U_{\wt m})\right],
\end{aligned}
\end{equation}
where $(U_1,\ldots,U_m,\wt{U}_1,\ldots,\wt{U}_{\wt m}) \sim \mu$. Here, $\{Y^\orc_{t,h}:h\in [m]\}$ and $\{\wt Y_{0,\ell}:\ell\in [\wt m]\}$ are defined in \eqref{eq:y_orcs}.
Furthermore, for all $h \in [m]$, if we consider any sequence of deterministic subsets $\mathcal F_{R,h} \subset [p_{h}]$ that satisfy $\frac{|\mathcal F_{R,h}|}{p_{h}} \rightarrow \lambda_{R,h} \in (0,1]$ as $p_h \rightarrow \infty$, and any pseudo-Lipschitz function $\varphi_h:\mathbb R^{2r_h} \rightarrow \mathbb R$, we have,
\begin{align}
\label{eq:orc_iterate_lim_v}
    \lim_{p_h \rightarrow \infty}\frac{1}{|\mathcal F_{R,h}|}\sum_{j \in \mathcal F_{R,h}}\varphi_h((\bm V_{t,h})_{j*},(\bm V_{h})_{j*}) \overset{a.s.}{=} \mathbb{E}_{\nu_h}\left[\varphi_h(Y^{R,\orc}_{t,h},V_h)\right],
\end{align}
where $V_h \sim \nu_h$ and $Y^{R,\orc}_{t,h}$ is defined in \eqref{eq:y_orc_2}.
    \end{enumerate} 
\end{thm}

Considering $\mathcal F_L=[N]$ and $\mathcal F_{R,h}=[p_h]$ for all $h \in [m]$, we can synopsize the last claim of the above theorem as follows:
\begin{align}
\label{eq:equiv_rep_amp_it}
    \bm U_{t,h} \approx \bm U_h(\hslm_{t,h})^\top+\bm Z^L_h(\hslv_{t,h})^{1/2} 
\quad \mbox{and}\quad 
	 \bm V_{t,h} \approx \bm V_h(\hsrm_{t,h})^\top+{\bm Z^R_h}(\hsrv_{t,h})^{1/2},
\end{align}
where $\bm Z^L_h \in \R^{N \times r_h}$ and $\bm Z^R_h \in \R^{p_h \times r_h}$ are matrices with iid~$N(0,1)$
entries 
for all $h \in [m]$.

Theorem \ref{thm:em_bayes_state_evol} implies the following corollary characterizing the mean-squared-error (MSE)
risk of estimating $\{\bm U_h,\bm V_h: h \in [m]\}$ with $\{\bu_{t,h},\bv_{t,h}: h \in [m]\}$.
\begin{cor}
    \label{thm:mmse_svd}
	Under Assumption \ref{asm:prior_1_mom},
for all $t \ge 0$, the iterates $\{\bu_{t,h}, \bv_{t,h}: h \in [m]\}$ satisfy:
\begin{align}
\label{eq:limit_svds}
  \lim\limits_{N \rightarrow \infty}\frac{1}{N^2}\|\bu_{t,h}\bu^\top_{t,h}-\bm U_h\bm U^\top_h\|^2_F&\,\overset{a.s}{=}\,\mathrm{Tr}\left(\bm I_{r_h}-(\hsrvor_{t+1,h})^2\right),\nonumber\\
  \lim\limits_{p_h \rightarrow \infty}\frac{1}{p^2_h}\|\bv_{t,h}\bv^\top_{t,h}-\bm V_h\bm V^\top_h\|^2_F&\,\overset{a.s}{=}\,
  \mathrm{Tr}\left(\bm I_{r_h}-\frac{1}{\gamma^2_h}{(\hslvor_{t,h})^2}\right).
\end{align}
\end{cor}

The counterpart for estimating $\{\bm{\wt U}_\ell: \ell \in [\wt m]\}$ with $\{\wt{\bm U}_{t,\ell}: \ell \in [\wt m]\}$ are as follows.
 \begin{cor}
     \label{thm:risk_ldim}
Under Assumption \ref{asm:prior_1_mom},	for all $t \ge 0$,
	the iterates $\{\wt{\bm U}_{t,\ell}: \ell \in [\wt m]\}$ satisfy:
     \begin{align}
     \label{eq:risk_ldim}
     \lim\limits_{N \rightarrow \infty}\frac{1}{N^2}\|\wt{\bm U}_{t,\ell}\wt{\bm U}^\top_{t,\ell}-\wt{\bm U}_\ell\wt{\bm U}^\top_\ell\|^2_F \,\overset{a.s.}{=}\,
     \mathrm{Tr}\left(\bm I_{\wt r_\ell}-(\wt{\bm \Sigma}^{\circ}_{t,\ell})^2\right),
     \end{align}
     where 
     $\wt{\bm \Sigma}^{\circ}_{t,\ell}=\mathbb E_\mu\left[\mathbb E_{\mu}
[\wt{U}_\ell\,\Big|\,Y^{\orc}_{t,1},\ldots,Y^{\orc}_{t,m},\wt{Y}_{0,1},\ldots,\wt{Y}_{0,m}
]^{\otimes 2}\right]$.
 \end{cor}

\subsubsection{Bayes optimality of OrchAMP estimators
}

Fix any set of positive definite matrices $\{\bm B^L_h, \bm B^R_h \in \R^{r_h \times r_h}: h\in [m]\}$.
Conditional on the latent factors $({U}_1,\ldots,{U}_m,\wt{U}_1,\ldots,\wt{U}_{\wt m}) \sim \mu$ and $V_h \sim \nu_{h}$ for $h\in [m]$, define random vectors $\rmxl_h = \rmxl_h(\bm B^L_h)$, $\rmxr_h = \rmxr_h(\bm B^{R}_h)$, $h\in [m]$ and $\rmtxl_\ell$, $\ell\in [{\wt m}]$ as
\begin{equation}
	\label{eq:bayes-opt-rv}
    \rmxl_h(\bm B^L_h) \sim N_{r_h}(U_h,(\bm B^L_h)^{-1}),\,\,
	\rmxr_h(\bm B^{R}_h) \sim N_{r_h}(V_h,(\bm B^{R}_h)^{-1}),\,\,
    \rmtxl_\ell \sim N_{{\wt r}_\ell}(\bm L_\ell \wt{U}_\ell,\bm I_{\wt r_{\ell}}).
\end{equation}
Next, for all $h\in [m]$, define the matrices $ \rmufunc_h(\bm B^L_1,\ldots,\bm B^L_m;\mu) \in \mathbb{R}^{r_h\times r_h}$ and $\rmvfunc_h(\bm B^R_h;\nu_h) \in \mathbb{R}^{r_h\times r_h} $ as follows:
\begin{align}
\label{eq:u_d}
    \rmufunc_h(\bm B^L_1,\ldots,\bm B^L_m;\mu) & := \mathbb{E}_\mu[(U_h-\mathbb{E}_\mu[U_h|\rmxl_1(\bm B^{L}_1),\ldots,\rmxl_m(\bm B^{L}_m),\rmtxl_1,\ldots,\rmtxl_{\wt m}])^{\otimes 2}], \\
\label{eq:v_d}
    \rmvfunc_h(\bm B^R_h;\nu_h) & := \mathbb{E}_{\nu_h}[(V_h-\mathbb{E}_{\nu_h}[V_h|\rmxr_h(\bm B^{R}_h)])^{\otimes 2}],
\end{align}
and
matrices $\bm \Gamma^R_{t,h},\bm \Gamma^L_{t,h} \in \R^{r_h \times r_h}$ as follows:
\begin{equation}
\label{eq:def_gamma}	
\begin{aligned}
    \gammab_{t,h}&:= \frac{1}{\gamma_h}\bm D^{-1/2}_h(\hslmor_{t,h})^\top(\hslvor_{t,h})^{-1}\hslmor_{t,h}\bm D^{-1/2}_h,\\
    \gamman_{t,h} &:= \bm D^{-1/2}_h(\hsrmor_{t,h})^\top(\hsrvor_{t,h})^{-1}\hsrmor_{t,h}\bm D^{-1/2}_h.
\end{aligned}
\end{equation}
Using the foregoing definitions, we \revzms{can reformulate}
the \emph{state evolution} recursions \revzms{as follows}. 
\begin{lem}
\label{lem:redef_state_evol}
The state evolution recursions \eqref{eq:state_evol_gen} can be equivalently characterized 
as 
\begin{align}
\label{eq:state_evol_mod}
   \gammab_{t,h}=\bm D_h-\rmvfunc_h(\gamman_{t,h};\bm D^{1/2}_h\nu_h), \quad \gamman_{t+1,h}=\bm D_h-\rmufunc_h(\gamma_1\gammab_{t,1},\ldots,\gamma_m\gammab_{t,m};\bm D^{1/2}\mu),
\end{align}
where $\bm D = \mathsf{diag}(\bm D_1,\ldots,\bm D_m,\bm I_{\wt r_1},\ldots,\bm I_{\wt r_{\wt m}})$. 
\end{lem}

The following theorem characterizes the large-$t$ asymptotics of the state evolutions.
\begin{thm}
    \label{thm:improvement_in_error}
    There exists matrices $\gamman_{\infty,h},\gammab_{\infty,h} \in \R^{r_h \times r_h}$ for $h \in [m]$, satisfying 
 \begin{align}
 \label{eq:fixed_point}
   \gammab_{\infty,h}=\bm D_h-\rmvfunc_h(\gamman_{\infty,h};\bm D^{1/2}_h\nu_h), \,
   \gamman_{\infty,h}=\bm D_h-\rmufunc_h(\gamma_1\gammab_{\infty,1},\ldots,\gamma_m\gammab_{\infty,m};\bm D^{1/2}\mu),
\end{align}   
    such that $\lim_{t \rightarrow \infty}\gammab_{t,h}=\gammab_{\infty,h}$ and $\lim_{t \rightarrow \infty}\gamman_{t,h}=\gamman_{\infty,h}$, for all $h \in [m]$.
\end{thm}

\begin{rem}
The system of equations \eqref{eq:fixed_point} is guaranteed to have at least one set of fixed points $\gamman_{\infty,h}$ and $\gammab_{\infty,h}$, while the set of fixed points is not necessarily unique. 
Whether the set is
unique is determined by the specific choice of the priors $\mu$ and $\{\nu_h: h \in [m]\}$.
\end{rem}

Next, consider the Bayes estimators of the signals $\bm U_h\bm D_h\bm V^\top_h$ for $h \in [m]$. The I-MMSE identity as in Proposition 3 of \cite{article}
implies that 
\begin{equation*}
\begin{aligned}
    &\lim_{N \rightarrow \infty}\frac{1}{Np_h}\mathbb E[\|\bm U_h\bm D_h\bm V^\top_h-\mathbb E[\bm U_h\bm D_h\bm V^\top_h|\bm X_1,\ldots,\bm X_m,\wt{\bm X}_1,\ldots,\wt{\bm X}_{\wt m}]\|^2_F]\\
	&\hskip 5em =\mathrm{Tr}(\bm D^2_h)-\mathrm{Tr}\,(\gammab_{\star,h}\gamman_{\star,h}),
\end{aligned}	
\end{equation*}
for some fixed point $\{(\gammab_{\star,h}\gamman_{\star,h}):h \in [m]\}$ of the recursion \eqref{eq:state_evol_mod}.

Though the estimators $\mathbb E[\bm U_h\bm D_h\bm V^\top_h|\bm X_1,\ldots,\bm X_m,\wt{\bm X}_1,\ldots,\wt{\bm X}_{\wt m}]$ give the optimal 
MSE
risk for estimating $\bm U_h\bm D_h\bm V^\top_h$,
they may not be computationally tractable (in the sense of polynomial time complexity) for 
general
priors $\mu$ and $\{\nu_h:h\in [m]\}$. 
In contrast, the following theorem shows that after sufficiently many iterative refinements of the empirical singular vectors using 
\eqref{eq:orc_amp_emp_bayes}, the reconstruction error of the OrchAMP estimators $\bu_{t,h}\bm D_h\bv^\top_{t,h}$ for any $h \in [m]$ is 
asymptotically
equal to the Bayes optimal risk as long as $\gammab_{\infty,h}$ and $\gamman_{\infty,h}$ are unique.
\begin{thm}
    \label{thm:improvement_in_error_2}
	If there exists unique fixed points of the recursive equations \eqref{eq:fixed_point} for the priors $\mu$ and $\{\nu_h:h \in [m]\}$ given by $\{\gammab_{\infty,h},\gamman_{\infty,h}:h \in [m]\}$, then 
the OrchAMP estimators $\{\bu_{t,h}\bm D_h\bv^\top_{t,h}:h\in [m]\}$ satisfy \revzms{that for all $h\in [m]$,}
 \begin{align}
    \lim_{t \rightarrow \infty}\lim_{n \rightarrow \infty}\frac{1}{Np_h}\|\bm U_h\bm D_h\bm V^\top_h-\bu_{t,h}\bm D_h\bv^\top_{t,h}\|^2_F
    \overset{a.s.}{=}\mathrm{Tr}(\bm D^2_h)-\mathrm{Tr}\,(\gammab_{\infty,h}\gamman_{\infty,h}).
\end{align}   
\end{thm}
We refer interested readers to \cite{10.5555/3157096.3157144,montanari2021,AbbeMonYash} for examples of priors with unique fixed points when $h=1,\ell=0$ and $r_h=1$. Further, one can find an example of a set of priors with unique fixed points when $h=2,\ell=0$ and $r_1=r_2=1$ in \cite{ma_nandy}.

\subsection{Asymptotic coverage of prediction sets}
To begin with, let us characterize the asymptotic distributions of $\{\hatuq_{h_k}:k\in [d]\}$ defined in \eqref{eq:point_est_cell_eff}.
 \begin{thm}
     \label{thm:point_pred}
     Consider the predictors $\{\hatuq_{h_k}:k\in [d] \}$ defined in \eqref{eq:hat_uq_def}. 
	 Suppose as $N \rightarrow \infty$, $|\calf_{h_k}|/p_{h_k} \rightarrow \lambda_{h_k}\in (0,1]$ for $k\in [d]$.
	 Then there exist $Z_{h_k}\sim N_{r_{h_k}}(0, {\bm I}_{r_{h_k}})$, $k\in [d]$, which are mutually independent and independent of $({U}_1^{\mathrm{Q}},\dots,{U}_m^{\mathrm{Q}}, \wt{U}_1^{\mathrm{Q}},\dots, \wt{U}_{\wt{m}}^{\mathrm{Q}})$, such that conditional on $({U}_1^{\mathrm{Q}},\dots,{U}_m^{\mathrm{Q}}, \wt{U}_1^{\mathrm{Q}},\dots, \wt{U}_{\wt{m}}^{\mathrm{Q}})$,
	 we have jointly over $k\in [d]$,
     \[
     \hatuq_{h_k}\, \stackrel{d}{\to}\, U^{\mathrm Q}_{h_k}+\lambda^{-1/2}_{h_k}\revsns{\left(\bm D^{-1}_{h_k}(\hslvor_{T,h_k})^{-1}\bm D^{-1}_{h_k}\right)^{1/2}}Z_{h_k}.
     \]	 
      \end{thm}

To obtain the asymptotic coverage probability for the prediction set of the query's full latent representation
$(U^\rmq_1,\ldots,U^\rmq_m,\wt{U}^\rmq_1,\ldots,\wt{U}^\rmq_{\wt m})$, 
we establish the following lemma proving that the class of closed  convex subsets of $\R^{r+\wt r}$, denoted by $\mathscr{C}_{r+\wt r}$, 
is Glivenko--Cantelli
with respect to the sequence of measures $\wh \mu$ that converges weakly to the true underlying prior $\mu$.
 \begin{lem}
 \label{lem:glivenko_cantelli}
     If $\mu \in \mathcal P$, then as $N \rightarrow \infty$,
     $\sup_{A \in \mathscr{C}_{r+\wt r}}|\mu(A)-{\wh \mu}(A)| \xrightarrow{a.s.} 0$.
\end{lem}

Since for any $\alpha \in (0,1)$, our prediction set $\calc_{\alpha}$, by definition, is a closed convex subset of $\R^{r+\wt r}$, Lemma \ref{lem:glivenko_cantelli} implies the following theorem.
 \begin{thm}
 \label{thm:pred_query_sets_full}
Suppose that $\mu \in \mathcal P$ and that $\wh \mu$ is any sequence of priors converging to $\mu$ weakly as $N \rightarrow \infty$  
For any $\alpha\in (0,1)$, the nominal $100\times (1-\alpha)\%$ prediction set $\calc_{\alpha}$ defined in \eqref{eq:full_pred_set} satisfies
     \[
     \lim_{N \rightarrow \infty}\mathbb P_\mu\left[(U^\rmq_1,\ldots,U^\rmq_m,\wt{U}^\rmq_1,\ldots,\wt{U}^\rmq_{\wt m}) \in \calc_{\alpha}\right]\ge 1-\alpha.
     \]
 \end{thm}

\begin{section}{Acknowledgement}
The research presented in this manuscript has been supported by NSF awards DMS-2345215 and DMS-2245575.
\end{section}

\bibliographystyle{apalike} 
\bibliography{amp_multimodal_data}

\begin{thebibliography}{}

\bibitem[Argelaguet et~al., 2020]{argelaguet2020mofa+}
Argelaguet, R., Arnol, D., Bredikhin, D., Deloro, Y., Velten, B., Marioni,
  J.~C., and Stegle, O. (2020).
\newblock Mofa+: a statistical framework for comprehensive integration of
  multi-modal single-cell data.
\newblock {\em Genome Biology}, 21(1):111.

\bibitem[Baik et~al., 2005]{baik2005}
Baik, J., Ben~Arous, G., and P{\'e}ch{\'e}, S. (2005).
\newblock Phase transition of the largest eigenvalue for nonnull complex sample
  covariance matrices.
\newblock {\em Annals of Probability}, 33(5):1643--1697.

\bibitem[Barbier et~al., 2016]{10.5555/3157096.3157144}
Barbier, J., Dia, M., Macris, N., Krzakala, F., Lesieur, T., and Zdeborov\'{a},
  L. (2016).
\newblock Mutual information for symmetric rank-one matrix estimation: A proof
  of the replica formula.
\newblock In {\em Proceedings of the 30th International Conference on Neural
  Information Processing Systems}, NIPS'16, pages 424--432, Red Hook, NY, USA.
  Curran Associates Inc.

\bibitem[Barbier et~al., 2019]{doi:10.1073/pnas.1802705116}
Barbier, J., Krzakala, F., Macris, N., Miolane, L., and Zdeborov{\'a}, L.
  (2019).
\newblock Optimal errors and phase transitions in high-dimensional generalized
  linear models.
\newblock {\em Proceedings of the National Academy of Sciences},
  116(12):5451--5460.

\bibitem[Bayati and Montanari, 2011]{BM11journal}
Bayati, M. and Montanari, A. (2011).
\newblock The dynamics of message passing on dense graphs, with applications to
  compressed sensing.
\newblock {\em IEEE Transactions on Information Theory}, 57(2):764--785.

\bibitem[Benaych-Georges and Nadakuditi, 2012]{BENAYCHGEORGES2012120}
Benaych-Georges, F. and Nadakuditi, R.~R. (2012).
\newblock The singular values and vectors of low rank perturbations of large
  rectangular random matrices.
\newblock {\em Journal of Multivariate Analysis}, 111:120--135.

\bibitem[Berthier et~al., 2019]{Berthier}
Berthier, R., Montanari, A., and Nguyen, P.-M. (2019).
\newblock {State evolution for approximate message passing with non-separable
  functions}.
\newblock {\em Information and Inference: A Journal of the IMA}, 9(1):33--79.

\bibitem[Blondel et~al., 2008]{blondel2008fast}
Blondel, V.~D., Guillaume, J.-L., Lambiotte, R., and Lefebvre, E. (2008).
\newblock Fast unfolding of communities in large networks.
\newblock {\em Journal of statistical mechanics: theory and experiment},
  2008(10):P10008.

\bibitem[Bolthausen, 2014]{Bol12}
Bolthausen, E. (2014).
\newblock An iterative construction of solutions of the tap equations for the
  sherrington--kirkpatrick model.
\newblock {\em Communications in Mathematical Physics}, 325(1):333--366.

\bibitem[Cao et~al., 2018]{doi:10.1126/science.aau0730}
Cao, J., Cusanovich, D.~A., Ramani, V., Aghamirzaie, D., Pliner, H.~A., Hill,
  A.~J., Daza, R.~M., McFaline-Figueroa, J.~L., Packer, J.~S., Christiansen,
  L., Steemers, F.~J., Adey, A.~C., Trapnell, C., and Shendure, J. (2018).
\newblock Joint profiling of chromatin accessibility and gene expression in
  thousands of single cells.
\newblock {\em Science}, 361(6409):1380--1385.

\bibitem[Chen et~al., 2019a]{Chen2019a}
Chen, H., Ye, F., and Guo, G. (2019a).
\newblock Revolutionizing immunology with single-cell rna sequencing.
\newblock {\em Cellular {\&} Molecular Immunology}, 16(3):242--249.

\bibitem[Chen et~al., 2019b]{Chen2019}
Chen, S., Lake, B.~B., and Zhang, K. (2019b).
\newblock High-throughput sequencing of the transcriptome and chromatin
  accessibility in the same cell.
\newblock {\em Nature Biotechnology}, 37(12):1452--1457.

\bibitem[Denault et~al., 2025]{denault2025covariate}
Denault, W.~R., Tayeb, K., Carbonetto, P., Willwerscheid, J., and Stephens, M.
  (2025).
\newblock Covariate-moderated empirical bayes matrix factorization.
\newblock In {\em The Thirty-ninth Annual Conference on Neural Information
  Processing Systems}.

\bibitem[Deshpande et~al., 2016]{AbbeMonYash}
Deshpande, Y., Abbe, E., and Montanari, A. (2016).
\newblock Asymptotic mutual information for the balanced binary stochastic
  block model.
\newblock {\em Information and Inference: A Journal of the IMA}, 6(2):125--170.

\bibitem[Donoho et~al., 2013]{6409458}
Donoho, D.~L., Johnstone, I., and Montanari, A. (2013).
\newblock Accurate prediction of phase transitions in compressed sensing via a
  connection to minimax denoising.
\newblock {\em IEEE Transactions on Information Theory}, 59(6):3396--3433.

\bibitem[Donoho et~al., 2009]{doi:10.1073/pnas.0909892106}
Donoho, D.~L., Maleki, A., and Montanari, A. (2009).
\newblock Message-passing algorithms for compressed sensing.
\newblock {\em Proceedings of the National Academy of Sciences},
  106(45):18914--18919.

\bibitem[Fletcher and Rangan, 2018]{FletcherRangan2018}
Fletcher, A.~K. and Rangan, S. (2018).
\newblock Iterative reconstruction of rank-one matrices in noise.
\newblock {\em Information and Inference: A Journal of the IMA}, 7(3):531--562.

\bibitem[Gaddis et~al., 2024]{gaddis2024lungmap}
Gaddis, N., Fortriede, J., Guo, M., Bardes, E.~E., Kouril, M., Tabar, S.,
  Burns, K., Ardini-Poleske, M.~E., Loos, S., Schnell, D., et~al. (2024).
\newblock Lungmap portal ecosystem: Systems-level exploration of the lung.
\newblock {\em American Journal of Respiratory Cell and Molecular Biology},
  70(2):129--139.

\bibitem[Gayoso et~al., 2021]{gayoso2021joint}
Gayoso, A., Steier, Z., Lopez, R., Regier, J., Nazor, K.~L., Streets, A., and
  Yosef, N. (2021).
\newblock Joint probabilistic modeling of single-cell multi-omic data with
  totalvi.
\newblock {\em Nature methods}, 18(3):272--282.

\bibitem[Gerbelot and Berthier, 2023]{gerbelot2022graphbased}
Gerbelot, C. and Berthier, R. (2023).
\newblock Graph-based approximate message passing iterations.
\newblock {\em Information and Inference: A Journal of the IMA},
  12(4):2562--2628.

\bibitem[Goodfellow et~al., 2016]{GoodBengCour16}
Goodfellow, I.~J., Bengio, Y., and Courville, A. (2016).
\newblock {\em Deep Learning}.
\newblock MIT Press, Cambridge, MA, USA.
\newblock \url{http://www.deeplearningbook.org}.

\bibitem[Gui et~al., 2025]{gui2025multi}
Gui, Y., Ma, C., and Ma, Z. (2025).
\newblock Multi-modal contrastive learning adapts to intrinsic dimensions of
  shared latent variables.
\newblock In {\em The Thirty-ninth Annual Conference on Neural Information
  Processing Systems}.

\bibitem[Hao et~al., 2021]{HAO20213573}
Hao, Y., Hao, S., Andersen-Nissen, E., Mauck, W.~M., Zheng, S., Butler, A.,
  Lee, M.~J., Wilk, A.~J., Darby, C., Zager, M., Hoffman, P., Stoeckius, M.,
  Papalexi, E., Mimitou, E.~P., Jain, J., Srivastava, A., Stuart, T., Fleming,
  L.~M., Yeung, B., Rogers, A.~J., McElrath, J.~M., Blish, C.~A., Gottardo, R.,
  Smibert, P., and Satija, R. (2021).
\newblock Integrated analysis of multimodal single-cell data.
\newblock {\em Cell}, 184(13):3573--3587.e29.

\bibitem[Hubert and Arabie, 1985a]{HubertArabie1985}
Hubert, L. and Arabie, P. (1985a).
\newblock Comparing partitions.
\newblock {\em Journal of Classification}, 2(1):193--218.

\bibitem[Hubert and Arabie, 1985b]{Hubert1985}
Hubert, L. and Arabie, P. (1985b).
\newblock Comparing partitions.
\newblock {\em Journal of Classification}, 2(1):193--218.

\bibitem[Javanmard and Montanari, 2013]{JM12}
Javanmard, A. and Montanari, A. (2013).
\newblock State evolution for general approximate message passing algorithms,
  with applications to spatial coupling.
\newblock {\em Information and Inference: A Journal of the IMA}, 2(2):115--144.

\bibitem[Keup and Zdeborová, 2025]{keup2024optimal}
Keup, C. and Zdeborová, L. (2025).
\newblock Optimal thresholds and algorithms for a model of multi-modal learning
  in high dimensions.
\newblock {\em Journal of Statistical Mechanics: Theory and Experiment},
  2025(9):093302.

\bibitem[Kim et~al., 2020]{kim2020citefuse}
Kim, H.~J., Lin, Y., Geddes, T.~A., Yang, J. Y.~H., and Yang, P. (2020).
\newblock Citefuse enables multi-modal analysis of cite-seq data.
\newblock {\em Bioinformatics}, 36(14):4137--4143.

\bibitem[Korsunsky et~al., 2019]{Korsunsky2019}
Korsunsky, I., Millard, N., Fan, J., Slowikowski, K., Zhang, F., Wei, K.,
  Baglaenko, Y., Brenner, M., Loh, P.-r., and Raychaudhuri, S. (2019).
\newblock Fast, sensitive and accurate integration of single-cell data with
  harmony.
\newblock {\em Nature Methods}, 16(12):1289--1296.

\bibitem[Lesieur et~al., 2017]{Lesieur2017ConstrainedLM}
Lesieur, T., Krzakala, F., and Zdeborová, L. (2017).
\newblock Constrained low-rank matrix estimation: phase transitions,
  approximate message passing and applications.
\newblock {\em Journal of Statistical Mechanics: Theory and Experiment},
  2017(7):073403.

\bibitem[Li et~al., 2023]{yuting}
Li, G., Fan, W., and Wei, Y. (2023).
\newblock Approximate message passing from random initialization with
  applications to z 2 synchronization.
\newblock {\em Proceedings of the National Academy of Sciences},
  120(31):e2302930120.

\bibitem[{Liu} et~al., 2019]{Rush}
{Liu}, H., {Rush}, C., and {Baron}, D. (2019).
\newblock An analysis of state evolution for approximate message passing with
  side information.
\newblock In {\em 2019 IEEE International Symposium on Information Theory
  (ISIT)}, pages 2069--2073.

\bibitem[Lock et~al., 2013]{lock2013jive}
Lock, E.~F., Hoadley, K.~A., Marron, J.~S., and Nobel, A.~B. (2013).
\newblock Joint and individual variation explained (jive) for integrated
  analysis of multiple data types.
\newblock {\em The Annals of Applied Statistics}, 7(1):523--542.

\bibitem[Ma and Nandy, 2023]{ma_nandy}
Ma, Z. and Nandy, S. (2023).
\newblock Community detection with contextual multilayer networks.
\newblock {\em IEEE Transactions on Information Theory}, 69(5):3203--3239.

\bibitem[McInnes et~al., 2018]{mcinnes2018umap-software}
McInnes, L., Healy, J., Saul, N., and Grossberger, L. (2018).
\newblock Umap: Uniform manifold approximation and projection.
\newblock {\em The Journal of Open Source Software}, 3(29):861.

\bibitem[M{\'e}zard and Montanari,
  2009]{10.1093/acprof:oso/9780198570837.001.0001}
M{\'e}zard, M. and Montanari, A. (2009).
\newblock {\em {Information, Physics, and Computation}}.
\newblock Oxford University Press.

\bibitem[Mezard et~al., 1986]{doi:10.1142/0271}
Mezard, M., Parisi, G., and Virasoro, M. (1986).
\newblock {\em Spin Glass Theory and Beyond}.
\newblock WORLD SCIENTIFIC.

\bibitem[Miolane, 2017]{article}
Miolane, L. (2017).
\newblock Fundamental limits of low-rank matrix estimation: the non-symmetric
  case.
\newblock {\em arXiv preprint arXiv:1702.00473}.

\bibitem[Montanari and Venkataramanan, 2021]{montanari2021}
Montanari, A. and Venkataramanan, R. (2021).
\newblock Estimation of low-rank matrices via approximate message passing.
\newblock {\em Annals of Statistics}, 49(1):321--345.

\bibitem[Nandy and Sen, 2025]{nandy2023bayes}
Nandy, S. and Sen, S. (2025).
\newblock Bayes optimal learning in high-dimensional linear regression with
  network side information.
\newblock {\em IEEE Trans. Inf. Theor.}, 71(1):565–591.

\bibitem[Papalexi and Satija, 2017]{Papalexi2017-dn}
Papalexi, E. and Satija, R. (2017).
\newblock Single-cell {RNA} sequencing to explore immune cell heterogeneity.
\newblock {\em Nat Rev Immunol}, 18(1):35--45.

\bibitem[Pearl, 1988]{PEARL198829}
Pearl, J. (1988).
\newblock Chapter 2 - bayesian inference.
\newblock In Pearl, J., editor, {\em Probabilistic Reasoning in Intelligent
  Systems}, pages 29--75. Morgan Kaufmann, San Francisco (CA).

\bibitem[Radford et~al., 2021]{radford2021learning}
Radford, A., Kim, J.~W., Hallacy, C., Ramesh, A., Goh, G., Agarwal, S., Sastry,
  G., Askell, A., Mishkin, P., Clark, J., Krueger, G., and Sutskever, I.
  (2021).
\newblock Learning transferable visual models from natural language
  supervision.
\newblock In Meila, M. and Zhang, T., editors, {\em Proceedings of the 38th
  International Conference on Machine Learning, {ICML} 2021, 18-24 July 2021,
  Virtual Event}, volume 139 of {\em Proceedings of Machine Learning Research},
  pages 8748--8763. {PMLR}.

\bibitem[Reeves et~al., 2018]{8437326}
Reeves, G., Pfister, H.~D., and Dytso, A. (2018).
\newblock Mutual information as a function of matrix snr for linear gaussian
  channels.
\newblock In {\em 2018 IEEE International Symposium on Information Theory
  (ISIT)}, pages 1754--1758.

\bibitem[Rosenberg and Hirschberg, 2007]{rosenberg2007vmeasure}
Rosenberg, A. and Hirschberg, J. (2007).
\newblock V-measure: A conditional entropy-based external cluster evaluation
  measure.
\newblock In {\em Proceedings of the 2007 Joint Conference on Empirical Methods
  in Natural Language Processing and Computational Natural Language Learning
  (EMNLP-CoNLL)}, pages 410--420.

\bibitem[Rousseeuw, 1987]{rousseeuw1987silhouettes}
Rousseeuw, P.~J. (1987).
\newblock Silhouettes: a graphical aid to the interpretation and validation of
  cluster analysis.
\newblock {\em Journal of Computational and Applied Mathematics}, 20:53--65.

\bibitem[Rozenblatt-Rosen et~al., 2020]{rozenblatt2020human}
Rozenblatt-Rosen, O., Regev, A., Oberdoerffer, P., Nawy, T., Hupalowska, A.,
  Rood, J.~E., Ashenberg, O., Cerami, E., Coffey, R.~J., Demir, E., et~al.
  (2020).
\newblock The human tumor atlas network: charting tumor transitions across
  space and time at single-cell resolution.
\newblock {\em Cell}, 181(2):236--249.

\bibitem[Schniter, 2010]{5464920}
Schniter, P. (2010).
\newblock Turbo reconstruction of structured sparse signals.
\newblock In {\em 2010 44th Annual Conference on Information Sciences and
  Systems (CISS)}, pages 1--6.

\bibitem[Schwarz, 1978]{schwarz1978bic}
Schwarz, G. (1978).
\newblock Estimating the dimension of a model.
\newblock {\em The Annals of Statistics}, 6(2):461--464.

\bibitem[Snyder et~al., 2019]{hubmap2019human}
Snyder, M., Cai, L., Shendure, J., Trapnell, C., Lin, S., Jackson, D., Yuan,
  G.-C., Zhu, Q., and Dries, R. (2019).
\newblock The human body at cellular resolution: the nih human biomolecular
  atlas program.
\newblock {\em Nature}, 574.

\bibitem[Stein, 1981]{Stein1981}
Stein, C.~M. (1981).
\newblock Estimation of the mean of a multivariate normal distribution.
\newblock {\em The Annals of Statistics}, 9(6):1135--1151.

\bibitem[Stoeckius et~al., 2017]{Stoeckius2017}
Stoeckius, M., Hafemeister, C., Stephenson, W., Houck-Loomis, B.,
  Chattopadhyay, P.~K., Swerdlow, H., Satija, R., and Smibert, P. (2017).
\newblock Simultaneous epitope and transcriptome measurement in single cells.
\newblock {\em Nature Methods}, 14(9):865--868.

\bibitem[Swanson et~al., 2021]{10.7554/eLife.63632}
Swanson, E., Lord, C., Reading, J., Heubeck, A.~T., Genge, P.~C., Thomson, Z.,
  Weiss, M.~D., Li, X.-j., Savage, A.~K., Green, R.~R., Torgerson, T.~R.,
  Bumol, T.~F., Graybuck, L.~T., and Skene, P.~J. (2021).
\newblock Simultaneous trimodal single-cell measurement of transcripts,
  epitopes, and chromatin accessibility using tea-seq.
\newblock {\em eLife}, 10:e63632.

\bibitem[Teichmann and Efremova, 2020]{teichmann2020method}
Teichmann, S. and Efremova, M. (2020).
\newblock Method of the year 2019: single-cell multimodal omics.
\newblock {\em Nat. Methods}, 17(1):2020.

\bibitem[van~der Vaart and Wellner, 1996]{vanderVaart1996}
van~der Vaart, A.~W. and Wellner, J.~A. (1996).
\newblock {\em Weak Convergence and Empirical Processes: With Applications to
  Statistics}, pages 154--165.
\newblock Springer New York, New York, NY.

\bibitem[Vershynin, 2010]{vershynin_2012}
Vershynin, R. (2010).
\newblock Introduction to the non-asymptotic analysis of random matrices.
\newblock {\em arXiv preprint arXiv:1011.3027}.

\bibitem[Wang et~al., 2022]{wang_zhong_fan}
Wang, T., Zhong, X., and Fan, Z. (2022).
\newblock Universality of approximate message passing algorithms and tensor
  networks.
\newblock {\em arXiv preprint arXiv:2206.13037}.

\bibitem[Wang and Stephens, 2021]{JMLR:v22:20-589}
Wang, W. and Stephens, M. (2021).
\newblock Empirical bayes matrix factorization.
\newblock {\em Journal of Machine Learning Research}, 22(120):1--40.

\bibitem[Zhong et~al., 2022]{eb_pca}
Zhong, X., Su, C., and Fan, Z. (2022).
\newblock {Empirical Bayes PCA in High Dimensions}.
\newblock {\em Journal of the Royal Statistical Society Series B: Statistical
  Methodology}, 84(3):853--878.

\end{thebibliography}

\appendix

\section{\revzms{Additional results from the TEA-seq data example}}
\label{sec:app_tea_seq}

\subsection{Scree plots of the ATAC and the RNA data matrices}

\revzm{In Fig \ref{fig:sidebyside_scree_plot}, we plot the leading singular values of the ATAC data matrix $\bm X_1$ and the RNA data matrix $\bm X_2$ in the TEA-seq data. 
The singular values were computed after the pre-processing steps described at the beginning of Section 4. 
The scree plots demonstrate the validity of our low-rank modeling approach on this data example.}

\begin{figure}[tbh]
    \centering
    \begin{minipage}{0.35\linewidth}
        \centering
        \includegraphics[width=\linewidth]{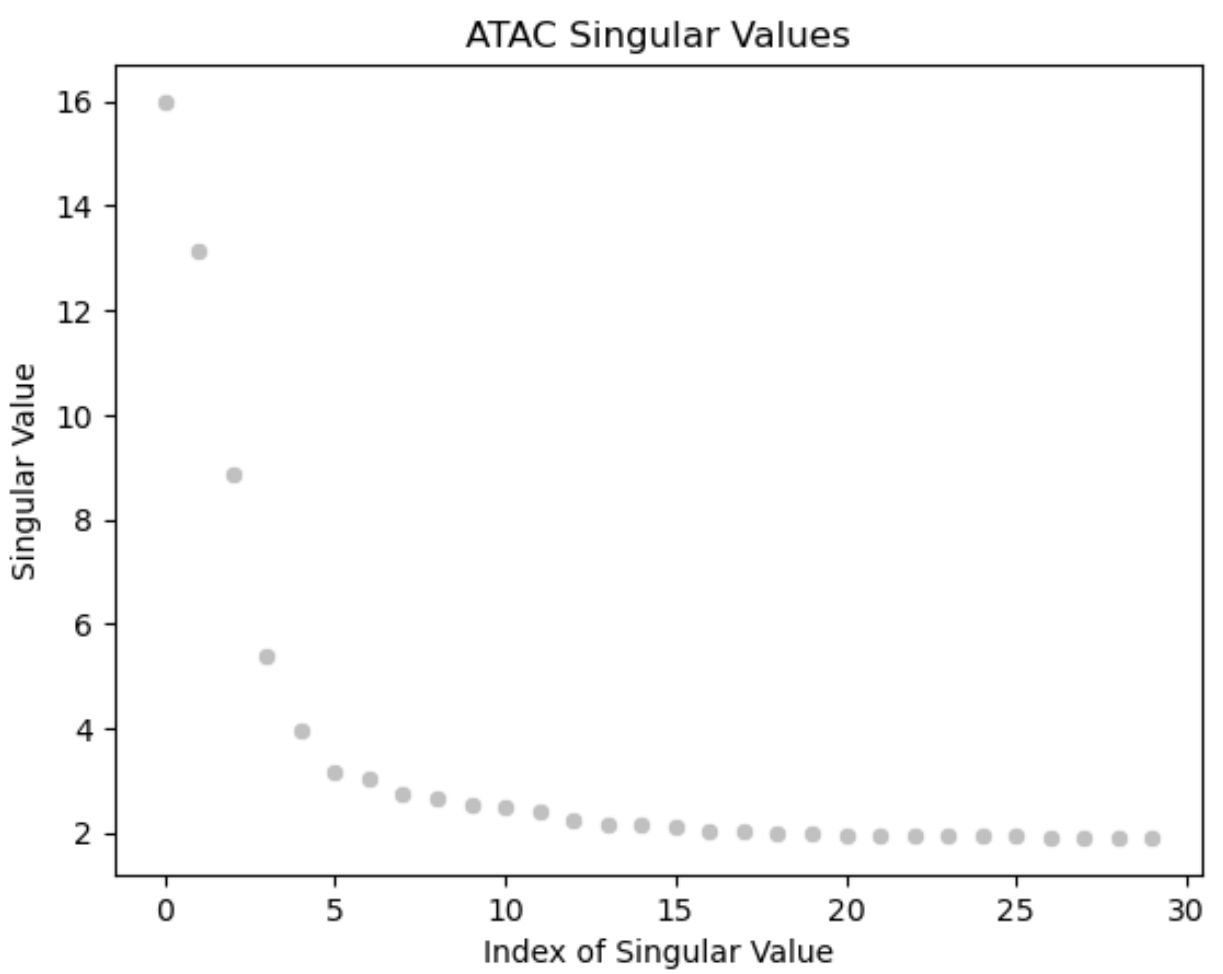}
    \end{minipage}
    \begin{minipage}{0.35\linewidth}
        \centering
        \includegraphics[width=\linewidth]{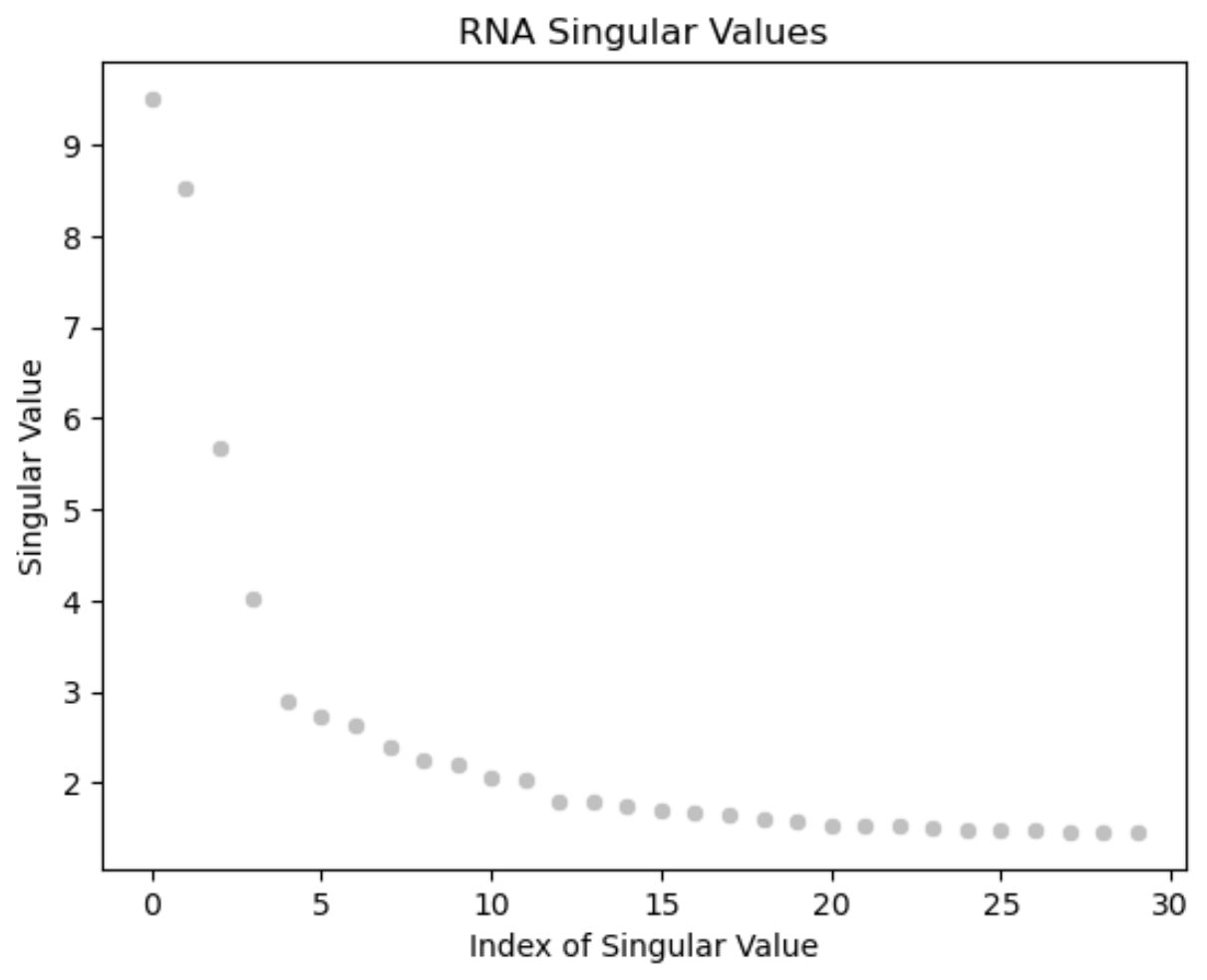}
    \end{minipage}
    \caption{Scree plots of the empirical singular values of the ATAC measurement matrix $\bm X_1$ (left) and the RNA measurement matrix $\bm X_2$ (right), after preprocessing. 
    }
    \label{fig:sidebyside_scree_plot}
\end{figure}

\subsection{Prediction sets for the \revzms{held-out} CD8 effector cell and the \revzms{held-out} pre-B cell}
\revsn{In Figures~\ref{fig:atlas_tea_seq_pred_cd8e} and \ref{fig:atlas_tea_seq_pred_pbc}, we illustrate the $95\%$ prediction sets for the latent embedding of the held-out CD8 effector cell and the held-out pre-B cell, respectively. 
These figures are set up in the same way as Figure 2.
\revzm{Overall, these prediction sets exhibit comparable behaviors to those described in Section 4.2, which suggests that the proposed prediction set construction approach is not only robust with respect to different querying modalities, but also with respect to the cell state of the querying cell.}
}


\begin{figure}[tbh]
    \centering
    \includegraphics[width = 0.7\textwidth]{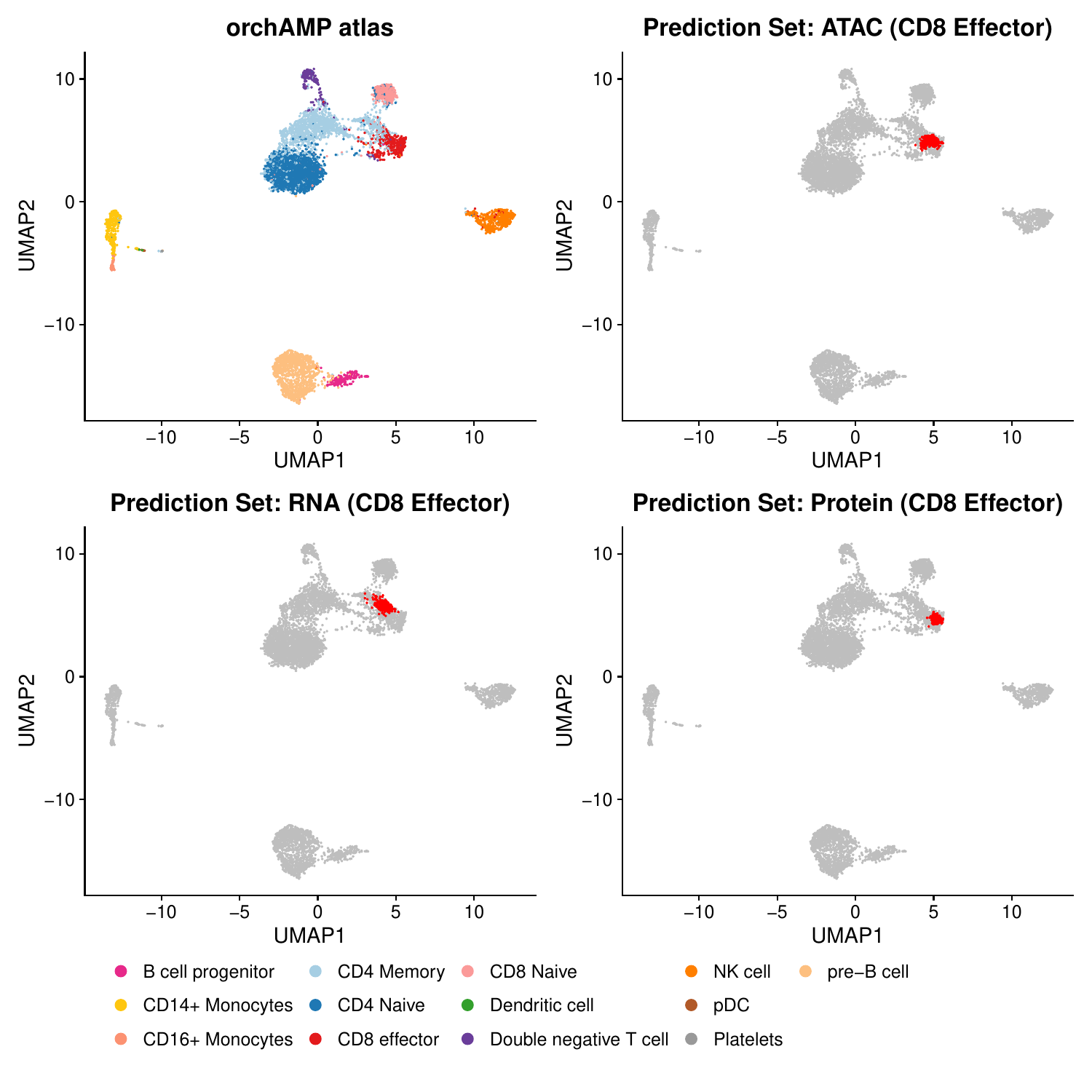}
    \caption{
	Visualization of prediction sets constructed by Algorithm 2. 
		Top left: OrchAMP cell atlas constructed 
	    on $6320$ TEA-seq cells. Cells are colored according to their cell type annotations in \cite{10.7554/eLife.63632}.
		\revzms{Visualizations of the $95\%$ prediction set ($500$ randomly sampled points from the set in red and atlas cells in grey)} when queried by the ATAC observation of a held-out \textsf{CD8 effector} cell \revzms{(top right), its RNA observation (bottom left), and its Protein observation (bottom right), respectively}.
	}
    \label{fig:atlas_tea_seq_pred_cd8e} 
\end{figure}

\begin{figure}[tbh]
    \centering
    \includegraphics[width = 0.7\textwidth]{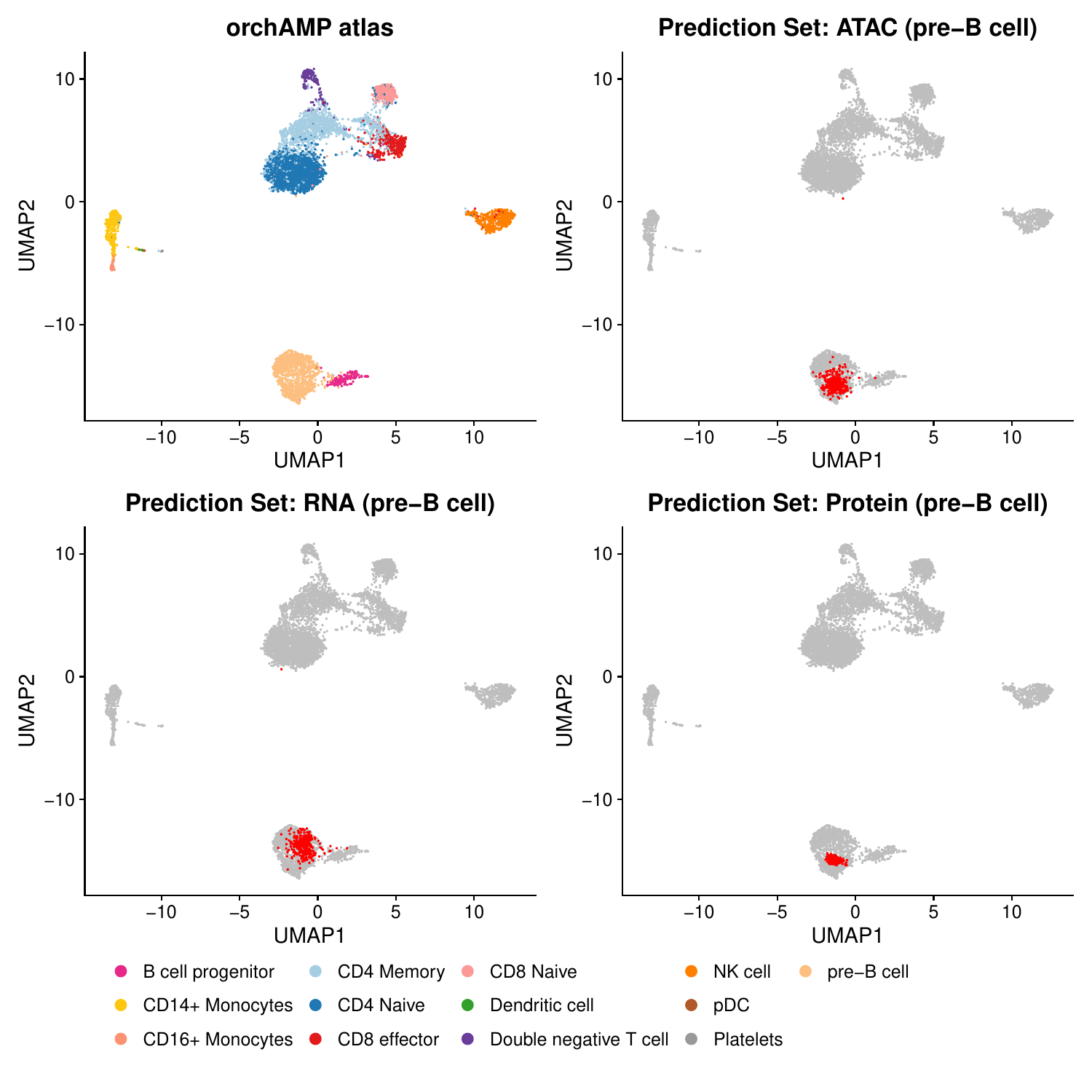}
    \caption{
	Visualization of prediction sets constructed by Algorithm 2. 
		Top left: OrchAMP cell atlas constructed 
	    on $6320$ TEA-seq cells. Cells are colored according to their cell type annotations in \cite{10.7554/eLife.63632}.
		\revzms{Visualizations of the $95\%$ prediction set ($500$ randomly sampled points from the set in red and atlas cells in grey)} when queried by the ATAC observation of a held-out \textsf{pre-B} cell \revzms{(top right), its RNA observation (bottom left), and its Protein observation (bottom right), respectively}.
		}
    \label{fig:atlas_tea_seq_pred_pbc}
\end{figure}

\section{\revzm{Benchmarking atlas construction methods on a CITE-seq data example}}
\label{sec:bench_cite_seq}

\revzm{To further benchmark OrchAMP against state-of-the-art atlas construction methods for multi-omics data, we present an additional CITE-seq data example which has two modalities (RNA, with roughly $20000$ transcripts, and Protein, with $224$ surface proteins), which consists of \revsn{$10000$ peripheral blood mononuclear cells (PBMCs) randomly sampled from the large-scale CITE-seq dataset in \citet{HAO20213573} after 
standard pre-processing 
recommended in the source study.}}
\revsn{
We retained $5000$ highly variable genes and $30$ highly variable protein markers for our analysis.} 
\revzm{Naturally, RNA is the high-dimensional modality and Protein low-dimensional.}
\revsn{\revzm{When the data was initially collected,} cells were processed across eight lanes of a 10x Genomics Chip B, which introduced noticeable batch effects. 
We employed Harmony \citep{Korsunsky2019}
to remove lane-specific artifacts \revzm{in both modalities.
}}
\revzm{In addition to WNN and MOFA+, we also included in our comparison below totalVI \cite{gayoso2021joint}, which is an atlas construction method initially designed specifically for CITE-seq data.
Harmony-corrected values were supplied to all atlas construction methods except for totalVI, which has a built-in batch correction capacity and only requires batch indices as part of its input.}




\subsection{Cell atlas construction}
\revsn{Based on a comparison between the empirical singular value distribution and the appropriately rescaled Marchenko–Pastur law, we had set the latent dimension of the high-dimensional modality to $r_1=50$. We then adopted an Empirical Bayes approach to estimate the priors: $\mu$ for the concatenated rows of $\bm U_1$ and $\wt{\bm U}_1$, and the rows of $\bm V$. Both priors were modeled using Gaussian mixtures, with $8$ mixture components for $\mu$ and $4$ for $\nu$. These choices were determined by the same procedure as described in Section 4, \revzm{and were used in the OrchAMP atlas construction reported below.}}

\revsn{Figure~\ref{fig:atlas_cite_seq} plots the UMAP representations \cite{mcinnes2018umap-software} of the \revzm{OrchAMP} atlas (top right), alongside the UMAP representations obtained from three other methods: the WNN-based integration of \citet{HAO20213573} (top left), using 50 principal components for the RNA modality and 30 for the protein modality, MOFA+ \cite{argelaguet2020mofa+} (bottom right) with 20 latent factors,  and totalVI \cite{gayoso2021joint} (bottom left) with 45 latent dimensions in the auto-encoder.}
\revzm{Visual inspection of the UMAPs suggests comparable performance of all four methods in separating the major cell types.}

\begin{figure}
    \centering
    \includegraphics[width = 0.7\textwidth]{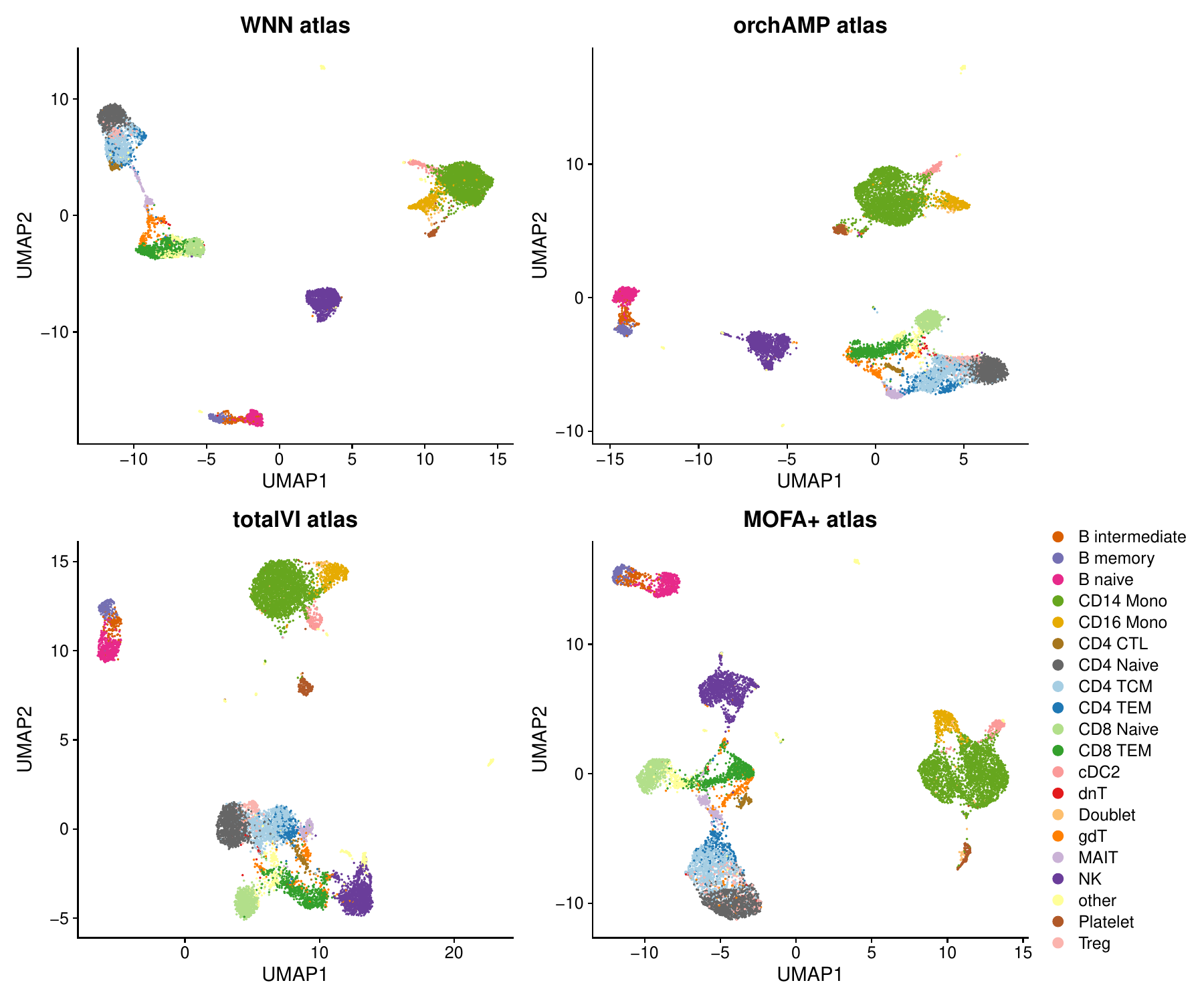}
    \caption{\revzm{UMAPs of cell atlases constructed from $10000$ CITE-seq PBMC cells by WNN (top left), OrchAMP (top right), totalVI (bottom left), and MOFA+ (bottom right). All cells colored according to cell type annotations in \cite{HAO20213573}.}}
    \label{fig:atlas_cite_seq}
\end{figure}


\revzm{Table \ref{tab:clustering_cite_seq} reports the four atlas quality metrics of the methods in comparison on the CITE-seq data, with the best of each metric in bold.
Definitions of these metrics (the average silhouette score \cite{rousseeuw1987silhouettes}, the V-measure \cite{rosenberg2007vmeasure}, the adjusted Rand index (ARI) \cite{Hubert1985} and the cell-type Local Inverse Simpson’s Index (LISI) \cite{Korsunsky2019}) are given in Section \ref{sec:clusreing_metrics}.
These metrics further confirmed that the performance of OrchAMP matched state of the art, while OrchAMP uniquely enjoys statistical optimality guarantees and theoretically justifiable prediction set coverage in downstream querying tasks.}


\begin{table}[t]
\centering
\scriptsize
\setlength{\tabcolsep}{3pt}
\renewcommand{\arraystretch}{0.95}
\caption{Comparison of clustering metrics among WNN, OrchAMP, totalVI and MOFA+ on the CITE-seq dataset. Higher average silhouette score, ARI and V-measure indicate better separation, while lower cLISI indicates improved segregation of clusters.}
\begin{tabular}{ccccc}
\toprule
Clustering Metric & WNN  & OrchAMP & totalVI & MOFA+ \\
\midrule
Average silhouette score & {\bf 0.393}  & 0.343 & 0.329& 0.329\\
V-measure & 0.791  & {\bf 0.796} & 0.775 & 0.794 \\
Adjusted Rand Index & {\bf 0.622}  & 0.610 & 0.619 & 0.605 \\
cLISI & 1.223  & {\bf 1.186} & 1.188 & 1.217 \\
\toprule
\end{tabular}
\label{tab:clustering_cite_seq}
\end{table}

\section{Definitions of the clustering metrics}
\label{sec:clusreing_metrics}
\revsn{To assess the quality of the constructed atlases in terms of separating cells of different 
states
into well-defined clusters, we have employed four metrics, namely average silhouette score \cite{rousseeuw1987silhouettes}, adjusted Rand index \cite{Hubert1985}, V-measure \cite{rosenberg2007vmeasure}, and cell-type LISI \cite{Korsunsky2019}. In this subsection, we provide brief descriptions of these clustering metrics.}
\begin{enumerate}
    \item \textbf{Average Silhouette score.} \revsn{This metric
    measures how well a point lies within its assigned cluster compared to other clusters. If we define $a(i)$ to be the average intra-cluster Euclidean distance of the  point $i$ and $b(i)$ to be the minimum average Euclidean distance of the point $i$ to the other clusters, then the silhouette score for point $i$ is given by
    \[
    s(i)=\frac{b(i)-a(i)}{\max\{a(i),b(i)\}}.
    \]
    In this paper, we report the average silhouette score for all points computed using their UMAP embeddings,
    with cluster labels 
    being the ground truth
    cell types.}

    \item \textbf{Adjusted Rand Index.} \revsn{This is a measure of how well two cluster partitions of a point set align. The reconstructed atlas shows some segregation of the points into clusters, which we recovered by Louvain clustering on the atlas points. Next, let $c$ be the number of point pairs assigned to the same cluster by both the cluster labels obtained by clustering the atlas points and those specified in the cell-type metadata. Similarly, $d$ be the pairs of points assigned to different clusters by the two clustering schemes. Then the Rand Index ($\mathrm{RI}$) is given by
    \[
    \mathrm{RI}=\frac{c+d}{{n \choose 2}}.
    \]
    Let $\mathbb E[\mathrm{RI}]$ be the expected value of the $\mathrm{RI}$ under random clustering but fixed cluster sizes, then the Adjusted Rand Index is defined as:
    \[
    \mathrm{ARI}=\frac{\mathrm{RI}-\mathbb E[\mathrm{RI}]}{\max(\mathrm{RI})-\mathbb E[\mathrm{RI}]},
    \]
    where $\max(\mathrm{RI})$ is the maximum possible value of the Rand Index given the cluster sizes of the two cluster partitions. }

    \item \textbf{V-measure.} \revsn{The V-measure provides an entropy-based evaluation scheme of two sets of cluster labels. In our case, these sets are the cluster labels obtained using Louvain clustering on the atlas points and the original cell-type metadata. The V-measure computes two quantities, namely, homogeneity $(h)$ and completeness $(c)$. Let $K$ be the true cell labels and $C$ be the Louvain clusters. Then 
    \[
    h=1-\frac{H(K\mid C)}{H(K)}, \quad \mbox{and} \quad c=1-\frac{H(C\mid K)}{H(C)} ,
    \]
    where $H(\cdot)$ refers to the Shannon entropy of a distribution. Then, the V-measure is defined as
    \[
    \text{V-measure}= \frac{2\times h \times c}{h+c}.
    \]
    Higher values of V-measure indicate that the reconstructed atlas adheres to the cluster structure encoded in the metadata.}

    \item \textbf{cLISI.} \revsn{The cell-type Local Inverse Simpson's Index (cLISI) is used to characterize how homogenous the local neighborhood of each cell is when such neighborhoods are determined by the atlas embeddings. In that direction, we computed the true cell-type composition 
    of the local neighborhood of each cell (10 nearest neighbors). Let $p_{\ell,i}$ be the proportion of cells of type $\ell$ in the neighborhood of cell $i$. Then $\mathrm{LISI}_i$ is defined as
    \[
    \mathrm{LISI}_i=\frac{1}{\sum_{\ell}p^2_{\ell,i}}.
    \]
    We report the average of the $\mathrm{LISI}_i$ over all cells as the cLISI measure.}
\end{enumerate}

\section{Proof of Proposition \ref{prop:singular_values} }
Let us denote the set of true latent factors by 
\[
\mathscr S_{m,\wt m} = \{\bm U_1,\ldots,\bm U_m,\wt{\bm U}_1,\ldots,\wt{\bm U}_{\wt m},\bm V_1,\ldots,\bm V_m\}.
\]
Observe that since $\frac{|\mathcal F_L|}{N} \rightarrow \lambda_L \in (0,1]$, using the definition of pseudo-Lipschitz functions \eqref{eq:pseudo_lips} and the Cauchy-Schwarz inequality, we can get a constant $C>0$ such that for sufficiently large values of $N$,
\begin{align*}
    &\Bigg|\frac{1}{|\mathcal F_L|}\sum_{i\in \mathcal F_L}\psi((\upca_{0,1})_{i*},\dots,(\upca_{0,m})_{i*},(\wt{\bm X}_{1})_{i*},\dots,(\wt{\bm X}_{\wt m})_{i*},(\bm U_{1})_{i*},\ldots,(\bm U_{m})_{i*},\\
    &\hskip 25em (\wt{\bm U}_{1})_{i*},\ldots,(\wt {\bm U}_{\wt m})_{i*})\\
    &\hskip 1em -\frac{1}{|\mathcal F_L|}\sum_{i\in \mathcal F_L}\psi\bigg(\hslmorp_{0,1}(\bm U_1)_{i*}+(\hslvorp_{0,1})^{1/2}(\bm Z^L_1)_{i*},\dots,\hslmorp_{0,m}(\bm U_m)_{i*}+(\hslvorp_{0,m})^{1/2}(\bm Z^L_m)_{i*},\\
    &\hskip 12em (\wt{\bm X}_{1})_{i*},\dots,(\wt{\bm X}_{\wt m})_{i*},(\bm U_{1})_{i*},\ldots,(\bm U_{m})_{i*},(\wt{\bm U}_{1})_{i*},\ldots,(\wt {\bm U}_{\wt m})_{i*}\bigg)\Bigg|\\
    & \le \frac{C}{\lambda_L}\Bigg(1+\sum_{h=1}^{m}\frac{\|\bm U_h\|_F}{\sqrt{N}}+\sum_{h=1}^{m}\frac{\|\upca_{0,h}\|_F}{\sqrt{N}}+\sum_{h=1}^{m}\frac{\|\bm U_h(\hslmorp_{0,h})^\top+\bm Z^L_h(\hslvorp_{0,h})^{1/2}\|_F}{\sqrt{N}} \\
    &\hskip 25em+\sum_{\ell=1}^{\wt m}\frac{\|\wt{\bm X}_\ell\|_F}{\sqrt{N}}+\sum_{\ell=1}^{\wt m}\frac{\|\wt{\bm U}_\ell\|_F}{\sqrt{N}}\Bigg)\\
    &\hskip 5em\times \left(\sum_{h=1}^{m}\frac{1}{\sqrt{N}}\|\upca_{0,h} - \bm U_h(\hslmorp_{0,h})^\top-\bm Z^L_h(\hslvorp_{0,h})^{1/2}\|_F\right).
\end{align*}
In the above asymptotic statements, 
conditional on
$\mathscr S_{m,\wt m}$, 
the random matrices $\{\bm Z^L_h \in \R^{N \times r_h}:h\in [m]\}$ and $\{\bm Z^R_h \in \R^{p_h \times r_h}:h \in [m]\}$ are mutually independent with iid $N(0,1)$ entries.
Therefore, to prove \eqref{eq:lim_u_pca}, it suffices to show that for all $h \in [m]$:
\begin{align}
\label{eq:sing_vec_F}
    \frac{1}{N}\left\|\upca_{0,h} - \bm U_h(\hslmorp_{0,h})^\top-\bm Z^L_h(\hslvorp_{0,h})^{1/2} \right\|^2_F \xrightarrow{a.s.} 0.
\end{align}
The rest of the proof follows using the Strong Law of Large Numbers. Similarly, since $\frac{|\mathcal F_{R,h}|}{p_h} \rightarrow \lambda_{R,h} \in (0,1]$ for all $h \in [m]$, to show \eqref{eq:lim_v_pca}, it suffices to show that
\begin{align}
\label{eq:sing_vec_G}
    \frac{1}{p_h}\left\|\vpca_{0,h} - \bm V_{h}(\hsrmorp_{0,h})^\top-\bm Z^R_h(\hsrvorp_{0,h})^{1/2} \right\|^2_F \xrightarrow{a.s.} 0, \quad \mbox{for all $h \in [m]$.}
\end{align}
In what follows, 
we focus on proving \eqref{eq:sing_vec_F}.
The proof of \eqref{eq:sing_vec_G} is similar.

First, let us condition on the set of true signal matrices given by $\mathscr S_{m,\wt m}$. Then the matrices $\wb{\bm X}_1,\ldots,\wb{\bm X}_m$ are independent. Let $\gsu_1,\ldots,\gsu_m$, be the matrices obtained by applying Gram-Schmidt orthogonalization to $\bm U_1,\ldots,\bm U_m$, respectively, with column norms set at $\sqrt{N}$. 
For each $h \in [m]$, let the projection matrix onto the column space of $\bm U_h$ and that onto the orthogonal complement of the column space of $\bm U_h$ be denoted by $\projpar_{h}$ and $\projperp_h$, respectively. Then, by the definition of projection matrices we have
\[
\projpar_h = \frac{\gsu_h(\gsu_h)^\top}{N} \quad \mbox{and} \quad  \projperp_h=\bm I_{N}-\frac{\gsu_h(\gsu_h)^\top}{N}.
\]
Since, by assumption we have $(\bm D_{h})_{kk}>\gamma^{-1/4}_h$ for all $k \in [r_h]$, using Lemma A.3 of \cite{eb_pca} and Theorem 2.9 of \cite{BENAYCHGEORGES2012120}, we have
\begin{align}
\label{eq:mean_projection_pca}
\lim_{N \rightarrow \infty}\frac{1}{N}(\gsu_h)^\top\upca_{0,h} = \lim_{N \rightarrow \infty}\frac{1}{N}\bm U^\top_h\upca_{0,h} \overset{a.s.}{=} \hslmorp_{0,h}, \quad \mbox{for all $h \in [m]$.}
\end{align}
Let us fix a set $\mathcal T$ such that on $\mathcal T$, \eqref{eq:mean_projection_pca} holds and $\mathbb P_{\mu}(\mathcal T)=1$.
Henceforth, we shall work on $\mathcal T$. Since, $\frac{1}{N}(\upca_{0,h})^\top\upca_{0,h} = \bm I_N$, we can conclude using \eqref{eq:mean_projection_pca} that
\[
\frac{1}{N}(\projperp_h\upca_{0,h})^\top(\projperp_h\upca_{0,h}) = \bm I_N-\hslmorp_{0,h}(\hslmorp_{0,h})^\top = \hslvorp_{0,h}.
\]
Here, the last equality follows from the definition of $\hslvorp_{0,h}$. Since $\hslvorp_{0,h}$ is positive definite for all $h \in [m]$, $\projperp\upca_{0,h}$ has full column rank for all $h \in [m]$.
Let $\gsuperp_h$ be the Gram-Schmidt orthogonalization of the columns of  $\projperp_h\upca_{0,h}$. For all $h \in [m]$, we have:
\begin{align}
\label{eq:decomp_pca}
\upca_{0,h} = \projpar_h\upca_{0,h}+\projperp_h\upca_{0,h} = \gsu_h\bm \Omega_h+\gsuperp_h\bm\Omega^\perp_h,
\end{align}
where $\bm \Omega_h=N^{-1}(\gsu_h)^\top\upca_{0,h}$ and $\bm\Omega^\perp_h=N^{-1}(\gsuperp_h)^\top\upca_{0,h}$ for all $h \in [m]$. 
Clearly, on $\mathcal T$, 
$\bm \Omega_{h}\rightarrow \hslmorp_{0,h}$ and $(\bm\Omega^\perp_h)^\top\bm\Omega^\perp_h \rightarrow \hslvorp_{0,h}$, as $N \rightarrow \infty$. 
Let $\{\bm O_h \in \R^{N \times N}:h \in [m]\}$ 
be a collection of orthogonal matrices such that for all $h \in [m]$, we have $\bm O_h\gsu_h=\gsu_h$. 
By Lemma A.3 of \cite{eb_pca}, this implies $\bm O_h\bm U_h=\bm U_h$ for all $h \in [m]$. As the conditional law of $\bm W_h$ given $\mathscr S_{m,\wt m}$ is invariant under multiplication of $\bm O_h$, we have for all $h \in [m]$:
\begin{align}
    \bm O_h\wb{\bm X}_h = \frac{1}{N}\bm O_h\bm U_h\bm D_h\bm V^\top_h+\frac{1}{\sqrt{N}}\bm O_h\bm W_h \overset{d}{=} \wb{\bm X}_h,
\end{align}
conditioned on $\mathscr S_{m,\wt m}$.
Since, $\upca_{0,h}$ contains the leading $r_h$ left singular vectors of $\wb{\bm X}_h$, this implies $\bm O_h\upca_{0,h}\overset{d}{=}\upca_{0,h}$ conditioned on $\mathscr S_{m,\wt m}$. This in turn implies that for all $h \in [m]$, conditioned on $\mathscr S_{m,\wt m}$, we have $\bm O_h\gsuperp_h \overset{d}{=} \gsuperp_h$. Therefore, conditioned on $\mathscr S_{m,\wt m}$, the matrices $\gsuperp_h$, for $h \in [m]$, are Haar uniformly distributed on the Steifel Manifolds 
\[
\mathcal U_h=\Bigg\{\gsuperp_h:\bm U^\top_h\gsuperp_h=0;\,\frac{1}{N}(\gsuperp_h)^\top\gsuperp_h=\bm I_{r_h}\Bigg\},
\]
respectively, and are independent. 
Using properties of the Haar measure on the Steifel manifold, we can construct a collection of independent matrices $\{\bm Z^L_h \in \R^{N \times r_h}: h \in [m]\}$, with the entries independently and identically distributed as $N(0,1)$ conditioned on $\mathscr S_{m,\wt m}$, such that for all $h \in [m]$ we have $\gsuperp_h=\projperp_h\bm Z^L_h\left(\frac{(\bm Z^L_h)^\top\projperp_h\bm Z^L_h}{N}\right)^{-1/2}$. Using \eqref{eq:decomp_pca}, for $h \in [m]$, we have the following relation:
\[
\upca_{0,h}=\gsu_h\bm \Omega_h+\projperp_h\bm Z^L_h\left(\frac{(\bm Z^L_h)^\top\projperp_h\bm Z^L_h}{N}\right)^{-1/2}\bm \Omega^\perp_h.
\]
By Strong Law of Large Numbers and the fact that $\projperp_h$ is idempotent, $\frac{(\bm Z^L_h)^\top\projperp_h\bm Z^L_h}{N} \stackrel{a.s.}{\rightarrow} \bm I_N$ and $(1/N)\|\projperp_h\bm Z^L_h-\bm Z^L_h\|^2_F 
\stackrel{a.s.}{\rightarrow}
0$, as $N 
\rightarrow 
\infty$. Hence, using Lemma A.3 of \cite{eb_pca}, 
we have
\[
\mathbb P\left(\left\{\frac{1}{N}\|\upca_{0,h} - \bm U_h(\hslmorp_{0,h})^\top-\bm Z^L_h(\hslvorp_{0,h})^{1/2}\|^2_F \rightarrow 0\; \forall h\in[m]\right\}\cap \mathcal T\bigg|\mathscr S_{m,\wt m}\right) = 1. 
\]
Since $\mathbb P(\mathcal T)=1$, 
we can conclude that
\[
\mathbb P\left(\frac{1}{N}\|\upca_{0,h} - \bm U_h(\hslmorp_{0,h})^\top-\bm Z^L_h(\hslvorp_{0,h})^{1/2}\|^2_F \rightarrow 0\;\forall h\in[m]\right) = 1.
\]
This completes the proof of \eqref{eq:sing_vec_F}.
  
\section{Consistency of nuisance parameter estimation}
\subsection{Proof of Lemma \ref{lem:consistency_nuisance}}
The first assertion follows using Assumption \ref{asm:prior_1_mom}(5) and \cite[Theorem 2.9]{BENAYCHGEORGES2012120}. Since $\mathbb{E}_\mu[(\wt{\bm U}_{\ell})_{i*}(\wt{\bm U}_{\ell})^\top_{i*}]=\bm I_{{\wt r}_\ell}$ for all $\ell \in [\wt m]$ and $i \in [N]$, and for all $\ell \in [\wt m]$ the matrices $\bm L_{\ell}$ are symmetric; the second assertion can be shown using the Strong Law of Large Numbers. Finally, the third assertion follows using the first assertion and Theorem \ref{prop:singular_values}.

\subsection{Proof of Lemma \ref{lem:weak_conv_emp_bayes}}
Consider $\{\hslmp_{0,h},\hslvp_{0,h},\hsrmp_{0,h},\hsrvp_{0,h}: h \in [m]\}$ and $\{\wh{\bm L}_\ell: \ell \in [\wt m]\}$ defined in \eqref{eq:scale-est} and \eqref{eq:sqrt_L}, respectively. For some sequences of priors $\{\mu_N\} \subseteq \mathcal P$ and $\{\nu_{h,N}\in \mathcal P_{\nu_h}: h \in [m], N \in \mathbb N\}$, let us define
the functions $\{f^L_N\}$ and $\{f^R_{h,N}: h \in [m]\}$ as follows: 
\begin{align}
\label{eq:f_l_n}
&f^L_N(x_1,\ldots,x_m,\wt{x}_1,\ldots,\wt{x}_{\wt m};\mu_N)\\
&:= \int_{u,\tilde{u}} \prod_{h=1}^{m}\exp\left(-\frac{[x_h-\hslmp_{0,h}u_h]^\top(\hslvp_{0,h})^{-1}[x_h-\hslmp_{0,h}u_h]}{2}\right)\times\\
& \hskip 4.5em \prod_{\ell=1}^{\wt m}\exp\left(-\frac{[\wt x_\ell-\wh{\bm L}_{\ell}\wt u_\ell]^\top[\wt x_\ell-\wh{\bm L}_{\ell}\wt u_\ell]}{2}\right)\,\times\\
& \hskip 8.5em d\mu_N(u_1,\ldots,u_m,\wt u_1,\ldots,\wt u_{\wt m}),
\end{align}
and
\[
f^R_{h,N}(y_h;\nu_{h,N}):=\int_{v_h}\exp\left(-\frac{[y_h-\hsrmp_{0,h}v_h]^\top(\hsrvp_{0,h})^{-1}[y_h-\hsrmp_{0,h}v_h]}{2}\right)d\nu_{h,N}(v_h).
\]
Similarly, define the functions $f^L$ and $\{f^R_{h}: h \in [m]\}$ as follows:
\begin{align}
\label{eq:f_l}
&f^L(x_1,\ldots,x_m,\wt{x}_1,\ldots,\wt{x}_{\wt m};\mu)\\
&:= \int_{u,\tilde{u}} \prod_{h=1}^{m}\exp\left(-\frac{[x_h-\hslmorp_{0,h}u_h]^\top(\hslvorp_{0,h})^{-1}[x_h-\hslmorp_{0,h}u_h]}{2}\right)\times\\
& \hskip 4.5em \prod_{\ell=1}^{\wt m}\exp\left(-\frac{[\wt x_\ell-\bm L_{\ell}\wt u_\ell]^\top[\wt x_\ell-\bm L_{\ell}\wt u_\ell]}{2}\right)\,\times\\
& \hskip 8.5em d\mu(u_1,\ldots,u_m,\wt u_1,\ldots,\wt u_{\wt m}),
\end{align}
and
\[
f^R_{h}(y_h;\nu_{h}):=\int_{v_h}\exp\left(-\frac{[y_h-\hsrmorp_{0,h}v_h]^\top(\hsrvorp_{0,h})^{-1}[y_h-\hsrmorp_{0,h}v_h]}{2}\right)d\nu_{h}(v_h).
\]

\begin{lem}
\label{lem:consistency_of_NPMLE}
If the sequence of estimators $\{\mu_N\}$ and $\{\nu_{h,N}: h \in [m], N \in \mathbb N\}$ are approximately maximum likelihood estimators of the true priors $\mu$ and $\{\nu_h: h \in [m]\}$, respectively, i.e., they satisfy the following relations
\begin{align}
\label{eq:kl_regularity_1}
    \liminf_{N \rightarrow \infty}\frac{1}{N}\sum_{i=1}^{N}\log\frac{f^L_N((\upca_{0,1})_{i*},\ldots,(\upca_{0,m})_{i*},(\wt{\bm X}_1)_{i*},\ldots,(\wt{\bm X}_{\wt m})_{i*};\mu_N)}{f^L((\upca_{0,1})_{i*},\ldots,(\upca_{0,m})_{i*},(\wt{\bm X}_1)_{i*},\ldots,(\wt{\bm X}_{\wt m})_{i*};\mu)} \ge 0,
\end{align}
and 
\begin{align}
\label{eq:kl_regularity_2}
   \liminf_{p_h \rightarrow \infty}\frac{1}{p_h}\sum_{j=1}^{p_h}\log\frac{f^R_{h,N}((\vpca_{0,h})_{j*};\nu_{h,N})}{f^R_h((\vpca_{0,h})_{j*};\nu_h)} \ge 0 \quad \mbox{for $h \in [m]$;}
\end{align}
then as $N,p_h \rightarrow \infty$, we have $\mu_N \overset{w}{\rightarrow} \mu$ and $ \nu_{h,N} \overset{w}{\rightarrow} \nu_{h}$ for all $h \in [m]$, almost surely.
\end{lem}
The proof of this theorem follows by the arguments presented in Lemma B.2 of \citet{eb_pca} and the consistency of $\hslmp_{0,h},\hsrmp_{0,h},\hslvp_{0,h},\hsrvp_{0,h}$ and $\hatl_\ell$ in estimating $\hslmorp_{0,h},\hsrmorp_{0,h},\hslvorp_{0,h},\hsrvorp_{0,h}$ and $\bm L_\ell$, respectively, for all $h \in [m]$ and $\ell \in [\wt m]$ (Lemma \ref{lem:consistency_nuisance}). Hence, we omit the details. Now, consider $\wh{\mu}$ and $\wh{\nu}_{h}$ defined by \eqref{eq:emp_bayes_1} and \eqref{eq:emp_bayes_2}. 

\begin{proof}[Proof of Lemma \ref{lem:weak_conv_emp_bayes}]
The proof of Lemma \ref{lem:weak_conv_emp_bayes} follows from Lemma \ref{lem:consistency_of_NPMLE}, the techniques used to prove Corollary B.7 of \cite{eb_pca}.
\end{proof}

\section{Asymptotics of the oracle AMP}
To analyze the asymptotics of the AMP iterates with the denoisers constructed using the estimated priors $\wh{\mu}$ and $\wh{\nu}_h$, we first analyze the asymptotics of the Oracle AMP iterates defined in \eqref{eq:orc_amp_basic}. Then we show that as $N \rightarrow \infty$, the differences between the two sets of iterates converge to zero. 

In that direction, let us consider the following theorem:
\begin{thm}
\label{thm:asymptotics of amp_iterates}
Consider the iterates $\{\uor_{t,h}, \vor_{t,h}: t \ge 0, h \in [m]\}$ given by \eqref{eq:orc_amp_basic}. Then, for any sequence of deterministic subsets $\mathcal F_L \subset [N]$ satisfying $\frac{|\mathcal F_L|}{N} \rightarrow \lambda_{L} \in (0,1]$, as $N \rightarrow \infty$ and any \emph{pseudo-Lipschitz} function $\psi:\mathbb R^{2(r+\wt r)} \rightarrow \mathbb R$, we have the following:
\begin{align}
     &\lim_{N \rightarrow \infty}\frac{1}{|\mathcal F_L|}\sum_{i \in \mathcal F_L}\psi((\uor_{t,1})_{i*},\ldots,(\uor_{t,m})_{i*},(\wt{\bm X}_1)_{i*},\ldots,(\wt{\bm X}_{\wt m})_{i*},(\bm U_{1})_{i*},\ldots,(\bm U_{m})_{i*},\\
     &\hskip 25em (\wt{\bm U}_{1})_{i*},\ldots,(\wt{\bm U}_{\wt m})_{i*})\\
     &\hskip 4em\overset{a.s.}{=} \mathbb{E}_\mu\left[\psi(Y^\orc_{t,1},\ldots,Y^\orc_{t,m},\wt{Y}_{0,1},\ldots,\wt{Y}_{0,\wt m},U_1,\ldots,U_m,\tilde U_1,\ldots,\wt U_{\wt m})\right].
\end{align}
where $\{Y^\orc_{t,h}:h \in [m]\}$ and $\{\wt{Y}_{0,\ell}:\ell \in [\wt m]\}$ are defined by \eqref{eq:y_orcs} and the collection of random vectors $(U_1,\ldots,U_m,\wt{U}_1,\ldots,\wt{U}_{\wt m}) \sim \mu$.
Next, for all $h \in [m]$, if we consider any sequence of deterministic subsets $\mathcal F_{R,h} \subset [p_{h}]$ that satisfy $\frac{|\mathcal F_{R,h}|}{p_{h}} \rightarrow \lambda_{R,h} \in (0,1]$ as $p_h \rightarrow \infty$ and \emph{pseudo-Lipschitz} functions $\varphi_h:\mathbb R^{2r_h} \rightarrow \mathbb R$, we have the following:
\begin{align}
    \lim_{p_h \rightarrow \infty}\frac{1}{|\mathcal F_{R,h}|}\sum_{i \in \mathcal F_{R,h}}\varphi_h((\vor_{t,h})_{i*},(\bm V_{h})_{i*}) \overset{a.s.}{=} \mathbb{E}_{\nu_h}\left[\varphi_h(Y^{R,\orc}_{t,h},V_h)\right],
\end{align}
where $\{Y^{R,\orc}_{t,h}:h \in [m]\}$ are defined by \eqref{eq:y_orc_2} and $V_h \sim \nu_h$.
\end{thm}

Let $\projperp_{\upca_{0,h}}$ and $\projperp_{\vpca_{0,h}}$ be the projection matrices onto the orthogonal complements of $\upca_{0,h}$ and $\vpca_{0,h}$. For all $h \in [m]$, let us consider the matrices $\adX_h$ defined as follows:
\begin{align}
    \adX_h = \frac{1}{N}\upca_{0,h}\bm D_{0,h}(\vpca_{0,h})^\top+\projperp_{\upca_{0,h}}\left(\frac{1}{N}\bm U_h \bm D_h \bm V_h^\top + \bm W^\perp_h\right)\bm P^{\perp}_{\vpca_{0,h}},
\end{align}
where $\bm W^\perp_h\overset{d}{=}\bm W_h/\sqrt{N}$ and is independent of $\bm X_1, \ldots, \bm X_m, \wt{\bm X}_1,\ldots,\wt{\bm X}_{\wt m}$. \revsn{Using the techniques from Lemma B.3 of \citet{montanari2021} and \citet{BENAYCHGEORGES2012120}, we can show that there exists a constant $\varepsilon_0 > 0$ such that for any $\varepsilon \in (0, \varepsilon_0)$, if the singular value matrices satisfy
\begin{align}
\label{eq:conv_sing_val_svd}
\max_{k \in [r_h]} \left| (\bm D_{0,h})_{kk} - \sqrt{ \frac{ \left( \gamma_h (\bm D_h)^2_{kk} + 1 \right)\left( (\bm D_h)^2_{kk} + 1 \right) }{ \gamma_h (\bm D_h)^2_{kk} } } \right| \le \varepsilon,
\quad \text{for all } h \in [m],
\end{align}
and the left singular vectors satisfy
\begin{align}
\label{eq:conv_u_svd}
\min_{k \in [r_h]} \left\{ \frac{1}{N} \left| (\upca_{0,h})^\top_{*k} (\bm U_h)_{*k} \right|^2 - (s^L_{0,h})_k \right\} \ge -\varepsilon,
\quad \text{for all } h \in [m],
\end{align}
and the right singular vectors satisfy
\begin{align}
\label{eq:conv_v_svd}
\min_{k \in [r_h]} \left\{ \frac{1}{p_h} \left| (\vpca_{0,h})^\top_{*k} (\bm V_h)_{*k} \right|^2 - (s^R_{0,h})_k \right\} \ge -\varepsilon,
\quad \text{for all } h \in [m],
\end{align}
for $(s^L_{0,h},s^R_{0,h})$ defined in \eqref{eq:initializers}, then the following coupling bound holds:
\begin{align}
\label{eq:coupling}
    &\left\| \mathbb{P}\left( \{ \bm X_h \}_{h=1}^{m}, \{ \wt{\bm X}_\ell \}_{\ell=1}^{\wt m} \in \cdot \;\middle|\; \{\bm D_{0,h}, \upca_{0,h}, \vpca_{0,h} \}_{h=1}^m \right) \right. \nonumber \\
    &\hspace{4em} - \left. \mathbb{P}\left( \{ \adX_h \}_{h=1}^{m}, \{ \wt{\bm X}_\ell \}_{\ell=1}^{\wt m} \in \cdot \;\middle|\; \{\bm D_{0,h}, \upca_{0,h}, \vpca_{0,h} \}_{h=1}^m \right) \right\|_{\mathrm{TV}} \nonumber \\
    &\hspace{18em} \le \frac{1}{c(\varepsilon)} e^{-N c(\varepsilon)},
\end{align}
with probability at least $1 - e^{-N c(\varepsilon)}$, where $c(\varepsilon) > 0$ is a constant depending only on $\varepsilon$. Fix $\varepsilon \in [\varepsilon_0/2,\varepsilon_0]$. Using (2.1) of \cite{eb_pca} and Proposition 5.1, we can get a $N_0(\varepsilon)>0$ such that \eqref{eq:conv_sing_val_svd}, \eqref{eq:conv_u_svd} and \eqref{eq:conv_v_svd} holds for all $N \ge N_0(\varepsilon)$. Let us consider the following modified AMP iterations:
\begin{align}
\label{eq:orc_amp_simpler}
    &\advs_{t,h} = v^\orc_{t,h}(\adv_{t,h};\nu_h),\\
    &\adu_{t,h} = \adX_h\advs_{t,h}-\adus_{t-1,h}\gamma_h\left(\jacror_{t,h}(\adv_{t,h};\nu_h)\right)^\top\\
    &\adus_{t,h} = u^\orc_{t,h}(\adus_{t,1},\ldots,\adus_{t,m},\wt{\bm X}_1,\ldots,\wt{\bm X}_{\wt m};\mu),\\
    &\adv_{t+1,h} = (\adX_h)^\top\adus_{t,h}-\advs_{t,h}\left(\jaclor_{t,h}(\adus_{t,1},\ldots,\adus_{t,m},\wt{\bm X}_1,\ldots,\wt{\bm X}_{\wt m};\mu)\right)^\top,
\end{align}
with $\adv_{0,h}=\vpca_{0,h}$, $\adu_{-1,h}=\upca_{0,h}(\hslvorp_{0,h})^{1/2}$ and the same set of state-evolution recursions given by \eqref{eq:state_evol_gen}. In addition, the matrix-valued functions $\{\jacror_{t,h},\jaclor_{t,h}: t \ge 0, h \in [m]\}$ are defined in \eqref{eq:jacror} and \eqref{eq:jaclor} respectively. Then, using steps similar to Eq.(A.10)– Eq.(A.11) of \cite{montanari2021} and \eqref{eq:coupling}, we can show that there exists a coupling of the laws of $\{\bm X_h\}_{h=1}^m, \{\adX_h\}_{h=1}^{m}$ and $\{\wt{\bm X}_\ell\}_{\ell=1}^{\wt m}$ such that for any \emph{pseudo-Lipschitz} function $\psi:\mathbb R^{2(r+\wt r)} \rightarrow \mathbb R$
\begin{align}
     &\mathbb P\bigg[\frac{1}{N}\sum_{i=1}^N\psi((\adu_{t,1})_{i*},\dots,(\adu_{t,m})_{i*},(\wt{\bm X}_1)_{i*},\dots,(\wt{\bm X}_{\wt m})_{i*},\\
     &\hskip 15em (\bm U_{1})_{i*},\ldots,(\bm U_{m})_{i*},(\wt{\bm U}_{1})_{i*},\ldots,(\wt{\bm U}_{\wt m})_{i*})\\
     &\hskip 1.5em\neq \frac{1}{N}\sum_{i=1}^N\psi((\bm U^\orc_{t,1})_{i*},\dots,(\bm U^\orc_{t,1})_{i*},(\wt{\bm X}_1)_{i*},\dots,(\wt{\bm X}_{\wt m})_{i*},\\
     &\hskip 15em (\bm U_{1})_{i*},\ldots,(\bm U_{m})_{i*},(\wt{\bm U}_{1})_{i*},\ldots,(\wt{\bm U}_{\wt m})_{i*})\bigg]\\
      &\le \frac{1}{c(\varepsilon)}e^{-Nc(\varepsilon)}, \quad \mbox{for $c(\varepsilon)>0$ defined in \eqref{eq:coupling} and for all $N \ge N_0(\varepsilon)$.}
\end{align}
Furthermore, for any \emph{pseudo-Lipschitz} functions $\varphi_h:\mathbb R^{2r_h} \rightarrow \mathbb R$
\begin{align}
   & \mathbb P\Bigg[\frac{1}{p_h}\sum_{i=1}^{p_h}\varphi_h((\adv_{t,h})_{i*},(\bm V_{h})_{i*}) 
    \neq \frac{1}{p_h}\sum_{i=1}^{p_h}\varphi_h((\bm V^\orc_{t,h})_{i*},(\bm V_{h})_{i*})\Bigg]\\
    & \le \frac{1}{c(\varepsilon)}e^{-Nc(\varepsilon)}, \quad \mbox{for all $p_h \ge \lceil N_0(\varepsilon)\gamma_h\rceil$.}
\end{align}
Since $p_h/N \rightarrow \gamma_h$ for all $h \in [m]$, using the Borel Cantelli Lemmas, the above displays imply
\begin{align}
\label{eq:modified_asymp_1}
     &\lim\limits_{N \rightarrow \infty}\frac{1}{N}\sum_{i=1}^N\psi((\adu_{t,1})_{i*},\dots,(\adu_{t,m})_{i*},(\wt{\bm X}_1)_{i*},\dots,(\wt{\bm X}_{\wt m})_{i*},(\bm U_{1})_{i*},\ldots,(\bm U_{m})_{i*},\nonumber\\
     & \hskip 15em(\wt{\bm U}_{1})_{i*},\ldots,(\wt{\bm U}_{\wt m})_{i*})\\
     &\overset{a.s.}{=} \lim\limits_{N \rightarrow \infty}\frac{1}{N}\sum_{i=1}^N\psi((\bm U^\orc_{t,1})_{i*},\dots,(\bm U^\orc_{t,1})_{i*},(\wt{\bm X}_1)_{i*},\dots,(\wt{\bm X}_{\wt m})_{i*},\\
     &\hskip 15em (\bm U_{1})_{i*},\ldots,(\bm U_{m})_{i*},(\wt{\bm U}_{1})_{i*},\ldots,(\wt{\bm U}_{\wt m})_{i*}),
\end{align}
and
\begin{align}
\label{eq:modified_asymp_2}
   \lim\limits_{p_h \rightarrow \infty}\frac{1}{p_h}\sum_{i=1}^{p_h}\varphi_h((\adv_{t,h})_{i*},(\bm V_{h})_{i*}) \overset{a.s.}{=} \lim\limits_{p_h \rightarrow \infty}\frac{1}{p_h}\sum_{i=1}^{p_h}\varphi((\bm V^\orc_{t,h})_{i*},(\bm V_{h})_{i*}).\nonumber\\
\end{align}
Since, $|\mathcal F_L|/N \rightarrow \lambda_L$ and $|\mathcal F_{R,h}|/p_h \rightarrow \lambda_{R,h}$ for all $h \in [m]$, therefore one can also conclude that
\begin{align}
\label{eq:modified_asymp_1}
     &\lim\limits_{N \rightarrow \infty}\frac{1}{|\mathcal F_L|}\sum_{i \in \mathcal F_L}\psi((\adu_{t,1})_{i*},\dots,(\adu_{t,m})_{i*},(\wt{\bm X}_1)_{i*},\dots,(\wt{\bm X}_{\wt m})_{i*},(\bm U_{1})_{i*},\ldots,(\bm U_{m})_{i*},\nonumber\\
     & \hskip 15em(\wt{\bm U}_{1})_{i*},\ldots,(\wt{\bm U}_{\wt m})_{i*})\\
     &\overset{a.s.}{=} \lim\limits_{N \rightarrow \infty}\frac{1}{|\mathcal F_L|}\sum_{i \in \mathcal F_L}\psi((\bm U^\orc_{t,1})_{i*},\dots,(\bm U^\orc_{t,1})_{i*},(\wt{\bm X}_1)_{i*},\dots,(\wt{\bm X}_{\wt m})_{i*},\\
     &\hskip 15em (\bm U_{1})_{i*},\ldots,(\bm U_{m})_{i*},(\wt{\bm U}_{1})_{i*},\ldots,(\wt{\bm U}_{\wt m})_{i*}),
\end{align}
and
\begin{align}
\label{eq:modified_asymp_2}
   \lim\limits_{p_h \rightarrow \infty}\frac{1}{|\mathcal F_{R,h}|}\sum_{i \in \mathcal F_{R,h}}\varphi_h((\adv_{t,h})_{i*},(\bm V_{h})_{i*}) \overset{a.s.}{=} \lim\limits_{ p_h\rightarrow \infty}\frac{1}{|\mathcal F_{R,h}|}\sum_{i \in \mathcal F_{R,h}}\varphi((\bm V^\orc_{t,h})_{i*},(\bm V_{h})_{i*}).\nonumber\\
\end{align}
Hence, it is enough to analyze the AMP iterations given by \eqref{eq:orc_amp_simpler}. }
\subsection{Proof of Theorem \ref{thm:asymptotics of amp_iterates}}
Due to the above-mentioned arguments, we shall analyze the AMP iterations given by \eqref{eq:orc_amp_simpler}. In that direction, for all $t\ge 0$, let us consider random vectors $\{\adxl_{t,h},\adxr_{t,h} \in \R^{r_h}:h \in [m]\}$ generated as follows:
\begin{align}
    \adxl_{t,h}&\sim N_{r_h}\bigg(\alphalt_{t,h}U_h+\betalt_{t,h}Y_{0,h},\tault_{t,h}\bigg)\\
    \adxr_{t,h}&\sim N_{r_h}\bigg(\alphart_{t,h}V_h+\betart_{t,h}Y^R_{0,h},\taurt_{t,h}\bigg),
\end{align}
In the above equation, $(U_1,\ldots,U_m,\wt U_1, \ldots, \wt U_{\wt m}) \sim \mu$, $V_h \sim \nu_h$ for all $h \in [m]$ and $\{Y_{0,h},Y^R_{0,h}:h \in [m]\}$ are defined in \eqref{eq:compund_decision_singular_vector} and \eqref{eq:compound_decision_model_2} respectively. Furthermore, for all $h \in [m]$, the matrices \[\alphalt_{t,h},\betalt_{t,h},\tault_{t,h},\alphart_{t,h},\betart_{t,h},\taurt_{t,h} \in \R^{r_h \times r_h}\] are recursively defined as follows:
\begin{align}
    &\alphart_{0,h}=\bm 0, \quad \betart_{0,h}=\bm I_{r_h}, \quad \taurt_{0,h}=\bm 0,\quad \alphalt_{-1,h}=\bm 0, \quad \betalt_{-1,h}=(\hslvorp_{0,h})^{1/2}, \quad \tault_{-1,h}=\bm 0,
\end{align}
\begin{align}
\label{eq:state_evol_aux}
    \alphalt_{t,h} &= \gamma_h\bigg\{\mathbb E_{\nu_h}[v^\orc_{t,h}(\adxr_{t,h};\nu_h)V^\top_h]\bm D_h-\mathbb E_{\nu_h}[v^\orc_{t,h}(\adxr_{t,h};\nu_h)(Y^R_{0,h})^\top]\hsrmorp_{0,h}\bm D_h\bigg\},\\
    \betalt_{t,h} &= \gamma_h\bigg\{\mathbb E_{\nu_h}[v^\orc_{t,h}(\adxr_{t,h};\nu_h)(Y^R_{0,h})^\top]\bm D^\ast_h-\mathbb E_{\nu_h}[v^\orc_{t,h}(\adxr_{t,h};\nu_h)V^\top_h]\bm D_h(\hslmorp_{0,h})^\top\\
    &\hskip 2em+\mathbb E_{\nu_h}[v^\orc_{t,h}(\adxr_{t,h};\nu_h)(Y^R_{0,h})^\top]\hsrmorp_{0,h}\bm D_h(\hslmorp_{0,h})^\top\\
    &\hskip 2em-\mathbb E_{\nu_h}[\mathsf{d}v^\orc_{t,h}(\adxr_{t,h};\nu_h)]\mathbb E_\mu[u^\orc_{t-1,h}(\adxl_{t-1,h},\ldots,\adxl_{t-1,m},\wt Y_{0,1},\ldots,\wt Y_{0,\wt m})(Y_{0,h})^\top]\bigg\},\\
    \tault_{t,h} &= \gamma_h\bigg\{\mathbb E_{\nu_h}[v^\orc_{t,h}(\adxr_{t,h};\nu_h)v^\orc_{t,h}(\adxr_{t,h};\nu_h)^\top]\\
    &\hskip 2em-\mathbb E_{\nu_h}[v_{t,h}(\adxr_{t,h};\nu_h)(Y^R_{0,h})^\top]\mathbb E_{\nu_h}[(Y^R_{0,h})v_{t,h}(\adxr_{t,h};\nu_h)^\top]\bigg\},\\
    \alphart_{t+1,h} &= \mathbb E_\mu[u^\orc_{t,h}(\adxl_{t,1},\ldots,\adxl_{t,m},\wt{Y}_{0,1},\ldots,\wt{Y}_{0,\wt m};\mu)U^\top_h]\bm D_h\\
    &\hskip2em-\mathbb E[u^\orc_{t,h}(\adxl_{t,1},\ldots,\adxl_{t,m},\wt{Y}_{0,1},\ldots,\wt{Y}_{0,\wt m};\mu)(Y_{0,h})^\top]\hslmorp_{0,h}\bm D_h,\\
    \betart_{t+1,h} &= \mathbb E_\mu[u^\orc_{t,h}(\adxl_{t,1},\ldots,\adxl_{t,m},\wt{Y}_{0,1},\ldots,\wt{Y}_{0,\wt m};\mu)Y^\top_{0,h}]\bm D^*_h\\
    &\hskip2em-\mathbb E_\mu[u^\orc_{t,h}(\adxl_{t,1},\ldots,\adxl_{t,m},\wt{Y}_{0,1},\ldots,\wt{Y}_{0,\wt m};\mu)U^\top_h]\bm D_h(\hsrmorp_{0,h})^\top\\
    &\hskip 2em+\mathbb E_\mu[u^\orc_{t,h}(\adxl_{t,1},\ldots,\adxl_{t,m},\wt{Y}_{0,1},\ldots,\wt{Y}_{0,\wt m};\mu)Y^\top_{0,h}]\hslmorp_{0,h}\bm D_h(\hsrmorp_{0,h})^\top\\
    &\hskip 2em-\mathbb E_\mu[\mathsf{d}_hu^\orc_{t,h}(\adxl_{t,1},\ldots,\adxl_{t,m},\wt{Y}_{0,1},\ldots,\wt{Y}_{0,\wt m};\mu)]\mathbb E_{\nu_h}[v^\orc_{t,h}(\adxr_{t,h})(Y^R_{0,h})^\top],\\
    \taurt_{t+1,h} &= \mathbb E_\mu[u^\orc_{t,h}(\adxl_{t,1},\ldots,\adxl_{t,m},\wt{Y}_{0,1},\ldots,\wt{Y}_{0,\wt m};\mu)^{\otimes 2}]\\
    &\hskip 2em-\mathbb E_\mu[u^\orc_{t,h}(\adxl_{t,1},\ldots,\adxl_{t,m},\wt{Y}_{0,1},\ldots,\wt{Y}_{0,\wt m};\mu)Y^\top_{0,h}]^{\otimes 2}
\end{align}
where $\lim_{N \rightarrow \infty}\bm D_{0,h}\overset{a.s.}{=}\bm D^*_h$, the matrices $\{\mathsf{d}v^\orc_{t,h}(\cdot):h \in [m]\}$ are the Jacobian matrices of the functions $\{v^\orc_{t,h}(\cdot):h \in [m]\}$ and $\{\mathsf{d}_hu^\orc_{t,h}(\cdot):h \in [m]\}$ are the Jacobian matrices of the function $\{u^\orc_{t,h}(\cdot):h \in [m]\}$ with respect to the variables corresponding to the $h$-th modality. By convention we set 
\[\mathbb E[u^\orc_{-1,h}(\adxl_{-1,1},\ldots,\adxl_{-1,m},\wt{Y}_{0,1},\ldots,\wt{Y}_{0,\wt m};\mu)Y^\top_{0,h}]=(\hslvorp_{0,h})^{1/2}.\]
Next, consider the following lemmas.
\begin{lem}
\label{lem:asymptotics_slln}
   Consider the sequence of deterministic subsets $\mathcal F_L \subset [N]$ satisfying $\frac{|\mathcal F_L|}{N} \rightarrow \lambda_{L} \in (0,1]$, as $N \rightarrow \infty$ and $\{\mathcal F_{R,h} \subset [p_h]:h \in [m]\}$ satisfying $\frac{|\mathcal F_{R,h}|}{p_h} \rightarrow \lambda_{R,h} \in (0,1]$, as $p_h \rightarrow \infty$ for all $h \in [m]$. Then, for any pseudo-Lipschitz functions $\psi:\R^{3r+2\wt r}\rightarrow \R$ and $\phi_h:\R^{3r_h} \rightarrow \R$ for $h \in [m]$, the following limits hold almost surely.
   \setlength{\jot}{-1ex}
    \begin{align}
        &\lim_{N \rightarrow \infty}\Bigg|\frac{1}{|\mathcal F_L|}\sum_{i \in \mathcal F_L}\psi\left(\{(\adu_{t,h})_{i*}\}_{h=1}^{m},\{(\upca_{0,h})_{i*}\}_{h=1}^{m},\{(\wt{\bm X}_h)_{i*}\}_{h=1}^{\wt m},\{(\bm U_{h})_{i*}\}_{h=1}^{m},\{(\wt{\bm U}_{\ell})_{i*}\}_{\ell=1}^{\wt m}\right)\\
    &\hskip 4em- \mathbb{E}_\mu\left[\psi\left(\{\alphalt_{t,h}U_{h}+\betalt_{t,h}Y_{0,h}+(\tault_{t,h})^{1/2}W^{\perp,L}_{t,h}\}_{h=1}^{m},\{Y_{0,h}\}_{h=1}^{m},\right.\right.\\
    &\hskip 15em \left.\left.\{\wt{Y}_{0,\ell}\}_{\ell=1}^{\wt m},\{U_{h}\}_{h=1}^{m},\{\wt U_\ell\}_{\ell=1}^{\wt m}\right)\right]\Bigg|=0,
    \end{align}
    \setlength{\jot}{0 pt}
    and
    \begin{align}
         &\lim_{p_h \rightarrow \infty}\Bigg|\frac{1}{|\mathcal F_{R,h}|}\sum_{i \in \mathcal F_{R,h}}\phi_h\left((\adv_{t,h})_{i*},(\vpca_{0,h})_{i*},(\bm V_{h})_{i*}\right)\\
         &\hskip 4em-\mathbb{E}_{\nu_h}\left[\phi_h\left(\alphart_{t,h}V_{h}+\betart_{t,h}Y^R_{0,h}+(\taurt_{t,h})^{1/2}W^{\perp,R}_{t,h},Y^R_{0,h},V_h\right)\right]\Bigg|=0,
    \end{align}
    where for all $t \ge 0$ and $h \in [m]$, the random variables $W^{\perp,L}_{t,h},W^{\perp,R}_{t,h} \sim N_{r_h}(0,\bm I_{r_h})$ and are mutually independent. In addition, $(U_1,\ldots,U_m,\wt U_1,\ldots,\wt U_{\wt m}) \sim \mu$ and $\{V_h \sim \nu_h:h \in [m]\}$. 
\end{lem}
\begin{lem}
\label{lem:asymptotic_final_limit}
     For any sequence of pseudo-Lipschitz functions $\psi:\R^{3r+2\wt r}\rightarrow \R$ and $\phi_h:\R^{3r_h} \rightarrow \R$ for $h \in [m]$, the following limits hold almost surely.
    \begin{align}
        &\lim_{N \rightarrow \infty}\Bigg|\mathbb{E}_\mu\left[\psi\left(\{\alphalt_{t,h}U_{h}+\betalt_{t,h}Y_{0,h}+(\tault_{t,h})^{1/2}W^{\perp,L}_{t,h}\}_{h=1}^{m},\{Y_{0,h}\}_{h=1}^{m},\right.\right.\\
         &\hskip 20em\left.\left.\{\wt{Y}_{0,\ell}\}_{\ell=1}^{\wt m},\{U_{h}\}_{h=1}^{m},\{\wt U_\ell\}_{\ell=1}^{\wt m}\right)\right]\\
         &\hskip 6em-\mathbb{E}_\mu\left[\psi\left(\{\hslmor_{t,h}U_h+(\hslvor_{t,h})^{1/2}Z^{\perp,L}_{t,h}\}_{h=1}^{m},\{Y_{0,h}\}_{h=1}^{m},\right.\right.\\
         &\hskip 20em\left.\left.\{\wt{Y}_{0,\ell}\}_{\ell=1}^{\wt m},\{U_{h}\}_{h=1}^{m},\{\wt U_\ell\}_{\ell=1}^{\wt m}\right)\right]\Bigg|=0,
    \end{align}
    and
    \begin{align}
         &\lim_{p_h \rightarrow \infty}\Bigg|\mathbb{E}_{\nu_h}\left[\phi_h\left(\alphart_{t,h}V_{h}+\betart_{t,h}Y^R_{0,h}+(\taurt_{t,h})^{1/2}W^{\perp,R}_{t,h},Y^R_{0,h},V_h\right)\right]-\\      &\hskip4em\mathbb{E}_{\nu_h}\left[\phi_h\left(\hsrmor_{t,h}V_h+(\hsrvor_{t,h})^{1/2}Z^{\perp,R}_{t,h},Y^R_{0,h},V_h\right)\right]\Bigg|=0,
    \end{align} 
    where for all $t \ge 0$ and $h \in [m]$, the random variables $Z^{\perp,L}_{t,h},Z^{\perp,R}_{t,h} \sim N_{r_h}(0,\bm I_{r_h})$ and are independent.
\end{lem}

\vskip 1em
\begin{proof}[Proof of Theorem \ref{thm:asymptotics of amp_iterates}]
Using \eqref{eq:modified_asymp_1} and \eqref{eq:modified_asymp_2} along with Lemmas \ref{lem:asymptotics_slln} and \ref{lem:asymptotic_final_limit}, the result follows.
\end{proof}

\subsection{Proof of Lemma \ref{lem:asymptotic_final_limit}}
Observe that by definition, as $N \rightarrow \infty$
\[
\|\hsrmorp_{0,h}\bm D_h-\bm D_h\hsrmorp_{0,h}\|^2_F\xrightarrow{a.s.}0, \quad \mbox{and} \quad \|\hslmorp_{0,h}\bm D_h-\bm D_h\hslmorp_{0,h}\|^2_F\xrightarrow{a.s.}0.
\]
Using the definitions of $\hslmorp_{0,h},\hsrmorp_{0,h}$ from \eqref{eq:initializers} and that of $\{\betalt_{t,h},\betart_{t,h}\}_{t \ge -1}$ we get the following limits using the Stein's Lemma \cite{Stein1981}.
\begin{align}
\label{eq:rel_beta}
    \lim_{N \rightarrow \infty}\|\betart_{t+1,h}\hsrmorp_{0,h}-\mathbb E_\mu[u^\orc_{t,h}(\adxl_{t,1},\ldots,\adxl_{t,m},\wt{Y}_{0,1},\ldots,\wt{Y}_{0,\wt m};\mu)Y^\top_{0,h}]\hslmorp_{0,h}\bm D_h\|^2_F \overset{a.s.}{=}0,\\
    \lim_{p_h \rightarrow \infty}\|\betalt_{t+1,h}\hslmorp_{0,h}-\gamma_h\mathbb E_{\nu_h}[v^\orc_{t,h}(\adxr_{t,h};\nu_h)(Y^R_{0,h})^\top]\hsrmorp_{0,h}\bm D_h\|^2_F \overset{a.s.}{=}0.
\end{align}
For $h \in [m]$, let us define:
\begin{align}
    &\funcm^L_{t,h}(\{\bm M_h\}_{h=1}^m,\{\bm Q_h\}_{h=1}^m)\\
    &\equiv\mathbb E_\mu\bigg[u^\orc_{t,h}\bigg(\{\bm M_hU_h+\bm Q^{1/2}_hZ_h\}_{h=1}^m,\{\bm L_\ell\wt{U}_\ell+\wt{Z}_\ell\}_{\ell=1}^{\wt m};\mu\bigg)U^\top_h\bigg]\bm D_h,\\
    &\funcv^L_{t,h}(\{\bm M_h\}_{h=1}^m,\{\bm Q_h\}_{h=1}^m)\equiv\mathbb E_\mu\bigg[u^\orc_{t,h}\bigg(\{\bm M_hU_h+\bm Q^{1/2}_hZ_h\}_{h=1}^m,\{\bm L_\ell\wt{U}_\ell+\wt{Z}_\ell\}_{\ell=1}^{\wt m};\mu\bigg)^{\otimes 2}\bigg],\\
   & \funcm^R_{t,h}(\bm M_h,\bm Q_h)\equiv\gamma_h\mathbb E_{\nu_h}[v^\orc_{t,h}(\bm M_hV_h+\bm Q^{1/2}_hZ_h;\nu_h)V^\top_h]\bm D_h\\
    &\funcv^R_{t,h}(\bm M_h,\bm Q_h)\equiv\gamma_h\mathbb E_{\nu_h}[v^\orc_{t,h}(\bm M_hV_h+\bm Q^{1/2}_hZ_h;\nu_h)^{\otimes 2}].
\end{align}
Let us define for $h \in [m]$, the parameters \[
\tildSl_{t,h}\equiv\alphalt_{t,h}+\betalt_{t,h}\hslmorp_{0,h}\quad \mbox{and}\quad \tildSgl_{t,h}\equiv\tault_{t,h}+\betalt_{t,h}\hslvorp_{0,h}(\betalt_{t,h})^\top.\]
 By construction, it is evident that for all $h \in [m]$ we have:
\begin{align}
        &\mathbb{E}_\mu\left[\psi\left(\{\alphalt_{t,h}U_{h}+\betalt_{t,h}Y_{0,h}+(\tault_{t,h})^{1/2}W^{\perp,L}_{t,h}\}_{h=1}^{m},\right.\right.\\
        &\hskip 10em \left.\left.\{Y_{0,h}\}_{h=1}^{m},\{\wt Y_{0,\ell}\}_{\ell=1}^{\wt m},\{U_{h}\}_{h=1}^{m},\{\wt U_{\ell}\}_{\ell=1}^{\wt m}\right)\right]\\
         &\hskip 4em=\mathbb{E}_\mu\left[\psi\left(\{\tildSl_{t,h}U_h+(\tildSgl_{t,h})^{1/2}Z^{\perp,L}_{t,h}\}_{h=1}^{m},\right.\right.\\
         &\hskip 15em\left.\left.\{Y_{0,h}\}_{h=1}^{m},\{\wt{Y}_{0,\ell}\}_{\ell=1}^{\wt m},\{U_{h}\}_{h=1}^{m},\{\wt U_\ell\}_{\ell=1}^{\wt m}\right)\right],
    \end{align}
From the definition of the functions $\funcm^R_{t,h}$ and $\funcv^R_{t,h}$, using \eqref{eq:rel_beta} we can show that
\begin{align}
    \lim_{N \rightarrow \infty}\|\tildSl_{t,h}-\funcm^R_{t,h}(\tildSr_{t,h},\tildSgr_{t,h})\|_F=0,\\
    \lim_{N \rightarrow \infty}\|\tildSgl_{t,h}-\funcv^R_{t,h}(\tildSr_{t,h},\tildSgr_{t,h})\|_F=0.
\end{align}
Similarly, we can show that
\begin{align}
    \lim_{N \rightarrow \infty}\|\tildSr_{t+1,h}-\funcm^L_{t,h}(\{\tildSl_{t,h}\}_{h=1}^m,\{\tildSgl_{t,h}\}_{h=1}^m)\|_F=0,\\
    \lim_{N \rightarrow \infty}\|\tildSgr_{t,h}-\funcv^L_{t,h}(\{\tildSl_{t,h}\}_{h=1}^m,\{\tildSgl_{t,h}\}_{h=1}^m)\|_F=0.
\end{align}
Finally, using the definition of the state evolution parameters \eqref{eq:state_evol_gen}, and $\tildSr_{0,h}=\hsrmor_{0,h}$, $\tildSgr_{0,h}=\hsrvor_{0,h}$ for all $h \in [m]$, using continuity of the functions $\funcm^R_{t,h}$ and $\funcv^R_{t,h}$ we can inductively conclude that
\[
\lim_{N \rightarrow \infty}\|\tildSl_{t,h}-\hslmor_{t,h}\|_F = 0 \quad \mbox{and} \quad \lim_{N \rightarrow \infty}\|\tildSgl_{t,h}-\hslvor_{t,h}\|_F = 0.
\]
Hence, the first assertion follows using the pseudo-Lipschitz property of $\psi$. The second assertion also follows using similar arguments.

\subsection{Proof of Lemma \ref{lem:asymptotics_slln}}
Let us consider the functions
\begin{align}
    &\tildfuncu_{t,h}(x_1,\ldots,x_m,\wt{x}_1,\ldots,\wt{x}_{\wt m};\mu)\\
    &=u^\circ_{t,h}(x_1,\ldots,x_m,\wt{x}_1,\ldots,\wt{x}_{\wt m};\mu)\\
    &\hskip 1.5em-\upca_{0,h}\,\mathbb E_\mu[Y_{0,h}u^\circ_{t,h}(\{\alphalt_{t,h}U_{h}+\betalt_{t,h}Y_{0,h}+(\tault_{t,h})^{1/2}W^{\perp,L}_{t,h}\}_{h=1}^{m},\{\wt Y_{0,\ell}\}_{\ell=1}^{\wt m})^\top],
\end{align}
and
\begin{align}
    &\tildfuncv_{t,h}(y;\nu_h)=v^\orc_{t,h}(y;\nu_h)-\vpca_{0,h}\,\mathbb E_{\nu_h}[Y^R_{0,h}v^\orc_{t,h}(\alphart_{t,h}V_{h}+\betart_{t,h}Y^R_{0,h}+(\taurt_{t,h})^{1/2}W^{\perp,R}_{t,h})^\top].
\end{align}
For all $h \in [m]$, let us define the sequence of auxiliary iterations $\{\nbaru_{t,h},\nbarv_{t,h}\}_{t \ge 0}$ as follows:
\begin{align}
\label{eq:aux_amp}
    &\nbarv_{t,h} = \tildfuncv_{t,h}(\nperpv_{t,h}+\bm V_h(\alphart_{t,h})^\top+\vpca_{0,h}(\betart_{t,h})^\top;\nu_h),\\
    &\nperpu_{t,h} = \bm W^{\perp}_h\nbarv_{t,h}-\nbaru_{t-1,h}\,\gamma_h\left(\jacrperp_{t,h}(\nperpv_{t,h}+\bm V_h(\alphart_{t,h})^\top+\vpca_{0,h}(\betart_{t,h})^\top;\nu_h)\right)^\top,\\
    &\nbaru_{t,h} = \tildfuncu_{t,h}(\nperpu_{t,1}+\bm U_1(\alphalt_{t,1})^\top+\upca_{0,1}(\betalt_{t,1})^\top,\ldots,\\
    & \hskip 10em \nperpu_{t,m}+\bm U_m(\alphalt_{t,m})^\top+\upca_{0,m}(\betalt_{t,m})^\top,\wt{\bm X}_1,\ldots,\wt{\bm X}_{\wt m};\mu),\\
    & \nperpv_{t+1,h} = (\bm W^{\perp}_h)^\top\nbaru_{t,h}-\nbarv_{t,h}\,\bigg(\jaclperp_{t,h}(\nperpu_{t,1}+\bm U_1(\alphalt_{t,1})^\top+\upca_{0,1}(\betalt_{t,1})^\top,\ldots,\\
    & \hskip 10em \nperpu_{t,m}+\bm U_m(\alphalt_{t,m})^\top+\upca_{0,m}(\betalt_{t,m})^\top,\wt{\bm X}_1,\ldots,\wt{\bm X}_{\wt m};\mu)\bigg)^\top.
\end{align}
In the above definitions $\{\bm W^{\perp}_h \in \R^{N \times p_h}:(\bm W^{\perp}_h)_{ij} \overset{iid}{\sim} N(0,N^{-1}), h \in [m]\}$ and is independent of the observed data.  Further, the matrix-valued functions $\jacrperp_{t,h}(\cdot):\mathbb R^{
p_h \times r_h} \rightarrow  \mathbb R^{
r_h \times r_h}$ is defined similarly as \eqref{eq:jacror} with $\tildfuncv_{t,h}(\cdot)$ replacing $v^\orc_{t,h}(\cdot)$. Similarly, the matrix-valued functions $\jaclperp_{t,h}(\cdot):\prod_{k=1}^m\mathbb R^{N \times r_k
} \times \prod_{\ell=1}^{\wt m}\mathbb R^{N \times {\wt r}_\ell
} \rightarrow  \mathbb R^{
r_h \times r_h}$ is defined similarly as \eqref{eq:jaclor} with $\tildfuncu_{t,h}(\cdot)$ replacing $u^\orc_{t,h}(\cdot)$.
We initialize the iterations \eqref{eq:aux_amp} as follows:
\begin{align}
    \nperpv_{0,h}=\adv_{0,h}-\bm V_{h}(\alphart_{0,h})^\top-\vpca_{0,h}(\betart_{0,h})^\top=0,\\
    \nperpu_{-1,h}=\adu_{-1,h}-\bm U_{h}(\alphalt_{-1,h})^\top-\upca_{0,h}(\betalt_{-1,h})^\top=0,
\end{align}
for all $h \in [m]$. Now we consider the following lemma.
\begin{lem}
\label{lem:amp_aux}
     Let $\tildgenu_{t,h}:\R^{r+\wt r} \rightarrow \R^{r_h}$, $\tildgenv_{t,h}:\R^{r_h} \rightarrow \R^{r_h}$ be Lipschitz functions with Lipschitz Jacobian matrices. Let $\genericm_h \in \R^{N \times p_h}$ for $h \in [m]$ be random matrices with iid Gaussian entries with zero mean and variance $N^{-1}$. The matrices $\genericm_h$ are independent of one another for $h \in [m]$. Consider the sequence of AMP iterates defined as follows:
     \begin{align}
    &\genvb_{t,h} = \tildgenv_{t,h}(\genv_{t,h}),\\
    &\genu_{t,h} = \genericm_{h}\genvb_{t,h}-\genub_{t-1,h}\,\gamma_h\left(J^{\mathtt{g,R}}_{t,h}(\genv_{t,h})\right)^\top,\\
    &\genub_{t,h} = \tildgenu_{t,h}(\genu_{t,1},\ldots,\genu_{t,m},\wt{\bm X}_1,\ldots,\wt{\bm X}_{\wt m}),\\ 
    &\genv_{t+1,h} =\genericm^\top_{h}\genub_{t,h}-\genvb_{t,h}\,\left(J^{\mathtt{g,L}}_{t,h}(\genu_{t,1},\ldots,\genu_{t,m},\wt{\bm X}_1,\ldots,\wt{\bm X}_{\wt m})\right)^\top.
\end{align}
where $\genv_{0,h}=\bm 0$ and $\genub_{-1,h}=\bm 0$. In the above definitions, the matrix-valued functions $J^{\mathtt{g,R}}_{t,h}:\mathbb R^{p_h \times r_h} \rightarrow \R^{r_h \times r_h}$ corresponding to the function $\tildgenv_{t,h}(\bm x)$ is defined as follows:
\[
J^{\mathtt{g,R}}_{t,h}(\genv_{t,h})=\frac{1}{p_h}\sum_{i=1}^{p_h}\frac{\bm \partial \tildgenv_{t,h}}{\bm \partial \bm x}((\genv_{t,h})_{i*})
\]
Similarly, the matrix-valued functions $J^{\mathtt{g,L}}_{t,h}:\prod_{k=1}^m\mathbb R^{N \times r_k} \times \prod_{\ell=1}^{\wt m}\mathbb R^{N \times {\wt r}_\ell} \rightarrow  \mathbb R^{r_h \times r_h}$ corresponding to the function $\tildgenu_{t,h}(\bm x_1,\ldots,\bm x_m,\wt{\bm x}_1,\ldots,\wt{\bm x}_{\wt m})$ is defined as follows:
\begin{align}
&J^{\mathtt{g,L}}_{t,h}(\genu_{t,1},\ldots,\genu_{t,m},\wt{\bm X}_1,\ldots,\wt{\bm X}_{\wt m})\\
&\quad =\frac{1}{N}\sum_{i=1}^{N}\frac{\bm \partial \tildgenu_{t,h}}{\bm \partial \bm x_h}((\genu_{t,1})_{i*},\ldots,(\genu_{t,m})_{i*},(\wt{\bm X}_1)_{i*},\ldots,(\wt{\bm X}_{\wt m})_{i*}).
\end{align}
Next, consider any sequence of deterministic subsets $\mathcal F_L \subset [N]$ satisfying $\frac{|\mathcal F_L|}{N} \rightarrow \lambda_{L} \in (0,1]$, as $N \rightarrow \infty$ and $\{\mathcal F_{R,h} \subset [p_h]:h \in [m]\}$ satisfying $\frac{|\mathcal F_{R,h}|}{p_h} \rightarrow \lambda_{R,h} \in (0,1]$, as $p_h \rightarrow \infty$ for all $h \in [m]$.
Then, for any pseudo Lipschitz functions $\psi:\R^{3r+2\wt r} \rightarrow \R$ and $\phi_h:\R^{3r_h} \rightarrow \R$ we have:
\begin{align}
        &\hskip -0.25em\lim_{N \rightarrow \infty}\Bigg|\frac{1}{|\mathcal F_L|}\sum_{i \in \mathcal F_L}\psi\left(\{(\genu_{t,h})_{i*}\}_{h=1}^{m},\{(\wt{\bm X}_\ell)_{i*}\}_{\ell=1}^{\wt m},\{(\upca_{0,h})_{i*}\}_{h=1}^{m},\{(\bm U_{h})_{i*}\}_{h=1}^{m},\{(\wt{\bm U}_{\ell})_{i*}\}_{\ell=1}^{\wt m}\right)\\
    &\hskip 4em- \mathbb{E}_\mu\left[\psi\left(\{(\taul_{t,h})^{1/2}W^{\mathtt{g,L}}_{t,h}\}_{h=1}^{m},\{\wt{Y}_{0,\ell}\}_{\ell=1}^{\wt m},\{Y_{0,h}\}_{h=1}^{m},\{U_{h}\}_{h=1}^{m},\{\wt U_\ell\}_{\ell=1}^{\wt m}\right)\right]\Bigg|=0,
    \end{align}
    and
    \begin{align}
         &\lim_{p_h \rightarrow \infty}\Bigg|\frac{1}{|\mathcal F_{R,h}|}\sum_{i \in \mathcal F_{R,h}}\phi_h\left((\genv_{t,h})_{i*},(\bm V_{0,h})_{i*},(\bm V_{h})_{i*}\right)\\
         & \hskip 10em -\mathbb{E}\left[\phi\left((\taur_{t,h})^{1/2}W^{\mathtt{ g,R}}_{t,h},V_{0,h},V_{h}\right)\right]\Bigg|=0.
    \end{align} 
Here, $W^{\mathtt{g,L}}_{t,1},\ldots,W^{\mathtt{g,L}}_{t,m}$ and $W^{\mathtt{g,R}}_{t,1},\ldots,W^{\mathtt{g,R}}_{t,\wt m}$ are independent; $W^{\mathtt{g,L}}_{t,h}, W^{\mathtt{g,R}}_{t,h} \sim N_{r_h}(0,\bm I_{r_h})$ for $h \in [m]$, $V_h \sim \nu_h$, $(U_1,\ldots,U_m,\wt{U}_1,\ldots,\wt{U}_{\wt m}) \sim \mu$ and for all $h \in [m]$, the parameters $\taul_{t,h}$ and $\taur_{t,h}$ are defined as follows:
\begin{align}
\label{eq:def_gen_amp_se}
    &\taul_{t,h}=\gamma_h\mathbb E[\tildgenv_{t,h}((\taur_{t,h})^{1/2}W^{\mathtt{g,R}}_{t,h})^{\otimes 2}], \quad \mbox{and} \\
    &\taur_{t+1,h}=\mathbb E[\tildgenu_{t,h}((\taul_{t,1})^{1/2}W^{\mathtt{g,L}}_{t,1},\ldots,(\taul_{t,m})^{1/2}W^{\mathtt{g,L}}_{t,m},\wt{Y}_{0,1},\ldots,\wt{Y}_{0,\wt m})^{\otimes 2}],
\end{align}
where $\taur_{0,h}=\bm 0$.
\end{lem}
Observe that the AMP orbits defined by \eqref{eq:aux_amp} conform to the set-up described in Lemma \ref{lem:amp_aux}. Hence, we can conclude that for any sequence of pseudo-Lipschitz functions $\wt{\psi}:\R^{3r+2\wt r}\rightarrow \R$ and $\wt{\phi}_h:\R^{3r_h} \rightarrow \R$, almost surely we have:
\setlength{\jot}{-1ex}
 \begin{align}
        &\hskip -0.5em\lim_{N \rightarrow \infty}\Bigg|\frac{1}{|\mathcal F_L|}\sum_{i \in \mathcal F_L}\tilde{\psi}\left(\{(\nperpu_{t,h})_{i*}\}_{h=1}^{m},\{(\upca_{0,h})_{i*}\}_{h=1}^{m},\{(\wt{\bm X}_\ell)_{i*}\}_{\ell=1}^{\wt m},\{(\bm U_{h})_{i*}\}_{h=1}^{m},\{(\wt{\bm U}_{\ell})_{i*}\}_{\ell=1}^{\wt m}\right)\\
    &\hskip 4em- \mathbb{E}_\mu\left[\wt\psi\left(\{(\tault_{t,h})^{1/2}W^{\perp,L}_{t,h}\}_{h=1}^{m},\{Y_{0,h}\}_{h=1}^{m},\{\wt{Y}_{0,\ell}\}_{\ell=1}^{\wt m},\{U_{h}\}_{h=1}^{m},\{\wt U_\ell\}_{\ell=1}^{\wt m}\right)\right]\Bigg|=0,
    \end{align}
    and
    \begin{align}
         &\lim_{p_h \rightarrow \infty}\Bigg|\frac{1}{|\mathcal F_{R,h}|}\sum_{i \in \mathcal F_{R,h}}\wt{\phi}_h\left((\nperpv_{t,h})_{i*},(\vpca_{0,h})_{i*},(\bm V_{h})_{i*}\right)\\
         & \hskip 10em -\mathbb{E}_{\nu_h}\left[\wt{\phi}_h\left((\taurt_{t,h})^{1/2}W^{\perp,R}_{t,h},Y^R_{0,h},V_{h}\right)\right]\Bigg|=0,
    \end{align}
    \setlength{\jot}{0pt}
where $\{\tault_{t,h}: h \in [m]\}$ and $\{\taurt_{t,h}: h \in [m]\}$ are defined by \eqref{eq:state_evol_aux}. Let us choose the following sequence of pseudo-Lipschitz functions:
\begin{align}
    &\tilde{\psi}\left(\{(\nperpu_{t,h})_{i*}\}_{h=1}^{m},\{(\upca_{0,h})_{i*}\}_{h=1}^{m},\{(\wt{\bm X}_\ell)_{i*}\}_{\ell=1}^{\wt m},\{(\bm U_{h})_{i*}\}_{h=1}^{m},\{(\wt{\bm U}_{\ell})_{i*}\}_{\ell=1}^{\wt m}\right)\\
    &\hskip 1em =\psi\bigg(\{(\nperpu_{t,h})_{i*}+\alphalt_{t,h}(\bm U_{h})_{i*}+\betalt_{t,h}(\upca_{0,h})_{i*}\}_{h=1}^{m},\{(\upca_{0,h})_{i*}\}_{h=1}^{m},\{(\wt{\bm X}_\ell)_{i*}\}_{\ell=1}^{\wt m},\\
    &\hskip 10em\{(\bm U_{h})_{i*}\}_{h=1}^{m}, \{(\wt{\bm U}_{\ell})_{i*}\}_{\ell=1}^{\wt m}\bigg)\\
    &\tilde{\phi}_h\left((\nperpv_{t,h})_{i*},(\vpca_{0,h})_{i*},(\bm V_{h})_{i*}\right) \\
    &\hskip 1em =\phi_h\left((\nperpv_{t,h})_{i*}+\alphart_{t,h}(\bm V_{h})_{i*}+\betart_{t,h}(\vpca_{0,h})_{i*},(\vpca_{0,h})_{i*},(\bm V_{h})_{i*}\right)
\end{align}
Consequently, for all $h \in [m]$, almost surely we have
    \setlength{\jot}{-1 ex}
\begin{align}
        &\lim_{N \rightarrow \infty}\Bigg|\frac{1}{|\mathcal F_L|}\sum_{i \in \mathcal F_L}\psi\bigg(\{(\nperpu_{t,h}+\bm U_h(\alphalt_{t,h})^\top+\upca_{0,h}(\betalt_{t,h})^\top)_{i*}\}_{h=1}^{m},\\
    &\hskip 12em\{(\upca_{0,h})_{i*}\}_{h=1}^{m},\{(\wt{\bm X}_\ell)_{i*}\}_{\ell=1}^{\wt m},\{(\bm U_{h})_{i*}\}_{h=1}^{m},\{(\wt{\bm U}_{\ell})_{i*}\}_{\ell=1}^{\wt m}\bigg)\\
    &\hskip 5em- \mathbb{E}\left[\psi\left(\{\alphalt_{t,h} U_{h}+\betalt_{t,h}Y_{0,h}+(\tault_{t,h})^{1/2}W^{\perp,L}_{t,h}\}_{h=1}^{m},\{Y_{0,h}\}_{h=1}^{m},\right.\right.\\
    &\hskip 14em\left.\left.\{\wt Y_{0,\ell}\}_{\ell=1}^{\wt m},\{U_{h}\}_{h=1}^{m},\{\wt U_{\ell}\}_{\ell=1}^{\wt m}\right)\right]\Bigg|=0,
    \end{align}
        \setlength{\jot}{0pt}
    and
    \begin{align}
         &\lim_{p_h \rightarrow \infty}\Bigg|\frac{1}{|\mathcal F_{R,h}|}\sum_{i \in \mathcal F_{R,h}}\phi_h\left((\nperpv_{t,h})_{i*},(\vpca_{0,h})_{i*},(\bm V_{h})_{i*}\right)\\
         & \hskip 4em-\mathbb{E}_{\nu_h}\left[\phi_h\left(\alphart_{t,h}V_{h}+\betart_{t,h}Y^R_{0,h}+(\taurt_{t,h})^{1/2}W^{\perp,R}_{t,h},Y^R_{0,h},V_{h}\right)\right]\Bigg|=0.
    \end{align} 
So, to establish the result, it is enough to show:
    \setlength{\jot}{-1px}
\begin{align}
    \label{eq:mean_u_perp}
  &\lim_{N \rightarrow \infty}\Bigg|\frac{1}{|\mathcal F_L|}\sum_{i \in \mathcal F_L}\psi\bigg(\{(\nperpu_{t,h}+\bm U_h(\alphalt_{t,h})^\top+\upca_{0,h}(\betalt_{t,h})^\top)_{i*}\}_{h=1}^{m},\{(\upca_{0,h})_{i*}\}_{h=1}^{m},\nonumber\\
    &\hskip 14em\{(\wt{\bm X}_\ell)_{i*}\}_{\ell=1}^{\wt m},\{(\bm U_{h})_{i*}\}_{h=1}^{m},\{(\wt{\bm U}_{\ell})_{i*}\}_{\ell=1}^{\wt m}\bigg)\\
    &\hskip 4em- \frac{1}{|\mathcal F_L|}\sum_{i \in \mathcal F_L}\psi\bigg(\{(\adu_{t,h})_{i*}\}_{h=1}^{m},\{(\upca_{0,h})_{i*}\}_{d=1}^{m},\\
    &\hskip 14em\{(\wt{\bm X}_\ell)_{i*}\}_{\ell=1}^{\wt m},\{(\bm U_h)_{i*}\}_{h=1}^{m},\{(\wt{\bm U}_{\ell})_{i*}\}_{\ell=1}^{\wt m}\bigg)\Bigg|=0,  
\end{align}
    \setlength{\jot}{0pt}
and 
\begin{align}
    \label{eq:mean_v_perp}
         &\lim_{p_h \rightarrow \infty}\Bigg|\frac{1}{|\mathcal F_{R,h}|}\sum_{i \in \mathcal F_{R,h}}\phi_h\left((\nperpv_{t,h}+\bm V_h(\alphart_{t,h})^\top+\vpca_{0,h}(\betart_{t,h})^\top)_{i*},(\vpca_{0,h})_{i*},(\bm V_{h})_{i*}\right)\nonumber\\
    &\hskip 4em-\frac{1}{|\mathcal F_{R,h}|}\sum_{i \in \mathcal F_{R,h}}\phi_h\left((\adv_{t,h})_{i*},(\vpca_{0.h})_{i*},(\bm V_{h})_{i*}\right)\Bigg|=0.
    \end{align} 
To show \eqref{eq:mean_u_perp} and \eqref{eq:mean_v_perp}, for $h \in [m]$, we define the discrepancy between the sets of iterates $\{\adu_{t,h},\adv_{t,h}\}_{t \ge 0}$ and $\{\nperpu_{t,h}+\bm U_h(\alphalt_{t,h})^\top+\upca_{0,h}(\betalt_{t,h})^\top,\nperpv_{t,h}+\bm V_h(\alphart_{t,h})^\top+\vpca_{0,h}(\betart_{t,h})^\top\}_{t \ge 0}$ as follows.
\begin{align}
    \bm \Delta_{1,t,h}&=\adu_{t,h}-(\nperpu_{t,h}+\bm U_h(\alphalt_{t,h})^\top+\upca_{0,h}(\betalt_{t,h})^\top),\\
    \bm \Delta_{2,t,h}&=\adv_{t,h}-(\nperpv_{t,h}+\bm V_h(\alphart_{t,h})^\top+\vpca_{0,h}(\betart_{t,h})^\top).
\end{align}
We shall inductively show the following:
\begin{enumerate}
    \item For any sequence of deterministic subsets $\mathcal F_L \subset [N]$ satisfying $\frac{|\mathcal F_L|}{N} \rightarrow \lambda_L \in (0,1]$, as $N \rightarrow \infty$ and pseudo-Lipschitz function $\psi:\R^{3r+2\wt r} \rightarrow \R$, for all $t \ge -1$, almost surely
    \begin{align}
  &\lim_{N \rightarrow \infty}\Bigg|\frac{1}{|\mathcal F_L|}\sum_{i \in \mathcal F_L}\psi\bigg(\{(\nperpu_{t,h}+\bm U_h(\alphalt_{t,h})^\top+\upca_{0,h}(\betalt_{t,h})^\top)_{i*}\}_{h=1}^{m},\\
    &\hskip 10em\{(\upca_{0,h})_{i*}\}_{h=1}^{m},\{(\wt{\bm X}_\ell)_{i*}\}_{\ell=1}^{\wt m},\{(\bm U_{h})_{i*}\}_{h=1}^{m},\{(\wt{\bm U}_{\ell})_{i*}\}_{\ell=1}^{\wt m}\bigg)\\
    &\hskip 1em- \frac{1}{|\mathcal F_L|}\sum_{i \in \mathcal F_L}\psi\left(\{(\adu_{t,h})_{i*}\}_{h=1}^{m},\{(\upca_{0,h})_{i*}\}_{h=1}^{m},\{(\wt{\bm X}_\ell)_{i*}\}_{\ell=1}^{\wt m},\right.\\
    &\hskip 10em\left.\{(\bm U_h)_{i*}\}_{h=1}^{m},\{(\wt{\bm U}_{\ell})_{i*}\}_{\ell=1}^{\wt m}\right)\Bigg|=0,  
\end{align}
\item For any sequence of deterministic subsets $\mathcal F_{R,h} \subset [p_h]$ satisfying $\frac{|\mathcal F_{R,h}|}{p_h} \rightarrow \lambda_{R,h} \in (0,1]$, as $p_h \rightarrow \infty$ and any pseudo-Lipschitz functions $\phi_h:\R^{3r_h} \rightarrow \R$ for $h \in [m]$, for all $t \ge -1$, almost surely
\begin{align}
         &\lim_{p_h \rightarrow \infty}\Bigg|\frac{1}{|\mathcal F_{R,h}|}\sum_{i \in \mathcal F_{R,h}}\phi_h\left((\nperpv_{t,h}+\bm V_h(\alphart_{t,h})^\top+\vpca_{0,h}(\betart_{t,h})^\top)_{i*},(\vpca_{0,h})_{i*},(\bm V_{h})_{i*}\right)\\
    &\hskip 4em-\frac{1}{|\mathcal F_{R,h}|}\sum_{i \in \mathcal F_{R,h}}\phi_h\left((\adv_{t,h})_{i*},(\bm V_{0.h})_{i*},(\bm V_{h})_{i*}\right)\Bigg|=0.
\end{align}
\item For $t \ge -1$, we have
\[
\lim\limits_{N \rightarrow \infty}\frac{1}{N}\|\bm \Delta_{1,t,h}\|^2_F = 0.
\]
\item For $t \ge -1$ and for all $h \in [m]$, we have
\[
\lim\limits_{p_h \rightarrow \infty}\frac{1}{p_h}\|\bm \Delta_{2,t,h}\|^2_F = 0.
\]
\item For $t \ge -1$, we have
\[
\lim\limits_{N \rightarrow \infty}\frac{1}{N}\|\adu_{t,h}\|^2_F < \infty, \quad \mbox{and} \quad \lim\limits_{N \rightarrow \infty}\frac{1}{N}\|\nperpu_{t,h}+\bm U_h(\alphalt_{t,h})^\top+\upca_{0,h}(\betalt_{t,h})^\top\|^2_F < \infty.
\]
\item For $t \ge 0$ and $h \in [m]$, we have
\[
\lim\limits_{p_h \rightarrow \infty}\frac{1}{p_h}\|\adv_{t+1,h}\|^2_F < \infty, \; \mbox{and}  \; \lim\limits_{p_h \rightarrow \infty}\frac{1}{p_h}\|\nperpv_{t+1,h}+\bm V_{h}(\alphart_{t+1,h})^\top+\vpca_{0,h}(\betart_{t+1,h})^\top\|^2_F < \infty.
\]
\end{enumerate}
To show the above assertions, we first introduce the following notations:
\begin{align}
    &f_{t,h}(\adu_{t,1},\ldots,\adu_{t,m},\wt{\bm X}_1,\ldots,\wt{\bm X}_{\wt m})
    =u^\orc_{t,h}(\adu_{t,1},\ldots,\adu_{t,m},\wt{\bm X}_1,\ldots,\wt{\bm X}_{\wt m})-\bm P_{\upca_{0,h}}\adus_{t,h},\\
    &g_{t,h}(\adv_{t,h})=v^\orc_{t,h}(\adv_{t,h})-\bm P_{\vpca_{0,h}}\advs_{t,h},\\
    &\bm \delta_{t,1,h}=\bm P_{\vpca_{0,h}}(\bm W^{\perp}_h)^\top f_{t,h}(\adu_{t,1},\ldots,\adu_{t,m},\wt{\bm X}_1,\ldots,\wt{\bm X}_{\wt m}) , \quad 
    \bm \delta_{t,2,h}=\bm P_{\upca_{0,h}}\bm W^{\perp}_hg_{t,h}(\adv_{t,h}).
\end{align}
Here, $\bm P_{\upca_{0,h}} = N^{-1}\upca_{0,h}(\upca_{0,h})^\top$ and $\bm P_{\vpca_{0,h}} = p_h^{-1}\vpca_{0,h}(\vpca_{0,h})^\top$.
Observe that for $t \ge 0$ and $h \in [m]$:
\begin{align}
   \adu_{t,h}&=\upca_{0,h}\bigg[\bm D_{0,h}\frac{(\vpca_{0,h})^\top v^\orc_{t,h}(\adv_{t,h})}{N}-\frac{(\upca_{0,h})^\top\adus_{t-1,h}}{N}\gamma_h\left(\jacror_{t,h}(\adv_{t,h};\nu_h)\right)^\top\\
   &\hskip 10em-\frac{(\upca_{0,h})^\top\bm U_h}{N}\bigg(\frac{\bm D_h\bm V^\top_h\bm P^\perp_{\vpca_{0,h}}v^\orc_{t,h}(\adv_{t,h})}{N}\bigg)\bigg]\\
   &\hskip 1em+\frac{1}{N}\bm U_h\bm D_h\bm V^\top_h\bm P^\perp_{\vpca_{0,h}}v^\orc_{t,h}(\adv_{t,h})\\
   &\hskip 8em -f_{t-1,h}(\adu_{t-1,1},\ldots,\adu_{t-1,m},\wt{\bm X}_1,\ldots,\wt{\bm X}_{\wt m})\gamma_h\left(\jacror_{t,h}(\adv_{t,h};\nu_h)\right)^\top\\
   &\hskip 1em+\bm W^\perp_hg_{t,h}(\adv_{t,h})-\bm \delta_{t,2,h},\\
   \adv_{t,h}&=\vpca_{0,h}\bigg[\bm D_{0,h}\frac{(\upca_{0,h})^\top u^\orc_{t,h}(\adu_{t,1},\ldots,\adu_{t,m},\wt{\bm X}_1,\ldots,\wt{\bm X}_{\wt m})}{N}\\
   &\hskip 5em-\frac{(\vpca_{0,h})^\top\advs_{t,h}}{p_h}\left(\jaclor_{t,h}(\adu_{t,1},\ldots,\adu_{t,m},\wt{\bm X}_1,\ldots,\wt{\bm X}_{\wt m})\right)^\top\\
   &\hskip 6em-\frac{(\vpca_{0,h})^\top\bm V_h}{p_h}\bigg(\frac{\bm D_h\bm U^\top_h\bm P^\perp_{\upca_{0,h}}u^\orc_{t,h}(\adu_{t,1},\ldots,\adu_{t,m},\wt{\bm X}_1,\ldots,\wt{\bm X}_{\wt m})}{N}\bigg)\bigg]\\
   &\hskip 1em+\frac{1}{N}\bm V_h\bm D_h\bm U^\top_h\bm P^\perp_{\upca_{0,h}}u^\orc_{t,h}(\adu_{t,1},\ldots,\adu_{t,m},\wt{\bm X}_1,\ldots,\wt{\bm X}_{\wt m})\\
   &\hskip 8em+(\bm W^\perp_h)^\top f_{t,h}(\adu_{t,1},\ldots,\adu_{t,m},\wt{\bm X}_1,\ldots,\wt{\bm X}_{\wt m})\\
   &\hskip 1em-g_{t,h}(\adv_{t,h})\left(\jaclor_{t,h}(\adu_{t,1},\ldots,\adu_{t,m},\wt{\bm X}_1,\ldots,\wt{\bm X}_{\wt m})\right)^\top-\bm \delta_{t,1,h},
\end{align}
with the convention $\tildfuncu_{-1,h}(\adu_{-1,1},\ldots,\adu_{-1,m},\wt{\bm X}_1,\ldots,\wt{\bm X}_{\wt m})=\adus_{-1,h}$. Consequently, we have
\begin{align}
\label{eq:delta_rewrite}
    \bm \Delta_{t,1,h}&=\upca_{0,h}\bigg[\bm D_{0,h}\frac{(\vpca_{0,h})^\top v^\orc_{t,h}(\adv_{t,h})}{N}-\frac{(\upca_{0,h})^\top\adus_{t-1,h}}{N}\gamma_h\left(\jacror_{t,h}(\adv_{t,h};\nu_h)\right)^\top\\
    &\hskip 1em -\frac{(\upca_{0,h})^\top\bm U_h}{N}\bigg(\frac{\bm D_h\bm V^\top_h\bm P^\perp_{\vpca_{0,h}}v^\orc_{t,h}(\adv_{t,h})}{N}\bigg)-(\betalt_{t,h})^\top\bigg]\\
    &\hskip 2em+\bm U_h\bigg[\bm D_h\frac{\bm V^\top_h\bm P^\perp_{\vpca_{0,h}}v^\orc_{t,h}(\adv_{t,h})}{N}-(\alphalt_{t,h})^\top\bigg]\\
    &\hskip 2em+\bm W^{\perp}_h(g_{t,h}(\adv_{t,h})-\tildfuncv_{t,h}(\nperpv_{t,h}+\bm V_h(\alphart_{t,h})^\top+\vpca_{0,h}(\betart_{t,h})^\top))\\
    &\hskip 2em+\bigg(\nbaru_{t-1,h}\,\gamma_h\left(\jacrperp_{t,h}(\nperpv_{t,h}+\bm V_h(\alphart_{t,h})^\top+\vpca_{0,h}(\betart_{t,h})^\top;\nu_h)\right)^\top\\
    &\hskip 2em-f_{t-1,h}(\adu_{t-1,1},\ldots,\adu_{t-1,m},\wt{\bm X}_1,\ldots,\wt{\bm X}_{\wt m})\gamma_h\left(\jacror_{t,h}(\adv_{t,h};\nu_h)\right)^\top\bigg)-\bm \delta_{t,2,h}\\
    \bm \Delta_{t+1,2,h}&=\vpca_{0,h}\bigg[\bm D_{0,h}\frac{(\upca_{0,h})^\top u^\orc_{t,h}(\adu_{t,1},\ldots,\adu_{t,m},\wt{\bm X}_1,\ldots,\wt{\bm X}_{\wt m})}{N}\\
    &\hskip 2em-\frac{(\vpca_{0,h})^\top\advs_{t,h}}{p_h}(\jaclor_{t,h}(\adu_{t,1},\ldots,\adu_{t,m},\wt{\bm X}_1,\ldots,\wt{\bm X}_{\wt m},\mu))^\top\\
   &\hskip 1em-\frac{(\vpca_{0,h})^\top\bm V_h}{p_h}\bigg(\frac{\bm D_h\bm U^\top_h\bm P^\perp_{\upca_{0,h}}u^\orc_{t,h}(\adu_{t,1},\ldots,\adu_{t,m},\wt{\bm X}_1,\ldots,\wt{\bm X}_{\wt m})}{N}\bigg)-(\betart_{t,h})^\top\bigg]\\
    &+\bm V_h\bigg[\bm D_h\frac{\bm U^\top_h\bm P^\perp_{\upca_{0,h}}u^\orc_{t,h}(\adu_{t,1},\ldots,\adu_{t,m},\wt{\bm X}_1,\ldots,\wt{\bm X}_{\wt m})}{N}-(\alphart_{t,h})^\top\bigg]-\bm \delta_{t,2,h}\\
    &+(\bm W^\perp_h)^\top\bigg(f_{t,h}(\adu_{t,1},\ldots,\adu_{t,m},\wt{\bm X}_1,\ldots,\wt{\bm X}_{\wt m})\\
    &-\tildfuncu_{t,h}(\nperpu_{t,1}+\bm U_h(\alphalt_{t,1})^\top+\upca_{0,1}(\betalt_{t,1})^\top,\ldots,\\
    &\hskip 10em\nperpu_{t,m}+\bm U_m(\alphalt_{t,m})^\top+\upca_{0,m}(\betalt_{t,m})^\top,\wt{\bm X}_1,\ldots,\wt{\bm X}_{\wt m})\bigg)\\
    &+\bigg(\nbarv_{t,h}\,\bigg(\jaclperp_{t,h}(\nperpu_{t,1}+\bm U_h(\alphalt_{t,h})^\top+\upca_{0,h}(\betalt_{t,h})^\top,\ldots,\\
    &\hskip 10em\nperpu_{t,h}+\bm U_h(\alphalt_{t,h})^\top+\upca_{0,h}(\betalt_{t,h})^\top,\wt{\bm X}_1,\ldots,\wt{\bm X}_{\wt m})\bigg)^\top\\
    &\hskip 10em-g_{t,h}(\adv_{t,h})\left(\jaclor_{t,h}(\adu_{t,1},\ldots,\adu_{t,m},\wt{\bm X}_1,\ldots,\wt{\bm X}_{\wt m})\right)^\top\bigg)\\
\end{align}
The base case for $t=-1$ follows by the definitions of $\adu_{-1,h}$ and $\adv_{0,h}$. Let us assume that the assertions (1)-(6) hold for $t=-1,\ldots,r$. First, let us consider Assertion (3) for $t=r+1$. Using \eqref{eq:state_evol_aux}, the assertions (2), (4), (6) for $t=r$ and arguments similar to the proof of (B.46) of \citet{montanari2021}, we can show 
\begin{align}
&\lim_{N \rightarrow \infty}\bigg[\bm D_{0,h}\frac{(\vpca_{0,h})^\top v^\orc_{t,h}(\adv_{t,h})}{N}-\frac{(\upca_{0,h})^\top\adus_{t-1,h}}{N}\gamma_h\left(\jacror_{t,h}(\adv_{t,h};\nu_h)\right)^\top\\
&\hskip 4em-\frac{(\upca_{0,h})^\top\bm U_h}{N}\bigg(\frac{\bm D_h\bm V^\top_h\bm P^\perp_{\vpca_{0,h}}v^\orc_{t,h}(\adv_{t,h})}{N}\bigg)-(\betalt_{t,h})^\top\bigg]=0
\end{align}
Using the arguments similar to the proof of (B.47) of \citet{montanari2021} and \eqref{eq:state_evol_aux}, we can further show that for all $h \in [m]$
\begin{align}
\lim_{N \rightarrow \infty}\bigg[\bm D_h\frac{\bm V^\top_h\bm P^\perp_{\vpca_{0,h}}v^\orc_{t,h}(\adv_{t,h})}{N}-(\alphalt_{t,h})^\top\bigg]=0.
\end{align}
Using the fact that $\|\bm W^\perp_h\|^2 \rightarrow 4$ along with the Lipschitz property of the functions $v^\orc_{0,h}$'s and their derivatives and the assertions (2), (4), (5), (6) for $t=r-1$, we can show along the lines of the proof of (B.48) and (B.49) of \citet{montanari2021} that
\[
\lim\limits_{N \rightarrow \infty}\frac{1}{N}\|\bm W^{\perp}_h(g_{t,h}(\adv_{t,h})-\tildfuncv_{t,h}(\nperpv_{t,h}+\bm V_h(\alphart_{t,h})^\top+\vpca_{0,h}(\betart_{t,h})^\top))\|_F=0,
\]
and
\begin{align}
    &\lim\limits_{N \rightarrow \infty}\frac{1}{N}\bigg\|\nbaru_{t-1,h}\,\gamma_h\left(\jacrperp_{t,h}(\nperpv_{t,h}+\bm V_h(\alphart_{t,h})^\top+\vpca_{0,h}(\betart_{t,h})^\top;\nu_h)\right)^\top\\
    &\hskip 4em-f_{t-1,h}(\adu_{t-1,1},\ldots,\adu_{t-1,m},\wt{\bm X}_1,\ldots,\wt{\bm X}_{\wt m})\gamma_h\left(\jacror_{t,h}(\adv_{t,h};\nu_h)\right)^\top\bigg\|_F=0.
\end{align}
Again, using the Lipschitz property of the denoisers and the previously established assertions, we can prove along the lines of (B.50) of \citet{montanari2021} that $\lim\limits_{N \rightarrow \infty}\frac{1}{N}\|\bm \delta_{t,2,h}\|=0$.
This proves assertion (3).

Next, we consider Assertion (5) for $t=r+1$. Consider the pseudo Lipschitz function $\psi(x,u,u_0)=\|x+\alphalt_{t,h}u+\betalt_{t,h}u_0\|^2$. Then by Lemma \ref{lem:amp_aux}, we have 
\[
\lim\limits_{N \rightarrow \infty}\frac{1}{N}\|\nperpu_{t,h}+\bm U_h(\alphalt_{t,h})^\top+\upca_{0,h}(\betalt_{t,h})^\top\|^2_F < \infty.
\]
Then, using Assertion (3) for $t=r+1$, we get
\[
\lim\limits_{N \rightarrow \infty}\frac{1}{N}\|\adu_{t,h}\|^2_F < \infty.
\]
To prove assertion (1), observe that, since $\frac{|\mathcal F_L|}{N} \rightarrow \lambda_L \in (0,1]$, using Cauchy-Schwarz inequality and properties of pseudo Lipschitz functions \eqref{eq:pseudo_lips}, we can get a constant $C>0$ (possibly depending on $\lambda_L$),
\begin{align}
\label{eq:typical_scaling_ineq}
     &\lim_{N \rightarrow \infty}\Bigg|\frac{1}{|\mathcal F_L|}\sum_{i \in \mathcal F_L}\psi\bigg(\{(\nperpu_{t,h}+\bm U_h(\alphalt_{t,h})^\top+\upca_{0,h}(\betalt_{t,h})^\top)_{i*}\}_{h=1}^{m},\{(\upca_{0,h})_{i*}\}_{h=1}^{m},\nonumber\\
    &\hskip 15em\{(\wt{\bm X}_\ell)_{i*}\}_{\ell=1}^{\wt m},\{(\bm U_{h})_{i*}\}_{h=1}^{m},\{(\wt{\bm U}_{\ell})_{i*}\}_{\ell=1}^{\wt m}\bigg)\\
    &\hskip 4em- \frac{1}{|\mathcal F_L|}\sum_{i \in \mathcal F_L}\psi\left(\{(\adus_{t,h})_{i*}\}_{h=1}^{m},\{(\upca_{0,h})_{i*}\}_{h=1}^{m},\{(\wt{\bm X}_\ell)_{i*}\}_{\ell=1}^{\wt m},\right.\\
    &\hskip 15em\left.\{(\bm U_h)_{i*}\}_{h=1}^{m},\{(\wt{\bm U}_{\ell})_{i*}\}_{\ell=1}^{\wt m}\right)\Bigg|\\
    & \le C\bigg(1+\sum_{h=1}^{m}\frac{\|\bm \Delta_{t,h}\|_F}{\sqrt{N}}+\sum_{h=1}^{m}\frac{\|\bm U_{h}\|_F}{\sqrt{N}}+\sum_{h=1}^{m}\frac{\|\upca_{0,h}\|_F}{\sqrt{N}}+\sum_{\ell=1}^{\wt m}\frac{\|\wt{\bm X}_\ell\|_F}{\sqrt{N}}+\sum_{\ell=1}^{\wt m}\frac{\|\wt{\bm U}_\ell\|_F}{\sqrt{N}}\\
    &\hskip 5em +\sum_{h=1}^{m}\frac{\|\nperpu_{t,h}+\bm U_h(\alphalt_{t,h})^\top+\upca_{0,h}(\betalt_{t,h})^\top\|_F}{\sqrt{N}}\bigg)\times\sum_{h=1}^{m}\frac{\|\bm \Delta_{t,h}\|_F}{\sqrt{N}}.
\end{align}
Combining Assertions (3) and (5) for $t=r+1$ with Proposition \ref{prop:singular_values}, Assertion (1) follows. 
The proof of Assertion (4) follows by retracing the steps of the proof of Assertion (3) while that of Assertion (6) follows using Assertion (3) and the techniques of the proof of Assertion (5). Finally, we can show Assertion (2) using Assertions (4) and (6) and the steps of the proof of Assertion (1).

\subsection{Proof of Lemma \ref{lem:amp_aux}}
To prove this lemma, we shall use the Gaussian conditioning technique from \citet{Bol12} as in \citet{Berthier}. Let us denote the $\sigma$-algebra generated by 
\begin{align}
&\{\genub_{0,h},\ldots,\genub_{t-1,h},\genu_{0,h},\ldots
\genu_{t-1,h},\\
& \hskip 10em\genvb_{1,h},\ldots,\genvb_{t,h},\genv_{1,h},\ldots,\genv_{t,h},
 \wt{\bm X}_\ell:h \in [m],\, \ell \in [\wt m]\}
\end{align}
by $\mathcal{G}_{t,t}$ and the $\sigma$-algebra generated by 
\begin{align}
&\{\genvb_{1,h},\ldots,\genvb_{t,h},\genv_{1,h},\ldots,\genv_{t,h},\\
& \hskip 10em\genub_{0,h},\ldots,\genub_{t,h},
\genu_{0,h},\ldots,\genu_{t,h}, \wt{\bm X}_\ell:h \in [m],\, \ell \in [\wt m]\}
\end{align}
by $\mathcal{G}_{t+1,t}$. 
Let us simplify the notations by introducing the following notations:
\[
\jacv_{t,h}=\gamma_{h}J^{\mathtt{g,R}}_{t,h}(\genv_{t,h}), \quad \jacu_{t,h}=J^{\mathtt{g,L}}_{t,h}(\genu_{t,1},\ldots,\genu_{t,m},\wt{\bm X}_1,\ldots,\wt{\bm X}_{\wt m}).
\]
Observe that conditional distribution of $\bm G_h$ given $\mathcal{G}_{t+1,t}$ is equivalent to the conditional distribution of $\bm G_h$ given 
\[
\underbracket{[\genv_{1,h}+\genvb_{0,h}(\jacu_{0,h})^\top|\cdots|\genv_{t,h}+\genvb_{t-1,h}(\jacu_{t-1,h})^\top]}_{=\sfnr_{t,h}}=\bm G^\top_h\underbracket{[\genub_{0,h}|\cdots|\genub_{t-1,h}]}_{=\sfnbl_{t,h}}
\]
and
\[
\underbracket{[\genu_{0,h}|\cdots|\genu_{t,h}+\genub_{t-1,h}(\jacv_{t,h})^\top]}_{=\sfnl_{t+1,h}}=\bm G_h\underbracket{[\genvb_{0,h}|\cdots|\genvb_{t,h}]}_{=\sfnbr_{t+1,h}}.
\]
Using Lemma 11 and 12 of \cite{BM11journal}, we can show that 
\[
\bm G_h|_{\mathcal G_{t+1,t}}\overset{d}{=}\sfE_{t+1,t,h}+\sfP_{t+1,t,h}(\wt{\bm G}_h),
\]
where 
\begin{align}
   &\sfE_{t+1,t,h}=\sfnl_{t+1,h}((\sfnbr_{t+1,h})^\top\sfnbr_{t+1,h})^{-1}(\sfnbr_{t+1,h})^\top+\sfnbl_{t,h}((\sfnbl_{t,h})^\top\sfnbl_{t,h})^{-1}\sfnr_{t,h}\\
   &\hskip 5em-\sfnbl_{t,h}((\sfnbl_{t,h})^\top\sfnbl_{t,h})^{-1}(\sfnbl_{t,h})^\top\sfnl_{t+1,h}((\sfnbr_{t+1,h})^\top\sfnbr_{t+1,h})^{-1}(\sfnbr_{t+1,h})^\top,\\
   &\sfP_{t+1,t,h}(\wt{\bm G}_h)=\bm P^\perp_{\sfnbl_{t,h}}\wt{
   \bm G}_h\bm P^\perp_{\sfnbr_{t+1,h}}.
\end{align}
Here $\wt{\bm G}_h \overset{d}{=} \bm G_h$ and $\wt{\bm G}_h$ is independent of $\bm G_h$ for all $h \in [m]$. Furthermore, conditioned on $\mathcal G_{t,t}$ the distribution of $\bm G_h$ for $h \in [m]$ is given by
\[
\bm G_h|_{\mathcal G_{t,t}}=\sfE_{t,t,h}+\sfP_{t,t,h}(\wt{\bm G}_h),
\]
where 
\begin{align}
   &\sfE_{t,t,h}=\sfnl_{t,d}((\sfnbr_{t,h})^\top\sfnbr_{t,h})^{-1}(\sfnbr_{t,h})^\top+\sfnbl_{t,h}((\sfnbl_{t,h})^\top\sfnbl_{t,h})^{-1}\sfnr_{t,h}\\
   &\hskip 8em-\sfnbl_{t,h}((\sfnbl_{t,h})^\top\sfnbl_{t,h})^{-1}(\sfnbl_{t,h})^\top\sfnl_{t,d}((\sfnbr_{t,h})^\top\sfnbr_{t,h})^{-1}(\sfnbr_{t,h})^\top,\\
   &\sfP_{t,t,d}(\wt{\bm G}_h)=\bm P^\perp_{\sfnbl_{t,h}}\wt{
   \bm G}_h\bm P^\perp_{\sfnbr_{t,h}}.
\end{align}
Again, the matrices $\wt{\bm G}_h \overset{d}{=} \bm G_h$ and $\wt{\bm G}_h$ are independent of $\bm G_h$ for all $h \in [m]$.
Note that in this analysis we are only considering the non-degenerate case when $(\sfnbl_{t,h})^\top\sfnbl_{t,h}$ and $(\sfnbr_{t,h})^\top\sfnbr_{t,h}$ are of full rank. The degenerate case can be handled by adding a small random perturbation to the denoisers as shown in Section 5.4 of \citet{Berthier}.

Define the following matrices:
\begin{align}
    \bm \Xi^L_{t,h} &= \bigg(\frac{(\sfnbl_{t,h})^\top\sfnbl_{t,h}}{N}\bigg)^{-1}\frac{(\sfnbl_{t,h})^\top\genub_{t,h}}{N}, \\
    & \mbox{and} \quad \bm \Xi^R_{t,h} = \bigg(\frac{(\sfnbr_{t,h})^\top\sfnbr_{t,h}}{p_h}\bigg)^{-1}\frac{(\sfnbr_{t,h})^\top\genvb_{t,h}}{p_h}.
\end{align}

Following the techniques used in the proof of Lemma E.1 of \citet{ma_nandy}, we can inductively show the following statements for all $t \ge 0$.
\begin{enumerate}
    \item If $\rfF_{L,t,h}=[\genu_{0,h}|\cdots|\genu_{t-1,h}]$ and $\rfF_{R,t,h}=[\genv_{1,h}|\cdots|\genv_{t,h}]$, then for $h \in [m]$ and $N \rightarrow \infty$, we have
    \begin{align}
    \label{eq:proj_gen_1}
        \genu_{t,h}|_{\mathcal{G}_{t,t}}=\rfF_{L,t,h}\bm \Xi^R_{t,h}+\wt{\bm G}_h\bm P^\perp_{\sfnbr_{t,h}}\genvb_{t,h}+ \wt{\boldsymbol{\mathrm N}}_{R,t,h}\,o(1),\\
        \label{eq:proj_gen_2}
        \genv_{t+1,h}|_{\mathcal{G}_{t+1,t}}=\rfF_{R,t,h}\bm \Xi^L_{t,h}+\wt{\bm G}^\top_h\bm P^\perp_{\sfnbl_{t,h}}\genub_{t,h}+\wt{\boldsymbol{\mathrm N}}_{L,t,h}\,o(1),
    \end{align}
where the matrices $\wt{\bm G}_h \overset{d}{=} \bm G_h$ and $\wt{\bm G}_h$ are independent of $\bm G_h$ for all $h \in [m]$. Furthermore, for all $h \in [m]$, the matrices $\wt{\boldsymbol{\mathrm N}}_{L,t,h}$ and $\wt{\boldsymbol{\mathrm N}}_{R,t,h}$ is formed by Gram-Schmidt orthogonalization of the columns of $\sfnbl_{t,h}$ and $\sfnbr_{t,h}$, respectively. 
    
    \item For any sequence of deterministic subsets $\mathcal F_L \subset [N]$ satisfying $\frac{|\mathcal F_L|}{N} \rightarrow \lambda_L \in (0,1]$, as $N \rightarrow \infty$ and $\mathcal F_{R,h} \subset [p_h]$, for $h \in [m]$ satisfying $\frac{|\mathcal F_{R,h}|}{p_h} \rightarrow \lambda_{R,h} \in (0,1]$, as $p_h \rightarrow \infty$. For any pseudo Lipschitz functions $\psi:\R^{3r+2\wt r} \rightarrow \R$ and $\phi_h:\R^{3r_h} \rightarrow \R$ for $h \in [m]$, we have:
\begin{align}
\label{eq:state_ev_g_amp_1}
        &\lim_{N \rightarrow \infty}\Bigg|\frac{1}{|\mathcal F_L|}\sum_{i \in \mathcal F_L}\psi\left(\{(\genu_{t,h})_{i*}\}_{h=1}^{m},\{(\upca_{0,h})_{i*}\}_{h=1}^{m},\{(\wt{\bm X}_\ell)_{i*}\}_{\ell=1}^{\wt m},\right.\nonumber\\
        &\hskip15em \left.\{(\bm U_{h})_{i*}\}_{h=1}^{m},\{(\wt{\bm U}_{\ell})_{i*}\}_{\ell=1}^{\wt m}\right)\nonumber\\
    &\hskip 1em- \mathbb{E}\left[\psi\left(\{(\taul_{t,h})^{1/2}W^{\mathtt{g,L}}_{t,h}\}_{h=1}^{m},\{Y_{0,h}\}_{h=1}^{m},\{\wt Y_{\ell}\}_{\ell=1}^{\wt m},\{U_{h}\}_{h=1}^{m},\{\wt U_{\ell}\}_{\ell=1}^{\wt m}\right)\right]\Bigg|=0,\nonumber\\
    &
    \end{align}
    and
    \begin{align}
    \label{eq:state_ev_g_amp_0}
         &\lim_{p_h \rightarrow \infty}\Bigg|\frac{1}{|\mathcal F_{R,h}|}\sum_{i \in \mathcal F_{R,h}}\phi_h\left((\genv_{t,h})_{i*},(\vpca_{0,h})_{i*},(\bm V_{h})_{i*}\right)\\
         &\hskip 10em-\mathbb{E}\left[\phi\left((\taur_{t,h})^{1/2}W^{\mathtt{g,R}}_{t,h},Y^R_{0,h},V_{h}\right)\right]\Bigg|=0.\nonumber\\
    \end{align} 
 \item For all $k,\ell \in \{0,\ldots,t\}$, we have for $h \in [m]$:
 \begin{align}
     &\lim_{p_h \rightarrow \infty}\frac{(\genv_{k+1,h})^\top\genv_{\ell+1,h}}{p_h}=\lim_{N\rightarrow \infty}\frac{(\genub_{k,h})^\top\genub_{\ell,h}}{N},\\
     &\lim_{N \rightarrow \infty}\frac{(\genu_{k,h})^\top\genu_{\ell,h}}{N}=\gamma_h\lim_{N\rightarrow \infty}\frac{(\genvb_{k,h})^\top\genvb_{\ell,h}}{N}.
 \end{align}
 \item For all $h \in [m]$, let $\varphi_h(x_1,\ldots,x_m,{\wt x}_1,\ldots,{\wt x}_{\wt m}):\R^{r+\wt r} \rightarrow \R$ and $\vartheta_h(x):\R^{r_h} \rightarrow \R$ be Lipschitz functions defined on appropriate domains with Lipschitz Jacobians. Let $\mathsf d\varphi_h$ and $\mathsf d\vartheta_h$ be the Jacobians of $\varphi_h$ and $\vartheta_h$ with respect to the $x_h$ and $y_h$ respectively. Then we have the following:
 \begin{align}
     &\lim_{N \rightarrow \infty}\frac{(\genu_{k,h})^\top\varphi_h(\{\genu_{\ell,j}\}_{j=1}^{m},\{\wt{\bm X}_{j'}\}_{j'=1}^{\wt m})}{N}\\
     &\hskip 3em =\lim_{N \rightarrow \infty}\frac{(\genu_{k,h})^\top\genu_{\ell,h}}{N}\lim\limits_{N \rightarrow\infty}\bigg(\frac{1}{N}\sum_{i=1}^{N}\frac{\partial \varphi_h}{\partial x_h}(\{(\genu_{\ell,j})_{i*}\}_{j=1}^{m},\{(\wt{\bm X}_{j'})_{i*}\}_{j'=1}^{\wt m})\bigg)^\top,\\
     &\lim_{p_h \rightarrow \infty}\frac{(\genv_{k+1,h})^\top\vartheta_h(\genv_{k+1,h})}{p_h}=\lim_{N\rightarrow \infty}\frac{(\genv_{k+1,h})^\top\genv_{\ell+1,h}}{p_h}\lim\limits_{N \rightarrow\infty}\bigg(\frac{1}{p_h}\sum_{i=1}^{p_h} \frac{\partial \vartheta_h}{\partial x}(\genv_{k,h})\bigg)^\top.
 \end{align}
 \item For all $h \in [m]$ and $i,j \in [r_h]\times[r_h]$, we have
 \[
 \lim_{N \rightarrow \infty}\bigg[\frac{(\genu_{t,h})^\top\genu_{t,h}}{N}\bigg]_{ij}<\infty \quad \mbox{and} \quad \lim_{p_h \rightarrow \infty}\bigg[\frac{(\genv_{t+1,h})^\top\genv_{t+1,h}}{p_h}\bigg]_{ij}<\infty.
 \]
 \item For all $0 \le s \le t-1$ and $0 \le k \le t$, we have positive definite matrices $\boldsymbol{\mathsf{Z}}_{s,h} \in \R^{r_h \times r_h}$ and $\boldsymbol{\mathsf{K}}_{k,h} \in \R^{r_h \times r_h}$ such that 
 \[
 \lim_{N \rightarrow \infty}\frac{(\genub_{s,h})^\top\bm P^\perp_{\sfnbl_{t,h}}\genub_{s,h}}{N}\succ \boldsymbol{\mathsf{Z}}_{s,h} \quad \mbox{and} \quad \lim_{p_h \rightarrow \infty}\frac{(\genvb_{s,h})^\top\bm P^\perp_{\sfnbr_{t,h}}\genvb_{s,h}}{p_h}\succ \boldsymbol{\mathsf{K}}_{s,h}.
 \]
\end{enumerate}
We shall show Assertions (1)-(6) inductively. For $t=0$, \eqref{eq:proj_gen_1} follows using the techniques used to show Assertion (a) for $\bm b^0$ in \cite{ma_nandy}. Observing that $\rfF_{L,0,h}=\bm 0$ and $\bm P^\perp_{\sfnbr_{0,h}}\genvb_{0,h}=\genvb_{0,h}$, from \eqref{eq:proj_gen_1} for $t=0$, we can conclude that for all $h \in [m]$, for all $i \in [N]$
\[
[\genu_{0,h}]_{i*}=\left(\frac{(\genvb_{0,h})^\top\genvb_{0,h}}{N}\right)^{1/2}Z,
\]
where $Z \sim N_{r_h}(0,\bm I_{r_h})$. Furthermore, the set of random vectors $\{[\genub_{0,h}]_{i*}:i \in [N]\}$ are independent. 
From the definition \eqref{eq:def_gen_amp_se}, we can conclude that for all $N \in \mathbb N$
\[
\frac{(\genvb_{0,h})^\top\genvb_{0,h}}{N} = \taul_{t,0}.
\]
Therefore, using Proposition \ref{prop:singular_values} and the Strong Law of Large Numbers we can conclude \eqref{eq:state_ev_g_amp_1}. The Assertions (3)-(6) for $\{\genu_{0,h}:h \in [m]\}$ follows by retracing the techniques used to prove Lemma E.1 of \cite{ma_nandy}. We omit the details to avoid the repetition of similar arguments. Furthermore, Assertion (1) for $\{\genv_{1,h}:h \in [m]\}$ follows using the techniques used to prove Lemma E.1 of \cite{ma_nandy}. Proceeding as in the proof of (3.36) of \cite{ma_nandy}, we can show that
\[
\genv_{1,h}|_{\mathcal G_{1,0}}= \wt{\bm G}_h\genub_{0,h}+o(1)\genvb_{0,h}.
\]
Since $\frac{|\mathcal F_{R,h}|}{p_h} \rightarrow \lambda_{R,h} \in (0,1]$, using \eqref{eq:pseudo_lips} and Cauchy Schwarz inequality we can get constants $C_h>0$ (possibly depending on $\lambda_{R,h}$) such that
\begin{align}
    \label{eq:state_ev_g_amp_2}
         &\lim_{p_h \rightarrow \infty}\Bigg|\frac{1}{|\mathcal F_{R,h}|}\sum_{i \in \mathcal F_{R,h}}\phi_h\left((\wt{\bm G}_h\genub_{0,h}+o(1)\genvb_{0,h})_{i*},(\vpca_{0,h})_{i*},(\bm V_{h})_{i*}\right)\\
         &\hskip 7em-\frac{1}{|\mathcal F_{R,h}|}\sum_{i \in \mathcal F_{R,h}}\phi_h\left((\wt{\bm G}_h\genub_{0,h})_{i*},(\vpca_{0,h})_{i*},(\bm V_{h})_{i*}\right)\Bigg|\\
         &\le C_h\left(1+\frac{\|\wt{\bm G}_h\genub_{0,h}\|_F}{\sqrt{N}}+\frac{\|\genvb_{0,h}\|_F}{\sqrt{N}}+\frac{\|\vpca_{0,h}\|_F}{\sqrt{N}}+\frac{\|\bm V_{h}\|_F}{\sqrt N}\right)\times\frac{\|\genvb_{0,h}\|_F}{\sqrt{N}}\times o(1)\xrightarrow{a.s} 0.
    \end{align}
Consequently, in \eqref{eq:state_ev_g_amp_2}, given $\mathcal G_{1,0}$, we can replace $\genv_{1,h}$ by $\wt{\bm N}^{G,R}_{1,h}$ for all $h \in [m]$ defined as follows
\[
\wt{\bm N}^{G,R}_{1,h}|_{\mathcal G_{1,0}} \overset{d}{=} \wt{\bm G}_h\genub_{0,h}.
\]
Consequently, using Theorem 3 of \cite{BM11journal}, Proposition \ref{prop:singular_values} and retracing the steps of the proof of Lemma E.1 of \cite{ma_nandy}, we can show \eqref{eq:state_ev_g_amp_2}. Again, one can show Assertions (3)-(6) for $\genv_{1,h}$ by retracing the steps of the proof of Lemma E.1 of \cite{ma_nandy}. Now, let us assume that Assertions (1)-(6) hold for $t=0,\ldots,s-1$. We shall show that they also hold for $t=s$. In fact, we shall only show Assertion (2) for $t=s$, as the proofs of all other assertions are similar to the proof of Lemma E.1 of \cite{ma_nandy}. To show the Assertion 2 for $\genu_{s,h}$, observe that, since $\frac{|\mathcal F_L|}{N} \rightarrow \lambda_L$, using Cauchy-Schwarz inequality and \eqref{prop:singular_values}, we have a constant $C>0$ (possibly depending on $\lambda_L$), such that
\begin{align}
        &\lim_{N \rightarrow \infty}\Bigg|\frac{1}{|\mathcal F_L|}\sum_{i \in \mathcal F_L}\psi\Bigg(\{(\rfF_{L,s,h}\bm \Xi^R_{s,h}+\wt{\bm G}_h\bm P^\perp_{\sfnbr_{s,h}}\genvb_{s,h}+ \wt{\boldsymbol{\mathrm N}}_{R,s,h}\,o(1))_{i*}\}_{h=1}^{m},\\
        & \hskip 10em \{(\upca_{0,h})_{i*}\}_{h=1}^{m},\{(\wt{\bm X}_\ell)_{i*}\}_{\ell=1}^{\wt m},\{(\bm U_{h})_{i*}\}_{h=1}^{m},\{(\wt{\bm U}_{\ell})_{i*}\}_{\ell=1}^{\wt m}\Bigg)\nonumber\\
    &\hskip 4em- \frac{1}{|\mathcal F_L|}\sum_{i \in \mathcal F_L}\psi\Bigg(\{(\rfF_{L,s,h}\bm \Xi^R_{s,h}+\wt{\bm G}_h\bm P^\perp_{\sfnbr_{s,h}}\genvb_{s,h})_{i*}\}_{h=1}^{m},\\
        & \hskip 13em \{(\upca_{0,h})_{i*}\}_{h=1}^{m},\{(\wt{\bm X}_\ell)_{i*}\}_{\ell=1}^{\wt m},\{(\bm U_{h})_{i*}\}_{h=1}^{m},\{(\wt{\bm U}_{\ell})_{i*}\}_{\ell=1}^{\wt m}\Bigg)\Bigg|,\nonumber\\
    & \le C\left(1+\sum_{h=1}^{m}\frac{\left\|\rfF_{L,s,h}\bm \Xi^R_{t,h}+\wt{\bm G}_h\bm P^\perp_{\sfnbr_{s,h}}\genvb_{s,h}\right\|_F}{\sqrt{N}}+\sum_{h=1}^{m}\frac{\|\wt{\boldsymbol{\mathrm N}}_{R,s,h}\|_F}{\sqrt{N}}+\sum_{h=1}^{m}\frac{\|\bm U_{h}\|_F}{\sqrt{N}}\right.\\
    &\left.\hskip 6em +\sum_{h=1}^{m}\frac{\|\upca_{0,h}\|_F}{\sqrt{N}}+\sum_{\ell=1}^{\wt m}\frac{\|\wt{\bm X}_\ell\|_F}{\sqrt{N}}+\sum_{\ell=1}^{\wt m}\frac{\|\wt{\bm U}_\ell\|_F}{\sqrt{N}}\right) \times \left(\sum_{h=1}^{m}\frac{\|\wt{\boldsymbol{\mathrm N}}_{R,s,h}\|_F}{\sqrt{N}}\right) \times ~o(1).
    \end{align}
    Furthermore, from induction hypotheses
    \begin{align}
        &\limsup_{N \rightarrow \infty}\sum_{h=1}^{m}\frac{\left\|\rfF_{L,s,h}\bm \Xi^R_{t,h}+\wt{\bm G}_h\bm P^\perp_{\sfnbr_{s,h}}\genvb_{s,h}\right\|_F}{\sqrt{N}} < \infty,\quad \limsup_{N \rightarrow \infty}\sum_{h=1}^{m}\frac{\|\upca_{0,h}\|_F}{\sqrt{N}}<\infty\\
        &\limsup_{N \rightarrow \infty}\sum_{h=1}^{m}\frac{\|\wt{\boldsymbol{\mathrm N}}_{R,s,h}\|_F}{\sqrt{N}}<\infty,\quad \limsup_{N \rightarrow \infty}\sum_{\ell=1}^{\wt m}\frac{\|\wt{\bm X}_\ell\|_F}{\sqrt{N}}<\infty\\
        &\mbox{and} \quad \limsup_{N \rightarrow \infty}\sum_{\ell=1}^{\wt m}\frac{\|\wt{\bm U}_\ell\|_F}{\sqrt{N}}<\infty.
    \end{align}
    Therefore, we can conclude that
    \begin{align}
        &\lim_{N \rightarrow \infty}\Bigg|\frac{1}{|\mathcal F_L|}\sum_{i \in \mathcal F_L}\psi\Bigg(\{(\rfF_{L,s,h}\bm \Xi^R_{s,h}+\wt{\bm G}_h\bm P^\perp_{\sfnbr_{s,h}}\genvb_{s,h}+ \wt{\boldsymbol{\mathrm N}}_{R,s,h}\,o(1))_{i*}\}_{h=1}^{m},\\
        & \hskip 10em \{(\upca_{0,h})_{i*}\}_{h=1}^{m},\{(\wt{\bm X}_\ell)_{i*}\}_{\ell=1}^{\wt m},\{(\bm U_{h})_{i*}\}_{h=1}^{m},\{(\wt{\bm U}_{\ell})_{i*}\}_{\ell=1}^{\wt m}\Bigg)\nonumber\\
    &\hskip 4em- \frac{1}{|\mathcal F_L|}\sum_{i \in \mathcal F_L}\psi\Bigg(\{(\rfF_{L,s,h}\bm \Xi^R_{s,h}+\wt{\bm G}_h\bm P^\perp_{\sfnbr_{s,h}}\genvb_{s,h})_{i*}\}_{h=1}^{m},\{(\upca_{0,h})_{i*}\}_{h=1}^{m},\\
        & \hskip 15em \{(\wt{\bm X}_\ell)_{i*}\}_{\ell=1}^{\wt m},\{(\bm U_{h})_{i*}\}_{h=1}^{m},\{(\wt{\bm U}_{\ell})_{i*}\}_{\ell=1}^{\wt m}\Bigg)\Bigg|=0.
    \end{align}
    Using the above observation, we can conclude \eqref{eq:state_ev_g_amp_1} by using Theorem 3 of \cite{BM11journal} and retracing the steps used to prove Lemma 5 of \cite{Berthier} and Lemma E.1 of \cite{ma_nandy}. Furthermore, \eqref{eq:state_ev_g_amp_2} can also be established using similar techniques. We omit the details for the sake of brevity.

\section{Proofs of AMP asymptotics with estimated priors}
Using Lemma B.4 in \citet{eb_pca} and Lemma \ref{lem:weak_conv_emp_bayes}, we have the following version of Corollary B.5 of \citet{eb_pca}.
\begin{lem}
\label{lem:conv_of_density}
 For $x=(x_1,\ldots,x_m)\in \R^{r}$ and $\wt{x}=(\wt x_1,\ldots,\wt x_{\wt m})\in \R^{\wt r}$
 \begin{align}
 \label{eq:conv_of_density_1}
     & \sup_{(x,\wt{x})}\left|f^L_N\left(x,\wt{x};\wh \mu\right)-f^L\left(x,\wt{x};\mu\right)\right|\rightarrow 0.
 \end{align}
Similarly, for all $h \in [m]$,
\begin{align}
 \label{eq:conv_of_density_2}
   \sup_{y_h \in \R^{r_h}}\left|f^R_{h,N}(y_h;\wh\nu_h)-f^R(y_h;\nu_h)\right|\rightarrow 0, 
\end{align}
where $f^L_N,f^L,f^R_{h,N}$ and $f^R_h$ are defined in Lemma \ref{lem:consistency_of_NPMLE}.
\end{lem}

\subsection{Proof of Theorem \ref{thm:em_bayes_state_evol}}
We shall prove Theorem \ref{thm:em_bayes_state_evol} using the following lemmas characterizing the smoothness of the estimated denoisers. Let us consider the denoisers $u^\orc_{t,h},v^\orc_{t,h}$ defined using the true priors $\mu$ and $\nu_h$; the state evolution matrices $\{\hslmor_{t,h},\hsrmor_{t,h},\hslvor_{t,h},\hsrvor_{t,h}:t\ge 0,\,h \in [m]\}$ and $\{\bm L_\ell:\ell \in [\wt m]\}$. Also consider the approximate denoisers constructed using the Empirical Bayes estimates of the priors $\wh{\mu}$ and $\wh{\nu}_h$, namely,  defined by \eqref{eq:emp_bayes_1}, \eqref{eq:emp_bayes_2}, the state evolution matrices $\{\hslm_{t,h},\hsrm_{t,h},\hslv_{t,h},\hsrv_{t,h}:t \ge 0\}$ defined by \eqref{eq:state_evol_emp_bayes} and the matrices $\{\hatl_\ell:\ell \in [\wt m]\}$ defined in \eqref{eq:sqrt_L}.

\begin{lem}
\label{lem:err_emp_bayes}
For all matrices $\bm F_h \in \R^{N \times r_h}$, $\wt{\bm X}_\ell \in \R^{N \times {\wt r}_\ell}$ and $\bm G_h \in \R^{p_h \times r_h}$ where $h \in [m]$ and $\ell \in [\wt m]$, as $N \rightarrow \infty$, we have
\begin{align}
    \label{eq:u_denoiser_1}
        \frac{1}{N}\|u_{t,h}(\bm F_1,\ldots,\bm F_m,\wt{\bm X}_1,\ldots,\wt{\bm X}_{\wt m};\wh{\mu})-u^\orc_{t,h}(\bm F_1,\ldots,\bm F_m,\wt{\bm X}_1,\ldots,\wt{\bm X}_{\wt m};\mu)\|^2_F \xrightarrow{a.s.} 0,
    \end{align}
    and
    \begin{align}
    \label{eq:u_denoiser_2}
      \frac{1}{p_h}\|v_{t,h}(\bm G_h;\wh{\nu}_h)-v^\orc_{t,h}(\bm G_h;\nu_h)\|^2_F \xrightarrow{a.s.} 0.
    \end{align}
    In the above equations, for all $h \in [m]$ the functions $u_{t,h}:\R^{r+\wt r} \rightarrow \R^{r_h}$ and $v_{t,h}:\R^{r_h} \rightarrow \R^{r_h}$ are defined in \eqref{eq:opt_denoisers} and the functions $u^\orc_{t,h}:\R^{r+\wt r} \rightarrow \R^{r_h}$ and $v^\orc_{t,h}:\R^{r_h} \rightarrow \R^{r_h}$ are defined in \eqref{eq:oracle_denoisers}.
Similarly, for the functions $\wt{u}_{t,h}:\R^{r+\wt r} \rightarrow \R^{\wt{r}_\ell}$ defined in \eqref{eq:theta_t_k} and $\wt{u}^\orc_{t,h}:\R^{r+\wt r} \rightarrow \R^{\wt{r}_\ell}$ defined in \eqref{eq:oracle_denoisers}, as $N \rightarrow \infty$ 
\begin{align}
\label{eq:conv_theta_n}
\frac{1}{N}\|\wt{u}_{t,h}(\bm F_1,\ldots,\bm F_m,\wt{\bm X}_1,\ldots,\wt{\bm X}_{\wt m};\wh{\mu})-\wt{u}^\orc_{t,h}(\bm F_1,\ldots,\bm F_m,\wt{\bm X}_1,\ldots,\wt{\bm X}_{\wt m};\mu)\|^2_F \xrightarrow{a.s.} 0.
\end{align}
\end{lem}
Next, we consider the following lemma characterizing the smoothness of the derivatives of the estimated denoisers.
\begin{lem}
    \label{lem:err_emp_bayes_derivative}
     Consider the same set-up as Lemma \ref{lem:err_emp_bayes}. Then as $N \rightarrow \infty$
    \begin{align}
    &\frac{1}{N}\sum_{i=1}^{N}\|\jacl_{t,h}((\bm F_1)_{i*},\ldots,(\bm F_m)_{i*},(\wt{\bm X}_1)_{i*},\ldots,(\wt{\bm X}_{\wt m})_{i*};\wh{\mu})\\
    &\hskip 8em -\jaclor_{t,h}((\bm F_1)_{i*},\ldots,(\bm F_m)_{i*},(\wt{\bm X}_1)_{i*},\ldots,(\wt{\bm X}_{\wt m})_{i*};\mu)\|_{\mathrm{op}} \xrightarrow{a.s.} 0,
    \end{align}
    and
    \[
    \frac{1}{p_h}\sum_{j=1}^{p_h}\|\jacr_{t,h}((\bm G_h)_{i*};\wh{\nu}_h)-\jacror_{t,h}((\bm G_h)_{i*};\nu_h)\|_{\mathrm{op}} \xrightarrow{a.s.} 0,
    \]
    for all matrices $\bm F_h \in \R^{N \times r_h}$, $\wt{\bm X}_\ell \in \R^{N \times {\wt r}_\ell}$ and $\bm G_h \in \R^{p_h \times r_h}$ where $h \in [m]$ and $\ell \in [\wt m]$.
\end{lem}
\begin{proof}[Proof of Theorem \ref{thm:em_bayes_state_evol}]We shall prove the first two parts of the theorem using induction. Let us formulate two sets of induction hypotheses $\mathcal{H}_t$ and $\bar{\mathcal H}_t$ for $t \in \mathbb N \cup \{0\}$, where $\mathcal{H}_t$ implies:
\begin{enumerate}
\item $\hsrm_{t,h} \rightarrow \hsrmor_{t,h}$ and $\hsrv_{t,h} \rightarrow \hsrvor_{t,h}$ as $N \rightarrow \infty$,
\item As $N \rightarrow \infty$, $\frac{1}{N}\|\bu_{t-1,h}-\buor_{t-1,h}\|^2_F \xrightarrow{a.s.} 0$, 
\item As $N \rightarrow \infty$, $\frac{1}{N}\|\wt{\bm U}_{t-1,\ell}-\wt{\bm U}^\orc_{t-1,\ell}\|^2_F \xrightarrow{a.s.} 0$, \quad \mbox{for $t \ge 1$,}
\item As $N \rightarrow \infty$, $\frac{1}{p_h}\|\bm V_{t,h}-\bm V^\orc_{t,h}\|^2_F \xrightarrow{a.s.} 0$, \quad \mbox{for all $h \in [m],\; \ell \in [\wt m]$,}
\end{enumerate}
 and $\bar{\mathcal H}_t$ implies:
\begin{enumerate}
\item $\hslm_{t,h} \rightarrow \hslmor_{t,h}$ and $\hslv_{t,h} \rightarrow \hslvor_{t,h}$ as $N \rightarrow \infty$.
\item As $p_h \rightarrow \infty$, $\frac{1}{p_h}\|\bv_{t,h}-\bvor_{t,h}\|^2_F \xrightarrow{a.s.} 0$, 
\item As $p_h \rightarrow \infty$, $\frac{1}{N}\|\bm U_{t,h}-\bm U^\orc_{t,h}\|^2_F \xrightarrow{a.s.} 0$, \quad \mbox{for all $h \in [m]$.}
\end{enumerate}

\paragraph{Step 1: Initialization $\mathcal{H}_0$.} By \eqref{eq:initializers},~\eqref{eq:est_sing_val} and the results from \cite{BENAYCHGEORGES2012120}, the first assertion holds. For the second assertion observe the following:
\[
\frac{1}{N}\|\bu_{-1,h}-\buor_{-1,h}\|^2_F \le \frac{1}{N}\|\upca_{0,h}\|^2_F\;\|(\hsrvp_{0,h})^{1/2}-(\hsrvorp_{0,h})^{1/2}\|^2_{\mathrm{op}}.
\]
By Proposition \ref{prop:singular_values}, $\limsup\limits_{N \rightarrow \infty}\frac{1}{N}\|\upca_{0,h}\|^2_F<\infty$ almost surely. Finally, since the plug-in estimator of $\hsrvp_{0,h}$ converges to $\hsrvorp_{0,h}$ as $N \rightarrow \infty$, the second assertion holds. Further, as both algorithms, given by \eqref{eq:orc_amp_emp_bayes} and \eqref{eq:orc_amp_basic}, have the same initialization for $\bm V_{0,h}$ for all $h \in [m]$ and $\ell \in [\wt m]$, the fourth assertions also hold. The third assertion is not applicable for $\mathcal H_0$.
\paragraph{Step 2: $\mathcal{H}_{t}  \implies \bar{\mathcal H}_t$.} This inductive step follows using the arguments of the proof of Theorem 5.4 in \citet{eb_pca}.
\paragraph{Step 3: $\bar{\mathcal H}_t  \implies \mathcal{H}_{t+1}$.} 
First, let us consider the second assertion. Observe that 
\begin{align}
\label{eq:triangle_ineq}
    & \frac{1}{N}\|\bu_{t,h}-\buor_{t,h}\|^2_F\\
    & = \frac{1}{N}\|u_{t,h}(\bm U_{t,1},\ldots,\bm U_{t,m},\wt{\bm X}_1,\ldots,\wt{\bm X}_{\wt m};\wh{\mu})-u^\orc_{t,h}(\bm U^\orc_{t,1},\ldots,\bm U^\orc_{t,m},\wt{\bm X}_1,\ldots,\wt{\bm X}_{\wt m};\mu)\|^2_F\\
    & = \frac{1}{N}\sum_{i=1}^{N}\|u_{t,h}((\bm U_{t,1})_{i*},\ldots,(\bm U_{t,m})_{i*},(\wt{\bm X}_{1})_{i*},\ldots,(\wt{\bm X}_{\wt m})_{i*};\wh{\mu})\\
    &\hskip 10em-u^\orc_{t,h}((\bm U^\orc_{t,1})_{i*},\ldots,(\bm U^\orc_{t,m})_{i*},(\wt{\bm X}_{1})_{i*},\ldots,(\wt{\bm X}_{\wt m})_{i*};\mu)\|^2_2\\
    & \le \frac{2}{N}\sum_{i=1}^{N}\|u_{t,h}((\bm U_{t,1})_{i*},\ldots,(\bm U_{t,m})_{i*},(\wt{\bm X}_{1})_{i*},\ldots,(\wt{\bm X}_{\wt m})_{i*};\wh{\mu})\\
    &\hskip 10em-u^\orc_{t,h}((\bm U_{t,1})_{i*},\ldots,(\bm U_{t,m})_{i*},(\wt{\bm X}_{1})_{i*},\ldots,(\wt{\bm X}_{\wt m})_{i*};\mu)\|^2_2\\
    & \hskip 1em + \frac{2}{N}\sum_{i=1}^{N}\|u^\orc_{t,h}((\bm U_{t,1})_{i*},\ldots,(\bm U_{t,m})_{i*},(\wt{\bm X}_{1})_{i*},\ldots,(\wt{\bm X}_{\wt m})_{i*};\mu)\\
    &\hskip 10em-u^\orc_{t,h}((\bm U^\orc_{t,1})_{i*},\ldots,(\bm U^\orc_{t,m})_{i*},(\wt{\bm X}_{1})_{i*},\ldots,(\wt{\bm X}_{\wt m})_{i*};\mu)\|^2_2.
\end{align}
Let us consider the first expression on the right-hand side of the above equation. By Lemma \ref{lem:err_emp_bayes} we have:
\begin{align}
    &\frac{2}{N}\sum_{i=1}^{N}\|u_{t,h}((\bm U_{t,1})_{i*},\ldots,(\bm U_{t,m})_{i*},(\wt{\bm X}_{1})_{i*},\ldots,(\wt{\bm X}_{\wt m})_{i*};\wh{\mu})\\
    &\hskip 10em-u^\orc_{t,h}((\bm U_{t,1})_{i*},\ldots,(\bm U_{t,m})_{i*},(\wt{\bm X}_{1})_{i*},\ldots,(\wt{\bm X}_{\wt m})_{i*};\mu)\|^2_2\xrightarrow{a.s.} 0.
\end{align}
Next, for second expression on the right-hand side of \eqref{eq:triangle_ineq}, we use the Lipschitz property of the denoisers implied by Assumption \ref{asm:prior_1_mom}, to conclude that for all $h \in [m]$ 
\begin{align}
   &\frac{2}{N}\sum_{i=1}^{N}\|u^\orc_{t,h}((\bm U_{t,1})_{i*},\ldots,(\bm U_{t,m})_{i*},(\wt{\bm X}_{1})_{i*},\ldots,(\wt{\bm X}_{\wt m})_{i*};\mu)\\
    &\hskip 10em-u^\orc_{t,h}((\bm U^\orc_{t,1})_{i*},\ldots,(\bm U^\orc_{t,m})_{i*},(\wt{\bm X}_{1})_{i*},\ldots,(\wt{\bm X}_{\wt m})_{i*};\mu)\|^2_2\\
   & \le \frac{C_1}{N}\sum_{h=1}^m\|\bm U_{t,h}-\bm U^\orc_{t,h}\|^2_F \xrightarrow{a.s.}0, \quad \mbox{as $N \rightarrow \infty$}.
\end{align}
Here $C_1>0$ is an absolute constant and the last relation follows using $\bar{\mathcal{H}}_t(3)$. For the third assertion, observe that 
\begin{align}
\label{eq:triangle_ineq_ld}
    & \frac{1}{N}\|\wt{\bm U}_{t,h}-\wt{\bm U}^\orc_{t,h}\|^2_F\nonumber\\
    & = \frac{1}{N}\|\wt{u}_{t,h}(\bm U_{t,1},\ldots,\bm U_{t,m},\wt{\bm X}_1,\ldots,\wt{\bm X}_{\wt m};\wh{\mu})-\wt{u}^\orc_{t,h}(\bm U^\orc_{t,1},\ldots,\bm U^\orc_{t,m},\wt{\bm X}_1,\ldots,\wt{\bm X}_{\wt m};\mu)\|^2_F\nonumber\\
    & = \frac{1}{N}\sum_{i=1}^{N}\|\wt{u}_{t,h}((\bm U_{t,1})_{i*},\ldots,(\bm U_{t,m})_{i*},(\wt{\bm X}_1)_{i*},\ldots,(\wt{\bm X}_{\wt m})_{i*};\wh{\mu})\nonumber\\
    &\hskip 10em-\wt{u}^\orc_{t,h}((\bm U^\orc_{t,1})_{i*},\ldots,(\bm U^\orc_{t,m})_{i*},(\wt{\bm X}_1)_{i*},\ldots,(\wt{\bm X}_{\wt m})_{i*};\mu)\|^2_2.
\end{align}
The right-hand side of \eqref{eq:triangle_ineq_ld} converges to $0$ by \eqref{eq:conv_theta_n}, Lemma \ref{lem:err_emp_bayes} and arguments as before.
Finally, consider the fourth assertion. Define \[\jacl_{t,h}\equiv\jacl_{t,h}(\bm U_{t,1},\ldots,\bm U_{t,m},\wt{\bm X}_1,\ldots,\wt{\bm X}_{\wt m};\wh{\mu})\] and \[\jaclor_{t,h}\equiv\jaclor_{t,h}(\bm U^\orc_{t,1},\ldots,\bm U^\orc_{t,m},\wt{\bm X}_1,\ldots,\wt{\bm X}_{\wt m};\mu).\] We can get a constant $C_2>0$ such that
\begin{align}
\label{eq:triangle_ineq_f}
    \frac{1}{p_h}\|\bm V_{t+1,h}-\bm V^\orc_{t+1,h}\|^2_F
    & \le \frac{C_2}{p_h}\|\bm X^\top_h\bu_{t,h}-\bm X^\top_h\bu^\orc_{t,h}\|^2_F+C_2\frac{\|\bv_{t,h}\|^2_F}{p_h}\|\jacl_{t,h}-\jaclor_{t,h}\|^2_{\mathrm{op}}\\
    &\hskip 3em +\frac{C_2}{p_h}\|\bv_{t,h}-\bv^\orc_{t,h}\|^2_F\|\jaclor_{t,h}\|^2_{\mathrm{op}}\\
    & \le \frac{C_2}{p_h}\|\bm X_h\|^2_{\mathrm{op}}\|\bu_{t,h}-\bu^{\orc}_{t,h}\|^2_F+C_2\frac{\|\bv_{t,h}\|^2_F}{p_h}\|\jacl_{t,h}-\jaclor_{t,h}\|^2_{\mathrm{op}}\\
    &\hskip 3em +\frac{C_2}{p_h}\|\bv_{t,h}-\bv^\orc_{t,h}\|^2_F\|\jaclor_{t,h}\|^2_{\mathrm{op}}.
\end{align}
By Corollary 5.35 of \citet{vershynin_2012} we have $\limsup_{N \rightarrow \infty}\|\bm X_h\|_{\mathrm{op}}<\infty$ almost surely. By $\mathcal{\bar{H}}_t(2)$, $\mathcal{H}_t(2)$ and Theorem \ref{thm:mmse_svd}, we can conclude that $\limsup_{p_h \rightarrow \infty}\frac{\|\bv_{t,h}\|^2_F}{p_h}<\infty$, almost surely. Further, observe that 
\begin{align}
    &\jacl_{t,h}-\jaclor_{t,h}\\
    &= \frac{1}{N}\sum_{i=1}^{N}\bigg[\jacl_{t,h}((\bm U_{t,1})_{i*},\ldots,(\bm U_{t,m})_{i*},(\wt{\bm X}_1)_{i*},\ldots,(\wt{\bm X}_{\wt m})_{i*};\wh{\mu})\\
    &\hskip 8em-\jaclor_{t,h}((\bm U^\orc_{t,1})_{i*},\ldots,(\bm U^\orc_{t,m})_{i*},(\wt{\bm X}_1)_{i*},\ldots,(\wt{\bm X}_{\wt m})_{i*};\mu)\bigg]\\
    \label{eq:def_I}
    & \le \frac{1}{N}\sum_{i=1}^{N}\bigg[\jacl_{t,h}((\bm U_{t,1})_{i*},\ldots,(\bm U_{t,m})_{i*},(\wt{\bm X}_1)_{i*},\ldots,(\wt{\bm X}_{\wt m})_{i*};\wh{\mu})\\
    &\hskip 8em-\jaclor_{t,h}((\bm U_{t,1})_{i*},\ldots,(\bm U_{t,m})_{i*},(\wt{\bm X}_1)_{i*},\ldots,(\wt{\bm X}_{\wt m})_{i*};\mu)\bigg]\\
    \label{eq:def_II}
    & + \frac{1}{N}\sum_{i=1}^{N}\bigg[\jaclor_{t,h}((\bm U_{t,1})_{i*},\ldots,(\bm U_{t,m})_{i*},(\wt{\bm X}_1)_{i*},\ldots,(\wt{\bm X}_{\wt m})_{i*};\mu)\\
    &\hskip 8em-\jaclor_{t,h}((\bm U^\orc_{t,1})_{i*},\ldots,(\bm U^\orc_{t,m})_{i*},(\wt{\bm X}_1)_{i*},\ldots,(\wt{\bm X}_{\wt m})_{i*};\mu)\bigg].
\end{align}
For the expression in \eqref{eq:def_I}, using Lemma \ref{lem:err_emp_bayes_derivative} can conclude that
\begin{align}
    &\Bigg\|\frac{1}{N}\sum_{i=1}^{N}\bigg[\jacl_{t,h}((\bm U_{t,1})_{i*},\ldots,(\bm U_{t,m})_{i*},(\wt{\bm X}_1)_{i*},\ldots,(\wt{\bm X}_{\wt m})_{i*};\wh{\mu})\\
    &\hskip 8em-\jaclor_{t,h}((\bm U_{t,1})_{i*},\ldots,(\bm U_{t,m})_{i*},(\wt{\bm X}_1)_{i*},\ldots,(\wt{\bm X}_{\wt m})_{i*};\mu)\bigg]\Bigg\|_{\mathrm{op}}\\
    &\le \frac{1}{N}\sum_{i=1}^{N}\bigg\|\jacl_{t,h}((\bm U_{t,1})_{i*},\ldots,(\bm U_{t,m})_{i*},(\wt{\bm X}_1)_{i*},\ldots,(\wt{\bm X}_{\wt m})_{i*};\wh{\mu})\\
    &\hskip 8em-\jaclor_{t,h}((\bm U_{t,1})_{i*},\ldots,(\bm U_{t,m})_{i*},(\wt{\bm X}_1)_{i*},\ldots,(\wt{\bm X}_{\wt m})_{i*};\mu)\bigg\|_{\mathrm{op}}\\
    & \hskip 2em \xrightarrow{a.s.} 0, \quad \mbox{as $N \rightarrow \infty$.}
\end{align}
To analyze the expression in \eqref{eq:def_II}, we first express the function $u^{\orc}_{t,h}(x_1,\ldots,x_m,\wt x_1,\ldots,\wt x_{\wt m};\mu)$ using Tweedie's formula as follows:
\begin{align}
    u^\circ_{t,h}(x_1,\ldots,x_m,\wt x_1,\ldots,\wt x_{\wt m};\mu)=x_h+\frac{\nabla_hf^L_{t}(x_1,\ldots,x_m,\wt x_1,\ldots,\wt x_{\wt m};\mu)}{f^L_{t}(x_1,\ldots,x_m,\wt x_1,\ldots,\wt x_{\wt m};\mu)},
\end{align}
where $\nabla_h$ refers to the gradient of a function $f:\R^{r+\wt r} \rightarrow \R$ with respect to the vector of argument $x_h$. The function $f^L_{t}:\R^{r+\wt r} \rightarrow \R$ is defined as follows:
\begin{align} f^L_{t}(x_1,\ldots,x_m,\wt x_1,\ldots,\wt x_{\wt m};\mu)\equiv&\int_{u,\tilde{u}}\prod_{h=1}^{m}\phi_{r_h}((\hslvor_{t,h})^{-1/2}(x_h-\hslmor_{t,h} u_h))\times\\
&\hskip 4em\prod_{\ell=1}^{\wt m}\phi_{\wt r_{\ell}}(\wt x_\ell-\bm L_\ell\wt{u}_\ell)d\mu(u_1,\ldots,u_m,\wt{u}_1,\ldots,\wt{u}_\ell).
\end{align}
By Assumption \ref{asm:prior_1_mom}, since $u^\orc_{t,h}$ is uniformly Lipschitz, therefore $\jaclor_{t,h}$ is uniformly bounded by a constant $C_L$. Further, for all $h \in [m]$ and $\ell \in [\wt m]$; $x_h \in \mathbb B_{r_h}(B)$ (here, $\mathbb B_{r}(B)$ denotes the open Euclidean ball of dimension $r$ and radius $B$) and $\wt x_\ell \in \mathbb B_{{\wt r}_\ell}(B)$, then there \revsn{exists} a constant $c>0$, such that $f_t(x_1,\ldots,x_m,\wt x_1,\ldots,\wt x_{\wt m};\mu)>c$ which implies that the derivative of $\jaclor_{t,h}$ with respect to the variable $x_h$ is uniformly bounded by a constant $C_B>0$ for all $x_h \in \mathbb B_{r_h}(B)$ and $\wt x_\ell \in \mathbb B_{{\wt r}_\ell}(B)$.
Hence, if $x_h \in \mathbb B_{r_h}(B)$ and $\wt x_\ell \in \mathbb B_{{\wt r}_\ell}(B)$ for all $h \in [m]$ and $\ell \in [\wt m]$, then $\jaclor_{t,h}$ is Lipschitz continuous with Lipschitz constant $C_B$. Now, consider the term in \eqref{eq:def_II}. 
\begin{align}
    &\frac{1}{N}\sum_{i=1}^{N}\bigg\|\jaclor_{t,h}((\bm U_{t,1})_{i*},\ldots,(\bm U_{t,m})_{i*},(\wt{\bm X}_1)_{i*},\ldots,(\wt{\bm X}_{\wt m})_{i*};{\mu})\\
    &\hskip 8em-\jaclor_{t,h}((\bm U^\orc_{t,1})_{i*},\ldots,(\bm U^\orc_{t,m})_{i*},(\wt{\bm X}_1)_{i*},\ldots,(\wt{\bm X}_{\wt m})_{i*};\mu)\bigg\|_{\mathrm{op}}\\
    &\le \frac{1}{N}\sum_{i=1}^{N}\bigg\|\jaclor_{t,h}((\bm U_{t,1})_{i*},\ldots,(\bm U_{t,m})_{i*},(\wt{\bm X}_1)_{i*},\ldots,(\wt{\bm X}_{\wt m})_{i*};{\mu})\\
    &\hskip 8em-\jaclor_{t,h}((\bm U^\orc_{t,1})_{i*},\ldots,(\bm U^\orc_{t,m})_{i*},(\wt{\bm X}_1)_{i*},\ldots,(\wt{\bm X}_{\wt m})_{i*};\mu)\bigg\|_{\mathrm{op}}\\
        \label{eq:def_IV}
    & \hskip 5em \times \prod_{h=1}^m\mathbb{I}(\|(\bm U_{t,h})_{i*}\|_2 \le B) \times \prod_{\ell=1}^n\mathbb{I}(\|(\wt{\bm X}_\ell)_{i*}\|_2 \le B) \times \prod_{h=1}^m\mathbb{I}(\|(\bm U^\orc_{t,h})_{i*}\|_2 \le B)\\
    & \hskip 2em+ \frac{C_L}{N}\sum_{i=1}^{N}\sum_{h=1}^{m}\mathbb{I}(\|(\bm U_{t,h})_{i*}\|_2 \ge B) + \frac{C_L}{N}\sum_{i=1}^{N}\sum_{\ell=1}^{\wt m}\mathbb{I}(\|(\wt{\bm X}_\ell)_{i*}\|_2 \ge B)\\
    &\hskip 8em+\frac{C_L}{N}\sum_{i=1}^{N}\sum_{h=1}^{m}\mathbb{I}(\|(\bm U^\orc_{t,h})_{i*}\|_2 \ge B).
\end{align}
Since $\|x\|^2_2 \ge B^2\mathbb{I}(\|x\|_2 \ge B)$, for all $x \in \R^{r_h}$, the last two terms in the above display can be bounded as follows:
\begin{align}
    &\frac{C_L}{N}\sum_{i=1}^{N}\sum_{h=1}^{m}\mathbb{I}(\|(\bm U_{t,h})_{i*}\|_2 \ge B) + \frac{C_L}{N}\sum_{i=1}^{N}\sum_{\ell=1}^{\wt m}\mathbb{I}(\|(\wt{\bm X}_\ell)_{i*}\|_2 \ge B)\\
    &\hskip 20em+\frac{C_L}{N}\sum_{i=1}^{N}\sum_{h=1}^{m}\mathbb{I}(\|(\bm U^\orc_{t,h})_{i*}\|_2 \ge B)\\
    & \le \frac{C_L}{NB^2}\sum_{h=1}^{m}\|\bm U_{t,h}\|^2_F + \frac{C_L}{NB^2}\sum_{h=1}^{m}\|\bm U^\orc_{t,h}\|^2_F + \frac{C_L}{NB^2}\sum_{\ell=1}^{\wt m}\|\wt{\bm X}_\ell\|^2_F.
\end{align}
By definition of $\wt{\bm X}_\ell$, $\mathcal{\wb{H}}_{t}(3)$ and Theorem \ref{thm:asymptotics of amp_iterates}, we can conclude that, almost surely
\[
\limsup_{N \rightarrow \infty}\left\{\frac{C_L}{N}\sum_{h=1}^{m}\|\bm U_{t,h}\|^2_F + \frac{C_L}{N}\sum_{h=1}^{m}\|\bm U^\orc_{t,h}\|^2_F+ \frac{C_L}{N}\sum_{\ell=1}^{\wt m}\|\wt{\bm X}_\ell\|^2_F\right\} < \infty.
\]
Hence, if we take $B \rightarrow \infty$, we get
\begin{align}
&\limsup_{B \rightarrow \infty}\left\{\frac{C_L}{N}\sum_{i=1}^{N}\sum_{h=1}^{m}\mathbb{I}(\|(\bm U_{t,h})_{i*}\|_2 \ge B) + \frac{C_L}{N}\sum_{i=1}^{N}\sum_{\ell=1}^{\wt m}\mathbb{I}(\|(\wt{\bm X}_\ell)_{i*}\|_2 \ge B)\right.\\
& \hskip 5em\left.+\frac{C_L}{N}\sum_{i=1}^{N}\sum_{h=1}^{m}\mathbb{I}(\|(\bm U^\orc_{t,h})_{i*}\|_2 \ge B)\right\}\overset{a.s.}{=}0
\end{align}
To analyze the other term we use the Lipschitz property of the Jacobian matrices of the denoisers when its arguments are confined to an Euclidean ball of appropriate dimensions and radius $B>0$. That condition implies 
\begin{align}
    & \frac{1}{N}\sum_{i=1}^{N}\bigg\|\jaclor_{t,h}((\bm U_{t,1})_{i*},\ldots,(\bm U_{t,m})_{i*},(\wt{\bm X}_1)_{i*},\ldots,(\wt{\bm X}_{\wt m})_{i*};{\mu})\\
    &\hskip 8em-\jaclor_{t,h}((\bm U^\orc_{t,1})_{i*},\ldots,(\bm U^\orc_{t,m})_{i*},(\wt{\bm X}_1)_{i*},\ldots,(\wt{\bm X}_{\wt m})_{i*};\mu)\bigg\|_{\mathrm{op}}\\
    & \hskip 5em \times \prod_{h=1}^m\mathbb{I}(\|(\bm U_{t,h})_{i*}\|_2 \le B) \times \prod_{\ell=1}^n\mathbb{I}(\|(\wt{\bm X}_\ell)_{i*}\|_2 \le B) \times \prod_{h=1}^m\mathbb{I}(\|(\bm U^\orc_{t,h})_{i*}\|_2 \le B)\\
    & \le \frac{C_{B}}{\sqrt{N}}\sum_{h=1}^{m}\|\bm U_{t,h}-\bm U^\orc_{t,h}\|_F \xrightarrow{a.s.} 0.
\end{align}
Hence, we can conclude that $\|\jacl_{t,h}-\jaclor_{t,h}\| \xrightarrow{a.s.}0$, as $N \rightarrow \infty$. Using similar arguments one can also show that $\limsup_{N \rightarrow \infty}\|\jaclor_{t,h}\|<\infty$, almost surely. Combining \eqref{eq:triangle_ineq_f}, $\mathcal{H}_{t+1}(2)$ and $\bar{\mathcal H}_t$, the third assertion of $\mathcal{H}_{t+1}$ follows.

To show the first assertion of $\mathcal{H}_{t+1}$, we observe that $u^{\otimes 2}_{t,h}$ is a pseudo-Lipschitz function. Hence, the assertion follows using the techniques described in the proof of Theorem 5.4 in \citet{eb_pca}. 

Finally, since $\frac{|\mathcal F_L|}{N} \rightarrow \lambda_L \in (0,1]$ and $\frac{|\mathcal F_{R,h}|}{p_h} \rightarrow \lambda_{R,h} \in (0,1]$ for all $h \in [m]$, the third part of the theorem follows using Theorem \ref{thm:asymptotics of amp_iterates} and the property of pseudo-Lipschitz functions (in particular, inequalities of form \eqref{eq:typical_scaling_ineq}).
\end{proof}

\subsection{Proof of Lemma \ref{lem:conv_of_density}}
We shall prove the statement \eqref{eq:conv_of_density_1}. The proof of \eqref{eq:conv_of_density_2} is similar. Denote by $$\wh{\bm \Theta}=\bigg(\{\hslmp_{0,h}\}_{h=1}^{m},\{\hslvp_{0,h}\}_{h=1}^{m},\{\hatl_{\ell}\}_{\ell=1}^{\wt m}\bigg),$$ and $$\bm \Theta=\bigg(\{\hslmorp_{0,h}\}_{h=1}^{m},\{\hslvorp_{0,h}\}_{h=1}^{m},\{\bm L_{\ell}\}_{\ell=1}^{\wt m}\bigg).$$ Consider the functions
\[
\xi^L_N((x,\wt x);(u,\wt u);\wh{\bm \Theta})=\prod_{h=1}^{m}\phi_{r_h}((\hslvp_{0,h})^{-1/2}(x_h-\hslmp_{0,h} u_h))\prod_{\ell=1}^{\wt m}\phi_{\wt r_{\ell}}(\wt x_\ell-\hatl_\ell\wt{u}_\ell),
\]
and
\[
\xi((x,\wt x);(u,\wt u);\bm \Theta)=\prod_{h=1}^{m}\phi_{r_h}((\hslvorp_{0,h})^{-1/2}(x_h-\hslmorp_{0,h}u_h))\prod_{\ell=1}^{\wt m}\phi_{\wt r_\ell}(\wt x_\ell-\bm L_\ell\wt{u}_\ell).
\]
Since, by Lemma \ref{lem:consistency_nuisance}, $\wh{\bm \Theta} \rightarrow \bm \Theta$ as $N \rightarrow \infty$, both $\xi^L_N((x,\wt x);(u,\wt u);\wh{\bm \Theta})$ and $\xi((x,\wt x);(u,\wt u);\bm \Theta)$ and their derivatives are uniformly bounded with respect to $(x,\wt x)$, $(u,\wt u)$ and $\wh{\bm \Theta}$, when $\wh{\bm \Theta}$ is contained in a finite radius neighborhood of $\bm \Theta$.  

Then, proceeding as in the proof of Lemma B.4 and Corollary B.5 of \citet{eb_pca}, the stated result follows.

\subsection{Proof of Lemma \ref{lem:err_emp_bayes}}
We shall prove \eqref{eq:u_denoiser_1} and the proof of \eqref{eq:u_denoiser_2} shall follow similarly. Let us consider the collection of matrices $\bm \Theta=\bigg(\{\hslmorp_{0,h}\}_{h=1}^{m},\{\hslvorp_{0,h}\}_{h=1}^{m},\{\bm L_{\ell}\}_{\ell=1}^{\wt m}\bigg)$ and $\wh{\bm \Theta}=\bigg(\{\hslmp_{0,h}\}_{h=1}^{m},\{\hslvp_{0,h}\}_{=1}^{m},\{\hatl_{\ell}\}_{\ell=1}^{\wt m}\bigg)$. By Tweedie's formula, for $h \in [m]$, we can write the posterior mean denoisers as 
\begin{align}
&u_{t,h}(x_1,\ldots,x_m,\wt x_1,\ldots,\wt x_{\wt m};\wh{\mu})\\
&=(\hslmp_{0,h})^{-1}\bigg(\hslvp_{0,h}\frac{\nabla_{x_h}f^L_N(x_1,\ldots,x_m,\wt x_1,\ldots,\wt x_{\wt m};\wh{\mu})^\top}{f^L_N(x_1,\ldots,x_m,\wt x_1,\ldots,\wt x_{\wt m};\wh{\mu})}+x_h\bigg),
\end{align}
and
\begin{align}
&u^\orc_{t,h}(x_1,\ldots,x_m,\wt x_1,\ldots,\wt x_{\wt m};\mu)\\
&=(\hslmorp_{0,h})^{-1}\bigg(\hslvorp_{0,h}\frac{\nabla_{x_h}f^L(x_1,\ldots,x_m,\wt x_1,\ldots,\wt x_{\wt m};\mu)^\top}{f^L(x_1,\ldots,x_m,\wt x_1,\ldots,\wt x_{\wt m};\mu)}+x_h\bigg),
\end{align}
where for any function $f$ defined on appropriate domain, $\nabla_{x_h}f(x_1,\ldots,x_m,\wt x_1,\ldots,\wt x_{\wt m};\wh{\mu})$ denotes the gradient of the function $f$ with respect to $x_h$. In the above equations the functions $f^L_N$ and $f^L$ are defined in \eqref{eq:f_l_n} and \eqref{eq:f_l}, respectively.
By Lemma \ref{lem:conv_of_density}, for $N \rightarrow \infty$, we have
\begin{align}
     & \sup_{(x,\wt x)}\Bigg|f^L_N(x_1,\ldots,x_m,\wt x_1,\ldots,\wt x_{\wt m};\wh{\mu})-f^L(x_1,\ldots,x_m,\wt x_1,\ldots,\wt x_{\wt m};\mu)\Bigg|\rightarrow 0.
 \end{align}
Using similar techniques as the proof of Lemma \ref{lem:conv_of_density}, we can also show that for $h \in [m]$:
\begin{align}
     & \sup_{(x,\wt x)}\Bigg\|\nabla_{x_h}f^L_N(x_1,\ldots,x_m,\wt x_1,\ldots,\wt x_{\wt m};\wh{\mu})-\nabla_{x_h}f^L(x_1,\ldots,x_m,\wt x_1,\ldots,\wt x_{\wt m};\mu)\Bigg\|_{\mathrm{op}}\rightarrow 0.
 \end{align}
 Since, by Lemma \ref{lem:consistency_nuisance}, $\wh{\bm \Theta} \rightarrow \bm \Theta$ as $N \rightarrow \infty$, we have an open bounded neighborhood of $\bm \Theta$ such that as $N\rightarrow \infty$, the matrices in the collection $\wh{\bm \Theta}$ lie in it. Then using Assumption \ref{asm:prior_1_mom}, we can verify along the lines of Proposition B.1 of \citet{eb_pca} that there exists a neighborhood of $(\bm \Theta,\mu)$ (with respect to the topology of weak convergence in $\mu$) where $u_{t,h}$ is uniformly Lipschitz in $(x,\wt x)$. Therefore, following the analysis of Corollary B.3 of \citet{eb_pca}, we can conclude \eqref{eq:u_denoiser_1}. The proofs of \eqref{eq:u_denoiser_2} and \eqref{eq:conv_theta_n} follow using similar arguments.
 
 \subsection{Proof of Lemma \ref{lem:err_emp_bayes_derivative}}
 The proof of this Lemma follows using the techniques used to prove Lemma \ref{lem:err_emp_bayes} and Proposition B.6 of \citet{eb_pca}.

\section{Proof of Bayes optimality of orchestrated AMP }
\subsection{Proof of Lemma \ref{lem:redef_state_evol}}
Recall the definitions of $\gamman_{t,h}$ and $\gammab_{t,h}$ from \eqref{eq:def_gamma}. Using \eqref{eq:state_evol_gen}, we can write
\begin{align}
    \gamman_{t,h}= \bm D^{1/2}_h\hsrvor_{t,h}\bm D^{1/2}_h, \quad \mbox{and} \quad
    \gammab_{t,h}= \frac{1}{\gamma_h}\bm D^{1/2}_h\hslvor_{t,h}\bm D^{1/2}_h.
\end{align}
Since
\begin{align}
    \rmvfunc(\gamman_{t,h};\nu_h)= \mathbb{E}[(\mathrm{V}_h-\mathbb{E}_{\nu_h}[\mathrm{V}_h|\rmxr(\gamman_{t,h})])^{\otimes 2}],
\end{align}
we have
\[
\mathbb E[\mathbb{E}_{\nu_h}[\mathrm{V}_h|\rmxr(\gamman_{t,h})])^{\otimes 2}]]=\bm I_{r_h}-\rmvfunc(\gamman_{t,h};\nu_h).
\]
Next, recall that:
\[
\hslvor_{t,h}=\gamma_h\mathbb E[\mathbb{E}_{\nu_h}[\mathrm{V}_h|\hsrmor_{t,h}\mathrm{V}_h+(\hsrvor_{t,h})^{1/2}\mathrm{Z}^R_h]^{\otimes 2}],
\]
where $\mathrm{Z}^R_h \sim N_{r_h}(0,\bm I_{r_h})$ and is independent of $\mathrm{V}_h \sim \nu_h$. In both the equations above and in similar contexts later, the outer expectations are always with respect to the distributions of both $\mathrm V_h$ and the Gaussian random variables $\mathrm Z^R_h$. Again, using \eqref{eq:state_evol_gen}, the above equation can be re-expressed as follows:
\begin{align}
\hslvor_{t,h}&=\gamma_h\mathbb E[\mathbb{E}_{\nu_h}[\mathrm{V}_h|\bm D^{1/2}_h\mathrm{V}_h+\bm D^{1/2}_h(\hsrmor_{t,h})^{-1}(\hsrvor_{t,h})^{1/2}\mathrm{Z}^R_h]^{\otimes 2}]\\
&=\gamma_h\mathbb E[\mathbb{E}_{\nu_h}[\mathrm{V}_h|\bm D^{1/2}_h\mathrm{V}_h+(\gamman_{t,h})^{-1/2}\mathrm{Z}^R_h]^{\otimes 2}].
\end{align}
Hence, we get
\begin{align}
\gammab_{t,h}&=\mathbb E[\mathbb{E}[\bm D^{1/2}_h\mathrm{V}_h|\bm D^{1/2}_h\mathrm{V}_h+(\gamman_{t,h})^{-1/2}\mathrm{Z}^R_h]^{\otimes 2}]\\
&=\bm D_h-\rmvfunc_h(\gamman_{t,h};\bm D^{1/2}_h\nu_h).
\end{align}
The other equation follows similarly.

\subsection{Proof of Theorem \ref{thm:improvement_in_error}}
First, observe that $\mathbb E[\mathbb E[\Theta|\Theta+\bm \Gamma^{-1/2}Z]^{\otimes 2}]=\mathbb E[\mathbb E[\Theta|\bm \Gamma^{1/2}\Theta+Z]^{\otimes 2}]$ for any $\bm Q$ positive semi-definite and $Z\sim\mbox{Gaussian}$. 
Using Lemma 9 in \cite{article}, we can show that, for all $h \in [m]$:
\begin{align}
\label{eq:in_v}
   \rmvfunc_{h}(\gamman_1;\bm D^{1/2}_h\nu_h) \succeq \rmvfunc_{h}(\gamman_2;\bm D^{1/2}_h\nu_h) \quad \mbox{for $\gamman_1 \preceq \gamman_2$ and $\gamman_1,\gamman_2 \in \mathbb S^{r_h}_+$.} 
\end{align}
Further, using the techniques used to prove Lemma 9 in \cite{article}, we can show that for all $k \in [m]$ and pairs of positive semi-definite matrices $\gammab_{k},\wb{\bm\Gamma}^L_{k} \in \mathbb S^{r_k}_+$ satisfying $\gammab_{k} \preceq \wb{\bm\Gamma}^L_{k}$, if the matrices $\{\gammab_{g} \in \mathbb S^{r_g}_+:g \in [m]\setminus\{k\}\}$ are fixed, then we have:
\begin{align}
\label{eq:in_u_1}
&\rmufunc_{h}(\gamma_1\gammab_1,\ldots,\gamma_{k-1}\gammab_{k-1},\gamma_k\gammab_{k},\gamma_{k+1}\gammab_{k+1},\ldots,\gamma_m\gammab_m;\bm D^{1/2}\mu)\\
&\succeq \rmufunc_h(\gamma_1\gammab_1,\ldots,\gamma_{k-1}\gammab_{k-1},\gamma_k\wb{\bm\Gamma}^L_{k},\gamma_{k+1}\gammab_{k+1},\ldots,\gamma_m\gammab_m;\bm D^{1/2}\mu),
\end{align}
for all $h \in [m]$. This implies that for the collections of matrices given by $\{(\gammab_{h},\wb{\bm\Gamma}^L_{h}) \in \mathbb S^{r_h}_+ \times \mathbb S^{r_h}_+:\gammab_{h}\preceq\wb{\bm\Gamma}^L_{h}, h \in [m]\}$, we have:
\begin{align}
\label{eq:in_u}
&\rmufunc_h(\gamma_1\gammab_{1},\ldots,\gamma_m\gammab_{m};\bm D^{1/2}\mu)\succeq \rmufunc_h(\gamma_1\wb{\bm\Gamma}^L_{1},\ldots,\gamma_m\wb{\bm\Gamma}^L_{m};\bm D^{1/2}\mu),
\end{align}
for all $h \in [m]$. Further, by Lemma 9 in \cite{article}, we get that for any $\{\gamman_h \in \mathbb S^{r_h}_+:h \in [m]\}$ and $\{\gammab_h \in \mathbb S^{r_h}_+:h \in [m]\}$, 
\begin{align}
\label{eq:mat_up_bound}
    \rmvfunc_h(\gamman_h;\bm D^{1/2}_h\nu_h) \preceq \bm D_h \quad \mbox{and} \quad \rmufunc_h(\gamma_1\gammab_1,\ldots,\gamma_m\gammab_{m};\bm D^{1/2}\mu) \preceq \bm D_h.
\end{align}
Using Theorem 2 in \cite{8437326}, we can show that for all $\gamman_h \in \mathbb S^{r_h}_+$ and $h \in [m]$,
\begin{align}
\rmvfunc_h(\gamman_h;\bm D^{1/2}_h\nu_h) \preceq (\bm D^{-1}_h+\gamman_h)^{-1}.
\end{align}
Next, for all $h \in [m]$ and $\{\gammab_h \in \mathbb S^{r_h}_+:h \in [m]\}$, using techniques similar to the proof of Theorem 2 in \cite{8437326}, we can show the following:
\begin{align}
\label{eq:u_up_bnd}
&\rmufunc_h(\gamma_1\gammab_1,\ldots,\gamma_m\gammab_{m};\bm D^{1/2}\mu)\\
&\preceq (\rmufunc_h(\gamma_1\gammab_{1},\ldots,\gamma_{h-1}\gammab_{h-1},\bm 0,\gamma_{h+1}\gammab_{h+1},\ldots,\gamma_m\gammab_{m};\bm D^{1/2}\mu)^{-1}+\gamma_h\gammab_h)^{-1}\\
&\preceq (\bm D^{-1}_h+\gamma_h\gammab_h)^{-1},
\end{align}
where the last relation follows from \eqref{eq:mat_up_bound}.

Next, let us recall that $\hsrvor_{0,h}=\hsrvorp_{0,h}$ where $\hsrvorp_{0,h}$ is defined in \eqref{eq:initializers} and construct $\gamman_{0,h}$ following the prescription of \eqref{eq:def_gamma}. From \eqref{eq:state_evol_mod} and Proposition 5.6 (a) of \cite{eb_pca}, we get for all $h \in [m]$
\begin{align}
    \gammab_{1,h}&=\bm D_h-\rmvfunc_h(\gamman_{0,h};\bm D^{1/2}_h\nu_h)\\
    & \succeq \bm D_h-(\bm D^{-1}_h+\gamman_{0,h})^{-1} = \mathsf{diag}\left(\frac{\gamma_h(\bm D^4_{h})_{ii}-1}{\gamma_h(\bm D_{h})_{ii}(1+(\bm D^2_{h})_{ii}}\right).
\end{align}
Therefore, again using \eqref{eq:state_evol_mod} and \eqref{eq:u_up_bnd}, we can conclude that
\begin{align}
    \gamman_{1,h}&=\bm D_h-\rmufunc_h(\gamma_1\gammab_{1,1},\ldots,\gamma_m\gammab_{1,m};\bm D^{1/2}\mu)\\
    & \succeq \bm D_h-\left(\bm D^{-1}_h+\gammab_{1,h}\right)^{-1}\\
    & \succeq \bm D_h-\left(\bm D^{-1}_h+\mathsf{diag}\left(\frac{\gamma_h(\bm D^4_{h})_{ii}-1}{(\bm D_{h})_{ii}(1+(\bm D^2_{h})_{ii})}\right)\right)^{-1} = \gamman_{0,h},
\end{align}
for all $h \in [m]$. Again, using the definitions of $\gammab_{1,h}$ and $\gammab_{0,h}$, we can conclude using \eqref{eq:in_v} that $\gammab_{1,h} \succeq \gammab_{0,h}$, for all $h \in [m]$. Consequently, using \eqref{eq:in_v} and \eqref{eq:in_u}, we can inductively conclude that for all $t \in \mathbb N$ and $h \in [m]$, $\gamman_{t+1,h} \succeq \gamman_{t,h}$ and $\gammab_{t+1,h} \succeq \gammab_{t,h}$. Further, using \eqref{eq:in_v} and \eqref{eq:in_u}, it can also be concluded that for all $t\in \mathbb N$ and $h \in [m]$, $\gamman_{t,h} \preceq \bm D_h$ and $\gammab_{t,h} \preceq \bm D_h$. Therefore, for all $h \in [m]$, there exists matrices $\gamman_{\infty,h}$ and $\gammab_{\infty,h}$, satisfying \eqref{eq:fixed_point}, such that the sequences $\{\gamman_{t,h}\}$ and $\{\gammab_{t,h}\}$ converges to $\gamman_{\infty,h}$ and $\gammab_{\infty,h}$, respectively, as $t \rightarrow \infty$.

\subsection{Proof of Theorem \ref{thm:improvement_in_error_2}}
By simplification, we get:
\begin{align}
\label{eq:expand_square}
    &\frac{1}{Np_h}\|\bm U_h\bm D_h\bm V^\top_h-\wb{\bm U}_{t,h}\bm D_h(\wb{\bm V}_{t,h})^\top\|^2_F\\
    & = \frac{1}{Np_h}\mathrm{Tr}(\bm V^\top_h\bm V_h\bm D_h\bm U^\top_h\bm U_h\bm D_h)+\frac{1}{Np_h}\mathrm{Tr}(({\bm V}_h)^\top\wb{\bm V}_{t,h} \bm D_h(\wb{\bm U}_{t,h})^\top\wb{\bm U}_{t,h}\bm D_h)\\
    &\hskip 4em -\frac{2}{Np_h}\mathrm{Tr}(\bm V^\top_h\wb{\bm V}_{t,h} \bm D_h(\wb{\bm U}_{t,h})^\top\bm U_h\bm D_h).
\end{align}
Using the Strong Law of Large Numbers, we have:
\[
\frac{\bm U^\top_h\bm U_h}{N} \xrightarrow{a.s.}\mathbb E_\mu[U_hU^\top_h] \quad \mbox{and} \quad \frac{\bm V^\top_h\bm V_h}{p_h} \xrightarrow{a.s.} \mathbb E_{\nu_h}[V_hV^\top_h], \quad \mbox{for all $h \in [m]$.}
\]
Consequently, we have:
\[
\frac{1}{Np_h}\mathrm{Tr}(\bm V^\top_h\bm V_h \bm D_h\bm U^\top_h\bm U_h\bm D_h) \xrightarrow{a.s.} \mathrm{Tr}\left[\mathbb E_\mu[U_hU^\top_h]\bm D_h\mathbb E_{\nu_h}[V_hV^\top_h]\bm D_h\right].
\]
Using the assumption that $\mathbb E_\mu[U_hU^\top_h]=\mathbb E_{\nu_h}[V_hV^\top_h]=\bm I_{r_h}$ for all $h \in [m]$, we have:
\[
\frac{1}{Np_h}\mathrm{Tr}(\bm V^\top_h\bm V_h \bm D_h\bm U^\top_h\bm U_h\bm D_h) \xrightarrow{a.s.} \mathrm{Tr}(\bm D^2_h).
\]
Next, using Theorem \ref{thm:em_bayes_state_evol} and \eqref{eq:state_evol_gen}, and the Lipschitz property of the denoisers, we get
\begin{align}
\label{eq:part_state_evol}
   \frac{(\wb{\bm U}_{t,h})^\top\wb{\bm U}_{t,h}}{N} \xrightarrow{a.s.}&\mathbb{E}\,[\mathbb{E}_{\mu}[U_h|Y^\orc_{t,1},\ldots,Y^\orc_{t,m},\wt{Y}_{0,1},\ldots,\wt{Y}_{0,\wt m}]^{\otimes 2}]=\hsrvor_{t+1,h}, \\
   \frac{(\wb{\bm V}_{t,h})^\top\wb{\bm V}_{t,h}}{p_h} \xrightarrow{a.s.}& \mathbb{E}\,[\mathbb{E}_{\nu_h}[V_h|Y^{R,\orc}_{t,h}]^{\otimes 2}]=\frac{\hslvor_{t,h}}{\gamma_h}.
\end{align}
Finally, using Theorem \ref{thm:em_bayes_state_evol} and $\mathbb E[Yg(X)|X]=g(X)\mathbb E[Y|X]$, we get 
\begin{align}
   \frac{(\wb{\bm U}_{t,h})^\top\bm U_h}{N} \xrightarrow{a.s.}\hsrvor_{t+1,h}, \quad \mbox{and} \quad \frac{(\wb{\bm V}_{t,h})^\top\bm V_{t,h}}{p_h} \xrightarrow{a.s.} \frac{\hslvor_{t,h}}{\gamma_h}.
\end{align}
Plugging in \eqref{eq:state_evol_gen} in \eqref{eq:def_gamma}, we get:
\[
\gamman_{t,h}:= \bm D^{1/2}_h\hsrvor_{t,h}\bm D^{1/2}_h, \quad \mbox{and} \quad
    \gammab_{t,h}:= \frac{1}{\gamma_h}\bm D^{1/2}_h\hslvor_{t,h}\bm D^{1/2}_h.
\]
Next, using Lemma \ref{lem:redef_state_evol}, Theorem \ref{thm:improvement_in_error}, and the fact that there exists unique fixed points to \eqref{eq:fixed_point}, we get positive definite matrices $\gammab_{\infty,h}$ and $\gamman_{\infty,h}$ such that:
\begin{align}
   \lim\limits_{t \rightarrow \infty}\hsrvor_{t,h} = \bm D^{-1/2}_h\gamman_{\infty,h}\bm D^{-1/2}_h, \quad \mbox{and} \quad \lim\limits_{t \rightarrow \infty}\hslvor_{t+1,h} = \gamma_h\bm D^{-1/2}_h\gammab_{\infty,h}\bm D^{-1/2}_h.
\end{align}
Substituting the above relation in \eqref{eq:part_state_evol} and then in \eqref{eq:expand_square}, we get the desired result.


\section{Proof of theoretical results for prediction sets} 

\subsection{Proof of Theorem \ref{thm:point_pred}}
Observe that as $\rcheck_{h_k}=\frac{1}{N}(\bv_{T,h_k})_{\calf_{h_k}*}\hd_{h_k}$, 
Theorem \ref{thm:em_bayes_state_evol} implies that for all $j,\ell \in [r_{h_{k}}]\times[r_{h_{k}}]$ and pseudo-Lipschitz functions $\varphi_{h_k,j,\ell}(x,v):\R^{2r_{h_k}} \rightarrow \R$ defined as $\varphi_{h_k,j,\ell}(x,v)=x_jx_\ell$,
\begin{align*}
&\lim_{N \rightarrow \infty}N\left(\left((\bv_{T,h_k})_{\calf_{h_k}*}\right)^\top_{*j}\left((\bv_{T,h_k})_{\calf_{h_k}*}\right)_{*\ell}\right) \\
& = \lim_{N \rightarrow \infty}\frac{p_{h_k}}{N}\times\frac{|\calf_{h_k}|}{p_{h_k}}\times\frac{1}{|\calf_{h_k}|}\sum_{i \in \calf_{h_k}}\varphi_{h_k,j,\ell}((\bv_{T,h_k})_{i*},(\bm V_{h_k})_{i*})\\
& \overset{a.s}{=}\, 	\lambda_{h_k}\gamma_{h_k}\mathbb{E}\left[\big[\mathbb{E}_{\nu_{h_k}}[V_{h_k}|Y^{R,\orc}_{T,h_k}]^{\otimes 2} \big]_{j\ell}\right],
\end{align*}
where $Y^{R,\orc}_{T,h_k}$ is defined in \eqref{eq:y_orc_2} and the outer expectation is with respect to the distribution of $Y^{R,\orc}_{T,h_k}$. Combining with, $\wh{\bm D}_{h_k} \overset{a.s.}{\rightarrow} \bm D_{h_k}$, we get  
\begin{align}
\label{eq:ci_pt_est}
    \lim_{N \rightarrow \infty} N(\rcheck^\top_{h_k}\rcheck_{h_k}) \,\overset{a.s}{=}\, 
	\lambda_{h_k}\gamma_{h_k}\bm D_{h_k}\mathbb{E}\big[\mathbb{E}_{\nu_{h_k}}[V_{h_k}|Y^{R,\orc}_{T,h_k}]^{\otimes 2} \big]\bm D_{h_k}.
\end{align}
Plugging in the definition of $\hslvor_{T,h_k}$ from \eqref{eq:state_evol_gen} in \eqref{eq:ci_pt_est}, we get
\begin{equation}
\label{eq:limit-unoise-cov1}	
\lim_{N \rightarrow \infty}
N(\rcheck^\top_{h_k}\rcheck_{h_k})
\, \overset{a.s}{=}\,
\lambda_{h_k}\bm D_{h_k}\hslvor_{T,h_k}\bm D_{h_k}.
\end{equation}
Next, for any $k \in [d]$, recall that 
\[
 \frac{1}{\sqrt{N}}\bq_{h_k}=\frac{1}{N}(\bm V_{h_k})_{\calf_{h_k}*}\bm D_{h_k}U^{\mathrm Q}_{h_k} + \frac{1}{\sqrt{N}}W^{\mathrm Q}_{h_k}
\,\in\, \mathbb{R}^{|\calf_{h_k}|}.
\]
Hence, for any $k \in [d]$, we have 
\[
\frac{1}{\sqrt{N}}(\rcheck^\top_{h_k}\rcheck_{h_k})^{-1}\rcheck^\top_{h_k}Q_{h_k} \,=\, \usig_{h_k} + \unoise_{h_k},
\]
where, conditional on $\{U^{\rmq}_{h_1},\ldots,U^{\rmq}_{h_d},\bar{\bm V}_{T,1},\ldots,\bar{\bm V}_{T,m},\wh{\bm D}_1,\ldots,\wh{\bm D}_{m},\revsag{\bm V_1,\ldots,\bm V_m}\}$, 
\newline $\usig_{h_k}$ and $\unoise_{h_k}$ are independent and can be expressed as 
\begin{align*} 
\usig_{h_k} & =\frac{1}{N}(\rcheck^\top_{h_k}\rcheck_{h_k})^{-1}\rcheck^\top_{h_k}(\bm V_{h_k})^\top_{\calf_{h_k*}}\bm D_{h_k}U^{\mathrm{Q}}_{h_k},	\\
\unoise_{h_k} & = 
\frac{1}{\sqrt{N}}(\rcheck^\top_{h_k}\rcheck_{h_k})^{-1}\rcheck^\top_{h_k}W^{\mathrm Q}_{h_k}.
\end{align*}
As $W^{\mathrm Q}_{h_k}\overset{ind}{\sim}N_{|\calf_{h_k}|}(0,\bm I_{|\calf_{h_k}|})$ for all $k \in [d]$,  we have that, conditional on \[\{U^{\rmq}_{h_1},\ldots,U^{\rmq}_{h_d},\bm V_{1},\ldots,\bm V_{m},
\bar{\bm V}_{T,1},\ldots,\newline \bar{\bm V}_{T,m},\wh{\bm D}_1,\ldots,\wh{\bm D}_{m}\},\] we have
\[
\unoise_{h_k} \stackrel{ind}{\sim} N_{r_{h_k}}(0,\varnoise_{h_k}),\quad k\in [d],
\]
where for each $k$,
\[
\varnoise_{h_k} = \frac{(\rcheck^\top_{h_k}\rcheck_{h_k})^{-1}}{N}\wh{\bm D}_{h_k}\frac{(\bv_{T,h_k})^\top_{\calf_{h_k*}}(\bv_{T,h_k})_{\calf_{h_k*}}}{N}\wh{\bm D}_{h_k}\frac{(\rcheck^\top_{h_k}\rcheck_{h_k})^{-1}}{N}.
\]
In view of \eqref{eq:limit-unoise-cov1},
\begin{equation}
	\label{eq:limit-unoise-cov}
\lim_{N\to\infty} \varnoise_{h_k} =  \lambda_{h_k}^{-1}
\bm D_{h_k}^{-1}(\hslvor_{T,h_k})^{-1}\bm D_{h_k}^{-1}.
\end{equation}

Using Theorem \ref{thm:em_bayes_state_evol} again, 
we obtain that
\begin{align}
    \label{eq:weak_conv_signal}
\lim_{N \rightarrow \infty}\usig_{h_k}
\,\overset{a.s.}{=}\,
U^{\mathrm Q}_{h_k}.
\end{align}
Combining \revzm{\eqref{eq:limit-unoise-cov}}
and \eqref{eq:weak_conv_signal}, we get the stated assertion.


\subsection{Proof of Lemma \ref{lem:glivenko_cantelli}}
Since $\wh \mu = \wh{\mu}_N \xrightarrow{w} \mu$, using arguments similar to Lemma \ref{lem:consistency_of_NPMLE}, we obtain that 
the sequence of probability measures $\{\wh \mu: N \in \mathbb N\}$ is uniformly tight.
Hence, for all $\varepsilon>0$, there exists a constant $M_{\mu,\varepsilon}>0$, such that for all $N \in \mathbb N$, 
\[
\wh \mu(B_{r+\wt r}(0,M_{\mu,\varepsilon})) \ge 1-\frac{\varepsilon}{8} \quad \mbox{and} \quad \mu(B_{r+\wt r}(0,M_{\mu,\varepsilon})) \ge 1-\frac{\varepsilon}{8}.
\]
This implies that, for all $N \in \mathbb N$ and $C \in \mathscr C_{r+\wt r}$,
\begin{equation}
\label{eq:tight}
\begin{aligned}
\wh \mu(C)-\mu(C) &= \wh \mu(B_{r+\wt r}(0,M_{\mu,\varepsilon}) \cap C)-\mu(B_{r+\wt r}(0,M_{\mu,\varepsilon}) \cap C)\\
&~~~ +\wh \mu(B^c_{r+\wt r}(0,M_{\mu,\varepsilon}) \cap C)-\mu(B^c_{r+\wt r}(0,M_{\mu,\varepsilon}) \cap C)\\
& \le \wh \mu(B_{r+\wt r}(0,M_{\mu,\varepsilon}) \cap C)-\mu(B_{r+\wt r}(0,M_{\mu,\varepsilon}) \cap C) + \frac{\varepsilon}{4}.
\end{aligned}
\end{equation}
Thus, it suffices to work with the following collection of sets
\[
\mathscr C_{r+\wt r,M_{\mu,\varepsilon}} = \{C \cap B_{r+\wt r}(0,M_{\mu,\varepsilon}): C \in \mathscr C_{r+\wt r}\}. 
\]
Note that this is a collection of compact and convex subsets of a bounded set in $\R^{r+\wt r}$. 
Since $\mu$ is absolutely continuous with respect to the Lebesgue measure with bounded density, Corollary 2.7.9 of \cite{vanderVaart1996} implies that the collection $\mathscr C_{r+\wt r,M_{\mu,\varepsilon}}$ has a finite bracketing number with respect to $\mu$. 
In other words, for any $\varepsilon'>0$, there exists a natural number, $N_{[\,]}(\mathscr C_{r+\wt r,M_{\mu,\varepsilon}},\varepsilon',\mu)<\infty$ and a collection of pairs of sets, $\mathscr C_{r+\wt r,M_{\mu,\varepsilon},\varepsilon',[\,]}=\left\{(C_{L,1},C_{U,1}),\ldots,(C_{L,N_{[\,]}(\mathscr C_{r+\wt r,M_{\mu,\varepsilon}},\varepsilon',\mu)},C_{U,N_{[\,]}(\mathscr C_{r+\wt r,M_{\mu,\varepsilon}},\varepsilon',\mu)})\right\} \subseteq \mathscr C_{r+\wt r,M_{\mu,\varepsilon}} \times \mathscr C_{r+\wt r,M_{\mu,\varepsilon}}$, such that for any $C \in \mathscr C_{r+\wt r,M_{\mu,\varepsilon}}$, we have pairs of sets $(C_{L,k},C_{U,k}) \in \mathscr C_{r+\wt r,M_{\mu,\varepsilon},\varepsilon',[\,]}$ satisfying
\[
C_{L,k} \subseteq C \subseteq C_{U,k}, \quad \mbox{and} \quad \mu(C_{U,k}\setminus C_{L,k})<\frac{\varepsilon'}{4}.
\]
Consequently, for a set in $C \in \mathscr C_{r+\wt r,M_{\mu,\varepsilon}}$, if its $\varepsilon'$-bracket is given by $(C_{L,k},C_{U,k})$, then we have
\begin{align}
    \wh \mu(C)-\mu(C) &\le \wh \mu(C_{U,k})-\mu(C_{U,k})+\mu(C_{U,k})-\mu(C)\\
    &\le \wh \mu(C_{U,k})-\mu(C_{U,k}) + \frac{\varepsilon'}{4}\\
    &\le \sup\left\{\wh \mu(D_U)-\mu(D_U):(D_L,D_U) \in \mathscr C_{r+\wt r,M_{\mu,\varepsilon},\varepsilon',[\,]}\right\} + \frac{\varepsilon'}{4}.
\end{align}
This implies,
\[ 
\sup_{C \in \mathscr C_{r+\wt r,M_{\mu,\varepsilon}}}\left\{\wh \mu(C)-\mu(C)\right\} \le \sup\left\{\wh \mu(D_U)-\mu(D_U):(D_L,D_U) \in \mathscr C_{r+\wt r,M_{\mu,\varepsilon},\varepsilon',[\,]}\right\} + \frac{\varepsilon'}{4}.
\]
Similarly, we can show that
\[
\inf_{C \in \mathscr C_{r+\wt r,M_{\mu,\varepsilon}}}\left\{\wh \mu(C)-\mu(C)\right\} \ge \inf\left\{\wh \mu(D_L)-\mu(D_L):(D_L,D_U) \in \mathscr C_{r+\wt r,M_{\mu,\varepsilon},\varepsilon',[\,]}\right\} - \frac{\varepsilon'}{4}.
\]
Combining the last two displays and setting $\varepsilon'= \varepsilon$,
we obtain
\begin{align}
\sup_{C \in \mathscr C_{r+\wt r,M_{\mu,\varepsilon}}}\left|\wh \mu(C)-\mu(C)\right| &\le \max\bigg\{\sup_{(D_L,D_U) \in \mathscr C_{r+\wt r,M_{\mu,\varepsilon},\varepsilon,[\,]}}\left\{\wh \mu(D_U)-\mu(D_U)\right\},\\
& \hskip 5em -\inf_{(D_L,D_U) \in \mathscr C_{r+\wt r,M_{\mu,\varepsilon},\varepsilon,[\,]}}\left\{\wh \mu(D_L)-\mu(D_L)\right\}\bigg\}+\frac{\varepsilon}{2}.
\end{align}
Since, $\mu$ is absolutely continuous with respect to Lebesgue measure and supported on $\R^{r+\wt r}$, it assigns zero mass to the boundaries of compact sets. Consequently, as $|\mathscr C_{r+\wt r,M_{\mu,\varepsilon},\varepsilon,[\,]}|<\infty$ and $\wh \mu \xrightarrow{w} \mu$, 
we have
\[
\limsup_{N \rightarrow \infty} \sup_{C \in \mathscr C_{r+\wt r,M_{\mu,\varepsilon}}}\left|\wh \mu(C)-\mu(C)\right| \le \frac{3\varepsilon}{4}, \quad \mbox{almost surely.}
\]
Substituting the above inequality in \eqref{eq:tight}, we almost surely get
\[
\limsup_{N \rightarrow \infty} \sup_{C \in \mathscr C_{r+\wt r}}\left|\wh \mu(C)-\mu(C)\right| \le \varepsilon.
\]
Since, $\varepsilon>0$ is arbitrary, the conclusion of the lemma follows.

\subsection{Proof of Theorem \ref{thm:pred_query_sets_full}}
Let us observe that by definition, for $y^{\rmq} \in \mathbb R^{r^{\rmq}}$ and \newline $\wt y^{\rmq} \in \mathbb R^{\wt r^\rmq}$ satisfying \eqref{eq:observed_predictors}, we have
\begin{align}
\label{eq:end_points_pred_int}
    \mathbb P_{\wh \mu}\left[ (U^\rmq_1,\ldots,U^\rmq_m,\wt{U}^\rmq_1,\ldots,\wt{U}^\rmq_{\wt m}) \in \calc_{\alpha}\,\Big|\,Y^{\rmq}=y^{\rmq},\wt Y^{\rmq}=\wt y^{\rmq}\right] \ge 1-\alpha,
\end{align}
where $Y^{\rmq},\wt Y^{\rmq}$ are defined in \eqref{eq:pred_y_q} with prior $\wh{\mu}$.
Let us define the function:
\[
f(x;B,\pi):=\mathbb P_{\pi}\left[ (U^\rmq_1,\ldots,U^\rmq_m,\wt{U}^\rmq_1,\ldots,\wt{U}^\rmq_{\wt m}) \in B\,\Big|\,(Y^{\rmq},\wt Y^{\rmq})=x\right].
\]
In the above definition, $x \in \R^{r^\rmq+\wt r^\rmq}$, $B \subset \R^{r+\wt r}$ and $(Y^{\rmq},\wt Y^{\rmq})$ satisfy  \eqref{eq:pred_y_q} with 
\[
(U^\rmq_{h_1},\ldots,U^\rmq_{h_d},\wt U^\rmq_{\ell_1},\ldots,\wt U^\rmq_{\ell_{\wt d}}) \sim \pi_{ h_1,\ldots,h_d;\ell_1,\ldots,\ell_{\wt d}}.
\]
Here, $\pi_{ h_1,\ldots,h_d;\ell_1,\ldots,\ell_{\wt d}}$ is the marginal distribution of the modalities $h_1,\ldots,h_d$ and $\ell_1,\ldots,\ell_{\wt d}$ when the entire set of latent factors $(U^\rmq_1,\ldots,U^\rmq_m,\wt{U}^\rmq_1,\ldots,\wt{U}^\rmq_{\wt m}) \sim \pi$. 
Now, Theorem \ref{thm:point_pred}, and Skorokhod's representation theorem implies that 
$y^{\rmq}=s^\rmq+e^\rmq$, where $s^\rmq$ is a sample from the distribution of $Y^\rmq$ and $e^\rmq$ is a realization of a random variable $E^\rmq = (E^\rmq_{h_1},\ldots,E^\rmq_{h_d}) \in \R^{r^\rmq}$ defined as follows:

\[
E^\rmq_{h_k} = \left(\frac{1}{N}(\rcheck^\top_{h_k}\rcheck_{h_k})^{-1}\rcheck^\top_{h_k}(\bm V_{h_k})^\top_{\calf_{h_k*}}\bm D_{h_k}-\bm I_{r_{h_k}}\right)U^\rmq_{h_k}+\bm T_{h_k}Z_{h_k}.
\]
In the above definition, the random matrix $\bm T_{h_k} \xrightarrow{a.s} 0$ and $Z_{h_k} \sim N_{r_{h_k}}(0,\bm I_{r_{h_k}})$. 
Conditional on the matrices $\{\bar{\bm V}_{T,1},\ldots,\bar{\bm V}_{T,m},\wh{\bm D}_1,\ldots,\wh{\bm D}_{m},\bm V_1,\ldots,\bm V_m\}$,
\begin{align}
\label{eq:def_ermq}
E^\rmq_{h_k} \sim N\left(\left(\frac{1}{N}(\rcheck^\top_{h_k}\rcheck_{h_k})^{-1}\rcheck^\top_{h_k}(\bm V_{h_k})^\top_{\calf_{h_k*}}\bm D_{h_k}-\bm I_{r_{h_k}}\right)U^\rmq_{h_k},\bm T_{h_k} \revzms{{\bm T}_{h_k}^\top} \right),
\end{align}
and all quantities defining mean and variance in the last display, except for $U^\rmq_{h_k}$, are deterministic after conditioning.
However, from Lemma \ref{lem:consistency_nuisance} and Theorem \ref{thm:em_bayes_state_evol}, we can conclude that 
\[
\frac{1}{N}(\rcheck^\top_{h_k}\rcheck_{h_k})^{-1}\rcheck^\top_{h_k}(\bm V_{h_k})^\top_{\calf_{h_k*}}\bm D_{h_k}-\bm I_{r_{h_k}}
\xrightarrow{a.s} 0.
\]
Henceforth, we shall condition on the the reference data \[\{\bar{\bm V}_{T,1},\ldots,\bar{\bm V}_{T,m},\wh{\bm D}_1,\ldots,\wh{\bm D}_{m},\bm V_1,\ldots,\bm V_m\}.\] 
Since, almost surely $\wh\mu \stackrel{w}{\to} \mu$, for a given $\varepsilon>0$, we can get a constant $K>0$ (independent of $N$) such that
\begin{align}
\label{eq:tight_e}
\mathbb P_{\wh \mu}[\|(U^\rmq_{h_1},\ldots,U^\rmq_{h_d})\|>K]
< \frac{\varepsilon}{16} \quad \mbox{a.s.}
\quad \mbox{and} \quad 
\mathbb P_{\mu}[\|(U^\rmq_{h_1},\ldots,U^\rmq_{h_d})\|>K]<\frac{\varepsilon}{16}.
\end{align}
Therefore, by \eqref{eq:def_ermq}, \eqref{eq:tight_e} and the same $\varepsilon>0$ chosen in \eqref{eq:tight_e}, we can get a constant $\wb M>0$ (independent of $N$) such that 
\begin{align}
\label{eq:bound_error_pred_set}
\mathbb P_{\wh \mu}[\|E^\rmq\| > \wb M]< \frac{\varepsilon}{16} 
\quad \mbox{a.s.}
\quad \mbox{and} \quad 
\mathbb P_{\mu}[\|E^\rmq\| > \wb M]< \frac{\varepsilon}{16}.
\end{align}
Henceforth, we shall use \[\mathbb P_{\wh \mu}\left[ (U^\rmq_1,\ldots,U^\rmq_m,\wt{U}^\rmq_1,\ldots,\wt{U}^\rmq_{\wt m}) \in \calc_{\alpha}\,\Big|\,Y^{\rmq}+E^\rmq,\wt Y^{\rmq}\right]\] as a (slightly abused) alternative notation for $f((Y^{\rmq}+E^\rmq,\wt Y^{\rmq});\calc_{\alpha},\wh \mu)$ and \[\mathbb P_{\mu}\left[ (U^\rmq_1,\ldots,U^\rmq_m,\wt{U}^\rmq_1,\ldots,\wt{U}^\rmq_{\wt m}) \in \calc_{\alpha}\,\Big|\,Y^{\rmq}+E^\rmq,\wt Y^{\rmq}\right]\]
for $f((Y^{\rmq}+E^\rmq,\wt Y^{\rmq});\calc_{\alpha},\mu)$.
From \eqref{eq:end_points_pred_int}, we have
\[
\mathbb E_{\wh \mu}\left[\mathbb P_{\wh \mu}\left[ (U^\rmq_1,\ldots,U^\rmq_m,\wt{U}^\rmq_1,\ldots,\wt{U}^\rmq_{\wt m}) \in \calc_{\alpha}\,\Big|\,Y^{\rmq}+E^\rmq,\wt Y^{\rmq}\right]\right] \ge 1-\alpha.
\]
We shall show that, as $N \rightarrow \infty$,
\begin{equation}
\label{eq:plug_in_bias}
\begin{aligned}
&\bigg|\mathbb E_{\wh \mu}\left[\mathbb P_{\wh \mu}\left[ (U^\rmq_1,\ldots,U^\rmq_m,\wt{U}^\rmq_1,\ldots,\wt{U}^\rmq_{\wt m}) \in \calc_{\alpha}\,\Big|\,Y^{\rmq}+E^\rmq,\wt Y^{\rmq}\right]\right]\\
&\hskip 10em -\mathbb E_{\wh \mu}\left[\mathbb P_{\wh \mu}\left[ (U^\rmq_1,\ldots,U^\rmq_m,\wt{U}^\rmq_1,\ldots,\wt{U}^\rmq_{\wt m}) \in \calc_{\alpha}\,\Big|\,Y^{\rmq},\wt Y^{\rmq}\right]\right]\bigg| \xrightarrow{a.s} 0.
\end{aligned}    
\end{equation}
This implies
\[
\limsup_{n \rightarrow \infty}\mathbb P_{\wh \mu}\left[(U^\rmq_1,\ldots,U^\rmq_m,\wt{U}^\rmq_1,\ldots,\wt{U}^\rmq_{\wt m}) \in \calc_{\alpha}\right] \ge 1-\alpha, \quad \mbox{a.s.}
\]

Now since $\wh \mu \xrightarrow{w} \mu$, the sequences of random variables $\{(Y^\rmq,\wt Y^\rmq)\}$ and $\{(Y^\rmq+E^\rmq,\wt Y^\rmq)\}$ are tight. This implies, for the $\varepsilon>0$ chosen in \eqref{eq:tight_e}, there exist compact and convex sets $M \in \R^{r^\rmq+\wt r^\rmq}$ and $\wt M \in \R^{r+\wt r}$ such that
\begin{equation}
\label{eq:tightness}    
\begin{aligned}
    &\mathbb P_{\wh \mu}[(Y^\rmq,\wt Y^\rmq) \notin M] <\frac{\varepsilon}{16}, \quad \mathbb P_{\wh \mu}[(Y^\rmq+E^\rmq,\wt Y^\rmq) \notin M] <\frac{\varepsilon}{16},\\ 
    &\mathbb P_{\mu}[(Y^\rmq,\wt Y^\rmq) \notin M] <\frac{\varepsilon}{16},\quad
    \mathbb P_{\mu}[(Y^\rmq+E^\rmq,\wt Y^\rmq) \notin M] <\frac{\varepsilon}{16},\\ 
  & \wh \mu(\wt M^c)<\frac{\varepsilon}{16}, \quad \mbox{and} \quad \mu(\wt M^c)<\frac{\varepsilon}{16}.
\end{aligned}
\end{equation}
In the above definition, $\mathbb P_{\mu}[(Y^\rmq,\wt Y^\rmq) \in \cdot]$ refers to the probability of $(Y^\rmq,\wt Y^\rmq)$ when in the definition \eqref{eq:pred_y_q}, the latent factors $(U^\rmq_1,\ldots,U^\rmq_m,\wt{U}^\rmq_1,\ldots,\wt{U}^\rmq_{\wt m}) \sim \mu$. 
By properties of multivariate Gaussian distributions, 
for any compact set $M \in \R^{r^\rmq+\wt r^\rmq}$ and $t \in \R^{r^\rmq}$ satisfying $\|t\| \le \wb M$ (with $\wb M$ defined in \eqref{eq:bound_error_pred_set}), the function
\begin{align*}
&f((y+t,\wt y);B,\pi)\\
&\hskip 6em:=\mathbb P_{\pi}\left[ (U^\rmq_1,\ldots,U^\rmq_m,\wt{U}^\rmq_1,\ldots,\wt{U}^\rmq_{\wt m}) \in B\,\Big|\,(Y^{\rmq},\wt Y^{\rmq})=(y+t,\wt y)\right]\mathbbm 1_{[(y,\wt y) \in M]},
\end{align*}
is uniformly Lipschitz and bounded as a function of $(y,\wt y) \in M$. 
By Corollary 2.7.2 of \cite{vanderVaart1996}, the class of functions
\[
\mathcal F_M =\bigg\{f((y+t,{\wt y}); B,\pi)\mathbbm 1_{[(y,\wt y) \in M]}:\pi \in \{\wh \mu,\mu\},\,B \subset \wt M,\, \|t\|\leq \wb{M} \bigg\},
\]
has a finite bracketing number.
Therefore, proceeding as in the proof of Lemma \ref{lem:glivenko_cantelli}, we can show that
\begin{align}
\label{eq:first_plug_in_bias_removal}
&\sup_{\substack{\pi \in \{\wh \mu,\mu\}\\B \subset \wt M, \|t\| \le \wb M}}\bigg|\mathbb E_{\wh \mu}\left[f((Y^\rmq+t,\wt Y^\rmq),B,\pi)\mathbbm 1_{[(Y^\rmq,\wt Y^\rmq) \in M]}\right]\\
&\hskip 10em-\mathbb E_{\mu}\left[f((Y^\rmq+t,\wt Y^\rmq),B,\pi)\mathbbm 1_{[(Y^\rmq,\wt Y^\rmq) \in M]}\right]\bigg|\xrightarrow{a.s}0.
\end{align}
This implies
\begin{equation}
    \label{eq:t_1_zero}
\begin{aligned}
&\bigg|\mathbb E_{\wh \mu}\left[f((Y^\rmq+E^\rmq,\wt Y^\rmq),\calc_\alpha,\wh \mu)\mathbbm 1_{[(Y^\rmq,\wt Y^\rmq) \in M,\, \|E^\rmq\| \le \wb M]}\right]\\
&\hskip 6em -\mathbb E_{\mu}\left[f((Y^\rmq+E^\rmq,\wt Y^\rmq),\calc_\alpha,\wh \mu)\mathbbm 1_{[(Y^\rmq,\wt Y^\rmq) \in M,\, \|E^\rmq\| \le \wb M]}\right]\bigg|\xrightarrow{a.s}0.
\end{aligned}
\end{equation}
Similarly, we can show that
\begin{equation}
\label{eq:t_3_zero}    
\begin{aligned}
&\bigg|\mathbb E_{\wh \mu}\left[f((Y^\rmq,\wt Y^\rmq),\calc_\alpha,\wh \mu)\mathbbm 1_{[(Y^\rmq,\wt Y^\rmq) \in M]}\right]
\\
& \hskip 5em -\mathbb E_{\mu}\left[f((Y^\rmq,\wt Y^\rmq),\calc_\alpha,\wh \mu)\mathbbm 1_{[(Y^\rmq,\wt Y^\rmq) \in M]}\right]\bigg|\xrightarrow{a.s}0.
\end{aligned}
\end{equation}
Now, let us consider the following:
\begin{equation}
    \label{eq:grand_decomp}
\begin{aligned}
&\bigg|\mathbb E_{\wh \mu}\left[\mathbb P_{\wh \mu}\left[ (U^\rmq_1,\ldots,U^\rmq_m,\wt{U}^\rmq_1,\ldots,\wt{U}^\rmq_{\wt m}) \in \calc_{\alpha}\,\Big|\,Y^{\rmq}+E^\rmq,\wt Y^{\rmq}\right]\right]\\
&\hskip 10em -\mathbb E_{\wh \mu}\left[\mathbb P_{\wh \mu}\left[ (U^\rmq_1,\ldots,U^\rmq_m,\wt{U}^\rmq_1,\ldots,\wt{U}^\rmq_{\wt m}) \in \calc_{\alpha}\,\Big|\,Y^{\rmq},\wt Y^{\rmq}\right]\right]\bigg|\\
&= \bigg|\mathbb E_{\wh \mu}\left[\mathbb P_{\wh \mu}\left[ (U^\rmq_1,\ldots,U^\rmq_m,\wt{U}^\rmq_1,\ldots,\wt{U}^\rmq_{\wt m}) \in \calc_{\alpha}\,\Big|\,Y^{\rmq}+E^\rmq,\wt Y^{\rmq}\right]\mathbbm 1_{[(Y^\rmq,\wt Y^\rmq) \notin M]}\right]\\
& \hskip 0.5em+\mathbb E_{\wh \mu}\left[\mathbb P_{\wh \mu}\left[ (U^\rmq_1,\ldots,U^\rmq_m,\wt{U}^\rmq_1,\ldots,\wt{U}^\rmq_{\wt m}) \in \calc_{\alpha}\,\Big|\,Y^{\rmq}+E^\rmq,\wt Y^{\rmq}\right]\mathbbm 1_{[(Y^\rmq,\wt Y^\rmq) \in M]}\mathbbm 1_{[\|E^\rmq\| > \wb{M}]}\right]\\
& \hskip 1em+\mathbb E_{\wh \mu}\left[\mathbb P_{\wh \mu}\left[ (U^\rmq_1,\ldots,U^\rmq_m,\wt{U}^\rmq_1,\ldots,\wt{U}^\rmq_{\wt m}) \in \calc_{\alpha}\,\Big|\,Y^{\rmq}+E^\rmq,\wt Y^{\rmq}\right]\mathbbm 1_{[(Y^\rmq,\wt Y^\rmq) \in M, \|E^\rmq\| \leq \wb{M}]}\right]\\
&\hskip 2em -\mathbb E_{\wh \mu}\left[\mathbb P_{\wh \mu}\left[ (U^\rmq_1,\ldots,U^\rmq_m,\wt{U}^\rmq_1,\ldots,\wt{U}^\rmq_{\wt m}) \in \calc_{\alpha}\,\Big|\,Y^{\rmq},\wt Y^{\rmq}\right]\mathbbm 1_{[(Y^\rmq,\wt Y^\rmq) \in M]}\right]  \\
&\hskip 5em -\mathbb E_{\wh \mu}\left[\mathbb P_{\wh \mu}\left[ (U^\rmq_1,\ldots,U^\rmq_m,\wt{U}^\rmq_1,\ldots,\wt{U}^\rmq_{\wt m}) \in \calc_{\alpha}\,\Big|\,Y^{\rmq},\wt Y^{\rmq}\right]\mathbbm 1_{[(Y^\rmq,\wt Y^\rmq) \notin M]}\right] \bigg|\\
& \hskip 2em\le 2\,\mathbb P_{\wh \mu}[(Y^{\rmq},\wt Y^{\rmq}) \notin M]
+\mathbb P_{\wh \mu}[\|E^\rmq\|>\wb M]
+T_1+T_2+T_3\\
& \hskip 2em\le \frac{\varepsilon}{4}+T_1
+T_2+T_3,
\end{aligned}
\end{equation}
where 
 \begin{align*}
 T_1 & = \bigg|\mathbb E_{\wh \mu}\left[f((Y^\rmq+E^\rmq,\wt Y^\rmq),\calc_\alpha,\wh \mu)\mathbbm 1_{[(Y^\rmq,\wt Y^\rmq) \in M,\, \|E^\rmq\| \le \wb M]}\right]\\
 & \hskip 6em -\mathbb E_{\mu}\left[f((Y^\rmq+E^\rmq,\wt Y^\rmq),\calc_\alpha,\wh \mu)\mathbbm 1_{[(Y^\rmq,\wt Y^\rmq) \in M,\, \|E^\rmq\| \le \wb M]}\right]\bigg|,\\
 T_2 & = \bigg|\mathbb E_{\mu}\left[f((Y^\rmq+E^\rmq,\wt Y^\rmq),\calc_\alpha,\wh \mu)\mathbbm 1_{[(Y^\rmq,\wt Y^\rmq) \in M,\, \|E^\rmq\| \le \wb M]}\right]\\
 &\hskip 10em-\mathbb E_{\mu}\left[f((Y^\rmq,\wt Y^\rmq),\calc_\alpha,\wh \mu)\mathbbm 1_{[(Y^\rmq,\wt Y^\rmq) \in M]}\right]\bigg|,    \\
 T_3 & = \bigg|\mathbb E_{\wh \mu}\left[f((Y^\rmq,\wt Y^\rmq),\calc_\alpha,\wh \mu)\mathbbm 1_{[(Y^\rmq,\wt Y^\rmq) \in M]}\right]-\mathbb E_{\mu}\left[f((Y^\rmq,\wt Y^\rmq),\calc_\alpha,\wh \mu)\mathbbm 1_{[(Y^\rmq,\wt Y^\rmq) \in M]}\right]\bigg|.
 \end{align*}
 By \eqref{eq:t_1_zero}, we can conclude that $T_1 \xrightarrow{a.s}0$ and by \eqref{eq:t_3_zero}, we can conclude that $T_3 \xrightarrow{a.s}0$. Next, by the properties of multivariate Gaussian distributions and arguments similar to those leading to \eqref{eq:t_1_zero}, 
we can also show that, as $N \rightarrow \infty$,
\begin{align*}
&\bigg|f((Y^\rmq+E^\rmq,\wt Y^\rmq),\calc_\alpha,\wh \mu)\mathbbm 1_{[(Y^\rmq,\wt Y^\rmq) \in M,\, \|E^\rmq\| \le \wb M]}\\
&\hskip 10em-f((Y^\rmq+E^\rmq,\wt Y^\rmq),\calc_\alpha,\mu)\mathbbm 1_{[(Y^\rmq,\wt Y^\rmq) \in M,\, \|E^\rmq\| \le \wb M]}\bigg| \xrightarrow{a.s} 0,\\
&\bigg|f((Y^\rmq,\wt Y^\rmq),\calc_\alpha,\wh \mu)\mathbbm 1_{[(Y^\rmq,\wt Y^\rmq) \in M]}-f((Y^\rmq,\wt Y^\rmq),\calc_\alpha,\mu)\mathbbm 1_{[(Y^\rmq,\wt Y^\rmq) \in M]}\bigg| \xrightarrow{a.s} 0.
\end{align*}
Finally, as $\|E^\rmq\| \xrightarrow{a.s}0$, by the Dominated Convergence theorem and Lemma \ref{lem:consistency_nuisance}, we conclude that
\[
\bigg|f((Y^\rmq+E^\rmq,\wt Y^\rmq),\calc_\alpha,\mu)\mathbbm 1_{[(Y^\rmq,\wt Y^\rmq) \in M,\, \|E^\rmq\| \le \wb M]}-f((Y^\rmq,\wt Y^\rmq),\calc_\alpha,\mu)\mathbbm 1_{[(Y^\rmq,\wt Y^\rmq) \in M]}\bigg| \xrightarrow{a.s} 0.
\]
Therefore,
\[
\bigg|f((Y^\rmq+E^\rmq,\wt Y^\rmq),\calc_\alpha,\wh \mu)\mathbbm 1_{[(Y^\rmq,\wt Y^\rmq) \in M,\, \|E^\rmq\| \le \wb M]}-f((Y^\rmq,\wt Y^\rmq),\calc_\alpha,\wh\mu)\mathbbm 1_{[(Y^\rmq,\wt Y^\rmq) \in M]}\bigg| \xrightarrow{a.s} 0.
\]
Furthermore, since the functions in the above equation are uniformly dominated by $1$, by the Dominated Convergence Theorem, we have 
\begin{align}
\label{eq:t_2_zero}
&\bigg|\mathbb E_{\mu}\left[f((Y^\rmq+E^\rmq,\wt Y^\rmq),\calc_\alpha,\wh \mu)\mathbbm 1_{[(Y^\rmq,\wt Y^\rmq) \in M,\, \|E^\rmq\| \le \wb M]}\right]\\
&\hskip 10em-\mathbb E_{\mu}\left[f((Y^\rmq,\wt Y^\rmq),\calc_\alpha,\wh \mu)\mathbbm 1_{[(Y^\rmq,\wt Y^\rmq) \in M]}\right]\bigg| \rightarrow 0,
\end{align}
almost surely, given $\{\bar{\bm V}_{T,1},\ldots,\bar{\bm V}_{T,m},\wh{\bm D}_1,\ldots,\wh{\bm D}_{m},\bm V_1,\ldots,\bm V_m\}$. In other words, $T_2\stackrel{a.s.}{\to} 0$ conditionally.
Since, $\varepsilon$ chosen in \eqref{eq:bound_error_pred_set} is arbitrary, by \eqref{eq:grand_decomp}, we can get \eqref{eq:plug_in_bias} conditioning on 
\[
\{\bar{\bm V}_{T,1},\ldots,\bar{\bm V}_{T,m},\wh{\bm D}_1,\ldots,\wh{\bm D}_{m},\bm V_1,\ldots,\bm V_m\}.\]
As the conclusion holds regardless of the conditioning event, \eqref{eq:plug_in_bias} holds marginally.
Now, as $\calc_{\alpha} \in \mathscr{C}_{r+\wt r}$, Lemma \ref{lem:glivenko_cantelli} 
implies,
\begin{align}
\label{eq:target_result}
&|\mu(\calc_{\alpha})-\wh \mu( \calc_{\alpha})| \xrightarrow{a.s} 0, \quad \mbox{as $N \rightarrow \infty$,}
\end{align}
conditioned on $\{\bar{\bm V}_{T,1},\ldots,\bar{\bm V}_{T,m},\wh{\bm D}_1,\ldots,\wh{\bm D}_{m},\bm V_1,\ldots,\bm V_m\}$. Since the right-hand side of \eqref{eq:target_result} is non-random, \eqref{eq:target_result} also holds unconditionally.
Hence,
we almost surely have
\begin{align*}
&\lim_{N \rightarrow \infty}\mathbb P_{\mu}\left[(U^\rmq_1,\ldots,U^\rmq_m,\wt{U}^\rmq_1,\ldots,\wt{U}^\rmq_{\wt m}) \in \calc_{\alpha}\right]\\
& = \lim_{N \rightarrow \infty} \mu(\calc_{\alpha})
\ge\lim_{N \rightarrow \infty}\wh \mu(\calc_{\alpha})+\lim_{N \rightarrow \infty} \revzm{|}\wh \mu(\calc_{\alpha})-\mu(\calc_{\alpha})\revzm{|}\ge 1-\alpha.
\end{align*}
Again as, $\mathbb P_{\mu}\left[(U^\rmq_1,\ldots,U^\rmq_m,\wt{U}^\rmq_1,\ldots,\wt{U}^\rmq_{\wt m}) \in \calc_{\alpha}\right]$ is a non-random quantity, our result follows.

\section{Choice of prior classes}
\label{sec:npmle}
In this section, we provide examples of two classes of prior distributions that can be chosen as $\mathcal P$ and $\mathcal P_{\nu_h}$.

\subsection{Gaussian mixture models}
\label{gmm}

We can choose $\mathcal P$ in \eqref{eq:emp_bayes_1} and each $\mathcal P_{\nu_h}$ in \eqref{eq:emp_bayes_2} to be the set of distributions on $\mathbb R^{r+\wt r}$ and $\mathbb R^{r_h}$ respectively, which are absolutely continuous with respect to the Lebesgue measure and have bounded densities. 
From \cite[Chapter 3]{GoodBengCour16}, we know that the class of Gaussian mixture distributions with a large number of components can approximate any distribution up to negligible error. 
Inspired by this fact, 
we can set $\mathcal P$ and all $\mathcal P_{\nu_h}$'s as collections of Gaussian mixture distributions with large numbers of components. 
The maximum number of components in each class can be specified according to the size of the data and its complexity.
In other words, 
we can set $\mathcal{P}$ as
\begin{equation}
    \label{eq:prior-class-mu}
    {\mathcal{P}} = \left\{\sum_{k=1}^{K}\pi_k\,N_{r+\wt r}(m_k,\bm \Sigma_k): m_k \in \R^{r+\wt r},\,\bm \Sigma_k \in \mathbb S^{r+\wt r}_+,\,\pi_k \geq 0,\, \sum_{k=1}^{K}\pi_k=1\; \mbox{for $k \in [K]$}\right\},
\end{equation}
and $\mathcal P_{\nu_h}$ for each $h \in [m]$ as
\begin{align}
    \label{eq:prior-class-nu-h}
    {\mathcal{P}}_{\nu_h} &= \left\{\sum_{k=1}^{K_h}\pi_{k,h}\,N_{r_h}(m_{k,h},\bm \Sigma_{k,h}): m_{k,h} \in \R^{r_h},\,\bm \Sigma_{k,h} \in \mathbb S^{r_h}_+,\,\pi_{k,h} \geq 0,\, \right.\\
    &\hskip 5em\left.\sum_{k=1}^{K_h}\pi_{k,h}=1\; \mbox{for $k_h \in [K_h]$}\right\}.
\end{align}

By working with these prior classes,
estimating $\mu$ and any $\nu_h$ reduces to
estimating mixing proportions and component means and covariances in a multivariate Gaussian mixture.

If $r_h$ for some $h \in [m]$ is very large, we can model $\bm \Sigma_k$ and $\bm \Sigma_{k,h}$ by diagonal matrices. 
In such situations, it might be essential to have more components in the Gaussian mixture to account for simplification in the covariance structure. 

\subsection{Nonparametric maximum likelihood estimators}
An alternative approach for estimating the priors $\mu$ and $\{\nu_h:h\in[m]\}$ is to consider their nonparametric maximum likelihood estimators (NPMLEs). 
This approach was advocated in \cite{wang_zhong_fan}. 
It is important to observe that, this estimator can suffer from the ``curse of dimensionality" if $r+\wt r$, or any of the $r_h$, is large. However, if the signal strengths $(\bm D_h)_{ii}$ for $i \in [r_h]$ and $h \in [m]$ are sufficiently high, such ``curse of dimensionality can be mitigated. 
For constructing the NPMLEs, we approximate $\mathcal P$ and all of the $\mathcal P_{\nu_h}$'s using discrete distributions with an appropriately chosen support.
In particular, we can set the support of $\mu$, denoted by $\mathcal S_\mu \subseteq \R^{r+\wt{r}}$, as
\begin{align}
\mathcal{S}_\mu &:= \left\{((\hslmp_{0,1})^{-1}(\upca_{0,1})_{i*},\ldots,(\hslmp_{0,m})^{-1}(\upca_{0,m})_{i*},\hatl^{-1}_1(\ldm_1)_{i*},\ldots,\right.\\
&\hskip 25em \left.\hatl^{-1}_{\wt m}(\ldm_{\wt m})_{i*}): i \in [N]
\right\}.
\end{align}
In other words, each point in the support of $\mu$ is one row in the 
matrix
\begin{equation*}
\begin{bmatrix}
    (\hslmp_{0,1})^{-1}\upca_{0,1} & \ldots & (\hslmp_{0,m})^{-1}\upca_{0,m} & 
    \hatl^{-1}_1 \ldm_1 & \ldots & \hatl^{-1}_{\wt m}\ldm_{\wt m}
\end{bmatrix} \in \mathbb{R}^{n\times (r + \wt{r} )}.
\end{equation*}
Then, we set
\begin{equation}
    \label{eq:prior-class-mu_1}
    {\mathcal{P}} = \{\mbox{all prior $\mu$ supported on $\mathcal{S}_{\mu}$}\}.
\end{equation}
Similarly, for each $h \in [m]$, the support of $\nu_h$, denoted by $\mathcal S_{\nu_h} \subset \R^{r_h}$, can be set at 
\[
\mathcal{S}_{\nu_h}:= \left\{(\hsrmp_{0,h})^{-1}(\vpca_{0,h})_{i*}: i \in [p_h] \right\}.
\]
Thus, each point in the support of $\nu_h$ is a row in $(\hsrm_{0,h})^{-1}{\bm V}_{0,h}$, and we then set
\begin{equation}
    \label{eq:prior-class-nu-h_1}
    {\mathcal{P}}_{\nu_h} = \{\mbox{all prior $\nu_h$ supported on $\mathcal{S}_{\nu_h}$}\}.
\end{equation}
By working with these discrete priors with a fixed support, estimating $\mu$ and each $\nu_h$ reduces to
estimating the mixing proportions in a mixture of degenerate Gaussians where the number of components equals the number of points in support. While these priors do not satisfy Assumption 5.1, yet because of the discrete structure the same theoretical conclusions can also be proven for them.

\subsection{Effect of the numbers of retained principal components on atlas building in TEA-seq data}
\revsn{In order to study the impact of the choice of the number of principal components $r_1$ and $r_2$ used to approximate the signals in the high-dimensional modalities (ATAC and RNA) on the quality of the atlas constructed by Algorithm 1, we considered an experiment where we built the multimodal atlases following the technique outlined in Section 4  while varying the numbers of principal components in the high-dimensional modalities as follows: 10 ATAC PCs and 15 RNA PCs, 15 ATAC PCs and 20 RNA PCs, and 20 ATAC PCs and 20 RNA PCs.
We assessed the quality of the atlases in terms of the four clustering metrics, namely, Adjusted Rand Index, Average Silhouette Score, cLISI, and V-measure. We summarize our findings in Table \ref{tab:clustering_tea_seq_ablation}. The UMAP 
visualization of the atlases constructed in the three settings is shown in Figure \ref{fig:tea_seq_ablation}. 
Both the UMAP visualization and the clustering metrics suggest that the quality of the constructed atlas does not vary drastically with the choice of the number of principal components,
after the leading principal components explaining the maximum variability in the data are included. 
The numbers of  principal components can generally be determined by observing the scree plots of the empirical singular values of the data matrices and comparing the empirical spectrum
of the residual component (after subtracting out the low rank representation) with appropriately rescaled 
Marcenko-Pastur distribution.}
\begin{table}[t]
\centering
\scriptsize
\setlength{\tabcolsep}{3pt}
\renewcommand{\arraystretch}{0.95}
\caption{Comparison of clustering metrics among the three settings with different numbers of RNA and ATAC PC's. Higher average silhouette score, ARI and V-measure indicate better separation, while lower cLISI indicates improved segregation of clusters.}
\begin{tabular}{ccccc}
\toprule
Clustering Metric & (10 ATAC PC, 15 RNA PC) & (15 ATAC PC, 20 RNA PC) & (20 ATAC PC, 20 RNA PC) \\
\midrule
Average silhouette score & {\bf 0.416} & 0.395 & 0.366\\
          V-measure & 0.737 & {\bf 0.762} & 0.732\\
          Adjusted Rand Index & 0.549 & {\bf 0.634} & 0.548\\
          cLISI & {\bf 1.191} & 1.200 & 1.202\\
\toprule
\end{tabular}
\label{tab:clustering_tea_seq_ablation}
\end{table}

\begin{figure}
    \centering
    \includegraphics[width=\linewidth]{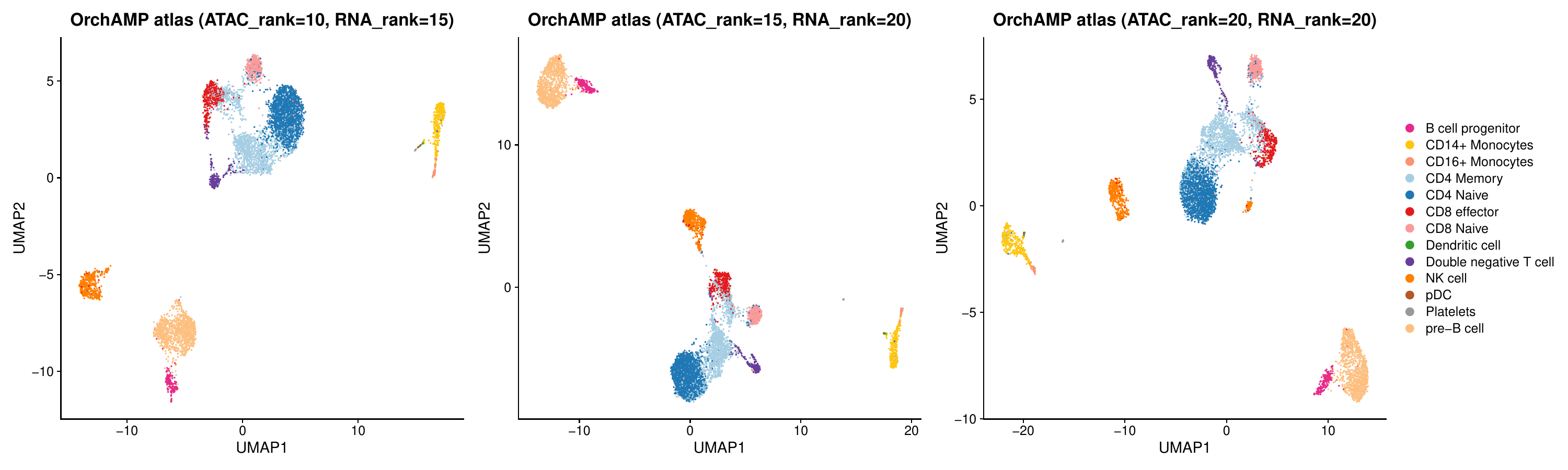}
    \caption{UMAP representation of the embeddings constructed using OrchAMP in the TEA-seq data with different numbers of principal components for the high-dimensional modalities (ATAC and RNA). Left: 10 ATAC PCs, 15 RNA PCs, Middle: 15 ATAC PCs, 20 RNA PCs, and Right: 20 ATAC PCs, 20 RNA PCs}
    \label{fig:tea_seq_ablation}
\end{figure}

\begin{table}[t]
\centering
\scriptsize
\setlength{\tabcolsep}{3pt}
\renewcommand{\arraystretch}{0.95}
\caption{Comparison of clustering metrics in CITE-seq atlas building among the three settings with different numbers of RNA PC's.}
\begin{tabular}{ccccc}
\toprule
Clustering Metric & 45 RNA PC & 50 RNA PC & 55 RNA PC \\
\midrule
Average silhouette score & 0.358 & 0.343 & {\bf 0.372}\\
          V-measure & 0.792 & {\bf 0.796} & 0.787\\
          Adjusted Rand Index & 0.592 & {\bf 0.610} & 0.587\\
          cLISI & 1.195 & 1.186 & {\bf 1.184}\\
\toprule
\end{tabular}
\label{tab:clustering_cite_seq_ablation}
\end{table}

\subsection{Effect of the number of retained principal components on atlas building in CITE-seq data}
\revsn{We conducted a similar experiment to study the impact of the choice of the number of principal components used to model the signal component in the RNA modality of the CITE-seq data described in Section \ref{sec:bench_cite_seq}. 
In that direction, we repeated the same atlas-building procedure described in Section \ref{sec:bench_cite_seq} while varying the number of principal components used to model RNA in $\{45,50,55\}$. 
The quantitative metrics are reported in Table \ref{tab:clustering_cite_seq_ablation} and the UMAP visualization of the three atlases in Figure \ref{fig:cite_seq_ablation}. We see that both the clustering metrics and the UMAP visualizations are comparable across the three choices of the number of retained principal components in the RNA modality, indicating that the relative robustness to the numbers of principal components used for atlas construction is consistent across data types and sequencing technologies.}

\begin{figure}
    \centering
    \includegraphics[width=\linewidth]{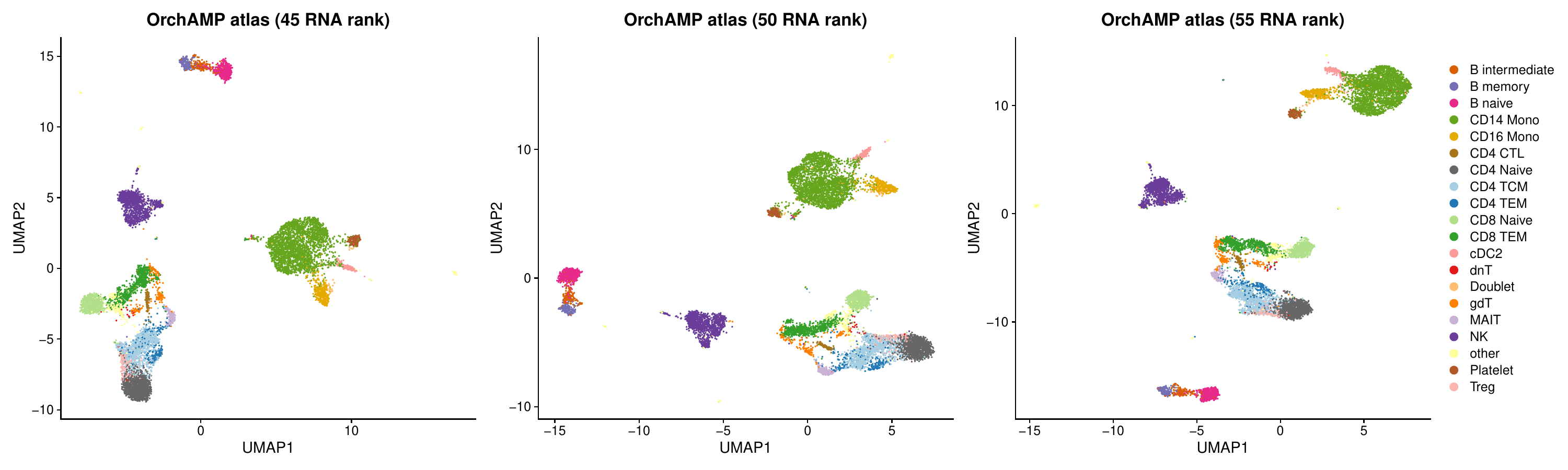}
    \caption{UMAP representation of the embeddings constructed using OrchAMP in the CITE-seq data with different numbers of principal components for the RNA modality. Left: 45 RNA PCs, Middle: 50 RNA PCs, and Right: 55 RNA PCs}
    \label{fig:cite_seq_ablation}
\end{figure}

\end{document}